\documentclass[a4paper]{article}

\usepackage[usenames,dvipsnames]{xcolor}
\usepackage{latexsym}
\usepackage{amsmath} 
\usepackage{amssymb}
\usepackage{gastex}
\usepackage{tikz}
\usetikzlibrary{arrows.meta}
\usepackage{stmaryrd} 
\usepackage[colorlinks=false]{hyperref}

\newcommand{\doi}[1]{\\\href{http://dx.doi.org/#1}{doi:#1}}

\renewcommand{\phi}{\varphi}
\renewcommand{\epsilon}{\varepsilon}
\newcommand{\nothing}{\varnothing}

\newcommand{\cA}{{\cal A}}
\newcommand{\cB}{{\cal B}}
\newcommand{\cC}{{\cal C}}
\newcommand{\cD}{{\cal D}}
\newcommand{\cF}{{\cal F}}
\newcommand{\cM}{{\cal M}}
\newcommand{\cN}{{\cal N}}
\newcommand{\cP}{{\cal P}}
\newcommand{\cR}{{\cal R}}
\newcommand{\cS}{{\cal S}}
\newcommand{\cX}{{\cal X}}
\newcommand{\nat}{{\mathbb N}}

\newcommand{\tup}[1]{\langle #1 \rangle}
\newcommand{\pair}[1]{(#1)}

\newcommand{\abb}[1]{{\sc \lowercase{#1}}}
\newcommand{\rank}{\operatorname{rank}}
\newcommand{\m}{{\mathit mx}}
\newcommand{\rt}{\mathrm{root}}
\newcommand{\childno}{\mathrm{child}}
\newcommand{\nod}[1]{N(#1)} 
\newcommand{\lab}{\mathrm{lab}}
\newcommand{\peb}{\mathrm{peb}}
\newcommand{\nex}{\mathrm{next}}
\newcommand{\first}{\mathrm{first}}
\newcommand{\last}{\mathrm{last}}
\newcommand{\leaf}{\mathrm{leaf}}
\newcommand{\str}{\mathrm{str}}

\newcommand{\suc}{\mathrm{succ}}

\newcommand{\tmark}{\operatorname{mark}}
\newcommand{\sem}[1]{\llbracket #1\rrbracket}
\newcommand{\semf}[1]{\llbracket #1\rrbracket_f}
\newcommand{\semt}[1]{\llbracket #1\rrbracket_t}
\newcommand{\xp}[1]{\operatorname{#1}}
\newcommand{\sit}{\xp{Sit}}
\newcommand{\con}{\xp{Con}}

\newcommand{\vis}{\xp{vis}}
\newcommand{\obs}{\xp{obs}}
\newcommand{\test}{\xp{test}}
\newcommand{\bool}{\xp{\beta}}
\newcommand{\rep}{\xp{rep}}

\newcommand{\stay}{{\rm stay}}
\newcommand{\up}{{\rm up}} 
\newcommand{\down}{{\rm down}}
\newcommand{\drop}{{\rm drop}}
\newcommand{\lift}{{\rm lift}}
\newcommand{\totop}{{\rm to\text{-}top}}
\newcommand{\enc}{{\rm enc}}
\newcommand{\dec}{{\rm dec}}
\newcommand{\at}{@}
\newcommand{\fl}{{\rm flat}}

\newcommand{\p}{{\rm pp}}
\newcommand{\g}{g}

\newcommand{\child}{{\tt child}}
\newcommand{\parent}{{\tt parent}}
\newcommand{\lft}{{\tt left}}
\newcommand{\rght}{{\tt right}}
\newcommand{\dropt}{{\tt drop}}
\newcommand{\liftt}{{\tt lift}}
\newcommand{\downt}{{\tt down}}
\newcommand{\upt}{{\tt up}}

\newcommand{\haslab}{{\tt haslabel}}
\newcommand{\haspeb}{{\tt haspebble}}
\newcommand{\isleaf}{{\tt isleaf}}
\newcommand{\isroot}{{\tt isroot}}
\newcommand{\isfirst}{{\tt isfirst}}

\newcommand{\islast}{{\tt islast}}
\newcommand{\ischild}{{\tt ischild}}
\newcommand{\loops}{{\tt loop}}

\newcommand{\family}[1]{\mbox{\sf #1}}
\newcommand{\REGT}{\family{REGT}}

\newcommand{\VIPTA}[1]{\family{V$_{#1}$I-PTA}}
\newcommand{\VIdPTA}[1]{\family{V$_{#1}$I-dPTA}}

\newcommand{\IdPTA}{\family{I-dPTA}}

\newcommand{\TT}{\family{TT}}
\newcommand{\VIPTT}[1]{\family{V$_{#1}$I-PTT}}
\newcommand{\VsIPTT}[1]{\family{V$^+_{#1}$I-PTT}}
\newcommand{\IPTT}{\family{I-PTT}}
\newcommand{\VPTT}[1]{\family{V$_{#1}$-PTT}}
\newcommand{\VsPTT}[1]{\family{V$^+_{#1}$-PTT}}

\newcommand{\dTT}{\family{dTT}}
\newcommand{\tdTT}{\family{tdTT}}
\newcommand{\dTTmso}{\family{dTT$^{\text{\abb{MSO}}}$}}
\newcommand{\tdTTmso}{\family{tdTT$^{\text{\abb{MSO}}}$}}
\newcommand{\VIdPTT}[1]{\family{V$_{#1}$I-dPTT}}
\newcommand{\VsIdPTT}[1]{\family{V$^+_{#1}$I-dPTT}}
\newcommand{\IdPTT}{\family{I-dPTT}}
\newcommand{\rIdPTT}{\family{r\,I-dPTT}}

\newcommand{\VdPTT}[1]{\family{V$_{#1}$-dPTT}}
\newcommand{\VsdPTT}[1]{\family{V$^+_{#1}$-dPTT}}

\newcommand{\VIPFT}[1]{\family{V$_{#1}$I-PFT}}
\newcommand{\VIdPFT}[1]{\family{V$_{#1}$I-dPFT}}
\newcommand{\IPFT}{\family{I-PFT}}
\newcommand{\IdPFT}{\family{I-dPFT}}
\newcommand{\DTL}{\family{DTL}}
\newcommand{\dDTL}{\family{dDTL}}
\newcommand{\TL}{\family{TL}}
\newcommand{\dTL}{\family{dTL}}
\newcommand{\DTLr}{\family{DTL$_\text{r}$}}
\newcommand{\dDTLr}{\family{dDTL$_\text{r}$}}
\newcommand{\TLr}{\family{TL$_\text{r}$}}
\newcommand{\dTLr}{\family{dTL$_\text{r}$}}
\newcommand{\TLl}{\family{TL$_\ell$}}
\newcommand{\dTLl}{\family{dTL$_\ell$}}
\newcommand{\TLlr}{\family{TL$_{\ell\text{r}}$}}
\newcommand{\dTLlr}{\family{dTL$_{\ell\text{r}}$}}

\newtheorem{theorem}{Theorem}
\newtheorem{lemma}[theorem]{Lemma}
\newtheorem{corollary}[theorem]{Corollary}
\newtheorem{proposition}[theorem]{Proposition}

\newcommand{\qed}{\ifvmode\penalty10000\vskip -\lastskip\fi%
  \penalty10000
  \ \ \hfill\hbox{\rm$\Box$}}
\newenvironment{proof}{\pagebreak[1]\noindent{\bf
Proof. }\nopagebreak}{\qed\medskip\pagebreak[3]}

\makeatletter
\newtheorem{ex@ple}[theorem]{Example}
\newenvironment{example}{\begin{ex@ple}\normalfont}{\qed\end{ex@ple}}
\makeatother

\newcommand{\smallpar}[1]{\smallskip\noindent {\bf #1.}}

\begin{document}

\title{XML Navigation and Transformation by \\
       Tree-Walking Automata and Transducers \\
       with Visible and Invisible Pebbles}

\author{
Joost Engelfriet\thanks{Email: {\tt j.engelfriet@liacs.leidenuniv.nl}} 
\and Hendrik Jan Hoogeboom\thanks{Email: {\tt h.j.hoogeboom@liacs.leidenuniv.nl}} \and \\
Bart Samwel\thanks{Email: {\tt bsamwel@gmail.com}} \\ \\
{\small LIACS, Leiden University, the Netherlands}
}

\date{}
\maketitle


\begin{abstract}
The pebble tree automaton and the pebble tree transducer are enhanced by
additionally allowing an unbounded number of
``invisible'' pebbles (as opposed to the usual ``visible''
ones). The resulting pebble tree automata recognize
the regular tree languages (i.e., can validate all
generalized DTD's) and hence can find all matches of
MSO definable patterns. Moreover, when viewed
as a navigational device, they lead to an XPath-like
formalism that has a path expression for every MSO
definable binary pattern. The resulting pebble tree
transducers can apply arbitrary MSO definable tests to
(the observable part of) their configurations, they
(still) have a decidable typechecking problem, and
they can model the recursion mechanism of XSLT. 
The time complexity of the typechecking problem for conjunctive
queries that use MSO definable patterns can
often be reduced through the use of invisible pebbles. 
\end{abstract}

\newpage
\tableofcontents

\newpage
\section{Introduction}

Pebble tree transducers, as introduced by Milo, Suciu,
and Vianu \cite{MilSucVia03}, are a formal model of XML
navigation and transformation for which typechecking is decidable. 
The pebble tree transducer is a tree-walking tree transducer with nested pebbles,
i.e., it walks on the input tree, dropping and lifting 
a bounded number of pebbles that have nested life times, whereas it produces the
output tree in a parallel top-down fashion. 
We enhance the power of the pebble tree transducer by
allowing an unbounded number of (coloured) pebbles,
still with nested life times, i.e., organized as a
stack. However, apart from a bounded number, the
pebbles are ``invisible'', which means that they can
be observed by the transducer only when they are on
top of the stack (and thus the number of observable
pebbles is bounded at each moment in time). 
We will call \abb{v-ptt} the pebble tree transducer of
\cite{MilSucVia03} (or rather, the one in \cite{EngMan03}:
an obvious definitional variant), and \abb{vi-ptt} the
enhanced pebble tree transducer. Moreover, \abb{i-ptt}
refers to the \abb{vi-ptt} that does not use visible
pebbles, which can be viewed as a generalization of 
the indexed tree transducer of~\cite{EngVog}. 
And \abb{tt} refers to the pebble tree transducer
without pebbles, i.e., to the tree-walking tree transducer,
cf.~\cite{Eng09} and~\cite[Section~8]{thebook}.
Tree-walking transducers were introduced in~\cite{AhoUll},
where they translate trees into strings.\footnote{In~\cite[Section~8]{thebook}
the \abb{tt} is called tree-walking transducer and the transducer of~\cite{AhoUll}
is called tree-walking tree-to-word transducer.
} 

The navigational part of the \abb{v-ptt}, i.e., the
behaviour of the transducer when no output is
produced, is the pebble tree automaton (\abb{v-pta}),
introduced in~\cite{jewels}, which is a tree-walking automaton 
with nested pebbles. It was shown in~\cite{jewels} that the \abb{v-pta}
recognizes regular tree languages only. In
\cite{expressive} the important result was proved
that not all regular tree languages can be recognized
by the \abb{v-pta}, and thus \cite{Don70,ThaWri} the
navigational power of the \abb{v-ptt} is below Monadic
Second Order (\abb{MSO}) logic, which is undesirable for a
formal model of XML transformation (see, e.g., \cite{NevSch}). 
One of the reasons
for introducing invisible pebbles is that the
\abb{vi-pta}, and even the \abb{i-pta}, recognizes exactly
the regular tree languages 
(Theorem~\ref{thm:regt}).
Thus, since the regular tree grammar is a formal model of 
DTD (Document Type Definition) in XML, the \abb{vi-pta} can
validate arbitrary generalized DTD's. 
We note that the \abb{i-pta} is
a straightforward generalization of the two-way
backtracking pushdown tree automaton of Slutzki
\cite{Slu}. 

Surveys on the use of tree-walking automata and transducers for XML 
can be found in~\cite{Nev,Sch}. 
For a survey on tree-walking automata see~\cite{Boj}.

It is easy to show that every regular tree language
can be recognized by an \mbox{\abb{i-pta}}, just simulating a
bottom-up finite-state tree automaton. The proof that all
\abb{vi-pta} tree languages are regular, is based on a
decomposition of the \abb{vi-ptt} into \abb{tt}'s
(Theorem~\ref{thm:decomp}),
similar to the one for the \abb{v-ptt} in \cite{EngMan03}.
Since the inverse type inference problem is solvable
for \abb{tt}'s (where a ``type'' is a regular tree
language), this shows that the domain of a
\abb{vi-ptt} is regular, and so even the alternating
\abb{vi-pta} tree languages are regular. It also shows that the
typechecking problem is decidable for \abb{vi-ptt}'s, by
the same arguments as used in \cite{MilSucVia03} for
\abb{v-ptt}'s. More precisely, we prove 
(Theorem~\ref{thm:typecheck}, based on~\cite[Theorem~3]{Eng09})
that a \abb{vi-ptt} with $k$ visible pebbles can be
typechecked in \mbox{($k+3$)}-fold exponential time. 
For varying~$k$ the complexity is non-elementary (as in~\cite{MilSucVia03}), 
but it is observed in~\cite{MolSch} that
``non-elementary algorithms on tree automata have previously 
been seen to be feasible in practice''. 

Generalizing the fact that the \abb{i-pta} can recognize the regular tree languages, 
we prove that the \abb{vi-pta} and the \abb{vi-ptt} can perform \abb{MSO} tests on
the observable part of their configuration, i.e., they
can check whether or not the observable pebbles on the
input tree (i.e., the visible ones, plus the top
pebble on the stack) satisfy certain \abb{MSO} requirements
with respect to the current position of the reading
head (Theorem~\ref{thm:mso}). 
If all the observable pebbles are visible this
is obvious (drop an additional visible pebble, simulate an \abb{i-pta}
that recognizes the regular tree
language corresponding to the \abb{mso} requirements, 
return to the pebble and lift it), but if the top
pebble is invisible (or if there is no visible pebble left) 
that does not work and a more complicated technique must be used. 
Consequently, the \abb{vi-pta} can match
arbitrary \abb{MSO} definable $n$-ary patterns, using 
$n$~visible pebbles to find all candidate matches as in
\cite[Example~3.5]{MilSucVia03}, and using invisible pebbles to perform
the \abb{MSO} test; the \abb{vi-ptt} can also output the matches. 
In fact, instead of the $n$~visible pebbles the \abb{vi-pta} 
can use $n-2$ visible pebbles, one invisible pebble (on top of the stack), 
and the reading head (Theorem~\ref{thm:matchall}).  

As the navigational part of the \abb{vi-ptt}, the
\abb{vi-pta} in fact computes a binary pattern on
trees, i.e., a binary relation between two nodes of a
tree: the position of the reading head of the
\abb{vi-ptt} before and after navigation. We prove
that also as a navigational device the \abb{vi-pta}
and the \abb{i-pta} have the same power as \abb{MSO} logic:
they compute exactly the \abb{MSO} definable binary patterns
(Theorem~\ref{thm:trips}).
This improves the result in \cite{trips} (where binary
patterns are called ``trips''), because the \abb{i-pta}
is a more natural automaton than the one considered in
\cite{trips}. 

One of the research goals of Marx and ten Cate
(see \cite{Mar05,GorMar05,Cat06,CatSeg} and the entertaining
\cite{Mar06}) has been to combine Core XPath of
\cite{GotKoc02} which models the navigational part of
XPath 1.0, with regular path expressions
\cite{AbiBunSuc00} (or caterpillar expressions
\cite{BruWoo00}) which naturally correspond to
tree-walking automata. An important feature of XPath
is the ``predicate'': it allows to test the context node
for the existence of at least one other node that
matches a given path expression. Thus, the path
expression $\alpha_1[\beta]/\alpha_2$ takes an
$\alpha_1$-walk from the context node to the new
context node $v$, checks whether there exists a
$\beta$-walk from $v$ to some other node, and then
takes an $\alpha_2$-walk from $v$ to the match node.
For tree automata this corresponds to the
notion of ``look-ahead''
(cf. \cite[Definition~6.5]{EngVog}). We prove 
(Theorem~\ref{thm:look-ahead})
that an \abb{i-pta}
$\cA$ can use another \abb{i-pta} $\cB$ as look-ahead
test, i.e., $\cA$ can test whether or not $\cB$ has a
successful computation when started in the current
configuration of $\cA$ (and similarly for \abb{vi-pta}
and \abb{vi-ptt}). 
Since XPath expressions can be nested arbitrarily, we even
allow $\cB$ to use yet another \abb{i-pta} as look-ahead test,
etcetera (Theorem~\ref{thm:iterated}). 
Due to this ``iterated look-ahead'' feature, we can use
Kleene's classical construction to translate the
$\mbox{\abb{i-pta}}$ into an XPath-like algebraic
formalism, which we call \emph{Pebble XPath}, with the
same expressive power as \abb{MSO} logic for defining binary
patterns (Theorem~\ref{thm:xpath}). 
In fact, Pebble XPath is the extension of
Regular XPath \cite{Mar05,Cat06} with a stack of invisible pebbles. 
It is proved in \cite{CatSeg} that Regular XPath
is not \abb{MSO} complete (see also~\cite{Mar06}).\footnote{To be precise,
it is proved in \cite{CatSeg} that 
Regular XPath with ``subtree relativisation'' 
is not \abb{mso} complete and 
has the same power as first-order logic with monadic transitive closure.
}
Other \abb{MSO} complete extensions of Regular XPath are considered 
in \cite{GorMar05,Cat06}.

To explain another reason for introducing invisible
pebbles we consider XQuery-like conjunctive queries of
the form  
\[
\mbox{\tt for } x_1,\dots,x_n \mbox{ \tt
where } \phi_1\wedge\dots\wedge\phi_m \mbox{ \tt
return } r,
\]  
where $x_1,\dots,x_n$ are variables,
each $\phi_\ell$ (with $1\leq \ell\leq m$) is an \abb{MSO} formula with two free
variables $x_i$ and $x_j$, and $r$ is an output tree
with variables at the leaves. As observed above, 
such pattern matching queries can be evaluated by a
\abb{vi-ptt} with $n-2$ visible pebbles, even if the
{\tt where}-clause contains an arbitrary \abb{MSO} formula.
In many cases, however, a much smaller number of
visible pebbles suffices (Theorem~\ref{thm:matching}). 
This is an enormous
advantage when typechecking the query, as for the time
complexity every visible pebble counts (viz. it counts
as an exponential). For instance if $j=i+1$ for every
$\phi_\ell$, then \emph{no} visible pebbles are needed, i.e.,
the query can be evaluated by an \abb{i-ptt}: we use
invisible pebbles $p_1,\dots,p_n$ on the stack (in
that order), representing the variables, and move them
through the input tree in document order, in a nested
fashion; just before dropping pebble $p_{i+1}$, each
formula $\phi_\ell(x_i,x_{i+1})$ can be verified by an
MSO test on the observable part of the configuration (which
consists of the top pebble $p_i$ and the reading head
position). 

The pebble tree transducer transforms ranked trees.
However, an XML document is not ranked; it is a
forest: a sequence of unranked trees. To model XML
transformation by \abb{ptt}'s,  forests are encoded as
binary trees in the usual way. For the input, it does
not make much of a difference whether the \abb{ptt}
walks on a binary tree or a forest. However, as opposed to
what is suggested in \cite{MilSucVia03}, for the output
it \emph{does} make a difference, as pointed out in
\cite{PerSei} for macro tree transducers. For that
reason we also consider pebble \emph{forest}
transducers (abbreviated with \abb{pft} instead of
\abb{ptt}) that walk on encoded forests, but construct
forests directly, using forest concatenation as
basic operation. As in \cite{PerSei}, \abb{pft} are
more powerful than \abb{ptt}, but the complexity of
the typechecking problem is the same, i.e.,
\abb{vi-pft} with $k$ visible pebbles can be
typechecked in ($k+3$)-fold exponential time 
(Theorem~\ref{thm:typecheck-pft}). 
In fact, \abb{pft} have all the properties 
mentioned before for \abb{ptt}.

The document transformation languages \abb{DTL} and
\abb{TL} were introduced in \cite{ManNev00} and
\cite{ManBerPerSei}, respectively, as a formal model
of the recursion mechanism in the template rules of
XSLT, with  \abb{MSO} logic rather than XPath to specify
matching and selection. Documents are modelled as forests. 
The language
\abb{DTL} has no variables or parameters, and
its only instruction is {\tt apply-templates}.
The language~\abb{TL} is the extension of \abb{DTL} with
accumulating parameters, i.e., the parameters of
XSLT~1.0 whose values are ``result tree fragments'' (and
on which no operations are allowed). 
We prove that
every \abb{DTL} program can be simulated, 
with forests encoded as binary trees, by an
\abb{i-ptt} (Theorem~\ref{thm:dtl}). 
More importantly, we prove that 
\abb{TL} and \abb{i-pft} have the same expressive
power (Theorem~\ref{thm:tl}). 
Thus, in its forest version, our new model the
\abb{vi-pft} can be viewed as the natural combination
of the pebble tree transducer of \cite{MilSucVia03}
(\abb{v-ptt}) and the \abb{TL} program of
\cite{ManBerPerSei} (\abb{i-pft}). Note that
\abb{v-ptt} and \abb{TL} have incomparable expressive
power. As claimed by \cite{ManBerPerSei}, \abb{TL}
can ``describe many real-world XML
transformations''. 
We show that it contains all deterministic \abb{vi-pft} transformations 
for which the size of the output document is linear 
in the size of the input document (Theorem~\ref{thm:lsi}). 
However, the visible pebbles seem to be a
requisite for the XQuery-like queries discussed above,
and we conjecture that not all such queries can be
programmed in \abb{TL} 
(though they \emph{can}, e.g.,
in the case that $j=i+1$ for every~$\ell$). 
As shown in~\cite{bex} (for a subset of \abb{MSO}), these queries can be
programmed in XSLT~1.0 using parameters that have
input nodes as values; however, with such parameters
even \abb{v-ptt} with \emph{non}nested pebbles can be
simulated, and typechecking is no longer decidable. 
In XSLT~2.0 \emph{all} (computable) queries can be 
programmed~\cite{JanKorBus}.
The main result of~\cite{ManBerPerSei} is that typechecking is decidable for \abb{tl} programs. 
Assuming that \abb{mso} formulas are represented by deterministic bottom-up finite-state tree automata,
the above relationship between \abb{tl} and \abb{i-pft} allows us to prove that
\abb{tl} programs can be typechecked in $4$-fold exponential time (Theorem~\ref{thm:tltypecheck}), 
which seems to be one exponential better than the algorithm in~\cite{ManBerPerSei}. 

In addition to the time complexity of typechecking a \abb{vi-ptt}, also the time complexity of
evaluating the queries realized by a \abb{vi-pta} or a \abb{vi-ptt} is of importance.
The binary pattern (or `trip') computed by a \abb{vi-pta}, i.e., 
the binary relation between two nodes of the input tree, 
can be evaluated in polynomial time. The same is true for every (fixed) expression of Pebble XPath
(see the last two paragraphs of Section~\ref{sec:xpath}).
Deterministic \abb{vi-ptt}'s have exponential time data complexity,
provided that the output tree can be represented by a DAG (directed acyclic graph).
To be precise, for every deterministic \abb{vi-ptt} there is an exponential time algorithm
that transforms any input tree of that \abb{vi-ptt} into a DAG 
that represents the corresponding output tree (Theorem~\ref{thm:expocom}). 
For the \abb{vi-ptt}'s that match \abb{mso} definable $n$-ary patterns (as discussed above) 
the algorithm is polynomial time (Theorem~\ref{thm:polycom}). 

Apart from the above results that are motivated by XML navigation and transformation,
we also prove some more theoretical results. 
We show that (as opposed to the \abb{v-ptt}) the \abb{i-ptt} 
can simulate the bottom-up tree transducer (Theorem~\ref{thm:bottomup}). 
We show that the composition of two deterministic \abb{tt}'s can be simulated by 
a deterministic \mbox{\abb{i-ptt}} (Theorem~\ref{thm:composition}). 
This even holds when the \abb{tt}'s are allowed to perform \abb{mso} tests on their configuration,
and then also vice versa, every deterministic \abb{i-ptt} can be decomposed
into two such extended \abb{tt}'s (Theorem~\ref{thm:charidptt}).

We show that every deterministic \abb{vi-ptt} can be decomposed 
into deterministic \abb{tt}'s (Theorem~\ref{thm:detdecomp}) and that, 
for the deterministic \abb{vi-ptt}, $k+1$ visible pebbles are more powerful 
than $k$ visible pebbles (Theorem~\ref{thm:dethier}).
Pebbles have to be lifted from the position where they were dropped;
however, in~\cite{fo+tc} it was convenient to consider a stronger type of pebbles
that can also be retrieved from a distance. Whereas \abb{i-ptt}'s with strong invisible 
pebbles can recognize nonregular tree languages, 
we show that \abb{vi-ptt}'s with strong visible pebbles can still be decomposed into 
\abb{tt}'s (Theorems~\ref{thm:decompplus} and~\ref{thm:decompplusI}) 
and hence their typechecking is decidable
(as already proved for \abb{v-ptt}'s with strong pebbles in~\cite{FulMuzFI}).
Similarly, deterministic \abb{vi-ptt}'s with strong visible pebbles can be decomposed into 
deterministic \abb{tt}'s (Theorems~\ref{thm:detdecompplus} and~\ref{thm:detdecompplusI}). 

Some of these theoretical results can be viewed as (slight) generalizations of 
existing results for formal models of compiler construction
(in particular attribute grammars), such as attributed tree transducers~\cite{Ful}, 
macro tree transducers~\cite{EngVog85}, and macro attributed tree transducers~\cite{KuhVog},
see also~\cite{FulVog}. As explained in~\cite[Section~3.2]{EngMan03}, attributed tree transducers 
are \abb{tt}'s that satisfy an additional requirement of ``noncircularity''. Similarly, 
as observed in~\cite{ManBerPerSei}, macro attributed tree transducers 
(that generalize both attributed tree transducers and macro tree transducers) 
are closely related to \abb{tl} programs, and hence to \abb{i-ptt}'s by Theorem~\ref{thm:tl}. 
For instance, Theorem~\ref{thm:composition} slightly generalizes the fact that the composition of 
two attributed tree transducers can be simulated by a macro attributed tree transducer, 
as shown in~\cite{KuhVog}. 

Most of the results of this paper were announced in the PODS'07 conference~\cite{invisible}.
The remaining results are based on technical notes of the authors from the years 2004--2008.
This paper has not been updated with the litterature of later years 
(with the exception of~\cite{thebook,Eng09,CatSeg}).

\section{Preliminaries}\label{sec:trees}

\smallpar{Sets, strings, and relations}
The set of natural numbers is $\nat=\{0,1,2,\dots\}$.
For $m,n\in \nat$, we denote the interval $\{k\in\nat\mid m\leq k\leq n\}$ by $[m,n]$. 
The cardinality or size of a set $A$ is denoted by $\#(A)$, 
and its powerset, i.e., the set of all its subsets, by $2^A$. 
The set of strings over $A$ is denoted by~$A^*$. 
It consists of all sequences $w=a_1\cdots a_m$ with $m\in\nat$ and 
$a_i\in A$ for every $i\in[1,m]$. 
The length~$m$ of $w$ is denoted by $|w|$. 
The empty string (of length $0$) is denoted by~$\epsilon$. 
The concatenation of two strings $v$ and $w$ is denoted by $v\cdot w$ or just~$vw$.
Moreover, $w^0=\epsilon$ and $w^{n+1}=w\cdot w^n$ for $n\in\nat$.
The composition of two binary relations $R\subseteq A\times B$ and 
$S\subseteq B\times C$ is 
$R\circ S = \{(a,c)\mid \exists\, b\in B: (a,b)\in R, \, (b,c)\in S\}$.
The~inverse of $R$ is $R^{-1}=\{(b,a)\mid (a,b)\in R\}$, and if $A=B$ then 
the transitive-reflexive closure of~$R$ is $R^*=\bigcup_{n\in\nat}R^n$
where $R^0=\{(a,a)\mid a\in A\}$ and $R^{n+1}=R\circ R^n$. The composition of 
two classes of binary relations $\cR$ and $\cS$ is 
$\cR\circ \cS = \{R\circ S\mid R\in\cR,\,S\in\cS\}$. 
Moreover, $\cR^1=\cR$ and $\cR^{n+1} = \cR\circ\cR^n$ for $n\geq 1$. 

\smallpar{Trees and forests}
An alphabet is a finite set of symbols. 
Let $\Sigma$ be an alphabet, or an arbitrary set.
Unranked trees and forests over $\Sigma$ are recursively defined
to be strings over the set $\Sigma\cup\{(,)\}$ 
consisting of the elements of $\Sigma$, the left parenthesis, 
and the right parenthesis, as follows.  
If $\sigma\in\Sigma$ and 
$t_1, \dots, t_m$ are unranked trees, with $m\in\nat$, then
their concatenation $t_1 \cdots t_m$ is a forest, and 
$\sigma(t_1 \cdots t_m)$ is an unranked tree.
For $m=0$, $t_1 \cdots t_m$ is the empty forest~$\epsilon$. 
For readability we also write 
the tree $\sigma(t_1 \cdots t_m)$ as $\sigma(t_1, \dots, t_m)$,
and even as~$\sigma$ when $m=0$. 
Obviously, the concatenation of two forests is again a forest.
It should also be noted that every nonempty forest can 
be written uniquely as $\sigma(f_1)f_2$
where $\sigma$ is in $\Sigma$ and $f_1$ and $f_2$ are forests. 
The set of forests over~$\Sigma$ is denoted $F_\Sigma$. 
For an arbitrary set $A$, disjoint with $\Sigma$, we denote by $F_\Sigma(A)$ the set
all forests $f$ over $\Sigma\cup A$ such that every node of $f$ 
that is labelled by an element of $A$, is a leaf. 

As usual trees and forests are viewed as directed labelled graphs. 
Here we distinguish between two types of edges: ``vertical'' and ``horizontal'' ones.
The root of the tree $t=\sigma( t_1, \dots, t_m )$ is labelled by $\sigma$.
It has vertical edges to the roots of subtrees $t_1, \dots, t_m$,
which are the children of the root of $t$ and have child number $1$ to $m$. 
The root of~$t$ is their parent.
The roots of $t_1, \dots, t_m$ are siblings, 
also in the case of the forest $ t_1 \cdots t_m $. 
There is a horizontal edge from each sibling to the next, i.e., 
from the root of $t_i$ to the root of $t_{i+1}$ for every $i\in [1,m-1]$. 
Thus, the vertical edges represent the usual parent/child relationship,
whereas the horizontal edges represent the linear order between children
(and between the roots in a forest), see Fig.~\ref{fig:forest}.\footnote{In 
informal pictures the horizontal edges are usually omitted because they are implicit 
in the left-to-right orientation of the page. Similarly, the arrows of 
the vertical edges are omitted because of 
the top-down orientation of the page.
}
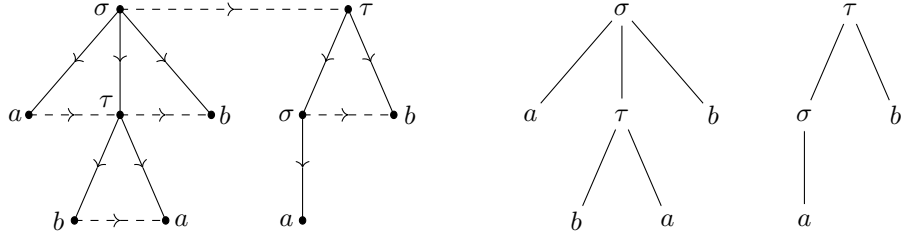
\begin{figure}[t]
\begin{center}
\begin{tikzpicture}
\begin{scope}[xscale=0.6,yscale=0.7]
%
\draw [fill] (1,0) circle [radius=0.07];
\draw [fill] (3,0) circle [radius=0.07];
\draw [fill] (6,0) circle [radius=0.07];
\draw [fill] (0,2) circle [radius=0.07];
\draw [fill] (2,2) circle [radius=0.07];
\draw [fill] (4,2) circle [radius=0.07];
\draw [fill] (6,2) circle [radius=0.07];
\draw [fill] (8,2) circle [radius=0.07];
\draw [fill] (2,4) circle [radius=0.07];
\draw [fill] (7,4) circle [radius=0.07];
\draw [->] (2,4) -- (1,3); \draw (1,3) -- (0,2); 
\draw [->] (2,4) -- (2,3); \draw (2,3) -- (2,2); 
\draw [->] (2,4) -- (3,3); \draw (3,3) -- (4,2);
\draw [->] (2,2) -- (1.5,1); \draw (1.5,1) -- (1,0);
\draw [->] (2,2) -- (2.5,1); \draw (2.5,1) -- (3,0);
\draw [->] (7,4) -- (6.5,3); \draw (6.5,3) -- (6,2); 
\draw [->] (6,2) -- (6,1); \draw (6,1) -- (6,0);
\draw [->] (7,4) -- (7.5,3); \draw (7.5,3) -- (8,2);
\draw [dashed] [->] (2,4) -- (4.5,4); \draw [dashed] (4.5,4) -- (7,4);
\draw [dashed] [->] (0,2) -- (1,2); \draw [dashed] (1,2) -- (2,2);
\draw [dashed] [->] (2,2) -- (3,2); \draw [dashed] (3,2) -- (4,2);
\draw [dashed] [->] (6,2) -- (7,2); \draw [dashed] (7,2) -- (8,2);
\draw [dashed] [->] (1,0) -- (2,0); \draw [dashed] (2,0) -- (3,0);
\node at (0.65,0) {$b$}; \node at (3.35,0) {$a$}; \node at (5.65,0) {$a$};
\node at (-0.3,2) {$a$}; \node at (1.7,2.25) {$\tau$}; \node at (4.3,2) {$b$};
\node at (5.65,2) {$\sigma$}; \node at (8.35,2) {$b$};
\node at (1.6,4) {$\sigma$}; \node at (7.4,4) {$\tau$};
%
%
\draw (13,4) -- (11,2); 
\draw (13,4) -- (13,2); 
\draw (13,4) -- (15,2);
\draw (13,2) -- (12,0);
\draw (13,2) -- (14,0);
\draw (18,4) -- (17,2); 
\draw (17,2) -- (17,0);
\draw (18,4) -- (19,2);
\draw [fill=white,white] (12,0) circle [radius=0.3];
\node at (12,0) {$b$}; 
\draw [fill=white,white] (14,0) circle [radius=0.3];
\node at (14,0) {$a$}; 
\draw [fill=white,white] (17,0) circle [radius=0.3];
\node at (17,0) {$a$};
\draw [fill=white,white] (11,2) circle [radius=0.3];
\node at (11,2) {$a$}; 
\draw [fill=white,white] (13,2) circle [radius=0.3];
\node at (13,2) {$\tau$}; 
\draw [fill=white,white] (15,2) circle [radius=0.3];
\node at (15,2) {$b$};
\draw [fill=white,white] (17,2) circle [radius=0.3];
\node at (17,2) {$\sigma$}; 
\draw [fill=white,white] (19,2) circle [radius=0.3];
\node at (19,2) {$b$};
\draw [fill=white,white] (13,4) circle [radius=0.3];
\node at (13,4) {$\sigma$};
\draw [fill=white,white] (18,4) circle [radius=0.3];
\node at (18,4) {$\tau$};
\end{scope}
\end{tikzpicture}
\end{center}
  \caption{Picture of the forest $\sigma(a,\tau(b,a),b)\,\tau(\sigma(a),b)$.
Formal at the left, with dotted lines for the horizontal edges and solid lines 
for the vertical edges, and informal at the right.}
  \label{fig:forest}
\end{figure}
For a tree $t$, its root is denoted by $\rt_t$, which is given child number $0$ 
for technical convenience. Its set of nodes is denoted by $\nod t$. 
For a forest $f= t_1 \cdots t_m $,  
the set of nodes $\nod f$ is the disjoint union 
of the sets $\nod{t_i}$, $i\in[1,m]$.
For a node $u$ of a tree~$t$ the subtree of $t$ with root $u$ is denoted $t|_u$,
and the $i$-th child of $u$ is denoted $ui$ 
(and similarly for a forest $f$ instead of $t$). 
The nodes of a tree $t$ correspond one-to-one to the positions 
of the elements of $\Sigma$ in the string $t$, i.e., 
for every $\sigma\in\Sigma$, each occurrence of $\sigma$ in $t$ 
corresponds to a node of $t$ with label $\sigma$. Since the positions 
of string $t$ are naturally ordered from left to right, this induces an order 
on the nodes of $t$, which is called pre-order (or document order, 
when viewing $t$ as an XML document). For example, the tree 
$\sigma(\tau(\alpha,\beta),\gamma))$ has five nodes which have the labels 
$\sigma$, $\tau$, $\alpha$, $\beta$, and $\gamma$ in pre-order. 

A \emph{ranked} alphabet (or set) $\Sigma$
has an associated mapping
$\rank_\Sigma : \Sigma \to \nat$.
The maximal rank of elements of $\Sigma$ is denoted
$\m_\Sigma$.
By $\Sigma^{(m)}$ we denote the elements of $\Sigma$ with rank $m$.
Ranked trees over $\Sigma$ are recursively defined
as above with the requirement that
$m = \rank_\Sigma(\sigma)$.
The set of ranked trees over~$\Sigma$ is denoted $T_\Sigma$. 
For an arbitrary set $A$, disjoint with $\Sigma$, we denote by $T_\Sigma(A)$ the set
$T_{\Sigma\cup A}$ where each element of $A$ has rank~$0$. 
We will not consider ranked forests. 

Forests over an alphabet $\Sigma$ can be encoded as binary trees, in the usual way: 
each node has a label in $\Sigma$, a ``vertical'' pointer to its first child, 
and a ``horizontal'' pointer to its next sibling;
the pointer is nil if there is no such child or sibling.
Such a binary tree can be modelled as a ranked tree 
over the ranked alphabet $\Sigma\cup\{e\}$
where every $\sigma\in\Sigma$ has rank~2 and $e$ is a symbol of rank~0 that 
represents the empty forest $\epsilon$ (or nil). 
Formally, the encoding of the empty forest
equals $\enc(\epsilon) = e$,
and recursively,
the encoding $\enc(f)$ of a forest 
$f=\sigma(f_1)f_2$ equals 
$\sigma(\enc(f_1),\enc(f_2))$.
Obviously, $\enc$ is a bijection between forests over $\Sigma$ and ranked trees 
over $\Sigma\cup\{e\}$. The decoding which is its inverse 
will be denoted by $\dec$. 
For an example of $\enc(f)$ see Fig.~\ref{fig:encforest} at the left. 
\begin{figure}[t]
\begin{center}
\begin{tikzpicture}
\begin{scope}[xscale=0.48,yscale=0.65]
%
%
\draw (7,10) -- (1.5,8) -- (0,6); 
\draw (1.5,8) -- (3,6) -- (1,4) -- (0,2);
\draw (1,4) -- (2,2) -- (1,0);
\draw (2,2) -- (3,0);
\draw (3,6) -- (5,4) -- (4,2);
\draw (5,4) -- (6,2);
\draw (7,10) -- (12.5,8) -- (11,6) -- (9,4) -- (8,2);
\draw (9,4) -- (10,2);
\draw (11,6) -- (13,4) -- (12,2);
\draw (13,4) -- (14,2);
\draw (12.5,8) -- (14,6);
\draw [fill=white,white] (1,0) circle [radius=0.3];
\node at (1,0) {$e$};
\draw [fill=white,white] (3,0) circle [radius=0.3];
\node at (3,0) {$e$};
\draw [fill=white,white] (0,2) circle [radius=0.3];
\node at (0,2) {$e$};
\draw [fill=white,white] (2,2) circle [radius=0.3];
\node at (2,2) {$a$};
\draw [fill=white,white] (4,2) circle [radius=0.3];
\node at (4,2) {$e$};
\draw [fill=white,white] (6,2) circle [radius=0.3];
\node at (6,2) {$e$};
\draw [fill=white,white] (8,2) circle [radius=0.3];
\node at (8,2) {$e$};
\draw [fill=white,white] (10,2) circle [radius=0.3];
\node at (10,2) {$e$};
\draw [fill=white,white] (12,2) circle [radius=0.3];
\node at (12,2) {$e$};
\draw [fill=white,white] (14,2) circle [radius=0.3];
\node at (14,2) {$e$};
\draw [fill=white,white] (1,4) circle [radius=0.3];
\node at (1,4) {$b$};
\draw [fill=white,white] (5,4) circle [radius=0.3];
\node at (5,4) {$b$};
\draw [fill=white,white] (9,4) circle [radius=0.3];
\node at (9,4) {$a$};
\draw [fill=white,white] (13,4) circle [radius=0.3];
\node at (13,4) {$b$};
\draw [fill=white,white] (0,6) circle [radius=0.3];
\node at (0,6) {$e$};
\draw [fill=white,white] (3,6) circle [radius=0.3];
\node at (3,6) {$\tau$};
\draw [fill=white,white] (11,6) circle [radius=0.3];
\node at (11,6) {$\sigma$};
\draw [fill=white,white] (14,6) circle [radius=0.3];
\node at (14,6) {$e$};
\draw [fill=white,white] (1.5,8) circle [radius=0.3];
\node at (1.5,8) {$a$};
\draw [fill=white,white] (12.5,8) circle [radius=0.3];
\node at (12.5,8) {$\tau$};
\draw [fill=white,white] (7,10) circle [radius=0.3];
\node at (7,10) {$\sigma$};
%
%
%
\draw (20,10) -- (18,8) -- (18,6); 
\draw (20,10) -- (22,8) -- (22,6);
\draw (18,6) -- (17,4) -- (17,2); \draw (18,6) -- (19,4);
\draw (22,6) -- (21,4); \draw (22,6) -- (23,4);
\draw [fill=white,white] (17,2) circle [radius=0.4];
\node at (17,2) {$a^{00}$};
\draw [fill=white,white] (17,4) circle [radius=0.4];
\node at (17,4) {$b^{01}$};
\draw [fill=white,white] (19,4) circle [radius=0.4];
\node at (19,4) {$b^{00}$};
\draw [fill=white,white] (21,4) circle [radius=0.4];
\node at (21,4) {$a^{00}$};
\draw [fill=white,white] (23,4) circle [radius=0.4];
\node at (23,4) {$b^{00}$};
\draw [fill=white,white] (18,6) circle [radius=0.4];
\node at (18.2,6.1) {$\tau^{11}$};
\draw [fill=white,white] (22,6) circle [radius=0.4];
\node at (22.2,6.1) {$\sigma^{11}$};
\draw [fill=white,white] (18,8) circle [radius=0.4];
\node at (18.2,8) {$a^{01}$}; 
\draw [fill=white,white] (22,8) circle [radius=0.4];
\node at (22.2,8) {$\tau^{10}$};
\draw [fill=white,white] (20,10) circle [radius=0.4];
\node at (20.3,10) {$\sigma^{11}$};
\end{scope}
\end{tikzpicture}
\end{center}
  \caption{Encoding of the forest of Fig.~\ref{fig:forest}  
  by $\enc$ (at the left) and by $\enc'$ (at the right).}
  \label{fig:encforest}
\end{figure}
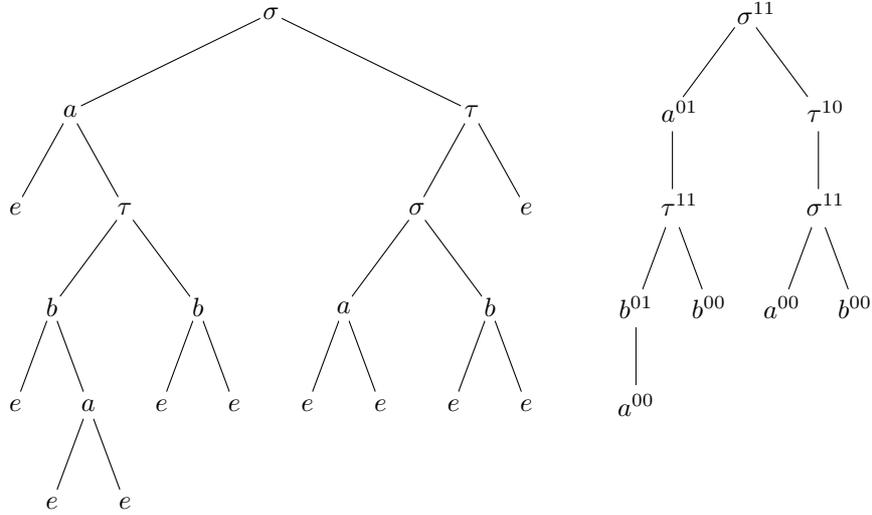

The disadvantage of this encoding is that the tree $\enc(f)$ has more nodes than the forest $f$, viz. all nodes with label $e$. That is inconvenient when comparing the behaviour of tree-walking automata on $f$ and $\enc(f)$. Thus, we will also use an encoding that preserves the number of nodes (and thus cannot encode the empty forest). For this we use the ranked alphabet
$\Sigma'$ consisting, for every $\sigma\in\Sigma$, of the symbols $\sigma^{11}$ of rank~2 (for a binary node without nil-pointers), $\sigma^{01}$ and $\sigma^{10}$ of rank~1 
(for a binary node with vertical or horizontal nil-pointer, respectively), 
and $\sigma^{00}$ of rank~0 (for a binary node with two nil-pointers). 
The encoding $\enc'(f)$ of a nonempty forest $f=\sigma(f_1)f_2$
equals $\sigma^{11}(\enc'(f_1),\enc'(f_2))$ or $\sigma^{01}(\enc'(f_2))$ 
or $\sigma^{10}(\enc'(f_1))$ or $\sigma^{00}$, 
where the first (second) superscript of~$\sigma$ equals $0$ if and only if
$f_1 = e$ ($f_2 = e$). 
Now, $\enc'$ is a bijection between nonempty forests over $\Sigma$ and ranked trees 
over $\Sigma'$. The decoding which is its inverse 
will be denoted by $\dec'$.
For an example of $\enc'(f)$ see Fig.~\ref{fig:encforest} at the right.
From the point of view of graphs, we assume that $\enc'(f)$ has the same nodes as $f$, 
i.e., $\nod{\enc'(f)}=\nod{f}$. The label of a node $u$ of $f$ is changed 
from $\sigma$ to $\sigma^{ij}$ where $i=1$ if and only if $u$ has at least one child, 
and $j=1$ if and only if $u$ has a next sibling. If $u$ has children, then its first child 
in $\enc'(f)$ is its first child in $f$, and its second child in $\enc'(f)$ 
is its next sibling (if it has one). 
If $u$ has no children, then its only child in $\enc'(f)$ is its next sibling (if it has one).
Although this encoding is intuitively clear, it is technically less attractive. 
We will use $\enc'$ for the input forest of automata and transducers,
and $\enc$ for the output forest of the transducers. 

We assume the reader to be familiar with the notion of 
a \emph{regular tree grammar}.
It is a context-free grammar~$G$ of which every rule is of the form 
$X_0\to \sigma(X_1\cdots X_m)$ where $X_i$ is a nonterminal and 
$\sigma$ is a terminal symbol of rank~$m$. Thus, $G$ generates a  
set $L(G)$ of ranked trees, which is called a regular tree language.  
The class of regular tree languages will be denoted $\REGT$.
We define a \emph{regular forest grammar} to be a context-free grammar $G$
of which every rule is of the form $X_0\to \sigma(X_1)X_2$ or $X\to\epsilon$,
where $\sigma$ is from an unranked alphabet. It generates 
a set $L(G)$ of (unranked) forests, which is called a regular forest language.
Obviously, $L$ is a regular forest language if and only if $\enc(L)$ 
is a regular tree language, and, as one can easily prove, if and only if 
$\enc'(L)$ is a regular tree language. 
The regular tree/forest grammar is a formal model of 
DTD (Document Type Definition) in XML.\footnote{In the litterature
regular forest languages are usually defined in a different way,
after which it is proved that $L$ is a regular forest language 
if and only if $\enc(L)$ is a regular tree language, 
thus showing the equivalence with our definition, 
see, e.g., \cite[Proposition~1]{Nev}.
}

\emph{Monadic second-order logic} (abbreviated as \abb{mso} logic) is used 
to describe properties of forests and trees. It views each forest or tree 
as a logical structure that has the set of nodes as domain. 
As basic properties of a forest over alphabet~$\Sigma$ it uses the atomic formulas 
$\lab_\sigma(x)$, $\down(x,y)$, and $\nex(x,y)$, meaning that 
node $x$ has label $\sigma\in\Sigma$, that $y$ is a child of $x$, and 
that $y$ is the next sibling of $x$, respectively.
Thus, $\down(x,y)$ and $\nex(x,y)$ represent the vertical and horizontal edges
of the graph representation of the forest.
For a ranked tree over ranked alphabet $\Sigma$ we could use the same 
atomic formulas, but it is customary to replace $\down(x,y)$ and $\nex(x,y)$
by the atomic formulas $\down_i(x,y)$, for every $i\in[1,\m_\Sigma]$,
meaning that $y$ is the $i$-th child of $x$. 
Additionally, \abb{mso}~logic has the atomic formulas 
$x=y$ and $x\in X$, where $X$ is a set of nodes. The formulas are built with the 
usual connectives $\neg$, $\wedge$, $\vee$, and $\to$; both node variables 
$x,y,\dots$ and node-set variables $X,Y,\dots$ can be quantified 
with~$\exists$ and~$\forall$.
For a forest (or ranked tree) $f$ over $\Sigma$ and a formula $\phi(x_1,\dots,x_n)$ 
with $n$ free node variables $x_1,\dots,x_n$, 
we write $f\models \phi(u_1,\dots,u_n)$ to mean that 
$\phi$ holds in $f$ for the nodes $u_1,\dots,u_n$ of $f$
(as values of the variables $x_1,\dots,x_n$ respectively).

We will occasionally use the following formulas: 
$\rt(x)$ and $\leaf(x)$ test whether node $x$ is a root or a leaf, 
and $\first(x)$ and $\last(x)$ test whether $x$ is a first or a last sibling.
Also, $\childno_i(x)$ tests whether $x$ is an $i$-th child, 
$\up(x,y)$ expresses that $y$ is the parent of $x$, 
and $\stay(x,y)$ expresses that $y$ equals $x$.
Thus, we define $\stay(x,y) \equiv x=y$ and 
\begin{enumerate}
\item[] $\rt(x) \equiv \neg\,\exists z(\down(z,x))$, 
        $\leaf(x) \equiv \neg\,\exists z(\down(x,z))$,
\item[] $\first(x) \equiv \neg\,\exists z(\nex(z,x))$,
        $\last(x) \equiv \neg\,\exists z(\nex(x,z))$,
\item[] $\childno_i(x) \equiv \exists z(\down_i(z,x))$,
        $\up(x,y) \equiv \down(y,x)$.
\end{enumerate}

\smallpar{Patterns}
Let $\Sigma$ be a ranked alphabet and $n\ge 0$.
An $n$-ary \emph{pattern} (or $n$-ary query) over $\Sigma$ is a set
$T \subseteq \{(t,u_1,\dots,u_n) 
   \mid t\in T_\Sigma, \,u_1,\dots,u_n \in N(t)\}$.
For $n=0$ this is a tree language, 
for $n=1$ it is a \emph{site} (trees with a distinguished node),
for $n=2$ it is a \emph{trip} \cite{trips}
(or a binary tree-node relation \cite{bloem}).

We introduce a new ranked alphabet $\Sigma \times \{0,1\}^n$,
the rank of $(\sigma,\ell)$ equals that of $\sigma$ in $\Sigma$.
For a tree $t$ over $\Sigma$ and $n$ nodes $u_1,\dots,u_n$
we define $\tmark(t,u_1,\dots,u_n)$ to be the tree over $\Sigma \times \{0,1\}^n$ 
that is obtained by adding to the label of each node $u$ in $t$ a vector 
$\ell \in \{0,1\}^n$ such that the $i$-th
component of $\ell$ equals $1$ if and only if $u=u_i$.
The $n$-ary pattern $T$ is \emph{regular} if
its marked representation 
is a regular tree language,
i.e., $\tmark(T)\in \REGT$.

An \abb{mso} formula $\phi(x_1,\dots,x_n)$ over $\Sigma$, 
with $n$ free node variables $x_1,\dots,x_n$, 
defines the $n$-ary pattern
$T(\phi) = \{(t,u_1,\dots,u_n) 
  \mid t \models \phi(u_1,\dots,u_n)\}$.
Note that $T(\phi)$ also depends on the order $x_1,\dots,x_n$
of the free variables of $\phi$. 
It easily follows from the result of 
Doner, Thatcher and Wright \cite{Don70,ThaWri}
that a pattern is \abb{mso} definable
if and only if it is regular (see~\cite[Lemma~7]{bloem}).

We will also consider patterns on forests. For an unranked alphabet $\Sigma$,
a (forest) pattern over $\Sigma$ is a subset of
$\{(f,u_1,\dots,u_n) \mid f\in F_\Sigma, \,u_1,\dots,u_n \in N(f)\}$.
As for ranked trees, an \abb{mso} formula $\phi(x_1,\dots,x_n)$ over $\Sigma$,  
defines the $n$-ary (forest) pattern
$\{(f,u_1,\dots,u_n) \mid f \models \phi(u_1,\dots,u_n)\}$.

\section{Automata and Transducers}\label{sec:autotrans}

In this section we define tree-walking automata and transducers with pebbles,
and discuss some of their properties.

\smallpar{Automata}
A \emph{tree-walking automaton with nested pebbles}
(pebble tree automaton for short, abbreviated \abb{pta})
is a finite state device with one reading head that walks from node to node
over its ranked input tree
following the vertical edges in either direction.
Additionally it has a supply of \emph{pebbles}
that can be used to mark the nodes of the tree.
The automaton may drop a pebble on the node
currently visited by the reading head,
but it may only lift any pebble from the current node if
that pebble was the last one dropped during the computation.
Thus, the life times of the pebbles on the tree are 
nested.
Here we consider two types of pebbles. First there are
a finite number of
``classical'' pebbles, which we here call \emph{visible} pebbles.
Each of these has a distinct colour, and 
at most $k$ visible pebbles (each with a different colour) can be 
present on the input tree during any computation,
where~$k$ is fixed.
Second there are \emph{invisible} pebbles. 
Again, these pebbles have a finite
number of colours
(distinct from those of the visible pebbles), 
but for each colour there is an
unlimited supply of pebbles that can be present on the input tree.
Visible pebbles can be observed by the automaton at
any moment when it visits the node where they were
dropped. An invisible pebble can only be observed
when it was the last pebble dropped on the tree 
during the computation.

The possible actions of the automaton 
are determined by its state,
the label of the current node, the child number of the node,
and the set of \emph{observable} pebbles on the current
node, that is, visible pebbles and an invisible
pebble when it was the last pebble dropped on the tree.
Unlike the \abb{PTA} from \cite{MilSucVia03},
our automata do \emph{not} branch (i.e., are not alternating).

The \abb{pta} is specified as a tuple 
$\cA = (\Sigma, Q, Q_0, F, C, C_\mathrm{v}, C_\mathrm{i}, R,k)$,
where $\Sigma$ is a ranked alphabet of input symbols,
$Q$ is a finite set of states,
$Q_0 \subseteq Q$ is the set of initial states,
$F \subseteq Q$ is the set of final states,
$C_\mathrm{v}$ and $C_\mathrm{i}$ 
are the finite sets of visible and invisible colours,
$C= C_\mathrm{v}\cup C_\mathrm{i}$,
$C_\mathrm{v} \cap C_\mathrm{i} = \nothing$,
$R$ is a finite set of rules, and $k\in \nat$.
Each rule is of the form
$\tup{q,\sigma,j,b} \to \tup{q',\alpha}$ 
such that $q,q'\in Q$, $\sigma \in \Sigma$, 
$j\in[0,\m_\Sigma]$,
$b\subseteq C$ with $\#(b\cap C_\mathrm{v})\leq k$ and 
$\#(b\cap C_\mathrm{i})\leq 1$,
and $\alpha$ is one of the following \emph{instructions}: 
\[
\begin{array}{ll}
\stay, & \\
\up & \text{provided } j\neq 0, \\
\down_i & \text{with } 1 \le i \le \rank_\Sigma(\sigma), \\
\drop_c & \text{with } c\in C 
  , \text{ and} \\
\lift_c & \text{with } c\in b,
\end{array}
\]
where the first three are \emph{move instructions} 
and the last two are \emph{pebble instructions}.
Note that, due to the nested life times of the pebbles, 
at most one pebble $c$ in $b$ can actually be lifted; 
however, the subscript $c$ of $\lift_c$
often increases the readability of a \abb{pta}.

A \emph{situation} $\tup{u,\pi}$ of the \abb{pta} $\cA$ on ranked tree $t$ over $\Sigma$ 
is given by the position $u$ of the head of $\cA$ on $t$,
and the stack $\pi$ containing the positions and colours of the
pebbles on the tree in the order in which they were dropped. Formally, 
$u\in N(t)$ and $\pi\in (N(t)\times C)^*$. The last element of $\pi$ represents 
the top of the stack. The set of all situations of $\cA$ on $t$ is denoted $\sit(t)$,
i.e., $\sit(t)=N(t)\times (N(t)\times C)^*$; note that it only depends on $C$.
A \emph{configuration} $\tup{q,u,\pi}$ 
of $\cA$ on $t$ additionally contains the state $q$ of $\cA$, $q\in Q$. 
It is \emph{final} when $q\in F$. 
An \emph{initial} configuration is of the form
$\tup{q_0,\rt_t,\epsilon}$
where $q_0 \in Q_0$,
$\rt_t$ is the root of $t$,
and $\epsilon$ is the empty stack.
The set of all configurations of $\cA$ on~$t$ is denoted $\con(t)$,
i.e., $\con(t)=Q\times N(t)\times (N(t)\times C)^*$.

We now define the computation steps of the \abb{pta} $\cA$, 
which lead from one configuration to another. 
For a given input tree $t$ they form a binary relation on $\con(t)$. 
A~rule $\tup{q,\sigma,j,b} \to \tup{q',\alpha}$ 
is \emph{relevant} to
every configuration $\tup{q,u,\pi}$ with state $q$
and with a situation $\tup{u,\pi}$ that satisfies the tests $\sigma$, $j$, and $b$,
i.e., $\sigma$ and $j$ are the label and child number of node $u$, 
and $b$ is the set of colours of the observable pebbles dropped on the node $u$. 
More precisely, $b$ consists of 
all $c\in C_\mathrm{v}$ such that $\pair{u,c}$ occurs in
$\pi$, plus $c \in C_\mathrm{i}$ if $\pair{u,c}$ is the topmost (i.e., last)
element of $\pi$. 
Application of the rule to such a configuration 
possibly leads to a new configuration $\tup{q',u',\pi'}$,
in which case we write $\tup{q,u,\pi}\Rightarrow_{t,\cA} \tup{q',u',\pi'}$.
The new state is $q'$ and the new situation $\tup{u',\pi'}$ is obtained from 
the situation $\tup{u,\pi}$ by the instruction $\alpha$. 
For the move instructions
$\alpha\!=\!\stay$,
$\alpha\!=\!\up$, and 
$\alpha\!=\!\down_i$
the pebble stack does not change, i.e., $\pi'=\pi$,  
and the new node $u'$
equals $u$, is the parent of $u$, or is the
$i$-th child of $u$, respectively.
For the pebble instructions
the node does not change, i.e., $u'=u$.
When $\alpha\!=\!\drop_c$, 
$\cA$~drops a pebble with colour $c$ on the current
node, thus  
the node-colour pair $\pair{u,c}$
is pushed onto the pebble stack~$\pi$, i.e., $\pi'=\pi(u,c)$, 
unless $c$ is a visible colour and the stack already
contains a pebble of that colour or already 
contains $k$ visible pebbles,
in which case the rule is not applicable.\footnote{To be precise,
the rule is not applicable if $c\in C_\mathrm{v}$, $\pi=(u_1,c_1)\cdots(u_n,c_n)$, and 
there exists $i\in[1,n]$ such that $c=c_i$, or $\#(\{i\in[1,n]\mid c_i\in C_\mathrm{v}\})=k$.
}
When $\alpha\!=\!\lift_c$, 
$\cA$~lifts a pebble with colour $c$ from the current node, 
only allowed if the topmost element of the
pebble stack is the pair $\pair{u,c}$, 
which is subsequently popped from the stack, i.e., $\pi=\pi'(u,c)$; 
otherwise this rule is not applicable.
We will also allow instructions like 
$\;\lift_c\,;\up\;$ with the obvious meaning
(first lift the pebble, then move up).
In this way we have defined the binary relation $\Rightarrow_{t,\cA}$ on $\con(t)$,
which represents the computation steps of~$\cM$. 
We will say informally that a computation step of $\cM$ \emph{halts successfully} 
if it leads to a final configuration. 

The \emph{tree language} $L(\cA)$ \emph{accepted by} \abb{pta}
$\cA$ consists of all ranked trees $t$ over~$\Sigma$
such that $\cA$ has a successful computation on $t$  
that starts in an initial configuration.
Formally, $L(\cA)=\{t\in T_\Sigma\mid 
\exists\, q_0\in Q_0, q_\infty\in F, \tup{u,\pi}\in \sit(t): 
\tup{q_0,\rt_t,\epsilon}\Rightarrow^*_{t,\cA} \tup{q_\infty,u,\pi}\}$.
Note that pebbles may remain in the final configuration 
and that the head need not return to the root.
Two \abb{PTA}'s $\cA$ and $\cB$ are \emph{equivalent} if $L(\cA)=L(\cB)$. 

By \abb{v$_k$i-pta} we denote a \abb{pta}
with last component $k$, i.e., 
that uses at most $k$ visible pebbles in its computations,
and an unbounded number of invisible pebbles,
and by \VIPTA{k} 
we denote the class of tree languages accepted by \abb{v$_k$i-pta}'s.
For $k=0$, automata that only use invisible pebbles, 
we also use the notation \abb{i-pta}, 
and for automata that only use $k$ visible pebbles
we use \abb{v$_k$-pta}. 
Moreover, \abb{ta} is used for tree-walking automata without
pebbles, i.e., \abb{v$_0$-pta}.
The lower case d or \family{d} is added 
when we only consider \emph{deterministic} automata,
which have a unique initial state, 
no final state in the left-hand side of a rule, and 
no two rules with the same left-hand side.  
Thus we have \abb{v$_k$i-}d\abb{pta}, \VIdPTA{k}, and variants.

\newpage
\smallpar{Properties of automata}
It is natural, and sometimes useful, to extend the \abb{v$_k$i-pta}
with the facility to test whether its pebble stack is nonempty, 
and if so, to test the colour of the topmost pebble. 
Thus, we define a \abb{pta} \emph{with stack tests} in the same way 
as an ordinary \abb{pta} 
except that its rules are of the form 
$\tup{q,\sigma,j,b,\gamma} \to \tup{q',\alpha}$ 
with $\gamma\in C\cup\{\epsilon\}$.  
Such a rule is relevant to a configuration $\tup{q,u,\pi}$ if, in addition, 
the pebble stack $\pi$ is empty if $\gamma=\epsilon$, and
the topmost pebble of $\pi$ has colour~$\gamma$ if $\gamma\in C$.\footnote{To be precise, for
$\pi=(u_1,c_1)\cdots(u_n,c_n)$ the requirements are the following: 
If $\gamma=\epsilon$ then $n=0$, i.e., $\pi=\epsilon$. 
If $\gamma\in C$ then $n\geq 1$ and $c_n=\gamma$. 
} 
All other definitions are the same. 
Note that, obviously, we may require for the above rule that 
$\gamma = c$ if $\alpha=\lift_c$,
which ensures that relevant rules with a lift-instruction 
are always applicable.\footnote{Additionally, we can require 
the following: If $\gamma=\epsilon$ then $b=\nothing$.
If $b\cap C_\mathrm{i}=\{c\}$ then $\gamma=c$. 
}

It is not difficult to see that these new tests do not extend 
the expressive power of the \abb{pta}. Informally we will say that 
the \abb{v$_k$i-pta} can \emph{perform stack tests}. 

\begin{lemma}\label{lem:stacktests} 
Let $k\geq 0$. For every \abb{v$_k$i-pta} with stack tests $\cA$ 
an equivalent (ordinary) \abb{v$_k$i-pta} $\cA'$ 
can be constructed in polynomial time. 
The construction preserves determinism and
the absence of invisible pebbles.\footnote{In other words, 
the statement of the lemma also holds for \abb{v$_k$i-}d\abb{pta}, 
\abb{v$_k$-pta} and \abb{v$_k$-}d\abb{pta}.} 
\end{lemma}

\begin{proof}
Let $\cA = (\Sigma, Q, Q_0, F, C, C_\mathrm{v}, C_\mathrm{i}, R,k)$.
The new automaton $\cA'$ stepwise simulates $\cA$ and, additionally, stores 
in its finite state whether or not the pebble stack is nonempty, and if so, 
what is the colour in $C$ of the topmost pebble. 
Thus, $Q'=Q\times (C\cup\{\epsilon\})$, 
$Q_0'= Q_0\times \{\epsilon\}$, 
and $F'=F\times (C\cup\{\epsilon\})$.
Moreover, the colour sets of $\cA'$ are 
$C_\mathrm{v}' = C_\mathrm{v} \times (C\cup\{\epsilon\})$
and $C_\mathrm{i}' = C_\mathrm{i} \times (C\cup\{\epsilon\})$.
In fact, if the pebble stack of $\cA$ is 
$\pi=(u_1,c_1)(u_2,c_2)\cdots(u_n,c_n)$, with $(u_n,c_n)$ being the topmost pebble,
then the stack of $\cA'$ is 
$\pi'=(u_1,(c_1,\epsilon))(u_2,(c_2,c_1))\cdots(u_n,(c_n,c_{n-1}))$, 
where $\epsilon$ is viewed as a bottom symbol. Thus, the new colour of a pebble 
contains its old colour together with the old colour of the previously dropped pebble
(or $\epsilon$ if there is none). This allows $\cA'$ to update 
its additional finite state component when $\cA$ lifts a pebble. 
More precisely, when $\cA$ is in configuration $\tup{q,u,\pi}$, 
the automaton $\cA'$ is in configuration $\tup{(q,\gamma),u,\pi'}$, 
where $\gamma=c_n$ if $n\geq 1$ and $\gamma=\epsilon$ otherwise.

The rules of $\cA'$ are defined as follows. 
Let $\tup{q,\sigma,j,b,\gamma} \to \tup{q',\alpha}$ be a rule of $\cA$,
and let $b'$ be (the graph of) a mapping from $b$ to $C\cup\{\epsilon\}$.
If $\alpha$ is a move instruction, then $\cA'$ has the rule 
$\tup{(q,\gamma),\sigma,j,b'} \to \tup{(q',\gamma),\alpha}$.
If $\alpha = \drop_c$, then $\cA'$ has the rule
$\tup{(q,\gamma),\sigma,j,b'} \to \tup{(q',c),\drop_{(c,\gamma)}}$. 
If $\alpha=\lift_c$, $\gamma=c$, and $(c,\gamma')\in b'$,
then $\cA'$ has the rule  
$\tup{(q,\gamma),\sigma,j,b'} \to \tup{(q',\gamma'),\lift_{(c,\gamma')}}$.

It should be clear that the construction of $\cA'$ takes polynomial time. 
Note that $k$ is fixed and $\#(b)\leq k+1$ in the left-hand side of 
the rule $\tup{q,\sigma,j,b,\gamma} \to \tup{q',\alpha}$ of $\cA$. 
\end{proof}

\abb{PTA}'s with stack tests will only be used 
in Sections~\ref{sec:look-ahead} and~\ref{sec:variations}.
The next two properties of \abb{PTA}'s will not be used in later sections, 
but are meant to clarify some of the details in the semantics of the \abb{pta}. 

A rule of a \abb{v$_k$i-pta} $\cA$ is \emph{progressive} if it is applicable 
to every reachable configuration\footnote{The configuration 
$\tup{q,u,\pi}$ on the tree $t$ is \emph{reachable} if 
$\tup{q_0,\rt_t,\epsilon}\Rightarrow^*_{t,\cA} \tup{q,u,\pi}$
for some $q_0\in Q_0$.
\label{foot:reachcon}
}
to which it is relevant.
The \abb{v$_k$i-pta} $\cA$ is progressive if all its rules are progressive. 
Intuitively this means that $\cA$ knows that 
its instructions can always be executed.
Clearly, according to the syntax of a \abb{pta}, 
every rule with a move instruction is progressive. 
The same is true for rules with a 
pebble instruction $\drop_c$ or $\lift_c$ with $c\in C_\mathrm{i}$:
an invisible pebble can always be dropped
and an observable invisible pebble can always be lifted.
Thus, only the dropping and lifting of visible pebbles is problematic.
It is easy to see that, for the \mbox{\abb{v$_k$i-pta}}~$\cA'$ constructed in the 
proof of Lemma~\ref{lem:stacktests}, every rule with a lift-instruction
is progressive.

A \abb{v$_k$i-pta} $\cA$ is \emph{counting} if 
$C_v=[1,k]$ and, in each reachable configuration, 
the colours of the visible pebbles on the tree are $1,\dots,\ell$ for some $\ell\in [0,k]$,
in the order in which they were dropped.\footnote{To be precise,  
for $\pi=(u_1,c_1)\cdots(u_n,c_n)$ we require that there exists $\ell\in [0,k]$
such that $(c_{i_1},\dots,c_{i_m})=(1,\dots,\ell)$ where 
$\{i_1,\dots,i_m\}=\{i\in[1,n]\mid c_i\in C_\mathrm{v}\}$ and $i_1<\cdots < i_m$. 
}
Note that in the litterature \mbox{\abb{v$_k$-pta}'s}
are usually counting. We have chosen to allow arbitrarily many visible colours
in a \abb{v$_k$i-pta} 
because we want to be able to store information in the pebbles,
as in the proof of Lemma~\ref{lem:stacktests}. 
It is straightforward to construct 
an equivalent counting \abb{v$_k$i-pta} $\cA'$ 
for a given \abb{v$_k$i-pta} $\cA$ (preserving determinism
and the absence of invisible pebbles). 
The automaton $\cA'$ stepwise simulates $\cA$ and, 
additionally, stores 
in its finite state the colours of the visible pebbles that are dropped on the tree,
in the order in which they were dropped. Thus, the states of $\cA'$ are of the form 
$(q,\phi)$ where $q$ is a state of $\cA$ and 
$\phi$ is a string over $C_\mathrm{v}$ without repetitions, of length at most $k$. 
The state $(q,\phi)$ is final if $q$ is final.
The initial states are $(q,\epsilon)$ where $q$ is an initial state of $\cA$.
The rules of $\cA'$ are defined as follows. 
Let $\tup{q,\sigma,j,b} \to \tup{q',\alpha}$ be a rule of $\cA$
and let $(q,\phi)$ be a state of $\cA'$ such that 
every $c\in b\cap C_\mathrm{v}$ occurs in $\phi$.
Moreover, let $b'\subseteq[1,k]\cup C_\mathrm{i}$ be obtained from $b$ by changing every 
$c\in C_\mathrm{v}$ into $i$, if $c$ is the $i$-th element of $\phi$. 
If $\alpha$ is a move instruction, 
or a pebble instruction $\drop_c$ or $\lift_c$ with $c\in C_\mathrm{i}$
then $\cA'$ has the rule 
$\tup{(q,\phi),\sigma,j,b'} \to \tup{(q',\phi),\alpha}$.
If $\alpha = \drop_c$ with $c\in C_\mathrm{v}$,  
$c$ does not occur in $\phi$, and $|\phi|<k$, 
then $\cA'$ has the rule
$\tup{(q,\phi),\sigma,j,b'} \to \tup{(q',\phi c),\drop_{|\phi|+1}}$. 
Finally, if $\alpha=\lift_c$ with $c\in C_\mathrm{v}$, 
and $\phi=\phi'c$ for some $\phi'\in C^*_\mathrm{v}$, 
then $\cA'$ has the rule  
$\tup{(q,\phi),\sigma,j,b'} \to \tup{(q',\phi'),\lift_{|\phi|}}$.
It should be clear that $\cA'$ is counting. 
Note also that all rules of $\cA'$ with a drop-instruction are progressive.
Thus, if we first apply the construction
in the proof of Lemma~\ref{lem:stacktests} and then the one above, 
we obtain an equivalent progressive \abb{v$_k$i-pta}.
Obviously, every progressive \abb{v$_k$i-pta} can be turned into an equivalent
\abb{v$_{k+1}$i-pta} by simply changing its last component $k$ into $k+1$, and hence 
$\VIPTA{k} \subseteq \VIPTA{k+1}$ and $\VIdPTA{k} \subseteq \VIdPTA{k+1}$.\footnote{In fact, 
these four classes are equal, as will be shown in Theorem~\ref{thm:regt}. 
}

\smallpar{Transducers}
A \emph{tree-walking tree transducer with nested pebbles}
(abbreviated \abb{ptt})
is a \abb{pta} without final states that additionally  
produces an output tree
over a ranked alphabet~$\Delta$.
Thus, omitting $F$, it is specified as a tuple 
$\cM = (\Sigma, \Delta, Q, Q_0,C, C_\mathrm{v}, C_\mathrm{i}, R,k)$,
where $\Sigma$, $Q$, $Q_0$, $C$, $C_\mathrm{v}$, $C_\mathrm{i}$, 
and~$k$ are as for the \abb{pta}. The rules of $\cM$ in the finite set $R$
are of the same form as for the \abb{pta}, except that 
$\cM$ additionally has \emph{output rules} of the form
$\tup{q,\sigma,j,b} \to 
\delta(\,\tup{q_1,\mathrm{stay}}, \dots, \tup{q_m,\mathrm{stay}}\,)$
with $\delta\in \Delta$, and 
$q_1,\dots , q_m \in Q$, 
where~$m$ is the rank of $\delta$.
Intuitively, the output tree is produced recursively. 
In other words, in a configuration to which the above output rule is relevant
(defined as for the \abb{pta}) the \abb{ptt}
$\cM$ outputs $\delta$, and for each child
$\tup{q_i,\mathrm{stay}}$ branches into a new process,
a copy of itself started in state $q_i$ at the current node,
retaining the same
stack of pebbles; thus, the stack is copied 
$m$ times. 
Note that a relevant output rule 
is always applicable. 
As a shortcut 
we may replace the stay-instruction in any $\tup{q_i,\mathrm{stay}}$
by another move instruction or a pebble instruction,
with obvious semantics. 

An \emph{output form} of the \abb{ptt} $\cM$ on ranked tree $t$ over $\Sigma$
is a tree in $T_\Delta(\con(t))$, where $\con(t)$ is defined as for the \abb{pta}. 
Intuitively, such an output form consists on the one hand of $\Delta$-labeled nodes 
that were produced by $\cM$ previously in the computation, using output rules, 
and on the other hand of leaves that represent the independent copies of $\cM$ 
into which the computation has branched previously, due to those output rules, 
where each leaf is labeled by the current configuration of that copy. 
Note that $\con(t)\subseteq T_\Delta(\con(t))$, i.e., 
every configuration of $\cM$ is an output form. 

The computation steps of the \abb{ptt} $\cM$ lead from one output form to another.
Let $s$ be an output form and let $v$ be a leaf of $s$ 
with label $\tup{q,u,\pi}\in \con(t)$.
If $\tup{q,u,\pi} \Rightarrow_{t,\cM} \tup{q',u',\pi'}$,
where the binary relation $\Rightarrow_{t,\cM}$ on $\con(t)$ 
is defined as for the \abb{pta} (disregarding the output rules of $\cM$), 
then we write $s \Rightarrow_{t,\cM} s'$ where $s'$ is obtained from $s$ 
by changing the label of $v$ into $\tup{q',u',\pi'}$. Moreover, 
for every output rule $\tup{q,\sigma,j,b} \to 
\delta(\,\tup{q_1,\mathrm{stay}}, \dots, \tup{q_m,\mathrm{stay}}\,)$
that is relevant to configuration $\tup{q,u,\pi}$,  
we write $s \Rightarrow_{t,\cM} s'$ where $s'$ is obtained from $s$ 
by replacing the node $v$ by the subtree 
$\delta(\tup{q_1,u,\pi},\dots,\tup{q_m,u,\pi})$.
In the particular case that $m=0$, $s'$ is obtained from $s$ by changing the label 
of $v$ into $\delta$. In that case we will say informally 
that $\cM$ \emph{halts successfully},
meaning that the copy of $\cM$ corresponding to the node $u$ of $s$ disappears. 
In this way we have extended $\Rightarrow_{t,\cM}$ to 
a binary relation on $T_\Delta(\con(t))$.

The \emph{transduction} $\tau_\cM$ \emph{realized by} $\cM$ consists of all pairs of trees 
$t$ over $\Sigma$ and $s$ over $\Delta$ such that $\cM$ has a (successful) computation on $t$ 
that starts in an initial configuration and ends with $s$. Formally, we define
$\tau_\cM = \{(t,s)\in T_\Sigma\times T_\Delta\mid 
\exists\, q_0\in Q_0: \tup{q_0,\rt_t,\epsilon}\Rightarrow^*_{t,\cM} s\}$. 
Two \abb{ptt}'s $\cM$ and $\cN$ are \emph{equivalent} if $\tau_\cM=\tau_\cN$.

The \emph{domain of} $\cM$ is defined to be the domain of $\tau_\cM$, 
i.e., the tree language
$L(\cM)=\{t\in T_\Sigma\mid \exists\, s\in T_\Delta: (t,s)\in\tau_\cM\}$. 
When $\cM$ is viewed as a recognizer of its domain, it is actually 
the same as an alternating \abb{pta}. Existential states in the alternation
correspond to the nondeterminism of the \abb{ptt},
universal states correspond to the recursive way in which
output trees are generated. More precisely, an output rule 
$\tup{q,\sigma,j,b}\to\delta(\,\tup{q_1,\mathrm{stay}}, \dots, \tup{q_m,\mathrm{stay}}\,)$
corresponds to a universal state $q$ that requires every state $q_i$ to have
a successful computation (and the output symbol $\delta$ is irrelevant).
An ordinary (non-alternating) \abb{pta} then corresponds to a \abb{ptt} 
for which every output symbol has rank 0; for $m=0$ the above output rule 
means that the \abb{pta} halts in a final state. 
We say that the \abb{ptt} $\cM$ is \emph{total} if $L(\cM)=T_\Sigma$, i.e., 
$\tau_\cM(t)\neq\emptyset$ for every input tree $t$.  

Similar to the notation \VIPTA{k} for tree languages,
we use the notation \VIPTT{k} for the class of transductions
defined by tree-walking tree transducers
 with $k$ visible
nested pebbles and an
unbounded number of invisible pebbles, 
as well as
the obvious variants
\VPTT{k}, and \IPTT.
Additionally \TT\ denotes the class of transductions 
realized by tree-walking tree transducers 
without pebbles, i.e., \VPTT{0}.
Such a transducer is specified as a tuple 
$\cM = (\Sigma, \Delta, Q, Q_0, R)$,
and the left-hand sides of its rules are written $\tup{q,\sigma,j}$, 
omitting $b=\nothing$. 
As for \abb{pta}'s, lower case \family{d} is added for
\emph{deterministic} transducers,
which have a unique initial state and 
no two rules with the same left-hand side. 
Moreover, lower case \family{td} is used for \emph{total deterministic} transducers,
i.e., transducers that are both total and deterministic.
Note that a deterministic \abb{ptt} realizes a function,
and a total deterministic \abb{ptt} a total function from $T_\Sigma$ to $T_\Delta$. 

\smallpar{Properties of transducers}
Stack tests are defined for the \abb{ptt} as for the \abb{pta},
and Lemma~\ref{lem:stacktests} and its proof carry over to \abb{ptt}'s.
If a given \abb{ptt} $\cM$ has the output rule $\tup{q,\sigma,j,b,\gamma} \to 
\delta(\tup{q_1,\stay},\dots,\tup{q_m,\stay})$,
and $b'$ is (the graph of) a mapping from~$b$ to $C\cup\{\epsilon\}$
as in the proof for \abb{pta}'s,
then the constructed \abb{ptt}~$\cM'$ has the rule $\tup{(q,\gamma),\sigma,j,b'} \to 
\delta(\tup{(q_1,\gamma),\stay},\dots,\tup{(q_m,\gamma),\stay})$. 

Progressive \abb{ptt}'s can be defined as for \abb{pta}'s,
based on the notion of a reachable configuration, cf. footnote~\ref{foot:reachcon}.
An output form $s$ of the \abb{ptt} $\cM$ on the input tree $t$ is \emph{reachable} if 
$\tup{q_0,\rt_t,\epsilon}\Rightarrow^*_{t,\cM} s$
for some $q_0\in Q_0$.   
A~configuration of $\cM$ on $t$ is \emph{reachable} 
if it occurs in some reachable output form of $\cM$ on $t$. 
Note that every \abb{i-ptt} is progressive. 

Also, counting \abb{ptt}'s can be defined as for \abb{pta}'s.
For every \abb{v$_k$i-ptt} $\cM$ an equivalent counting \abb{v$_k$i-ptt} $\cM'$
can be constructed, just as for \abb{pta}'s.
If $\tup{q,\sigma,j,b,\gamma} \to \delta(\tup{q_1,\stay},\dots,\tup{q_m,\stay})$ 
is an output rule of $\cM$,
and $\phi$ and $b'$ are as in the proof for \abb{pta}'s,
then $\cM'$ has the rule $\tup{(q,\phi),\sigma,j,b'} \to 
\delta(\tup{(q_1,\phi),\stay},\dots,\tup{(q_m,\phi),\stay})$. 
Thus, as for \abb{pta}'s, every \abb{v$_k$i-ptt}
can be turned into an equivalent progressive \abb{v$_k$i-ptt},
with determinism and the absence of invisible pebbles preserved. That implies that 
$\VIPTT{k} \subseteq \VIPTT{k+1}$ and $\VIdPTT{k} \subseteq \VIdPTT{k+1}$. 

\medskip
We end this section with an example of an \abb{i-ptt}. 

\begin{example}\label{ex:siberie}\leavevmode 
We want to generate itineraries for a trip 
along the Trans-Siberian Railway, starting in Moscow
and ending in Vladivostok, and optionally visiting
some cities along the way. An XML document
lists all the stops:
\begin{small}
\begin{verbatim}
<stop name="Moscow" large="1" initial="1">
 ...
  <stop name="Birobidzhan" large="0">
   ...
    <stop name="Vladivostok" large="1" final="1" />
   ...
  </stop>
 ...
</stop>
\end{verbatim}
\end{small}
The initial and final stops are marked, 
and for every stop the
\texttt{large} attribute indicates whether or not the
stop is in a large city. 
We want to generate a list 
\begin{small}
\begin{verbatim}
<result>it-1
        <result>it-2
            ...
                <result>it-n
                        <endofresults />
                </result>
            ...
        </result>
</result>
\end{verbatim}
\end{small}
where {\small\texttt{it-1,it-2,...,it-n}} are all itineraries
(i.e., lists of stops) that satisfy the constraint
that one does not visit a small city twice in a row. 
An example input XML document, with the corresponding
output XML document is given in Tables~\ref{tab:input} 
and~\ref{tab:output} (where, e.g.,  
{\small\verb+</stop>^3+} abbreviates {\small\verb+</stop></stop></stop>+}).
A deterministic \abb{i-ptt} $\cM_{\text{sib}}$ is able to perform this 
XML transformation by systematically enumerating all
possible lists of stops, marking each stop in the list
(except the initial and final stop) by a pebble.
Since the pebbles are invisible, $\cM_{\text{sib}}$ constructs a 
possible list of stops on the pebble stack 
\emph{in reverse}, so that the stops will appear in the
output tree in the correct order.

\begin{table*}[ht]
\begin{small}
\begin{verbatim}
     <?xml version="1.0" encoding="UTF-8"?>
     <?xml-stylesheet type="text/xsl" href="transsiberie.xsl"?>

     <stop name="Moscow" large="1" initial="1">
       <stop name="Stop 2" large="0">
         <stop name="Stop 3" large="0">
           <stop name="LargeStop 4" large="1">
             <stop name="Stop 5" large="0">
               <stop name="Vladivostok" large="1" final="1"/>
     </stop>^5
\end{verbatim}
\end{small}
\caption{Input}\label{tab:input}
\end{table*}

\begin{table*}[p]
\begin{scriptsize}
\begin{verbatim}
<result>
  <stop name="Moscow" large="1" initial="1">
    <stop name="Stop 3" large="0">
      <stop name="LargeStop 4" large="1">
        <stop name="Stop 5" large="0">
          <stop name="Vladivostok" large="1" final="1"/>
  </stop>^4
  <result>
    <stop name="Moscow" large="1" initial="1">
      <stop name="Stop 2" large="0">
        <stop name="LargeStop 4" large="1">
          <stop name="Stop 5" large="0">
            <stop name="Vladivostok" large="1" final="1"/>
    </stop>^4
    <result>
      <stop name="Moscow" large="1" initial="1">
        <stop name="LargeStop 4" large="1">
          <stop name="Stop 5" large="0">
            <stop name="Vladivostok" large="1" final="1"/>
      </stop>^3
      <result>
        <stop name="Moscow" large="1" initial="1">
          <stop name="Stop 5" large="0">
            <stop name="Vladivostok" large="1" final="1"/>
        </stop>^2
        <result>
          <stop name="Moscow" large="1" initial="1">
            <stop name="Stop 3" large="0">
              <stop name="LargeStop 4" large="1">
                <stop name="Vladivostok" large="1" final="1"/>
          </stop>^3
          <result>
            <stop name="Moscow" large="1" initial="1">
              <stop name="Stop 2" large="0">
                <stop name="LargeStop 4" large="1">
                  <stop name="Vladivostok" large="1" final="1"/>
            </stop>^3
            <result>
              <stop name="Moscow" large="1" initial="1">
                <stop name="LargeStop 4" large="1">
                  <stop name="Vladivostok" large="1" final="1"/>
              </stop>^2
              <result>
                <stop name="Moscow" large="1" initial="1">
                  <stop name="Stop 3" large="0">
                    <stop name="Vladivostok" large="1" final="1"/>
                </stop>^2
                <result>
                  <stop name="Moscow" large="1" initial="1">
                    <stop name="Stop 2" large="0">
                      <stop name="Vladivostok" large="1" final="1"/>
                  </stop>^2
                  <result>
                    <stop name="Moscow" large="1" initial="1">
                      <stop name="Vladivostok" large="1" final="1"/>
                    </stop>
                    <endofresults/>
                  </result>
                </result>
</result>^8
\end{verbatim}
\end{scriptsize}
\caption{Output}\label{tab:output}
\end{table*}

Since in this example the XML tags are ranked,
there is no need for a binary encoding of the XML documents.
The input alphabet $\Sigma$ of $\cM_{\text{sib}}$ consists of all
{\small\texttt{<stop~at>}} where {\small\texttt{at}} is a possible
value of the attributes. 
The rank of \mbox{{\small\texttt{<stop at>}}} is 0 if 
{\small\texttt{final="1"}} and 1 otherwise. The
output alphabet $\Delta$ consists of $\Sigma$, the tag
$r =\;${\small\texttt{<result>}} of rank~$2$, and the tag 
$e =\;${\small\texttt{<endofresults>}} of rank~0. 
The set of pebble colours is $C=C_{\mathrm i}=\{0, 1\}$, with $C_\mathrm{v}=\nothing$.
The transducer $\cM_{\text{sib}}$ will not use the attribute {\small\texttt{initial}},
as it can recognize the root by its child number~0. 
Also, it will disregard the attribute {\small\texttt{large}} 
of the initial and the final stop, and always consider them as large cities.
The set of states of $\cM_{\text{sib}}$ is 
$Q=\{q_\mathrm{start},q_1,q_0,q_\mathrm{out},q_\mathrm{next}\}$
with $Q_0=\{q_\mathrm{start}\}$. 

In the rules below the variables range over the following values:
$\sigma_0 \in \Sigma^{(0)}$,
$\sigma_1 \in \Sigma^{(1)}$,
$j,c\in \{0,1\}$,
and, 
for $i\in \{0,1\}$, 
$\lambda_i\in \{\texttt{\small{<stop at>}}\in\Sigma \mid
\texttt{\small{large}="}i\texttt{"}\}$.
The \abb{i-ptt} $\cM_{\text{sib}}$ first walks from Moscow to Vladivostok
in state $q_\mathrm{start}$:

$\tup{q_\mathrm{start},\sigma_1,j,\nothing} \to \tup{q_\mathrm{start},\mathrm{down}_1}$

$\tup{q_\mathrm{start},\sigma_0,1,\nothing} \to 
\tup{q_1,\mathrm{up}}$

\noindent
State $q_c$ remembers whether the most recently marked city
is small or large; when a new city is marked with a pebble,
it gets the colour $c$. In states $q_0$ and $q_1$ 
as many cities are marked as possible (in the second rule,
$c=1$ or $i=1$): 

$\tup{q_0,\lambda_0,1,\nothing} \to 
\tup{q_0,\mathrm{up}}$

$\tup{q_c,\lambda_i,1,\nothing} \to 
\tup{q_i,\mathrm{drop}_c;\mathrm{up}}$

$\tup{q_c,\sigma_1,0,\nothing} \to 
r(\tup{q_\mathrm{out},\mathrm{stay}}, 
  \tup{q_\mathrm{next},\mathrm{down}_1})$

\noindent
In state $q_\mathrm{out}$ an itinerary is generated as output,
while state $q_\mathrm{next}$ continues the search for itineraries
by unmarking the most recently marked city:

$\tup{q_\mathrm{out},\sigma_1,0,\nothing} \to 
\sigma_1(\tup{q_\mathrm{out},\mathrm{down}_1})$

$\tup{q_\mathrm{out},\sigma_1,1,\nothing} \to 
\tup{q_\mathrm{out},\mathrm{down}_1}$

$\tup{q_\mathrm{out},\sigma_1,1,\{c\}} \to 
\sigma_1(\tup{q_\mathrm{out},\mathrm{lift_c;down}_1})$

$\tup{q_\mathrm{out},\sigma_0,1,\nothing} \to \sigma_0$

$\tup{q_\mathrm{next},\sigma_1,1,\nothing} \to 
\tup{q_\mathrm{next},\mathrm{down}_1}$

$\tup{q_\mathrm{next},\sigma_1,1,\{c\}} \to 
\tup{q_c,\mathrm{lift_c;up}}$

$\tup{q_\mathrm{next},\sigma_0,1,\nothing} \to e$

\noindent
Note that this XML transformation cannot be realized by 
a \abb{v-ptt}, because the height of the output tree is,
in general, exponential in the size of the input tree,
whereas it is polynomial for \abb{v-ptt}'s 
(cf.~\cite[Lemma~7]{EngMan03}). 
\end{example}

\section{Decomposition}\label{sec:decomp}

In this section we decompose every \abb{ptt} into a sequence of
\abb{tt}'s, i.e., transducers without pebbles.
This is useful as it will give us information on 
the domain of a \abb{ptt}, see Theorem~\ref{thm:regt}, 
and on the complexity of typechecking the \abb{ptt}, 
see Theorem~\ref{thm:typecheck}.

It is possible to reduce the number of visible pebbles used,
by preprocessing the input tree with a total deterministic \abb{tt}.
This was shown in 
\cite[Lemma~9]{EngMan03}
for transducers with only visible pebbles.
The basic idea of that proof can be extended to include
invisible pebbles.  

\begin{lemma}\label{lem:decomp}
Let $k\ge 1$. For every \abb{v$_k$i-ptt} $\cM$ a total deterministic \abb{tt} $\cN$ 
and a \mbox{\abb{v$_{k-1}$i-ptt}} $\cM'$ 
can be constructed in polynomial time such that $\tau_{\cN}\circ\tau_{\cM'}= \tau_{\cM}$.
If $\cM$ is deterministic, then so is $\cM'$. Hence,
for every $k\ge 1$, 
$$\VIPTT{k} \subseteq \tdTT \circ \VIPTT{k-1}
\text{ and }
\VIdPTT{k} \subseteq \tdTT \circ \VIdPTT{k-1}.$$ 
\end{lemma}

\begin{proof}
Let $\cM = (\Sigma, \Delta, Q, Q_0,C, C_\mathrm{v}, C_\mathrm{i}, R,k)$ be a \abb{ptt} 
with $k$ visible pebbles. 
The construction of the \abb{tt} $\cN$ and the \abb{ptt} $\cM'$ with $k-1$ visible pebbles 
is a straightforward extension of the one in \cite[Theorem 5]{Eng09}, 
which slightly differs from the one in the proof of \cite[Lemma~9]{EngMan03}, 
but uses the same basic idea. 
For completeness sake we repeat a large part of the proof of \cite[Theorem 5]{Eng09}, 
adapted to the current formalism. 
The simple idea of the proof is to preprocess the input tree $t\in T_\Sigma$ in such a way that the dropping and lifting of the first visible pebble can be simulated by walking into and out of specific areas of the preprocessed input tree $\p(t)$. This preprocessing is independent of the given pebble tree transducer $\cM$. More precisely, $\p(t)$ is obtained from $t$ by attaching to each node $u$ of $t$, as an additional (last) subtree, a fresh copy of $t$ in which (the copy of) node $u$ is marked; let us denote this subtree by $t_u$. Thus, if $t$ has $n$ nodes, then $\p(t)$ has $n+n^2$ nodes. The subtrees $t_u$ of $\p(t)$ are the ``specific areas'' mentioned above. As long as there are no visible pebbles on~$t$, $\cM'$ stepwise simulates $\cM$ on the original nodes of $t$, which form the ``top level'' of $\p(t)$. When $\cM$ drops the first visible pebble $c$ on node $u$, $\cM'$ enters $t_u$ and walks to the marked node, storing $c$ in its finite state. As long as $\cM$ keeps pebble $c$ on the tree, $\cM'$ stays in $t_u$, stepwise simulating $\cM$ on $t_u$ rather than~$t$. Since $u$ is marked in $t_u$, $\cM$'s pebble $c$ at $u$ is visible to the transducer $\cM'$, not as a pebble but as a marked node. Thus, during this time, $\cM'$ only uses $k-1$ visible pebbles.
When $\cM$ lifts pebble~$c$ from $u$ (and hence all visible pebbles are lifted), $\cM'$ walks from the copy of $u$ out of $t_u$, back to the original node $u$, and continues simulating $\cM$ on the top level of $\p(t)$ until $\cM$ again drops a visible pebble.  
There is one problem: how does $\cM'$ know whether or not pebble~$c$ is on top of the stack 
when $\cM$ tries to lift it? To solve this problem, 
$\cM'$ uses an additional special invisible pebble $\odot$. It drops pebble~$\odot$ at the copy of $u$ and thus knows that pebble~$c$ is at the top of the stack (for $\cM$) when it observes pebble~$\odot$. Thus, at any moment of time, $\cM'$ has the same pebble stack as~$\cM$, except that $c$ is replaced by $\odot$
and, moreover, the (invisible) pebbles below $\odot$ are on the top level of $\p(t)$, 
whereas $\odot$ and the pebbles above it are on $t_u$. 

Unfortunately, this preprocessing cannot be realized by a \abb{tt} (though it can easily be realized by a \abb{v$_1$-ptt}). For this reason we ``fold'' $t_u$ at the node $u$, such that (the marked copy of) $u$ becomes its root; let us denote the resulting tree by $\hat{t}_u$. Roughly, $\hat{t}_u$ is obtained from $t_u$ by inverting the parent-child relationship between the ancestors of $u$ (including $u$), similarly as in the tree traversal algorithm sometimes known as ``link inversion'' \cite[p.562]{Knuth}. Appropriate information is added to the node labels of those ancestors to reflect this inversion. As these changes are local (i.e., each node keeps the same neighbours) and clearly marked in the tree, $\cM'$ can easily reconstruct the unfolded $t_u$, and simulate $\cM$ as before. Note also that, with this change of $\p(t)$, dropping or lifting of the first visible pebble can be simulated by $\cM'$ in one computation step, because the marked copy of $u$ is the last child of the original $u$. 

Now a \abb{tt} $\cN$ can compute $\p(t)$, as follows\footnote{See also \cite[Example~3.7]{MilSucVia03} where $\hat{t}_u$ occurs as ``a complex rotation of the input tree'' $t$, albeit for leaves $u$ only.}. 
It copies $t$ to the output (adding primes to its labels), but when it arrives at node $u$ it additionally outputs the copy $\hat{t}_u$ of $t$ in a side branch of the computation. 
Copying the descendants of $u$ ``down stream'' is an easy recursive task.
To invert the parent-child relationship between the nodes on the path from $u$ to $\rt_t$, 
$\cN$ uses a single process that walks along the nodes
of that path ``up stream'' to the root, inverting the relationships in the copy.
Copies of other siblings of children on the path are connected 
as in $t$, and their descendants are copied ``down stream''.
More precisely, if in $t$ the $i$-th child $v$ of parent $w$ is on the path, then, 
in the output $\hat{t}_u$, 
$v$ has an additional (last) child that corresponds to $w$, and $w$ has the same children 
(with their descendants) as in $t$, except that its $i$-th child is a node that is labeled by 
the bottom symbol $\bot$ of rank~0. For the sake of uniformity, 
$\rt_t$ is also given an additional (last) child, with label $\bot$. 
Note that the nodes of $t$ correspond one-to-one to the non-bottom nodes of $\hat{t}_u$; in particular, the path in $t$ from $u$ to $\rt_t$ corresponds to the
path in $\hat{t}_u$ from its root to the parent of its rightmost leaf. 
The bottom nodes of $\hat{t}_u$ will not be visited by $\cM'$.  

A picture of $\p(t)$ is given in 
Fig.~\ref{fig:encpeb}, where $\hat{t}_u$ is drawn for 
two nodes only. Note that in this picture the root of 
the copy of $t$ (which is also the root of $\p(t)$) is the top of 
the triangle, but the root of $\hat{t}_u$ is $u$ (and, 
of course, similarly for $v$). 
\begin{figure}
\centerline{\begin{tikzpicture}[scale=0.1,>=stealth,shorten >=.6pt]
\tikzset{dot/.style={circle,minimum size=1.4mm,inner sep=0pt,fill=black}}
\node at (45,50) {$t$};
\draw[fill=white] (40,30) -- (60,30) -- (50,50) -- cycle;
\node at (25,30) {$\hat{t}_u$};
\draw[fill=white] (20,10) -- (40,10) -- (30,30) -- cycle;
\node[dot,draw] (A) at (45,35) {};
\node at (45,32) {$u$};
\node[dot,draw] (B) at (25,15) {};
\node at (25,12) {$u$};
\path[draw,->] (A) to (B);
\node at (75,33) {$\hat{t}_v$};
\draw[fill=white] (60,13) -- (80,13) -- (70,33) -- cycle;
\node[dot,draw] (A) at (52,40) {};
\node at (52,37) {$v$};
\node[dot,draw] (B) at (72,23) {};
\node at (72,20) {$v$};
\path[draw,->] (A) to (B);
\end{tikzpicture}}
\caption{Output tree $\p(t)$ of the \abb{tt} $\cN$ of Lemma~\ref{lem:decomp}
for input tree $t$.}\label{fig:encpeb}
\end{figure}
As a concrete example, consider $t= \sigma(\delta(a,b),c)$ where $\sigma,\delta$ have rank~2 and 
$a,b,c$ rank~0. We will name the nodes of $t$ by their labels. Then 
$$\p(t) = \sigma'(\delta'(a'(\hat{t}_a),b'(\hat{t}_b),\hat{t}_\delta),c'(\hat{t}_c),\hat{t}_\sigma)$$ 
where 
$$
\begin{array}{lll}
\hat{t}_a &=& a_{0,1}(\delta_{1,1}(\bot,b,\sigma_{1,0}(\bot,c,\bot))),\\
\hat{t}_b &=& b_{0,2}(\delta_{2,1}(a,\bot,\sigma_{1,0}(\bot,c,\bot))),\\
\hat{t}_\delta &=& \delta_{0,1}(a,b,\sigma_{1,0}(\bot,c,\bot)),\\
\hat{t}_c &=& c_{0,2}(\sigma_{2,0}(\delta(a,b),\bot,\bot)), \mbox{ and}\\ 
\hat{t}_\sigma &=& \sigma_{0,0}(\delta(a,b),c,\bot).
\end{array}
$$
The subscripted node labels are on the rightmost paths of the $\hat{t}_u$'s; the subscripts contain ``reconstruction'' information, to be explained below. As another example, if $t$ is the monadic tree $a(b^m(c(d^n(e))))$ of height $m+n+3$, and $u$ is the \mbox{$c$-labelled} node, then $\hat{t}_u = c_{0,1}(s_1,s_2)$ with 
$s_1=d^n(e)$ and $s_2$ is the binary tree $b_{1,1}(\bot,b_{1,1}(\bot,\dots b_{1,1}(\bot,a_{1,0}(\bot,\bot))\cdots ))$ of height $m+2$. This shows more clearly that $\hat{t}_u$ is obtained by ``folding''.

We now formally define the deterministic \abb{tt} $\cN$ that, for given ranked alphabet $\Sigma$, realizes the preprocessing $\p$ (called EncPeb in \cite{EngMan03}). 
The definition is identical to the one in \cite[Section 6]{Eng09}. 
Since $\cN$ has no pebbles, we abbreviate the left-hand side $\tup{q,\sigma,j,\nothing}$ of a rule by $\langle q,\sigma,j\rangle$. 
To simplify the definition of~$\cN$ we additionally allow output rules of the form 
$\tup{q,\sigma,j}\to \delta(s_1,\dots,s_m)$ where $\delta$ is an output symbol of rank~$m$
and every $s_i$ is either the output symbol $\bot$ or it is of the form $\tup{q',\phi}$
where $\phi$ is $\stay$, $\up$, or $\down_i$ with $i\in[1,m]$. 
Such a rule should be replaced by the rules 
$\langle q,\sigma,j\rangle\to 
\delta(\langle p_1,\stay \rangle,\dots,\langle p_m,\stay \rangle)$ 
and $\langle p_j,\sigma,j\rangle\to s_j$ for all $j\in[1,m]$, where $p_1,\dots,p_m$ 
are new states. Obviously this replacement can be done in quadratic time. 

We introduce the states and rules of $\cN$ one by one; in what follows $\sigma$ ranges over $\Sigma$, with  $m=\rank_\Sigma(\sigma)$, $j$ ranges over $[0,\m_\Sigma]$, and $i$ over $[1,m]$.
First, $\cN$~has an ``identity'' state $d$ that just recursively copies the subtree of the current node to the output, using the rules 
$\langle d,\sigma,j\rangle\to \sigma(\langle d,\down_1\rangle,\dots,\langle d,\down_m\rangle)$.
Then, $\cN$ has initial state $\g$ that copies the input tree $t$ to the output (with primed labels) and at each node $u$ of $t$ ``generates'' a new copy $\hat{t}_u$ of the input tree by calling the state $f$ that computes $\hat{t}_u$ by ``folding'' $t_u$. The rules for $\g$ are 
$$\langle \g,\sigma,j\rangle\to
\sigma'(\langle \g,\down_1\rangle,\dots,\langle \g,\down_m\rangle, \langle f,\stay\rangle).$$
Note that $\sigma'$ has rank $m+1$: the root of $\hat{t}_u$ is attached to $u$ as its last child. 
The rules for $f$ are 
$$\langle f,\sigma,j\rangle\to
\sigma_{0,j}(\langle d,\down_1\rangle,\dots,\langle d,\down_m\rangle, \xi_j)$$
where $\xi_j = \langle f_j,\up\rangle$ for $j\neq 0$, and $\xi_0 = \bot$.
The ``reconstruction'' subscripts of $\sigma_{0,j}$ mean the following:  
subscript $0$ indicates that this node is the root of some $\hat{t}_u$, and subscript $j$ is the child number of $u$ in $t$. 
Note that $\sigma_{0,j}$ has rank $m+1$: its last child corresponds to the parent of $u$ in $t$ 
(viewing $\bot$ as the ``parent'' of $\rt_t$ in $t$).
The \abb{tt} $\cN$ walks up along the path from $u$ to the root of $t$ using ``folding'' states $f_i$, where the $i$ indicates that in the previous step $\cN$ was at the $i$-th child of the current node. The rules for $f_i$ are 
$$
\begin{array}{lll}
\langle f_i,\sigma,j\rangle & \to & \sigma_{i,j}( \\
&& \langle d,\down_1\rangle,\dots,\langle d,\down_{i-1}\rangle, \\
&& \bot, \\
&& \langle d,\down_{i+1}\rangle,\dots,\langle d,\down_m\rangle, \\
&& \xi_j)
\end{array}
$$
where $\xi_j$ is as above. 
If a node (in $\hat{t}_u$) with label $\sigma_{i,j}$ corresponds to the node $v$ in~$t$, then the 
``reconstruction'' subscript $i$ means that its parent corresponds to the $i$-th child of $v$ in $t$
(and its own $i$-th child is $\bot$), and, as above, ``reconstruction'' subscript $j$ is the child number of $v$. Just as $\sigma_{0,j}$, also $\sigma_{i,j}$ has rank $m+1$: its last child corresponds to the parent of $v$ in $t$.
Note that the copy $\hat{t}_u$ of the input tree is computed by the states
$f$, $f_i$ (for every $i$) and $d$, such that $f$ copies node $u$ to the output and 
the other states walk from $u$ to every other node $v$ of $t$ and copy $v$ to the output.
To be precise, $\cN$ walks from $u$ to $v$ along the shortest (undirected) 
path from $u$ to $v$, from $u$ up to the least common ancestor of $u$ and $v$ 
(in the states $f_i$), and then down to $v$ (in the state $d$). Arriving in a node $v$ from a neighbour of $v$, the transducer $\cN$ branches into a new process 
for every other neighbour of $v$. 

This ends the description of the \abb{tt} $\cN$. 
The output alphabet $\Gamma$ of $\cN$ (which will also be the input alphabet of $\cM'$) is the union of $\Sigma$, $\{\bot\}$, $\{\sigma'\mid \sigma\in\Sigma\}$, and $\{\sigma_{i,j}\mid \sigma\in\Sigma, i\in[0,\rank_\Sigma(\sigma)], j\in[0,\m_\Sigma]\}$. Thus, $\cN$ has $O(n^2)$ output symbols, where $n$ is the size of $\Sigma$.\footnote{We assume here that the rank of each symbol of the ranked alphabet $\Sigma$ is specified in unary rather than decimal notation, and thus $\m_\Sigma\leq n$; cf. the last paragraph of \cite[Section~2]{Eng09}.} 
So, since $\m_\Gamma= \m_\Sigma +1$, the size of $\Gamma$ is polynomial in $n$. 
The set of states of $\cN$ is $\{d,\g,f\} \cup \{f_i\mid i\in[1,\m_\Sigma]\}$, with initial state $\g$. Thus, it has $O(n)$ states and $O(n^3)$ rules; moreover, each of these rules is of size $O(n\log n)$. Hence, the size of $\cN$ is polynomial in the size of $\Sigma$, and it can be constructed in polynomial time.  

We now turn to the description of the \abb{v$_{k-1}$i-ptt} $\cM'$. It has input alphabet~$\Gamma$, output alphabet $\Delta$, set of states $Q\cup (Q\times C_{\mathrm{v}})$, 
and the same initial states and visible colours as $\cM$. Its invisible colour set is $C'_{\mathrm{i}} = C_{\mathrm{i}}\cup\{\odot\}$.
It remains to discuss the set $R'$ of rules of $\cM'$.
Let $\langle q,\sigma,j,b\rangle\to \zeta$ be a rule of $\cM$ 
with $\rank_\Sigma(\sigma)=m$. 
We consider four cases, depending on the variant $\sigma'$, $\sigma_{0,j}$, $\sigma_{i,j}$ with $i\neq 0$, or $\sigma$ in $\Gamma$ of the input symbol $\sigma\in\Sigma$. 

In the first case, we consider the behaviour of $\cM'$ in state $q$ on $\sigma'$, 
and we assume that $b\cap C_{\mathrm{v}}=\nothing$. 
If $\zeta = \tup{q',\drop_c}$ with $c\in C_{\mathrm{v}}$, then $R'$ contains the rule 
$\langle q,\sigma',j,b\rangle\to \tup{(q',c),\down_{m+1};\drop_\odot}$,\footnote{To be completely formal, 
this rule should be replaced by the two rules 
$\langle q,\sigma',j,b\rangle\to \tup{p,\down_{m+1}}$ and
$\langle p,\sigma_{0,j},m+1,\nothing\rangle\to \tup{(q',c),\drop_\odot}$, where $p$ is a new state.
}
and otherwise  $R'$ contains the rule $\langle q,\sigma',j,b\rangle\to \zeta$.
Thus, $\cM'$ simulates $\cM$ on the original (now primed) part of the input tree $t$ in $\p(t)$, 
until $\cM$ drops a visible pebble~$c$ on node $u$. Then $\cM'$ steps to the root of $\hat{t}_u$ where it drops the invisible pebble~$\odot$, and stores $c$ in its finite state. 

Next, we let $c\in C_{\mathrm{v}}$ and we consider the behaviour of $\cM'$ in state $(q,c)$ 
on the remaining variants of $\sigma$. 
Let $\zeta_c$ be the result of changing in $\zeta$ every occurrence of a state $q'$ into $(q',c)$. 

In the second case we assume that $c\in b$ (corresponding to the fact that $\sigma_{0,j}$ 
labels the marked node of some $\hat{t}_u$). 
If $b=\{c\}$ and $\zeta=\tup{q',\lift_c}$, then $R'$ contains the rule 
$\tup{(q,c),\sigma_{0,j},m+1,\{\odot\}} \to \tup{q',\lift_\odot;\up}$.\footnote{Again,
to be completely formal, this rule should be replaced by the two rules 
$\tup{(q,c),\sigma_{0,j},m+1,\{\odot\}} \to \tup{p,\lift_\odot}$ and 
$\tup{p,\sigma_{0,j},m+1,\nothing} \to \tup{q',\up}$, where $p$ is a new state.
} 
Thus, when $\cM$ lifts visible pebble~$c$ from node $u$, 
$\cM'$ lifts invisible pebble~$\odot$ and steps from the root of $\hat{t}_u$ back to node $u$. 
Otherwise, $R'$ contains the rules 
$$\tup{(q,c),\sigma_{0,j},m+1,b\setminus\{c\}\cup\{\odot\}} \to \zeta'_c$$ 
(provided $b\cap C_{\mathrm{i}}=\nothing$) and 
$$\tup{(q,c),\sigma_{0,j},m+1,b\setminus\{c\}} \to \zeta'_c,$$
where $\zeta'_c$ is obtained from $\zeta_c$ by changing $\up$ into $\down_{m+1}$. These two rules 
correspond to whether or not the invisible pebble~$\odot$ is observable. 
Note that the child number in $\p(t)$ of a node with label $\sigma_{0,j}$ is always $m+1$ 
(and the label of its parent is $\sigma'$). 

In the remaining two cases we assume that $c\notin b$ in the above rule of $\cM$.
In the third case, we consider $\sigma_{i,j}$ with $i\neq 0$. Then $R'$ contains the rules
$\tup{(q,c),\sigma_{i,j},j',b} \to \zeta'_c$ for every $j'\in[1,\m_\Gamma]$, where 
$\zeta'_c$ is now obtained from $\zeta_c$ by changing $\up$ into $\down_{m+1}$, and $\down_i$ into $\up$. 
In the fourth and final case, we consider $\sigma$ itself (in $\Gamma$). 
Then $R'$ contains the rule $\tup{(q,c),\sigma,j,b} \to \zeta_c$. 
Thus, $\cM'$ stepwise simulates $\cM$ on every $\hat{t}_u$.

This ends the description of the \abb{v$_{k-1}$i-ptt} $\cM'$. It should now be clear that $\tau_{\cM'}(\p(t)) = \tau_{\cM}(t)$ for every $t\in T_\Sigma$, and hence $\tau_{\cN}\circ\tau_{\cM'}= \tau_{\cM}$.
Each rule of $\cM$ is turned into at most 
$1+\#(C_{\mathrm{v}})\cdot(2+\m_\Sigma(\m_\Sigma +1))$ rules of $\cM'$, of the same size as that rule (disregarding the space taken by the occurrences of $c$ and $m+1$). Thus, $\cM'$ can be computed from $\cM$ in polynomial time.
\end{proof}

The tree $\p(t)$ that is used in the previous proof consists
of two levels of copies of the original input tree $t$; on the
first level a straightforward copy of $t$ 
(used until the first visible pebble is dropped)
and a second level of copies $\hat{t}_u$
(used to ``store'' the first visible pebble dropped).
It is tempting to add another level, meant as a way
to store the next visible pebble dropped.
The problem with this is that it would 
make the first visible pebble effectively unobservable
when the next one is dropped.
The idea \emph{can} be used for invisible pebbles,
for arbitrarily many levels.

\begin{lemma}\label{lem:nul-decomp}
For every \abb{i-ptt} $\cM$ a \abb{tt} $\cN$ 
and a \abb{tt} $\cM'$ 
can be constructed in polynomial time such that $\tau_{\cN}\circ\tau_{\cM'}= \tau_{\cM}$.
If $\cM$ is deterministic, then so is~$\cM'$. Hence,
$\IPTT \subseteq \TT \circ \TT$ and
$\IdPTT \subseteq \TT \circ \dTT$.
\end{lemma}

\begin{proof}
The computation of a \abb{PTT} $\cM$ with invisible pebbles 
on tree $t$ is simulated by a \abb{TT} $\cM'$ 
(without pebbles) on tree $t'$. 
The input tree $t$ is preprocessed in a nondeterministic way
by a \abb{TT} $\cN$ to obtain $t'$.
The top level of $t'$ is a copy of $t$, as before. 
On the next level, since the simulating transducer $\cM'$ cannot store 
the colours of all the pebbles in its finite state 
(as we did for one colour in the proof of Lemma~\ref{lem:decomp}),
$\cN$ does not attach one copy $\hat{t}_u$ of $t$ to each node $u$ of $t$
but $\#(C_\mathrm{i})$ such copies, one for each pebble colour. 
In this way, the child number in $t'$ of the root of $\hat{t}_u$ 
represents the pebble colour. 
In fact, in each node $u$ of $t$ 
the transducer $\cN$ nondeterministically
decides for each pebble colour $c$ 
whether or not to spawn a process that copies $t$ into
$\hat{t}_u$, and this is a recursive process:
in each node in each copy of $t$ it can be decided to 
spawn such processes that generate new copies. 

In this way a ``tree of trees'' is constructed. 
For an ``artist impression'' of such an output tree $t'$, see 
Fig.~\ref{fig:artists}.
\begin{figure}
\centerline{\begin{tikzpicture}[scale=0.1,>=stealth,shorten >=.6pt]
\tikzset{dot/.style={circle,minimum size=1mm,inner sep=0pt,fill=black}}
\node at (45,50) {$t$};
\draw[fill=white] (50,50) -- ++(-5,-10) -- ++(10,00) -- cycle;
\draw[fill=white] (38,36) -- ++(-5,-10) -- ++(10,00) -- cycle;
\node[dot,draw] (A) at (48,43){}; \node[dot,draw] (B) at (36,29){}; \path[draw,->] (A) to (B) ;
\draw[fill=white] (35,35) -- ++(-5,-10) -- ++(10,00) -- cycle;
\draw[fill=white] (25,20) -- ++(-5,-10) -- ++(10,00) -- cycle;
\draw[fill=white] (40,20) -- ++(-5,-10) -- ++(10,00) -- cycle;
\draw[fill=white] (35,05) -- ++(-5,-10) -- ++(10,00) -- cycle;
\draw[fill=white] (64,36) -- ++(-5,-10) -- ++(10,00) -- cycle;
\node[dot,draw] (A) at (51,45){}; \node[dot,draw] (B) at (65,31){}; \path[draw,->] (A) to (B) ;
\draw[fill=white] (60,35) -- ++(-5,-10) -- ++(10,00) -- cycle;
\draw[fill=white] (55,20) -- ++(-5,-10) -- ++(10,00) -- cycle;
\draw[fill=white] (84,23) -- ++(-5,-10) -- ++(10,00) -- cycle;
\node[dot,draw] (A) at (61,30){}; \node[dot,draw] (B) at (71,15){}; \path[draw,->] (A) to (B) ;
\node[dot,draw] (A) at (65,31){}; \node[dot,draw] (B) at (85,18){}; \path[draw,->] (A) to (B) ;
\draw[fill=white] (80,22) -- ++(-5,-10) -- ++(10,00) -- cycle;
\node[dot,draw] (A) at (65,31){}; \node[dot,draw] (B) at (81,17){}; \path[draw,->] (A) to (B) ;
\draw[fill=white] (74,21) -- ++(-5,-10) -- ++(10,00) -- cycle;
\node[dot,draw] (A) at (61,30){}; \node[dot,draw] (B) at (75,16){}; \path[draw,->] (A) to (B) ;
\draw[fill=white] (70,20) -- ++(-5,-10) -- ++(10,00) -- cycle;
\draw[fill=white] (65,05) -- ++(-5,-10) -- ++(10,00) -- cycle;
\node[dot,draw] (A) at (51,45){}; \node[dot,draw] (B) at (61,30){}; \path[draw,->] (A) to (B) ;
\node[dot,draw] (A) at (48,43){}; \node[dot,draw] (B) at (33,28){}; \path[draw,->] (A) to (B) ;
\node[dot,draw] (A) at (61,30){}; \node[dot,draw] (B) at (71,15){}; \path[draw,->] (A) to (B) ;
\node[dot,draw] (A) at (36,28){}; \node[dot,draw] (B) at (41,13){}; \path[draw,->] (A) to (B) ;
\node[dot,draw] (A) at (58,28){}; \node[dot,draw] (B) at (53,13){}; \path[draw,->] (A) to (B) ;
\node[dot,draw] (A) at (32,26){}; \node[dot,draw] (B) at (22,11){}; \path[draw,->] (A) to (B) ;
\node[dot,draw] (A) at (26,15){}; \node[dot,draw] (B) at (36,00){}; \path[draw,->] (A) to (B) ;
\node[dot,draw] (A) at (56,15){}; \node[dot,draw] (B) at (66,00){}; \path[draw,->] (A) to (B) ;
\end{tikzpicture}}
\caption{An output tree $t'$ of the \abb{tt} $\cN$ of Lemma~\ref{lem:nul-decomp} 
for input tree $t$.}\label{fig:artists}
\end{figure}
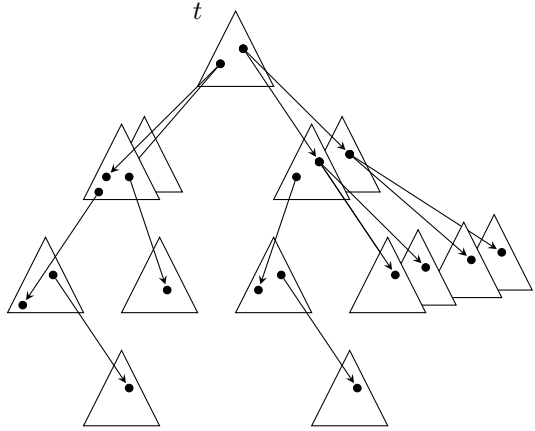
The child number in $t'$ of the root
of each copy $\hat{t}_u$ indicates an invisible pebble
of colour $c$
placed at node $u$ in the original tree $t$.
In each copy only one pebble is observable, the one represented 
by the child number of its root, exactly as the last pebble dropped
in the original computation.
In the simulation, 
moving down or up along the tree of trees corresponds to
dropping and lifting invisible pebbles.

In general there is no bound on the depth of the 
stack of pebbles during a computation of $\cM$.
The preprocessor $\cN$ nondeterministically constructs
$t'$. If $t'$ is not sufficiently deep,
the simulating transducer $\cM'$ aborts the computation.
Conversely, for every computation of $\cM$ a tree $t'$ of sufficient
depth can be constructed nonderministically from $t$.

We now turn to the formal definitions. 
Let $\cM = (\Sigma, \Delta, Q, Q_0,C, C_\mathrm{v}, C_\mathrm{i}, R,0)$ be an \abb{i-ptt}.
Without loss of generality we assume that $C=C_\mathrm{i}$ and that 
$C=[1,\gamma]$ for some $\gamma\in\nat$. This choice of $C$ simplifies 
the representation of colours by child numbers. 

First, we define the nondeterministic \abb{tt} $\cN$ that preprocesses the trees over~$\Sigma$.
It is a straightforward variant of the one in the proof of Lemma~\ref{lem:decomp}.
The output alphabet $\Gamma$ of $\cN$ is now the union of $\{\bot\}$, $\{\sigma'\mid \sigma\in\Sigma\}$, and $\{\sigma'_{i,j}\mid \sigma\in\Sigma, i\in[0,\rank_\Sigma(\sigma)], j\in[0,\m_\Sigma]\}$ where, 
for every $\sigma\in\Sigma$ of rank~$m$, $\sigma'$ has rank $m+\gamma$ and 
$\sigma'_{i,j}$ has rank $m+\gamma+1$,
because $\gamma$ processes are spawned at each node, and each of these processes generates, 
nondeterministically, either a copy $\hat{t}_u$ of $t$ or the bottom symbol $\bot$.
The set of states of $\cN$ is as before, 
except that the state $d$ is removed (with its rules). 
In the rules of $\cN$ we will use $\tup{f,\stay}^\gamma$ as an abbreviation of the sequence
$\tup{f,\stay},\dots,\tup{f,\stay}$ of length $\gamma$. 
The rules for the initial state~$\g$ are 
\[
\begin{array}{lll}
\langle \g,\sigma,j\rangle & \to &
\sigma'(\langle \g,\down_1\rangle,\dots,\langle \g,\down_m\rangle, \langle f,\stay\rangle^\gamma).
\end{array}
\]
The rules for $f$ are 
\[
\begin{array}{lll}
\langle f,\sigma,j\rangle & \to & \bot \\[1mm]
\langle f,\sigma,j\rangle & \to &
\sigma'_{0,j}(\langle\g,\down_1\rangle,\dots,\langle\g,\down_m\rangle,\tup{f,\stay}^\gamma,\xi_j)
\end{array}
\]
where, as before, $\xi_j = \langle f_j,\up\rangle$ for $j\neq 0$, and $\xi_0 = \bot$.
Finally, the rules for $f_i$ are 
\[
\begin{array}{lll}
\langle f_i,\sigma,j\rangle & \to & \sigma'_{i,j}( \\
&& \langle \g,\down_1\rangle,\dots,\langle \g,\down_{i-1}\rangle, \\
&& \bot, \\
&& \langle \g,\down_{i+1}\rangle,\dots,\langle \g,\down_m\rangle, \\
&& \tup{f,\stay}^\gamma, \\
&& \xi_j)
\end{array}
\]
where $\xi_j$ is as above. 
This ends the definition of $\cN$. 

Next, we define the simulating \abb{tt} $\cM'$. 
It has input alphabet~$\Gamma$ (the output alphabet of $\cN$), output alphabet $\Delta$, 
and the same set of states and initial states as $\cM$. 
The set $R'$ of rules of $\cM'$ is defined as follows.
Let $\langle q,\sigma,j,b\rangle\to \zeta$ be a rule of $\cM$ 
with $\rank_\Sigma(\sigma)=m$. 
Note that $b$ is either empty or a singleton. 
We consider three cases, that describe the behaviour of $\cM'$ on the symbols $\sigma'$, 
$\sigma'_{0,j}$, and $\sigma'_{i,j}$ with $i\neq 0$. 

In the first case we assume that $b=\nothing$ 
(and hence $\zeta$ does not contain a lift-instruction). 
Then $R'$ contains the rule $\langle q,\sigma',j\rangle\to \zeta'$
where $\zeta'$ is obtained from~$\zeta$ by changing $\drop_c$ into $\down_{m+c}$
for every $c\in C$.

In the second case we assume that $b=\{c\}$ for some $c\in C$.
Then $R'$ contains the rule $\tup{q,\sigma'_{0,j},m+c} \to \zeta'$
where $\zeta'$ is now obtained from $\zeta$ by changing $\up$ into $\down_{m+\gamma+1}$,
$\lift_c$ into $\up$, and $\drop_d$ into $\down_{m+d}$
for every $d\in C$.
Note that the child number in $t'$ of a node with label $\sigma'_{0,j}$ 
is always $m+c$ for some $c\in C$
(and the label of its parent is $\sigma'$ or $\sigma'_{i,j}$ 
for some $i\in[0,m]$). 

In the third case we assume (as in the first case) that $b=\nothing$.
Then $R'$ contains the rule 
$\langle q,\sigma'_{i,j},j'\rangle\to \zeta'$ for every $j'\in[1,\m_\Gamma]$,
where $\zeta'$ is now obtained from $\zeta$ by changing $\up$ into $\down_{m+\gamma+1}$, 
$\down_i$ into $\up$, and $\drop_c$ into $\down_{m+c}$
for every $c\in C$.

This ends the definition of $\cM'$. It should, again, be clear that 
for every $t\in T_\Sigma$ and $s\in T_\Delta$, 
$s\in \tau_{\cM}(t)$ if and only if there exists $t'\in \tau_{\cN}(t)$ 
such that $s\in \tau_{\cM'}(t')$. 
Hence $\tau_{\cN}\circ\tau_{\cM'}= \tau_{\cM}$.
It is straightforward to show, as in the proof of Lemma~\ref{lem:decomp}, 
that $\cN$ and $\cM'$ can be constructed in polynomial time from~$\cM$. 
Note that $\m_\Gamma = \m_\Sigma + \#(C_{\mathrm{i}})+1$ and so 
the size of $\Gamma$ is polynomial in the size of $\cM$. 
\end{proof}

Combining the previous two results we
can inductively decompose tree tranducers with 
(visible and invisible)
pebbles into tree transducers without pebbles.

\begin{theorem}\label{thm:decomp}
For every $k\ge 0$,
$\VIPTT{k} \subseteq \TT^{k+2}$.
For fixed $k$, the involved construction takes polynomial time. 
\end{theorem}

Observe that  $\VPTT{k} \subseteq \VIPTT{k-1}$
as the topmost pebble can be replaced by an invisible one, thus
$\VPTT{k} \subseteq \TT^{k+1}$,
which was proved in~\cite[Theorem~10]{EngMan03}, also for the deterministic case.

We do not know whether Theorem~\ref{thm:decomp} is optimal, 
i.e., whether or not $\VIPTT{k}$ is included in $\TT^{k+1}$. 
The deterministic version of Theorem~\ref{thm:decomp} (for $k\neq 0$)
will be proved in Section~\ref{sec:variations} (Theorem~\ref{thm:detdecomp}),
and we will show that it is optimal (after Theorem~\ref{thm:dethier}).

The nondeterminism of the ``preprocessing'' transducer $\cN$ 
in the proof of Lemma~\ref{lem:nul-decomp} is rather limited. 
The general form of the constructed tree 
is completely determined by the input tree, 
only the depth of the construction is
nondeterministically chosen.
At the same time it remains nondeterministic
even when we start with a deterministic \abb{pTT}
with invisible pebbles:
$\IdPTT \subseteq \TT \circ \dTT$.
However, we can obtain a 
deterministic transduction if the number of invisible
pebbles used by the transducer is bounded 
(over all input trees), cf. 
the M.~Sc.~Thesis of the third author \cite{Sam}
(where visible and invisible pebbles are called 
global and local pebbles, respectively).
In Section~\ref{sec:poweriptt} we will show that 
if we start with a deterministic tree transduction,
then the inclusions of Lemma~\ref{lem:nul-decomp} also
hold in the other direction (Theorem~\ref{thm:composition}).
In Section~\ref{sec:variations} we will show that $\IdPTT \subseteq \dTT^3$
(Corollary~\ref{cor:idptt-tt3}).

\section{Typechecking}\label{sec:typechecking}

The \emph{inverse type inference problem} is to construct, 
for a tree transducer $\cM$ and a regular tree grammar $G_{\rm out}$, 
a regular tree grammar $G_{\rm in}$ such that $L(G_{\rm in}) = \tau_\cM^{-1}(L(G_{\rm out}))$. 
The \emph{typechecking problem} asks, 
for a tree transducer $\cM$ and two regular
tree grammars $G_{\rm in}$ and $G_{\rm out}$, whether
or not $\tau_\cM(L(G_{\rm in}))\subseteq L(G_{\rm out})$. 
The inverse type inference problem can be used to solve the typechecking problem, 
because $\tau_\cM(L(G_{\rm in}))\subseteq L(G_{\rm out})$ if and only if 
$L(G_{\rm in}) \cap \tau_\cM^{-1}(L'_{\rm out})=\nothing$, where 
$L'_{\rm out}$ is the complement of $L(G_{\rm out})$. 

It was shown in \cite{MilSucVia03} (see also \cite[Section~7]{EngMan03}) 
that both problems are solvable for 
tree-walking tree transducers with visible pebbles, i.e., for \abb{v-ptt}'s, 
and hence in particular for tree-walking tree transducers without pebbles, 
i.e., for \abb{tt}'s.\footnote{Note however that our definition of inverse type inference 
differs from the one in \cite{MilSucVia03}, 
where it is required that $L(G_{\rm in}) = 
\{\;s\mid \tau_\cM(s)\subseteq L(G_{\rm out})\;\}$.  
The reason is that our definition is more convenient 
when considering compositions of tree transducers.
}
This was extended in \cite{Eng09} to compositions of such transducers and,
moreover, the time complexity of the involved algorithms was improved,
using a result of~\cite{Bar} for attributed tree transducers. 

We define a \emph{$k$-fold exponential} function to be a function of the form  
$2^{g(n)}$ where $g$ is a $(k\!-\!1)$-fold
exponential function;
a $0$-fold exponential function is a polynomial.

\begin{proposition}\label{prop:invtypeinf}
For fixed $k\geq 0$, 
the inverse type inference problem is solvable
\\
{\em (1)}
for compositions of $k$ \abb{tt}'s in $k$-fold exponential time, and
\\
{\em (2)}
for \abb{v$_k$-PTT}'s in $(k\!+\!1)$-fold exponential time.
\end{proposition}

\begin{proposition}\label{prop:typecheck}
For fixed $k\geq 0$, 
the typechecking problem is solvable
\\
{\em (1)}
for compositions of $k$ \abb{tt}'s in $(k\!+\!1)$-fold exponential time, and
\\
{\em (2)}
for \abb{v$_k$-PTT}'s in $(k\!+\!2)$-fold exponential time.
\end{proposition}

As also observed in \cite{Eng09}, one exponential can be taken off the results 
of Proposition~\ref{prop:typecheck} if we
assume that $G_{\rm out}$ is a total deterministic bottom-up finite-state tree automaton,
because that exponential is due to the complementation of $L(G_{\rm out})$. 

It is immediate from Theorem~\ref{thm:decomp} and 
Propositions~\ref{prop:invtypeinf}(1) and~\ref{prop:typecheck}(1) 
that both problems are also solvable for 
tree-walking tree transducers with invisible pebbles. 

\begin{theorem}\label{thm:typecheck}
For fixed $k\geq 0$, 
the inverse type inference problem and the typechecking problem are solvable
for \abb{v$_k$i-PTT}'s in $(k\!+\!2)$-fold and $(k\!+\!3)$-fold exponential time,
respectively.
\end{theorem}

The main conclusion from Proposition~\ref{prop:typecheck}(2) and Theorem~\ref{thm:typecheck}
is that the complexity of typechecking \abb{ptt}'s basically depends
on the number of visible pebbles used.
Thus we can improve the complexity of the problem by
changing visible pebbles into invisible ones as much as possible, 
see Section~\ref{sec:pattern}.

Note that the solvability of the inverse type inference problem for a tree transducer $\cM$ 
means in particular that its domain is a regular tree language, 
taking $L(G_{\rm out}) = T_\Delta$ 
where $\Delta$ is the output alphabet of $\cM$. Thus, 
it follows from Theorem~\ref{thm:typecheck} that the domains of \abb{PTT}'s are regular,
or in other words, that every alternating \abb{pta} accepts a regular tree language.

\begin{corollary}\label{cor:domptt}
For every \abb{ptt} $\cM$, its domain $L(\cM)$ is regular. 
\end{corollary}

\section{Trees, Tests and Trips}\label{sec:tests}

In this section we show that \abb{vi-pta}'s recognize the regular tree languages,
that they compute the \abb{mso} definable binary patterns (or trips),
and that they can perform \abb{mso} tests on the observable part of their configuration
(which consists of the position of the head and the positions of the observable pebbles). 

For ``classical'' tree-walking automata with 
a bounded number of visible pebbles, i.e., for \abb{v-pta}'s, 
it was shown in \cite[Section~5]{jewels} that these
automata accept regular tree languages only.
However, as proved in \cite{expressive}, they cannot accept all regular tree languages. 
One of the main reasons for introducing an unbounded number
of invisible pebbles is that they can be used to recognize every
regular tree language.
Recall that $\REGT$ denotes the class of regular tree languages.

\begin{lemma}\label{lem:regular}
$\REGT \subseteq \IdPTA$.
\end{lemma}

\begin{proof}
As the regular tree languages are recognized by deterministic
bottom-up finite-state tree automata, it suffices to explain how the
computation of such an automaton $\cA$ can be simulated by a
deterministic \abb{pta} $\cA'$ with invisible pebbles.
The computation of $\cA$
on the input tree can be reconstructed by a post-order evaluation
of the tree.
At the current node $u$, $\cA'$ uses an invisible pebble to store the states 
in which $\cA$ arrives at the first $m$ children of $u$, 
for some $m$.
The colour of the pebble represents the sequence of states.
For each ancestor $v$ of $u$ 
the pebble stack contains a similar pebble for the first $i-1$ children of $v$, where $vi$  
is the unique child of $v$ that is also an ancestor of $u$ (or $u$ itself). 
If $u$ has more than $m$ children, then $\cA'$ moves to its $(m+1)$-th child and  
drops a pebble that represents the empty sequence of states of $\cA$. 
Otherwise, $\cA'$ computes the state
assumed by $\cA$ in $u$ based on the states of the children, lifts the pebble at $u$, and 
moves to the parent of $u$ to update its pebble with that state. 
The post-order evaluation ensures that pebbles are used
in a nested fashion.

Formally, let $\cA= (\Sigma,P,F,\delta)$ where $\Sigma$ is a ranked alphabet, 
$P$ is a finite set of states, 
$F\subseteq P$ is the set of final states, and $\delta$ is the transition function 
that assigns a state $\delta(\sigma,p_1,\dots,p_m)\in P$ to every $\sigma\in\Sigma$
and $p_1,\dots,p_m\in P$ with $m=\rank_\Sigma(\sigma)$.
As pebble colours the \abb{i-pta} $\cA'$ has all strings in $P^*$ of length at most $\m_\Sigma$. 
Its states and rules are introduced one by one as follows, where $\sigma$ ranges over $\Sigma$,
$j$ and $m$ range over $[0,\rank(\sigma)]$, and $p,p_1,\dots,p_m$ range over~$P$.
The initial state $q_0$ does not occur in the right-hand side of any rule.
In the initial state, the automaton $\cA'$ drops a pebble at the root representing 
the empty sequence of states of $\cA$, and goes into the main state $q_\circ$. The rule is  
\[
\rho_1: \tup{q_0,\sigma,0,\nothing} \to \tup{q_\circ,\drop_\epsilon}.
\]
In state $q_\circ$, $\cA'$ consults the pebble 
to see whether or not all children have been evaluated, and acts accordingly.
For $m<\rank(\sigma)$ it has the rule
\[
\rho_2: \tup{q_\circ,\sigma,j,\{p_1\cdots p_m\}} \to  
                \tup{q_\circ,\down_{m+1};\drop_\epsilon}, 
\]
which handles the case that the state of $\cA$ is not yet known for all children of node $u$. 
For $m=\rank(\sigma)$ and $p=\delta(\sigma,p_1,\dots,p_m)$ it has the rules
\[
\begin{array}{llll}
\rho_3: \tup{q_\circ,\sigma,j,\{p_1\cdots p_m\}} & \to & 
                 \tup{\bar{q}_p,\lift_{p_1\cdots p_m};\up} 
     & \text{if }  j\neq 0, \\[1mm]
\rho_4: \tup{q_\circ,\sigma,0,\{p_1\cdots p_m\}} & \to & \tup{q_\mathrm{yes},\stay}
     & \text{if } p\in F, \\[1mm]
\rho_5: \tup{q_\circ,\sigma,0,\{p_1\cdots p_m\}} & \to & \tup{q_\mathrm{no},\stay}
     & \text{if } p\notin F,
\end{array}
\]
and for $m<\rank(\sigma)$ it has the rule
\[
\rho_6: \tup{\bar{q}_p,\sigma,j,\{p_1\cdots p_m\}} \to 
                       \tup{q_\circ,\lift_{p_1\cdots p_m};\drop_{p_1\cdots p_mp}}.
\]
Thus, if the states $p_1,\dots,p_m$ of $\cA$ at all the children of node $u$ are known, 
$\cA'$~computes the state $p=\delta(\sigma,p_1,\dots,p_m)$ of $\cA$ at $u$. 
If $u$ is not the root of the input tree, then $\cA'$ stores $p$ in its own state $\bar{q}_p$,
lifts the pebble $p_1\cdots p_m$, and moves up to the parent of $u$. 
Since the pebble at the parent is now observable, it can be updated.  
If $u$ is the root of the input tree, then $\cA'$ knows whether or not $\cA$ accepts that tree,
and correspondingly goes into state $q_\mathrm{yes}$ or state $q_\mathrm{no}$,
where $q_\mathrm{yes}$ is the unique final state of $\cA'$.
Note that there is one pebble left on the root of the tree.
\end{proof}

Adding an infinite supply of invisible
pebbles on the other hand 
does not lead out of the regular tree languages. 
It is possible to give a proof of this fact
by reducing \abb{v$_k$i-pta}'s to the
backtracking pushdown tree automata of \cite{Slu},
but here we deduce it from the results of the previous section.

\begin{theorem}\label{thm:regt}
For each $k\ge 0$,
$\VIPTA{k} = \VIdPTA{k} = \REGT$.
\end{theorem}

\begin{proof}
By Lemma~\ref{lem:regular}, $\REGT \subseteq \VIdPTA{k}$. 
Conversely, as observed before, 
a \abb{pta} $\cA$ is easily turned
into a \abb{ptt} $\cM$ that outputs 
single node tree $\delta$ (with $\rank(\delta)=0$) for 
trees accepted by $\cA$: for every final state $q$ of $\cA$
add all rules $\tup{q,\sigma,j,b} \to \delta$. 
Then $L(\cA)=L(\cM)$, the domain of $\cM$, 
which is regular by Corollary~\ref{cor:domptt}. 
\end{proof}

Note that an infinite supply of \emph{visible} pebbles 
could be used to mark $a$'s and $b$'s alternatingly
and thus accept the nonregular language $\{a^nb^n\mid n\in\nat\}$
(and similarly $\{a^nb^nc^n\mid n\in\nat\}$).
Note also that the stack of pebbles cannot be replaced by 
two independent stacks, one for visible and one for invisible pebbles.
Then we could accept $\{a^nb^n\mid n\in\nat\}$ 
with just one visible pebble: 
drop an invisible pebble on each $a$, and then use 
the visible pebble on the $b$'s to count the number of $a$'s,
by lifting one invisible pebble (in fact, the unique observable one) 
for each $b$.

Recall from Section~\ref{sec:trees} that 
an $n$-ary \emph{pattern} over a ranked alphabet $\Sigma$ is a set
$T \subseteq \{(t,u_1,\dots,u_n) 
   \mid t\in T_\Sigma, \,u_1,\dots,u_n \in N(t)\}$.
Recall also that the pattern~$T$ is said to be regular 
if its marked representation $\tmark(T)\subseteq T_{\Sigma\times\{0,1\}^n}$ 
is a regular tree language.
In fact, $T$ is regular if and only if it is \abb{mso} definable,
which means that there is an 
\abb{mso} formula $\phi(x_1,\dots,x_n)$ over $\Sigma$ such that $T = T(\phi)$, 
where $T(\phi) = \{(t,u_1,\dots,u_n) \mid t \models \phi(u_1,\dots,u_n)\}$.
Recall finally that a unary pattern ($n=1$) is called a \emph{site}, 
and a binary pattern ($n=2$) is called a \emph{trip}. 

With the help of an unbounded supply
of invisible pebbles 
tree-walking automata can recognize regular tree languages,
Lemma~\ref{lem:regular}.
Likewise \abb{v$_n$i-pta}'s can match 
arbitrary \abb{MSO} definable $n$-ary 
patterns $\phi$. 
When $n$ visible pebbles are dropped on a sequence
of $n$ nodes, the invisible pebbles can be used to evaluate
the tree, and test whether it belongs to the
regular tree language $\tmark(T(\phi))$.
In Section~\ref{sec:pattern} we will consider pattern matching in detail.

Ignoring the visible pebbles,
it is also possible to consider just the position of the head, and 
test whether the input tree together with that position belongs to a given
regular ``marked'' tree language.
We say that a family $\cF$ of \abb{pta}'s (or \abb{ptt}'s)
can \emph{perform \abb{mso} head tests} 
if, for a regular site $T$ over $\Sigma$, 
an automaton (or transducer) in $\cF$ can test whether or not 
$(t,h) \in T$, where $t$ is the input tree
and $h$ the position of the head at the moment of the test. 
Admittedly, this is a very informal definition. 
To formalize it we have to define a \abb{pta}$^{\text{\abb{mso}}}$ 
(or a \abb{ptt}$^{\text{\abb{mso}}}$),
i.e., a \abb{pta} (or \abb{ptt}) \emph{with \abb{mso} head tests},
that has rules of the form $\tup{q,\sigma,j,b,T} \to \zeta$
where $T$ is a regular site over $\Sigma$ (specified in some effective way). 
Such a rule is relevant to a configuration $\tup{q,h,\pi}$ on a tree $t$ if,
in addition, $(t,h) \in T$. Since the regular tree languages are closed under complement,
the complement $T^\mathrm{c}$ of~$T$ can be tested in a rule with left-hand side 
$\tup{q,\sigma,j,b,T^\mathrm{c}}$.
Such an automaton (or transducer) is deterministic
if for every two distinct rules $\tup{q,\sigma,j,b,T} \to \zeta$ and 
$\tup{q,\sigma,j,b,T'} \to \zeta'$, the site $T'$ is the complement of the site $T$. 
For a family $\cF$ of \abb{pta}'s (or \abb{ptt}'s), 
such as the \abb{v$_k$i-pta} or \abb{v$_k$i-}{\rm d}\abb{ptt} or \abb{v$_k$-pta}, we denote 
by $\cF^{\text{\abb{mso}}}$ the corresponding family of \abb{pta}$^{\text{\abb{mso}}}$'s
(or \abb{ptt}$^{\text{\abb{mso}}}$'s).
With this definition of \abb{pta}$^{\text{\abb{mso}}}$ we can formally define that 
a family $\cF$ of \abb{pta}'s can perform \abb{mso} head tests 
if for every \abb{pta}$^{\text{\abb{mso}}}$ in $\cF^{\text{\abb{mso}}}$ 
an equivalent \abb{pta} in $\cF$ can be constructed, and similarly for \abb{ptt}'s. 

Obviously, 
as \abb{v-pta}'s cannot recognize all regular tree languages,
they cannot perform \abb{mso} head tests either:
for any regular tree language $T$
the set $\{(t,\rt_t) \mid t\in T \}$
is a regular site.

The next result shows that any \abb{vi-pta} 
that uses \abb{mso} head tests as a built-in feature
(i.e., any \abb{vi-pta}$^{\text{\abb{mso}}}$) 
can be replaced by an equivalent \mbox{\abb{vi-pta}}
without such tests.
The result holds for \abb{vi-pta}'s
with any fixed number of visible pebbles, 
either deterministic or nondeterministic,
and it also holds for the corresponding \abb{vi-ptt}'s. 

\begin{lemma}\label{lem:sites}
For each $k\ge 0$, 
the \abb{v$_k$i-pta} can perform \abb{mso} head tests.
The same holds for the \abb{v$_k$i-}{\em d}\abb{pta},
\abb{v$_k$i-ptt}, and \abb{v$_k$i-}{\em d}\abb{ptt}.
\end{lemma}

\begin{proof}
Let $\cA_T$ be a deterministic bottom-up finite-state
tree automaton recognizing  the regular 
tree language $\tmark(T)$ over $\Sigma \times \{0,1\}$,
representing the site~$T$, trees with a single marked node.
We show how a deterministic \abb{i-pta} $\cA'_T$ can test 
whether or not the  
input tree with current head position $h$
is accepted by~$\cA_T$,
in a computation starting in configuration $\tup{q_0,h,\epsilon}$
and ending in configuration $\tup{q_\mathrm{yes},h,\epsilon}$
or $\tup{q_\mathrm{no},h,\epsilon}$, where $q_0$ is the initial state
and $\{q_\mathrm{yes},q_\mathrm{no}\}$ the set of final states of $\cA'_T$. 
Moreover, it starts the computation by dropping a pebble on $h$,
and it keeps a pebble on $h$ until the final computation step. 
It should be obvious that this \abb{i-pta} $\cA'_T$ can be used as a subroutine 
by any \abb{v$_k$i-pta} or \abb{v$_k$i-ptt} $\cA$, 
starting in configuration $\tup{(\tilde{q},q_0),h,\pi}$ and 
ending in configuration $\tup{(\tilde{q},q_\mathrm{yes}),h,\pi}$
or $\tup{(\tilde{q},q_\mathrm{no}),h,\pi}$, 
for every state $\tilde{q}$ and pebble stack $\pi$ of $\cA$. 
Just replace each rule $\tup{q,\sigma,j,b} \to \tup{q',\alpha}$ of $\cA'_T$ 
by all possible rules 
$\tup{(\tilde{q},q),\sigma,j,b\cup b'} \to \tup{(\tilde{q},q'),\alpha}$
where $b'$ is a set of visible pebble colours of $\cA$ 
(except that in the first rule of~$\cA'_T$, which drops a pebble on $h$,
the set $b'$ possibly contains an invisible pebble colour of~$\cA$). 

The post-order evaluation of Lemma~\ref{lem:regular} 
does not work here without precautions.
If we mark node $h$ with an invisible pebble
the pebble becomes unobservable during the evaluation.
In this way we cannot take the special ``marked'' position of $h$
into account.\footnote{Marking $h$ with a visible pebble would easily work,
showing that \abb{vi-pta} can perform \abb{mso} head tests.}
Instead, we first evaluate the subtree rooted at~$h$,
and subsequently the subtrees rooted at the ancestors of $h$, 
moving along the path from $h$ to the root of the input tree. 
At the start of the evaluation of a subtree, we ``paint'' its root $u$ 
by adding a special colour to the pebble on~$u$, and preserving that 
information when the pebble is updated.  
In this way it is always clear when the painted node is visited.
We paint node $h$ with the special additional colour $\odot$ and 
use the evaluation process of Lemma~\ref{lem:regular} to compute 
the state of $\cA_T$ at $h$, viewing the label $\sigma$ of each node as $(\sigma,0)$
except for the label $\sigma$ of $h$ which is treated as $(\sigma,1)$. 
We paint each ancestor $u$ of $h$ with an additional colour $(j,p)$ which indicates 
the child number $j$ of the previous ancestor of $h$ and the state $p$ at which $\cA_T$
arrives at that child of~$u$ (with $h$ as a marked node). 
Then we use, again, the evaluation process of Lemma~\ref{lem:regular}
to compute the state of~$\cA_T$ at $u$ (with every $\sigma$ viewed as $(\sigma,0)$), 
except that  
the information in the pebble $(j,p)$ is used for the state~$p$ of the $j$-th child of $u$,
which is the unique child that has $h$ as a descendant.
Repeating this process for each ancestor, we eventually
reach the root of the tree, and know the outcome of the test. 
Then we return to the original position~$h$ picking up the
pebbles left on the path from that position to the root.

Formally, let $\cA_T= (\Sigma \times \{0,1\},P,F,\delta)$. 
For convenience we will identify the symbols $(\sigma,0)$ and $\sigma$. 
The \abb{i-pta} $\cA'_T$ is an extension of the \abb{i-pta} $\cA'$ 
in the proof of Lemma~\ref{lem:regular}. It has the additional 
states $q_{\downarrow\mathrm{yes}}$ and $q_{\downarrow\mathrm{no}}$,
and in addition to the pebble colours $p_1\cdots p_m$ of $\cA'$
it has the pebble colours $(\mu,p_1\cdots p_m)$ where either 
$\mu=\odot$ or $\mu=(i,r)$ for some $i\in[1,\m_\Sigma]$ and $r\in P$.
The additional pebbles are used to ``paint'' $h$ (with $\mu=\odot$) 
and the ancestors of $h$ (with some $\mu=(i,r)$). 
The automaton $\cA'_T$ has all the rules of~$\cA'$, 
except that rules $\rho_4$ and~$\rho_5$ will become superfluous,
and rule $\rho_1$ is replaced by the rule 
\[
\begin{array}{lll}
\rho'_1: \tup{q_0,\sigma,j,\nothing} & \to & \tup{q_\circ,\drop_{(\odot,\epsilon)}}.
\end{array}
\] 
Thus, $\cA'_T$ starts by evaluating the subtree rooted at $h$, with $h$ as marked node. 
For $m <\rank(\sigma)$ and every $\mu$ as above, 
except when $\mu=(m+1,r)$ for some $r\in P$,  
$\cA'_T$ has the rules
\[
\begin{array}{lll}
\rho_2^\mu: \tup{q_\circ,\sigma,j,\{(\mu,p_1\cdots p_m)\}} & \to & 
                \tup{q_\circ,\down_{m+1};\drop_\epsilon} \\[1mm]
\rho_6^\mu: \tup{\bar{q}_p,\sigma,j,\{(\mu,p_1\cdots p_m)\}} & \to & 
                       \tup{q_\circ,\lift_{(\mu,p_1\cdots p_m)};\drop_{(\mu,p_1\cdots p_mp)}}
\end{array}
\]
which intuitively means that the pebble $(\mu,p_1\cdots p_m)$ is treated in the same way as $p_1\cdots p_m$
when not all children of the current node have been evaluated:
$\cA'_T$~moves to the $(m+1)$-th child and calls $\cA'$, and when $\cA'$ returns with the state $p$,
$\cA'_T$ adds $p$ to the sequence of states in the pebble. 
However, in the exceptional case where $m <\rank(\sigma)$ and $\mu=(m+1,r)$,
$\cA'_T$ has the rule
\[
\begin{array}{lll}
\rho_{2,6}^\mu: \tup{q_\circ,\sigma,j,\{(\mu,p_1\cdots p_m)\}} & \to & 
                       \tup{q_\circ,\lift_{(\mu,p_1\cdots p_m)};\drop_{(\mu,p_1\cdots p_mr)}}
\end{array}
\]
which means that for the $(m+1)$-th child $\cA'_T$ does not call $\cA'$ 
but uses the state~$r$ that was previously computed and stored in $\mu$. 

The remaining rules of $\cA'_T$ handle the situations that $\cA'_T$ has just evaluated 
the subtrees rooted at the children of $h$ or of one of the ancestors $u$ of $h$, in state~$q_\circ$. 
The automaton $\cA'_T$ computes the state $p$ of $\cA_T$ at the marked node~$h$ 
or the unmarked node $u$, and drops
the pebble $((j,p),\epsilon)$ at its parent $v$, where $j$ is the child number of $h$ or $u$, 
thus indicating that the subtree rooted at the $j$-th child of~$v$ 
(with~$h$ as a marked node) evaluates to $p$. 
Then $\cA'_T$ evaluates the subtree rooted at $v$. 

For $m=\rank(\sigma)$ and every $\mu$ as above, $\cA'_T$ has the rules
\[
\begin{array}{llll}
\rho_3^\mu: \tup{q_\circ,\sigma,j,\{(\mu,p_1\cdots p_m)\}} & \to & 
                 \tup{q_\circ,\up;\drop_{((j,p),\epsilon)}} 
     & \text{if }  j\neq 0, \\[1mm]
\rho_4^\mu: \tup{q_\circ,\sigma,0,\{(\mu,p_1\cdots p_m)\}} & \to & 
                 \tup{q_{\downarrow\mathrm{yes}},\stay}
     & \text{if } p\in F, \\[1mm]
\rho_5^\mu: \tup{q_\circ,\sigma,0,\{(\mu,p_1\cdots p_m)\}} & \to & 
                 \tup{q_{\downarrow\mathrm{no}},\stay}
     & \text{if } p\notin F.
\end{array}
\]
where $p=\delta((\sigma,1),p_1,\dots,p_m)$ if $\mu=\odot$ and 
$p=\delta(\sigma,p_1,\dots,p_m)$ otherwise. 

When $\cA'_T$ arrives at the root of the input tree, it knows whether or not $\cA_T$ accepts that tree
(with $h$ as a marked node), and moves down to $h$. 
For the outcome $x\in\{\mathrm{yes},\mathrm{no}\}$ the rules are 
\[
\begin{array}{llll}
\tup{q_{\downarrow x},\sigma,j,\{((i,r),p_1\cdots p_m)\}} & \to & 
     \tup{q_{\downarrow x},\lift_{((i,r),p_1\cdots p_m)};\down_i} \\[1mm]
\tup{q_{\downarrow x},\sigma,j,\{(\odot,p_1\cdots p_m)\}} & \to & 
     \tup{q_x,\lift_{(\odot,p_1\cdots p_m)}}.
\end{array}
\]
This ends the description of $\cA'_T$. 
\end{proof}

This result can easily be extended, using the same
proof technique:
\abb{pta}'s and \abb{ptt}'s can test their 
\emph{visible configuration},
the position of the head together with the
positions and colours of the visible pebbles.
Later we will show the more complicated result that \abb{pta}'s and \abb{ptt}'s can even test their 
\emph{observable configuration},
i.e., the visible configuration plus the topmost pebble 
(Theorem~\ref{thm:mso}).

Let $C$ be the set of colours of a \abb{pta} or \abb{ptt}. 
To represent the visible and observable configurations,  
we introduce a new ranked alphabet $\Sigma \times 2^C$,
such that the rank of $(\sigma,b)$ equals that of $\sigma$ in $\Sigma$.
A tree over $\Sigma \times 2^C$ is a ``coloured tree''. 
For each pebble stack $\pi$ on a tree $t$ over $\Sigma$ 
we define two coloured trees. 
The visible configuration tree $\vis(t,\pi)$ is
obtained by adding to the label of each node~$u$ of~$t$ 
the set $b\subseteq C$ such that $b$ contains~$c$ if and only if 
$(u,c)$ occurs in $\pi$ and $c\in C_{\rm v}$. 
Similarly for $\obs(t,\pi)$,
the observable configuration tree,
$b$~contains~$c$ if and only if 
$(u,c)$ occurs in $\pi$ and $c$ is observable
(i.e., $c\in C_{\rm v}$ or $(u,c)$ is the top element of $\pi$). 
Note that as long as a \abb{pta} does not change 
its pebble stack by a drop- or lift-instruction,
it behaves just as a \abb{ta} on $\obs(t,\pi)$.

We say that a family $\cF$ of \abb{pta}'s (or \abb{ptt}'s)
can \emph{perform \abb{mso} tests on the visible configuration} 
if, for a regular site $T$ over $\Sigma \times 2^C$, 
an automaton (or transducer) in $\cF$ can test whether or not 
$(\vis(t,\pi),h) \in T$, where $t$ is the input tree, 
$\pi$ the current pebble stack and $h$ the current position of the head.
A similar definition can be given for \emph{\abb{mso} tests on the observable configuration}.
These informal definitions could be formalized in a way explained 
for \abb{mso} head tests before Lemma~\ref{lem:sites}.

We now show that the \abb{vi-pta} and \abb{vi-ptt} can perform \abb{mso} tests on the visible configuration.
Note that for a regular site $T$ over $\Sigma \times 2^C$,
$\tmark(T)$ is a regular tree language over 
$\Sigma \times 2^C \times \{0,1\}$.

\begin{lemma}\label{lem:visiblesites}
For each $k\ge 0$, 
the \abb{v$_k$i-pta} and \abb{v$_k$i-}{\em d}\abb{pta} 
can perform \abb{mso} tests
on the visible configuration.
The same holds for the \abb{v$_k$i-ptt} and \abb{v$_k$i-}{\em d}\abb{ptt}.
\end{lemma}

\begin{proof}
As in the proof of Lemma~\ref{lem:sites}, let $\cA_T$ be a deterministic bottom-up finite-state
tree automaton recognizing the regular 
tree language $\tmark(T)$ over $\Sigma \times 2^C \times \{0,1\}$,
representing the site $T$, coloured trees with a single marked node.
As observed in the first paragraph of that proof the \abb{i-ptt} $\cA'_T$ 
(of that proof) can be turned into a subroutine 
for any \abb{v$_k$i-pta} or \abb{v$_k$i-ptt} $\cA$ 
with visible colour set~$C_\mathrm{v}$ by replacing each rule 
$\tup{q,\sigma,j,b} \to \tup{q',\alpha}$ of $\cA'_T$ (except $\rho'_1$)
by all possible rules 
$\tup{(\tilde{q},q),\sigma,j,b\cup b'} \to \tup{(\tilde{q},q'),\alpha}$
with $b'\subseteq C_\mathrm{v}$. This subroutine can easily be turned into one 
that tests whether or not $(\vis(t,\pi),h) \in T$ as follows. 
For the rules corresponding in this way to 
$\rho_3,\rho_4,\rho_5$ (in the proof of Lemma~\ref{lem:regular}),
change $p=\delta(\sigma,p_1,\dots,p_m)$ into $p=\delta((\sigma,b',0),p_1,\dots,p_m)$.
Similarly, for $\rho_3^\mu,\rho_4^\mu,\rho_5^\mu$
change $p=\delta((\sigma,1),p_1,\dots,p_m)$ into $p=\delta((\sigma,b',1),p_1,\dots,p_m)$
and, again, $p=\delta(\sigma,p_1,\dots,p_m)$ into $p=\delta((\sigma,b',0),p_1,\dots,p_m)$.
\end{proof}

We now turn to the \abb{pta} as a navigational device:
the \emph{trip $T(\cA)$ computed by a \abb{pta} $\cA$}
consists of all triples $(t,u,v)$ such that
$\cA$, on input tree $t$, started at node $u$ in an initial state 
without pebbles on the tree, 
walks to node $v$, and halts in a final state (possibly leaving pebbles on the tree).
Formally, $T(\cA)=\{(t,u,v)\in T_\Sigma\times \nod{t}\times \nod{t} \mid 
\exists\, q_0\in Q_0, q_\infty\in F, \pi\in (N(t)\times C)^*: 
\tup{q_0,u,\epsilon} \Rightarrow_\cA^* \tup{q_\infty,v,\pi}\}$. 
Two \abb{PTA}'s $\cA$ and $\cB$ are \emph{trip-equivalent} if $T(\cA)=T(\cB)$. 
Since clearly $L(\cA)=\{t\in T_\Sigma\mid \exists\, u\in\nod{t}: 
(t,\rt_t,u)\in T(\cA)\}$, trip-equivalence implies (language-)equivalence.  
A~trip $T$ is \emph{functional} if, for every $t$,
$\{ (u,v) \mid (t,u,v)\in T \}$ is a function.
Note that the trip computed by a deterministic \abb{pta}
is functional.

It is straightforward to check that Lemma~\ref{lem:stacktests} also holds for 
the \abb{pta} as navigational device, replacing equivalence by trip-equivalence.
Thus, \abb{v$_k$i-pta}'s can perform stack tests also when computing a trip.
Similarly, they can perform the \abb{mso} tests discussed in 
Lemmas~\ref{lem:sites} and~\ref{lem:visiblesites}, and 
to be discussed in Theorem~\ref{thm:mso}.

In \cite[Theorem~8]{bloem} it is shown that every \abb{mso} definable trip
(tree-node relation) can be computed by a \abb{ta}$^{\text{\abb{mso}}}$, i.e., 
a tree-walking automaton with \abb{mso} head tests (and vice versa).  
Moreover, by (the corrected version of) \cite[Theorem~9]{bloem}, 
if the trip is functional, then the automaton is deterministic.
We will also use the fact that, according to the proof of \cite[Theorem~8]{bloem},
the \abb{mso} definable trips can be computed in a special way. 

\begin{proposition}\label{prop:trips}
Every \abb{mso} definable trip can be computed by a tree-walking automaton 
with \abb{mso} head tests that has the following two properties:

$(1)$ it never walks along the same edge twice (in either direction), and 

$(2)$ it visits each node at most twice.

\noindent
If the trip is functional, then the automaton is deterministic. 
\end{proposition}

The first property means that, 
when walking from a node $u$ to a node $v$, the automaton always takes the 
shortest (undirected) path from $u$ to~$v$, i.e., 
the path that leads from $u$ up to the least ancestor of $u$ and~$v$, and then 
down to $v$.
The second property means that the automaton does not execute 
two consecutive stay-instructions. 

The next result 
provides a characterization of the \abb{mso} definable trips
by pebble automata that is more elegant than the one in \cite{trips},
which uses so-called marble/pebble automata,
a restricted kind of \abb{v$_1$i-pta}  
(marbles are invisible pebbles only dropped on the path
from the root to the current position of the head;
a single visible pebble may only be dropped and
picked up on a tree without marbles).

\begin{theorem}\label{thm:trips}
For each $k\ge 0$,
the trips computed by
 \abb{v$_k$i-pta}'s are
exactly the \abb{mso} definable  trips.
Similarly for \abb{V$_k$I-}{\em d}\abb{PTA}'s and functional trips.
\end{theorem}

\begin{proof}
Consider a trip $T$ computed by  \abb{v$_k$i-pta} $\cA$.
Thus, for any $(t,u,v)$ in~$T$, 
starting at node $u$ of input tree $t$, $\cA$ walks to node $v$ and halts.
Then $\tmark(T)$ can be recognized by another \abb{v$_k$i-pta} 
as follows. First it searches (deterministically) for the marked  
starting node $u$, then it simulates $\cA$, and
when $\cA$ halts in a final state, verifies that the marked node $v$ is reached.
By Theorem~\ref{thm:regt} this tree language is
regular and hence $T$ is \abb{mso} definable.

By Proposition~\ref{prop:trips} every \abb{mso} definable trip
can be computed by a 
tree-walking automaton $\cB$ with \abb{mso} head tests. 
Since (as observed above) Lemma~\ref{lem:sites} 
also holds for the \abb{pta} as a navigational device, 
it can therefore be computed by an \abb{I-PTA} $\cB'$.
Moreover, if the trip is functional,
then the automata $\cB$ and $\cB'$ are deterministic.
\end{proof}

Note that the automaton $\cB'$ in the above proof 
always removes all its pebbles before halting. 
Thus, that requirement could be added to the definition of the trip
computed by a \abb{v$_k$i-pta} (implying that not every 
\abb{v$_k$i-pta} computes a trip). This conforms to the idea 
that one should not leave garbage after a picknick. 

Using the above result, or rather Proposition~\ref{prop:trips}, 
we are now able to show that the \abb{pta} and \abb{ptt} 
can perform \abb{MSO} tests
on the observable configuration,
i.e., they can evaluate \abb{mso} formulas $\phi(x)$
on the observable configuration tree $\obs(t,\pi)$
with the variable $x$ assigned to the position of the
reading head.

\begin{theorem}\label{thm:mso}
For each $k \ge 0$, the \abb{v$_k$i-pta} and \abb{v$_k$i-}{\em d}\abb{pta} 
can perform \abb{mso} tests 
on the observable configuration.
The same holds for the \abb{v$_k$i-ptt} and \abb{v$_k$i-}{\em d}\abb{ptt}.
\end{theorem}

\begin{proof}
Let $T$ be a regular site over $\Sigma\times 2^C$, 
and let $\cA$ be a \abb{v$_k$i-pta} that uses $T$ as a test to find out 
whether or not $(\obs(t,\pi),h)\in T$. Our aim is to construct a trip-equivalent  
\abb{v$_k$i-pta} $\cA'$ that does not use \abb{mso} tests on the observable configuration.
The proof is exactly the same for the case where $\cA$ and $\cA'$ are \abb{v$_k$i-ptt} (with equivalence instead of trip-equivalence). 

Essentially, $\cA'$ simulates $\cA$. When $\cA$ uses the test $T$, there are two cases.
In the first case, either the pebble stack of $\cA$ is empty or 
the colour of the topmost pebble of $\cA$ is visible. 
Then the observable configuration equals the visible configuration,
and so $\cA'$ can use the test $T$ too, by Lemma~\ref{lem:visiblesites}.
The remaining, difficult case is that the colour $d$ of the topmost pebble of $\cA$ 
is invisible. To implement the test $T$ in this case
it seems that $\cA'$ cannot use any additional invisible pebbles
(as in the proof of Lemma~\ref{lem:visiblesites}), because 
they make pebble $d$ unobservable. However, 
this is not a problem as long as the additional pebbles carry 
sufficient information about the position $u$ of pebble~$d$. 
The solution is to view $T$ as a trip from $u$ to $h$ (the position of the head), 
and to keep track of an automaton $\cB_d$ that computes that trip.
Although $\cB_d$ is nondeterministic, it is straightforward for $\cA'$ to employ
the usual subset construction for finite-state automata. 

For every $d\in C_\mathrm{i}$, let $T_d$ be the trip over $\Sigma\times 2^C$
defined by $T_d = \{(s,u,h)\mid (s',h)\in T\}$, where $s'$ is obtained from $s$
by changing the label $(\sigma,b)$ of $u$ into $(\sigma,b\cup\{d\})$. 
Then $(\obs(t,\pi),h)\in T$ if and only if $(\vis(t,\pi),u,h)\in T_d$, 
if $(d,u)$ is the topmost element of $\pi$. 
It should be clear from the regularity of $T$ 
that $T_d$ is a regular trip. Hence, by Proposition~\ref{prop:trips}, 
there is a \abb{ta} with \abb{mso} head tests $\cB_d$ that computes $T_d$
and that has the special properties mentioned there. 
Therefore (see the paragraph after Proposition~\ref{prop:trips}),  
to keep track of the possible computations of $\cB_d$, the automaton $\cA'$ uses 
additional invisible pebbles to cover the shortest (undirected) path from $u$ to $h$. 
These pebbles will be called \emph{beads} to distinguish them
from $\cA$'s original pebbles.
Each bead carries state information on
computations of $\cB_d$ that start at position $u$ 
(in an initial state) and end at position $h$. More precisely, 
each bead is a triple $(S,\delta,d)$ 
where $S$ is a set of states of $\cB_d$ and 
$\delta\in \{\up,\stay\}\cup\{\down_i\mid i\in[1,\m_\Sigma]\}$.
There is one such bead $(S,\delta,d)$ 
on every node $v$ on the path from $u$ to $h$ (including $u$ and $h$) 
where $S$ is the set of states $p$ of $\cB_d$
such that $\cB_d$ has a computation on $\vis(t,\pi)$ starting at $u$ 
in an initial state and ending at $v$ in state $p$. 
Moreover, $\delta$ indicates the node $w$ just before $v$ on the path,
which is the parent or $i$-th child of $v$ if $\delta$ is $\up$ or $\down_i$, respectively, 
and which is nonexistent when $v=u$, if $\delta=\stay$. 
The bead at $v$ is on top of the bead at $w$ in the pebble stack of $\cA'$. 
Thus, the bead at $h$ is always on the top of the stack of $\cA'$ 
and hence is always observable.

The automaton $\cA'$ can still simulate $\cA$ because 
if the bead $(S,\delta,d)$ is at head position $h$, 
then the invisible pebble $d$ is observable at $h$ by $\cA$ 
if and only if $\delta=\stay$. 
If $\cA$ lifts $d$, then $\cA'$ lifts both $(S,\stay,d)$ and $d$. 
If $\cA$ drops another pebble $d'$ at $h$, then so does $\cA'$ 
(and starts a new chain of beads
on top of that pebble if $d'$ is invisible). When pebble $d'$ is lifted again, 
the beads for pebble $d$ are still available and can be used as before. 

Now, suppose that $\cA$ uses the test $T$ at position $h$.
If $\cA'$ does not see a bead at position $h$, 
then it uses $T$ as a test 
on the visible configuration. 
If $\cA'$ sees a bead $(S,\delta,d)$ at $h$, then $\cA'$ just
checks whether or not $S$ contains a final state of $\cB_d$, i.e., 
whether or not $(\vis(t,\pi),u,h)\in T_d$.

It remains to show how $\cA'$ computes the beads. 
The path of beads is initialized by $\cA'$ when $\cA$ drops invisible pebble $d$.
Then $\cA'$ also drops pebble $d$, computes the relevant set $S$ of states of $\cB_d$, 
and drops bead $(S,\stay,d)$. The set~$S$ contains all initial states of $\cB_d$,
plus all states that $\cB_d$ can reach from an initial state by applying one relevant rule 
with a stay-instruction (cf. the second property in Proposition~\ref{prop:trips}).
To find the latter states, $\cA'$ just simulates all those rules. 
Note that the \abb{mso} head tests of $\cB_d$ on $\vis(t,\pi)$ are  
\abb{mso} tests on the visible configuration of $\cA'$. That is because
during the simulation of $\cA$ by $\cA'$ 
the visible configuration $\vis(t,\pi')$ of $\cA'$ equals 
the visible configuration $\vis(t,\pi)$ of $\cA$:  
the pebble stack $\pi$ of $\cA$ 
is obtained from the corresponding pebble stack~$\pi'$ of $\cA'$ 
by removing all (invisible) beads. 

The path of beads is updated as follows. If we backtrack 
on the path from~$u$ to $h$, i.e., 
the current bead is $(S,\delta,d)$ with $\delta\neq\stay$ 
and we move in the direction~$\delta$, we just lift the current bead before moving.
If we move away from~$u$, we must compute new bead information.
Suppose the current bead on $h$ is $(S,\up,d)$ and we move down to the $i$-th child $hi$ of $h$. 
Then the bead at $hi$ is $(S',\up,d)$ where $S'$ can be computed 
in a similar way as the set $S$ above: $\cA'$~simulates all computations of $\cB_d$ 
that start at $h$ in a state of $S$ and end at $hi$ (and note that such a computation 
consists of one step, possibly followed by another step with a stay-instruction). 
Now suppose that the current bead is $(S,\down_i,d)$, which means that 
$u$ is a descendant of~$h$. If we move up to the parent $v$ of $h$, then the new bead 
is $(S',\down_j,d)$ where $j$ is the child number of~$h$. If we move down to a child $v$
of $h$ with child number $\neq i$, then the new bead is $(S',\up,d)$. In each of these cases
$S'$ can be computed as before, by simulating 
the computations of~$\cB_d$ from $h$ to $v$.  

In general, $\cA$ can of course use several regular sites $T_1,\dots,T_n$ 
as tests on the observable configuration. It should be obvious how to extend the proof
to handle that. The beads are then of the form $(S_1,\dots,S_n,\delta,d)$ where
$S_i$ is a set of states of a \abb{ta} with \abb{mso} head tests $\cB_{id}$ 
that computes the trip $T_{id}$. To test~$T_i$ in the presence of such a bead,
$\cA'$ just checks whether or not $S_i$ contains a final state of $\cB_{id}$. 
\end{proof}

\section{The Power of the I-PTT}\label{sec:poweriptt}

In this section we discuss some applications of the fact that 
the \abb{i-ptt} can perform \abb{mso} head tests (Lemma~\ref{lem:sites}).  
We prove that it can simulate the composition of two \abb{tt}'s 
of which the first is deterministic (cf. Lemma~\ref{lem:nul-decomp}),
and that it can simulate the bottom-up tree transducer.   

\smallpar{Composition of TT's}
We now prove that the inclusions of Lemma~\ref{lem:nul-decomp} also
hold in the other direction, provided that we start with a deterministic \abb{tt}. 

\begin{theorem}\label{thm:composition}
$\dTT \circ \dTT \subseteq \IdPTT$ and
$\dTT \circ \TT \subseteq \IPTT$.
\end{theorem}

\begin{proof}
Consider two deterministic \abb{TT}'s $\cM_1$ and $\cM_2$.
Assume that input tree~$t$ is translated into tree $s$ by 
transducer $\cM_1$.
We will simulate the computation of~$\cM_2$ on $s$
directly on $t$ using a  \abb{pTT} 
$\cM$ with invisible pebbles.
Any action taken by $\cM_2$ on node $v$ of tree $s$
will be simulated by $\cM$ on the node $u$ of~$t$ that was the
position of $\cM_1$ when it generated $v$.
This means that if $\cM_2$ moves down in the tree $s$
to one of the children of $v$, the computation of $\cM_1$
is simulated until it generates that child.
On the other hand, if $\cM_2$ moves up in the tree $s$
to the parent of $v$, it is necessary to backtrack on 
the computation of $\cM_1$, back to the moment that that parent
was generated. In this way, tree $s$ is never fully 
reconstructed as a whole, but at every moment $\cM$ has
access to a single node of $s$. 
The necessary node, the current node of $\cM_2$,
is continuously updated by moving back and forth along the
computation of $\cM_1$ on $t$. 

Moving forward on the computation of $\cM_1$
is straightforward. To be able to retrace,
$\cM$ uses its pebbles to record the output-generating steps of the computation
of $\cM_1$ on $t$. Each output rule of $\cM_1$ is represented
by a pebble colour, and is put on the node $u$ of $t$ where it was applied.
The pebble colour also codes the child number of the
generated node $v$ in $s$.
Thus the pebble stack represents a (shortest) path in $s$
from the root to $v$.
For each node on that path the stack contains a pebble 
with the rule of $\cM_1$ used to generate that node and 
with its child number,
from bottom to top.

Note that the determinism of $\cM_1$ is an essential
ingredient for this construction. 
Simulating $\cM_2$,  walking along the virtual tree $s$,
one has to ensure that each time a node $v$ is revisited,
the same rule of $\cM_1$ is applied to $u$.

The above intuitive description assumes that the input tree $t$ is 
in the domain $L(\cM_1)$ of $\cM_1$.
In fact, it suffices to construct an \abb{i-ptt} $\cM$ 
such that $\tau_\cM(t)=\tau_{\cM_2}(\tau_{\cM_1}(t))$ for every such $t$, 
because $\cM$ can then easily be adapted to start with an \abb{mso} head test verifying that 
the input tree is in $L(\cM_1)$, which is regular by Corollary~\ref{cor:domptt}.

Let us now give the formal definitions. 
Let $\cM_1 = (\Sigma, \Delta, P, \{p_0\}, R_1)$ be a deterministic \abb{tt} and let 
$\cM_2 = (\Delta, \Gamma, Q, Q_0, R_2)$ be an arbitrary \abb{tt}. 
To define the \abb{i-ptt} $\cM$ it is convenient to extend the definition of an \abb{i-ptt}
with a new type of instruction: we allow the right-hand side of a rule to be of the form 
$\tup{q',\totop}$, which when applied to a configuration $\tup{q,u,\pi}$ 
leads to the next configuration $\tup{q',v,\pi}$ where $v$ is the node in the topmost element of $\pi$.
Obviously this does not extend the expressive power of the \abb{i-ptt}: 
it is straightforward to write a subroutine that searches for the (unique observable) pebble 
on the tree, by first walking to the root and 
then executing a depth-first search of the tree until a pebble is observed. 

The \abb{i-ptt} $\cM$ has input alphabet $\Sigma$ and output alphabet $\Gamma$. 
Its set $C_\mathrm{i}$ of pebble colours consists of all pairs $(\rho,i)$ where 
$\rho$ is an output rule of $\cM_1$, i.e., a rule of the form 
$\tup{p,\sigma,j}\to \delta(\tup{p_1,\stay},\dots,\tup{p_m,\stay})$ with $p,p_1,\dots,p_m\in P$,
and $i$ is a child number of $\Delta$, i.e., $i\in[0,\m_\Delta]$.
The set of states of $\cM$ is defined to be $Q\cup (P\times[0,\m_\Delta]\times Q)$ 
and the set of initial states is $\{p_0\}\times \{0\}\times Q_0$. 
A state $q\in Q$ is used by $\cM$ when simulating a computation step of $\cM_2$,
and a state $(p,i,q)$ is used by $\cM$ when simulating the computation of $\cM_1$ 
that generates the $i$-th child of the current node of $\cM_2$ (keeping the state $q$ of $\cM_2$
in memory). Initially, $\cM$ simulates $\cM_1$ in order to generate the root of its output tree. 
The rules of $\cM$ are defined as follows. 

First we define the rules that simulate $\cM_1$. Let $\rho: \tup{p,\sigma,j}\to\zeta$ be a rule in~$R_1$.
If $\zeta=\tup{p',\alpha}$ and $\alpha$ is a move instruction, 
then $\cM$ has the rules $\tup{(p,i,q),\sigma,j,b}\to \tup{(p',i,q),\alpha}$ 
for every $i\in[0,\m_\Delta]$, $q\in Q$, and $b\subseteq C_\mathrm{i}$ with $\#(b)\leq 1$. 
If $\rho$ is an output rule with $\zeta=\delta(\tup{p_1,\stay},\dots,\tup{p_m,\stay})$, 
then $\cM$ has the rules $\tup{(p,i,q),\sigma,j,b}\to \tup{q,\drop_{(\rho,i)}}$
for every $i$, $q$, $b$ as above.
Thus, $\cM$ simulates $\cM_1$ until $\cM_1$ generates an output node, 
drops the corresponding pebble,
and continues simulating $\cM_2$. 

Second we define the rules that simulate $\cM_2$.  
Let $\tup{q,\delta,i}\to \zeta$ be a rule in~$R_2$ and 
let $\rho: \tup{p,\sigma,j}\to \delta(\tup{p_1,\stay},\dots,\tup{p_m,\stay})$ be an output rule in~$R_1$
(with the same $\delta$). 
Then $\cM$ has the rule $\tup{q,\sigma,j,\{(\rho,i)\}}\to \zeta'$ 
where $\zeta'$ is defined as follows. 
If $\zeta=\tup{q',\down_\ell}$, then $\zeta'= \tup{(p_\ell,\ell,q'),\stay}$.
If $\zeta=\tup{q',\up}$, then $\zeta'= \tup{q',\lift_{(\rho,i)};\totop}$. 
Otherwise, $\zeta'=\zeta$.
Thus, $\cM$ simulates every output rule or stay rule of $\cM_2$ without changing 
its current node and current pebble stack, 
because the current node of $\cM_2$ stays the same. 
To simulate a $\down_\ell$-instruction of $\cM_2$, $\cM$ starts simulating $\cM_1$ 
in state $p_\ell$ with the child number $\ell$ of the next node of~$\cM_2$. 
Finally, $\cM$ simulates an $\up$-instruction of~$\cM_2$ by lifting its topmost pebble and  
walking to the new topmost pebble, where it continues the simulation of~$\cM_2$. 
\end{proof}

Taking Theorem~\ref{thm:composition} and 
Lemma~\ref{lem:nul-decomp} together, 
we obtain that 
$\dTT \circ \dTT \subseteq \IdPTT \subseteq 
\IPTT \subseteq \TT \circ \TT$.
It is open whether or not the first and last inclusions are proper. 
A way to express $\IdPTT$ and $\IPTT$ in terms of tree-walking tree transducers 
(without pebbles) would be to allow those transducers to have infinite input and output trees. 
Let us denote by $\dTT^\infty$ the class of transductions realized by deterministic \abb{tt}'s
that have finite input trees but can output infinite trees. As a particular example, the \abb{tt} $\cN$ 
in the proof of Lemma~\ref{lem:nul-decomp} can be turned into such a deterministic \abb{tt} $\cN^\infty$
by removing all rules $\langle f,\sigma,j\rangle \to \bot$. 
This $\cN^\infty$ preprocesses every input tree $t$ into a unique ``tree of trees'' $t_\infty$ 
consisting of top level $t$ and infinitely many levels of copies $\hat{t}_u$ of $t$. 
Moreover, let us denote by ${}^\infty\TT$ the class of transductions realized by \abb{tt}'s 
that output finite trees but can walk on infinite input trees, and similarly for ${}^\infty\dTT$. 
It should be clear that the \abb{tt} $\cM'$ in the proof of Lemma~\ref{lem:nul-decomp} 
can also be viewed as working on input tree $t_\infty$ rather than a nondeterministically generated $t'$
(and thus never aborts its simulation of $\cM$). It should also be clear that the proof of 
Theorem~\ref{thm:composition} still works when $\cM_1$ produces an infinite output tree 
as input tree for $\cM_2$.\footnote{To see that $L(\cM_1)$ is regular,
construct an ordinary nondeterministic \abb{tt} $\cN$ by adding to $\cM_1$ all rules
$\tup{q,\sigma,j}\to \bot$ such that $\cM_1$ has no rule with left-hand side $\tup{q,\sigma,j}$,
and all rules $\tup{q,\sigma,j}\to \top$ such that $\cM_1$ has a rule with that left-hand side
(where $\bot$ and $\top$ are new output symbols of rank~0).
Then $L(\cM_1)$ is the complement of $\tau_{\cN}^{-1}(R)$ 
where $R$ is the set of output trees of $\cN$ 
with an occurrence of $\bot$. Now use Proposition~\ref{prop:invtypeinf}(1).
}
Taking these results together, we obtain that 
$\IdPTT = \dTT^\infty \circ {}^\infty\dTT$ and
$\IPTT = \dTT^\infty \circ {}^\infty\TT$.
The formal definitions are left to the reader. 
Other characterizations of $\IdPTT$ will be shown in Section~\ref{sec:variations} 
(Theorem~\ref{thm:charidptt}), where we also show that $\IdPTT \subseteq \dTT^3$
(Corollary~\ref{cor:idptt-tt3}).

\smallpar{Bottom-up tree transducers}
The classical top-down and bottom-up tree transducers are compared to the \abb{v-ptt}
at the end of~\cite[Section~3.1]{MilSucVia03}. 
Obviously, \abb{tt}'s generalize top-down tree transducers. In fact,
the latter correspond to \abb{tt}'s that do not use the move instructions $\up$ and $\stay$.
Moreover, the classical top-down tree transducers with regular look-ahead 
can be simulated by \abb{tt}'s with \abb{mso} head tests, and hence by \abb{i-ptt}'s.
In general, bottom-up tree transducers cannot be simulated 
by \abb{v-ptt}'s, because otherwise every regular tree language could be accepted
by a \abb{v-pta} (see below for the details), 
which is false as proved in~\cite{expressive}. 
We will show that every bottom-up tree transducer can be simulated by an \abb{i-ptt}. 
This will not be used in the following sections. 

A \emph{bottom-up tree transducer} is a tuple $\cM=(\Sigma,\Delta,P,F,R)$ where 
$\Sigma$ and $\Delta$ are ranked alphabets, 
$P$ is a finite set of states with a subset $F$ of final states, and 
$R$ is a finite set of rules of the form 
$\sigma(p_1(x_1),\dots,p_m(x_m))\to p(\zeta)$ such that $m\in\nat$, $\sigma\in\Sigma^{(m)}$,
$p_1,\dots,p_m,p\in P$ and $\zeta\in T_\Delta(\{x_1,\dots,x_m\})$. 
For $p\in P$, the sets $\tau_p\subseteq T_\Sigma\times T_\Delta$ 
are defined inductively as follows: 
the pair $(\sigma(t_1,\dots,t_m),s)$ is in~$\tau_p$ if there is a rule as above 
and there are pairs $(t_i,s_i)\in \tau_{p_i}$ for all $i\in[1,m]$ such that 
$s=\zeta[s_1,\dots,s_m]$, which is the result of substituting $s_i$ 
for every occurrence of $x_i$ in $\zeta$. 
The transduction $\tau_\cM$ realized by $\cM$ is 
the union of all~$\tau_p$ with $p\in F$.
The transducer $\cM$ is deterministic if it does not have 
two rules with the same left-hand side.
For more information see, e.g., \cite[Chapter~IV]{GecSte}.

For every regular tree language $L$ there is a 
deterministic bottom-up finite-state tree automaton $\cA=(\Sigma,P,F,\delta)$ 
(see the proof of Lemma~\ref{lem:regular}) that recognizes $L$
and hence there is a deterministic bottom-up tree transducer $\cM$ 
that realizes the transduction $\tau_L=\{(t,1)\mid t\in L\}\cup \{(t,0)\mid t\notin L\}$.
In fact, $\cM=(\Sigma,\{0,1\},P,F,R)$ where $0$ and $1$ have rank~0 and 
$R$ is the set of all rules $\sigma(p_1(x_1),\dots,p_m(x_m))\to p(i)$
such that $\delta(\sigma,p_1,\dots,p_m)=p$ and 
$i = 1$ if $p\in F$, $i=0$ otherwise. 
A \abb{v-ptt} that computes $\tau_L$ can be turned into a \abb{v-pta}
that accepts $L$ by removing every output rule $\tup{q,\sigma,j,b}\to 0$ and
changing every output rule $\tup{q,\sigma,j,b}\to 1$ into 
$\tup{q,\sigma,j,b}\to \tup{q_\mathrm{fin},\stay}$ 
where $q_\mathrm{fin}$ is the final state. 

Let $\family{B}$ ($\family{dB}$) denote the class of transductions 
realized by (deterministic) bottom-up tree transducers. 

\begin{theorem}\label{thm:bottomup}
$\family{B}\subseteq \IPTT$ and $\family{dB}\subseteq \IdPTT$.
\end{theorem}

\begin{proof}
Let $\cM=(\Sigma,\Delta,P,F,R)$ be a bottom-up tree transducer. 
Intuitively, for a given input tree $t$, the transducer $\cM$ visits each node $u$ of $t$ 
exactly once. It arrives at the children of $u$ in certain states $p_1,\dots,p_m$ 
with certain output trees $s_1,\dots,s_m$, 
and applies a rule $\sigma(p_1(x_1),\dots,p_m(x_m))\to p(\zeta)$ where $\sigma$ 
is the label of $u$. Thus, it arrives at $u$ in state $p$ with output $\zeta[s_1,\dots,s_m]$.  

We construct an \abb{i-ptt} $\cM'$ with \abb{mso} head tests 
such that $\tau_{\cM'} =\tau_\cM$ (see Lemma~\ref{lem:sites}). 
The transducer $\cM'$ uses the rules of $\cM$ as pebble colours.  
The behaviour of $\cM'$ on a given input tree $t$ is divided into two phases. 
In the first phase $\cM'$ walks through $t$ and (nondeterministically) 
drops one pebble $c$ on each node $u$ of $t$, in post-order. The input symbol $\sigma$ 
in the left-hand side of rule~$c$ must be the label of $u$. Intuitively, 
$c$ is the rule $\sigma(p_1(x_1),\dots,p_m(x_m))\to p(\zeta)$
applied by $\cM$ at $u$ during a possible computation. When $\cM$ drops $c$ on $u$
it uses \abb{mso} head tests to check that $\cM$ has a computation on $t$ 
that arrives at the $i$-th child $ui$ of $u$ in state $p_i$, for every $i\in[1,m]$.
This can be done because the state behaviour of $\cM$ on $t$ is that of a 
bottom-up finite-state tree automaton. Thus, the tree language 
$L_p = \{t\in T_\Sigma\mid \exists\, s: (t,s)\in\tau_p\}$ is regular for every $p\in P$
and hence the site $T_i=\{(t,u)\mid t|_{ui} \in L_{p_i}\}$ is also regular, 
as can easily be seen. Note that if $\cM$ is deterministic, then this first phase of $\cM'$ 
is deterministic too, because $\cM$ arrives at each node in a unique state
(during a successful computation). In the second, deterministic phase $\cM'$ moves top-down 
through $t$, checks that the states in the guessed rules are consistent, and 
computes the corresponding output. First $\cM'$ checks for the pebble 
$c=\sigma(p_1(x_1),\dots,p_m(x_m))\to p(\zeta)$ at the root $u$, that the state $p$ is in~$F$. 
If so, it starts a process that is the same for every node $u$ of $t$. 
It lifts pebble $c$ and goes into state $[c,\zeta]$, in which it will output 
the $\Delta$-labeled nodes of $\zeta$, without leaving $u$. 
In state $q=[c,\delta(\zeta_1,\dots,\zeta_n)]$, 
it uses the output rules $\tup{q,\sigma,j,\nothing}\to 
\delta(\tup{[c,\zeta_1],\stay},\dots,\tup{[c,\zeta_n],\stay})$. 
When $\cM'$ is in a state $[c,x_i]$, it calls a subroutine $S_i$.
Subroutine $S_i$ walks through the subtrees $t|_{um},\dots,t|_{u(i+1)}$ of $t$, 
depth-first right-to-left, lifts the pebbles at all the nodes of those trees 
in reverse post-order (which is possible because 
the pebbles were dropped in post-order), and returns control to $\cM'$,
which continues by moving in state $c$ to child~$ui$ 
where it observes the pebble at $ui$ (again, because of the post-order dropping).
Then $\cM$ checks that the state in the right-hand side of that pebble is $p_i$,
and repeats the above process for node $ui$ instead of $u$.
It should be clear that in this way $\cM'$ simulates the computations of $\cM$, 
and so $\tau_{\cM'} =\tau_\cM$. Note that the bottom-up transducer $\cM$ can disregard
computed output, because in a rule as above it may be that $x_i$ 
does not occur in $\zeta$. In such a case $\cM'$ clearly does not compute 
that output either, in the second phase, whereas it has checked in the first phase 
that $\cM$ indeed has a computation that arrives in state $p_i$ at the $i$-th child.
Note also that if $x_i$ occurs twice in $\zeta$, then $\cM'$ simulates in the second phase 
twice the same computation of $\cM$ on the $i$-th subtree 
(which was guessed in the first phase). 
\end{proof}

\section{Look-Ahead Tests}\label{sec:look-ahead}

The results on look-ahead in this section are only needed in 
the next section (and in a minor way in Section~\ref{sec:pft}). They 
also hold for the \abb{pta} as navigational device, computing a trip.  

We say that a family $\cF$ of \abb{PTA}'s (or \abb{PTT}'s)
can \emph{perform look-ahead tests}
if an automaton (or transducer) $\cA$ in $\cF$ can test whether or not a  
\abb{PTT} $\cB$ (not necessarily in $\cF$)
has a successful computation when started 
in the current situation of $\cA$
(i.e., position of the head and stack of pebbles).
We require that  
$\Sigma^\cA = \Sigma^\cB$, 
$C_\mathrm{v}^\cA \subseteq C_\mathrm{v}^\cB$,
$C_\mathrm{i}^\cA \subseteq C_\mathrm{i}^\cB$,
and $k^\cA \leq k^\cB$
(where $\Sigma^\cA$ is the input alphabet of $\cA$, and similarly for the other notation). 
Since we are only interested in the existence of a successful computation,
and not in its output tree, we are actually using alternating \abb{PTA}'s
as look-ahead device (cf. Section~\ref{sec:autotrans}). 
In particular, we also allow a \abb{PTA} to be used as look-ahead $\cB$,
viewing it as a \abb{PTT} as in the proof of Theorem~\ref{thm:regt}.

In the formal definition of a \abb{PTA} or \abb{PTT}
\emph{with look-ahead tests} (cf. the formal definition of \abb{mso} head tests 
before Lemma~\ref{lem:sites}), the rules are of the form 
$\tup{q,\sigma,j,b,\cB} \to \zeta$ or $\tup{q,\sigma,j,b,\neg\,\cB} \to \zeta$
which are relevant to a given configuration $\tup{q,h,\pi}$ of $\cA$ on tree $t$
if the transducer $\cB$ does or does not have a successful computation on $t$ 
that starts in the situation $\tup{h,\pi}$, i.e., 
if there do or do not exist $p_0\in Q_0^{\cB}$ and $s\in T_{\Delta^{\cB}}$
such that $\tup{p_0,h,\pi}\Rightarrow^*_{t,\cB} s$
(where $\Delta^{\cB}$ is the output alphabet of $\cB$), 
or in the case of a \abb{pta} $\cB$, if there do or do not exist
$p_0\in Q_0^{\cB}$, $p_f\in F^{\cB}$, 
and $\tup{u,\pi}\in \sit^{\cB}(t)$ such that  
$\tup{p_0,\rt_t,\epsilon}\Rightarrow^*_{t,\cB} \tup{p_f,u,\pi}$
(where $F^{\cB}$ is the set of final states of $\cB$).

\begin{theorem}\label{thm:look-ahead}
For each $k \ge 0$, the \abb{v$_k$i-pta} and \abb{v$_k$i-}{\em d}\abb{pta} 
can perform look-ahead tests.
The same holds for the \abb{v$_k$i-ptt} and \abb{v$_k$i-}{\em d}\abb{ptt}.
\end{theorem}

\begin{proof}
Let $\cA$ be a \abb{v$_k$i-pta} that performs 
a look-ahead test by calling some \mbox{\abb{v$_m$i-ptt}} $\cB$ (with $k\leq m$).
We wish to construct a trip-equivalent \abb{v$_k$i-pta} $\cA'$ 
that does not perform such look-ahead tests. 
By Lemma~\ref{lem:stacktests} we may construct $\cA'$ as a \abb{pta} with stack tests, 
i.e., a \abb{pta} that
can test whether its pebble stack is empty and if so, 
what the colour of the topmost pebble is. 
 
As usual, $\cA'$ simulates $\cA$. 
Suppose that $\cA$ uses the look-ahead test $\cB$ in situation $\tup{h,\pi}$.
When no pebbles are dropped, i.e., $\pi=\epsilon$, the test whether~$\cB$, 
started in that situation, has a successful computation,
is an \abb{mso} head test. Indeed, the site 
$T=\{(t,h)\mid \exists\, p_0\in Q_0^{\cB}, s\in T_{\Delta^{\cB}}:  
\tup{p_0,h,\epsilon}\Rightarrow^*_{t,\cB} s\}$
is regular,
as $\tmark(T)$ is the domain of the \abb{v$_m$i-ptt} $\cB'$ that
starts in the root, looks for the marked node $h$,
and then simulates $\cB$. Domains are regular 
by Corollary~\ref{cor:domptt}, and $\cA'$ can perform 
\abb{mso} head tests by Lemma~\ref{lem:sites}.

In general, one may imagine that $\cA'$ implements the look-ahead test 
by simulating $\cB$. However, when $\cA'$ 
is ready with the simulation of $\cB$, that started 
with the stack $\pi$ of $\cA$, 
$\cA'$ must be able to recover $\pi$ to continue the simulation of~$\cA$. 
Note that $\cB$ can inspect $\pi$, thereby possibly destroying 
part of $\pi$ and adding something else. 
For this reason the computations of
$\cB$ starting at the position of the topmost pebble of~$\pi$ will be precomputed.
With each pebble dropped by $\cA$, the automaton $\cA'$ stores the set $S$ of states
$p$ of $\cB$ for which $\cB$ has a successful computation when
started in state $p$ at the position $u$ of the topmost
stack element (and with the current stack of $\cA$). 
Now a successful computation of $\cB$ can be safely simulated,
consisting of a part where the pebbles of $\cB$ are on top
of the stack $\pi$ inherited from $\cA$, possibly followed by
a precomputed part where $\cB$ inspects $\pi$, starting with a visit to~$u$.
We discuss how these state sets are determined, and how they are used 
(by $\cA'$) to perform the look-ahead test. Rather then simulating~$\cB$,
$\cA'$ will use \abb{mso} tests on the observable configuration,
which is possible by Theorem~\ref{thm:mso}. 
The colour sets of $\cA'$ are $C'_\mathrm{v} = C_\mathrm{v} \times 2^{Q^\cB}$
and $C'_\mathrm{i} = C_\mathrm{i} \times 2^{Q^\cB}$.

If $\cA$ drops the first pebble $c$ (i.e., $\pi=(h,c)$), 
then $\cA'$ drops the pebble $(c,S)$ where it determines 
for every state $p$ of $\cB$ whether or not $p\in S$ 
using an \abb{mso} head test: construct $\cB'$ as above except that
it now drops $c$ at the marked node~$h$ before simulating $\cB$ in state $p$.
Thus, this time, the domain of $\cB'$ is $\tmark(T)$ with 
$T=\{(t,h)\mid \exists\, s\in T_{\Delta^{\cB}}:  
\tup{p,h,c}\Rightarrow^*_{t,\cB} s\}$. 

Suppose now that $\cA$ uses the look-ahead test $\cB$ when it is in situation $\tup{h,\pi}$
with $\pi\neq\epsilon$, and suppose that the
topmost pebble of $\pi$ has colour $d$ and 
that the set of visible pebble colours 
that occur in $\pi$ is $C_\mathrm{v}(\pi)=\{c_1,\dots,c_\ell\}\subseteq C_\mathrm{v}$,
with $\ell\in[0,k]$. 
Then the colour of the topmost pebble of the stack $\pi'$ of $\cA'$ is $(d,S)$ 
for some set $S$ of states of $\cB$, 
and the set of visible pebble colours that occur in $\pi'$ is 
$C_\mathrm{v}(\pi')=\{(c_1,S_1),\dots,(c_\ell,S_\ell)\}$ for some $S_1,\dots,S_\ell$.
Since $\cA'$ can perform stack tests, it can determine $(d,S)$.
Moreover, it should be clear that $\cA'$ can determine $C_\mathrm{v}(\pi')$, 
and hence $C_\mathrm{v}(\pi)$, by an \abb{mso} test on the visible configuration.
With this topmost colour $d$, this state information $S$,
and this set $C_\mathrm{v}(\pi)$ of visible pebbles, 
the look-ahead test can be performed by $\cA'$ as an \abb{mso} test on 
the observable configuration, as follows.
Consider the observable configuration tree 
$\obs(t,\pi')$ with the current node $h$ marked, 
see Theorem~\ref{thm:mso}.
We want to show that there is a regular site $T$ over $\Sigma \times 2^{C'}$
such that $(\obs(t,\pi'),h)\in T$ if and only if 
there exist $p_0\in Q_0^{\cB}$ and $s\in T_{\Delta^{\cB}}$
such that $\tup{p_0,h,\pi}\Rightarrow^*_{t,\cB} s$.
Indeed, $\tmark(T)$ is the domain of  
a \abb{v$_{m'}$i-ptt} $\cB'$, with $m'= m-\ell$. 
It first searches for the position
$u$ of the topmost pebble, which is the unique node of $\obs(t,\pi')$
of which the label contains the colour $(d,S)$. 
It drops the special invisible pebble $\odot$ on $u$,
and then proceeds to the marked node $h$,
starts simulating $\cB$ 
and halts successfully when 
it observes pebble $\odot$ at position $u$ with $\cB$ in a state of $S$,
or when it never has observed $\odot$ and $\cB$ halts successfully
(meaning that pebbles are still on top of $\odot$ when visiting $u$).
Note that $\cB'$ can simulate~$\cB$, 
which walks on $t$ with pebbles rather than on $\obs(t,\pi')$, 
because the colours in the labels of the nodes of $\obs(t,\pi')$
contain the observable pebbles on $t$ in the stack~$\pi$.
Also, $\cB'$ does not apply rules of $\cB$ that contain a $\drop_{c_i}$-instruction
with $c_i\in C_\mathrm{v}(\pi)$. 
The domain $\tmark(T)$ of $\cB'$ is regular 
and $\cA'$ can perform the \abb{mso} test $T$
on its observable configuration.

The same reasoning shows that the state set for the 
next pebble $c$ dropped by $\cA$ can be computed
by \abb{mso} tests on the observable configuration:
again $\cB'$ first drops the pebble $c$ on $h$ before starting the
simulation of $\cB$ in any state $p$.

Finally it should be clear that if $\cA$ uses the look-ahead tests 
$\cB_1,\dots,\cB_n$, then state information for every $\cB_i$ 
should be stored in the pebbles, i.e., they are of the form $(c,S_1,\dots,S_n)$
where $S_i$ is a set of states of $\cB_i$. 
\end{proof}

A natural question is now whether Theorem~\ref{thm:look-ahead} also holds
for \abb{pta}'s and \abb{ptt}'s that are allowed to perform stack tests, 
\abb{mso} head tests, and \abb{mso} tests on the visible and 
observable configuration. The answer is yes. 

Let us first consider the case of stack tests. 
Roughly speaking, if $\cA$ uses look-ahead tests $\cB_1,\dots,\cB_n$, 
then we just apply the construction of Lemma~\ref{lem:stacktests} 
to both $\cA$ and all $\cB_i$, $i\in[1,n]$,
and then apply Theorem~\ref{thm:look-ahead} to the resulting equivalent (ordinary)
\abb{pta} $\cA'$ that calls the (ordinary) \abb{ptt}'s $\cB'_1,\dots,\cB'_n$. 
It should be noted that even if $\cA$
does \emph{not} use stack tests but some $\cB_i$ \emph{does}, 
the construction of Lemma~\ref{lem:stacktests} must be applied to $\cA$ too,
because the stack that $\cB_i$ inherits 
from~$\cA$ must contain the necessary additional information concerning 
the colours of previously dropped pebbles. Vice versa, 
if $\cA$ (or another $\cB_j$) uses stack tests but $\cB_i$ does not,
then $\cB_i$ can just ignore the additional information in the stack of $\cA$,
but it is also correct to apply the construction of Lemma~\ref{lem:stacktests} 
to~$\cB_i$.  
However, not only the additional information in the stack 
should be passed from $\cA'$ to $\cB_1',\dots,\cB_n'$, 
but also the additional information 
in the finite state of~$\cA'$. Thus, to be more precise, if $\cA$ is in state $q$ 
and uses the look-ahead test $\cB_i$, then whenever $\cA'$ is in state 
$(q,\gamma)$, it should use the look-ahead test $\cB'_i(\gamma)$ 
that is obtained from $\cB'_i$ by changing its set 
$Q_0^{\cB_i} \times \{\epsilon\}$ of initial states into 
$Q_0^{\cB_i} \times \{\gamma\}$. 

For the case of \abb{mso} head tests and \abb{mso} tests on the visible configuration 
the proof is easier. The constructions of Lemmas~\ref{lem:sites} 
and~\ref{lem:visiblesites} can be applied to $\cA$ and $\cB_1,\dots,\cB_n$ independently, 
depending on whether they use such tests or not. The reason is that 
these tests are implemented by subroutines for which the pebble stack 
need not be changed. 
Finally, for the case of \abb{mso} tests on the observable configuration
the construction of Theorem~\ref{thm:mso} is again applied 
simultaneously to all of $\cA$ and $\cB_1,\dots,\cB_n$, with beads that take care of 
all the regular sites $T$ that are used by both $\cA$ and $\cB_1,\dots,\cB_n$ as tests.
That ensures that the beads of $\cA'$ also contain the information 
needed by $\cB'_1,\dots,\cB'_n$. Note that in this case 
(as opposed to the case of stack tests above) $\cA'$ does not carry 
any additional information in its finite state and thus, 
whenever $\cA$ uses $\cB_i$ as look-ahead test, 
$\cA'$ can use $\cB'_i$ as look-ahead test.

A similar natural question is whether Theorem~\ref{thm:look-ahead} also holds
for \abb{pta}'s and \abb{ptt}'s that use look-ahead, in particular whether we
can allow the look-ahead transducer to use another transducer  
as look-ahead test. The answer is again yes, with a similar solution. 
In fact it can be shown that the \abb{v$_k$i-pta} (and \abb{v$_k$i-ptt}) even can perform
\emph{iterated} look-ahead tests, that is, they can use look-ahead tests 
that use look-ahead tests that use $\dots$ look-ahead tests. 

Formally, we define for $n\geq 0$ the notion of 
a \abb{pta} or \abb{ptt} $\cA$ \emph{of (look-ahead) depth} $n$, by induction on $n$.
Simultaneously we define the finite sets $\test(\cA)$ and $\test^*(\cA)$ of \abb{ptt}'s,
where $\test(\cA)$ contains the look-ahead tests of $\cA$, 
and $\test^*(\cA)$ contains its iterated look-ahead tests plus $\cA$ itself. 
For $n=0$, a \abb{pta} or \abb{ptt} $\cA$ of depth~$0$ is 
just a \abb{pta} or \abb{ptt} (without look-ahead tests). Moreover, 
$\test(\cA)=\nothing$ and $\test^*(\cA)=\{\cA\}$.
For $n\geq 0$, a \abb{pta} or \abb{ptt} $\cA$ of depth~$n+1$ 
uses as look-ahead tests arbitrary \abb{ptt}'s of lower depth, 
i.e., it has rules 
$\tup{q,\sigma,j,b,\cB} \to \zeta$ or $\tup{q,\sigma,j,b,\neg\,\cB} \to \zeta$
where $\cB$ is a \abb{ptt} of depth~$m\leq n$. 
Furthermore, $\test(\cA)$ is the set of all \abb{ptt}'s of depth~$m\leq n$ 
that $\cA$ uses as look-ahead tests, and 
$\test^*(\cA)=\{\cA\} \cup \bigcup_{\cB\in \test(\cA)}\test^*(\cB)$.
A \abb{pta} or \abb{ptt} \emph{with iterated look-ahead tests} is one of depth $n$, 
for some $n\in\nat$. 
Note that a \abb{pta} (or \abb{ptt}) of depth~$1$ is the same as a 
\abb{pta} (or \abb{ptt}) with look-ahead tests. 
The definition of the semantics of a \abb{pta} or \abb{ptt} with iterated look-ahead tests
is by induction on the depth $n$, and  
is entirely analogous to the one for the case $n=1$ 
as given in the beginning of this section. 

\begin{theorem}\label{thm:iterated}
For each $k \ge 0$, the \abb{v$_k$i-pta} and \abb{v$_k$i-}{\em d}\abb{pta} 
can perform iterated look-ahead tests.
The same holds for the \abb{v$_k$i-ptt} and \abb{v$_k$i-}{\em d}\abb{ptt}.
\end{theorem}

\begin{proof}
We will show that 
for every \abb{v$_k$i-ptt} $\cC$ of depth $n\geq 1$ we can construct an equivalent 
\abb{v$_k$i-ptt} $\cC'$ of depth $n-1$. 
The result then follows by induction.
Since the construction generalizes the one of Theorem~\ref{thm:look-ahead}
(which is the case $n=1$),
we will need all \abb{ptt}'s in $\test^*(\cC')$ to use stack tests and 
\abb{mso} tests on the observable configuration. Thus, for the induction to work, 
we first have to prove that every \abb{v$_\ell$i-ptt} of depth $m\geq 1$ 
can perform such tests. For the case $m=1$ we have already argued this after 
Theorem~\ref{thm:look-ahead}, and the general case can be proved in a similar way. 
Let $\cD$ be a \abb{v$_\ell$i-ptt} of depth $m$ 
such that all $\cA\in\test^*(\cD)$ perform stack tests. 
We just apply the construction of 
Lemma~\ref{lem:stacktests} simultaneously to every
\abb{ptt} $\cA\in\test^*(\cD)$, resulting in the \abb{ptt} $\cA'$.
Moreover, for all $\cA,\cB\in \test^*(\cD)$, if $\cA$ is in state $q$ 
and uses look-ahead test $\cB$, then whenever $\cA'$ is in state $(q,\gamma)$,
it uses look-ahead test $\cB'(\gamma)$. 
Obviously, every $\cB'(\gamma)$ is of the same depth as $\cB$, 
and hence the resulting \abb{v$_\ell$i-ptt} $\cD'$ is of the same depth $m$ as~$\cD$. 
For the \abb{mso} tests the argument 
is completely analogous to the argument for $m=1$ after Theorem~\ref{thm:look-ahead},
applying the appropriate constructions simultaneously to 
all \abb{ptt} $\cA\in\test^*(\cD)$. 

Now let $\cC$ be a \abb{v$_k$i-ptt} of depth $n\geq 1$ and let us construct 
an equivalent \abb{v$_k$i-ptt} $\cC'$ of smaller depth. The argument is similar to 
those above. Let $P_0$ be the set of all $\cB\in \test^*(\cC)$ of depth~0,
i.e., all \abb{ptt} without look-ahead tests, and let $P_1$ contain all 
$\cA\in \test^*(\cC)$ of depth~$\geq 1$. We now apply the construction of 
Theorem~\ref{thm:look-ahead} simultaneously to every \abb{ptt} $\cA\in P_1$,
resulting in a \abb{ptt} $\cA'$ that 
stores state information of every $\cB\in P_0$ in the pebbles. 
If $\cA_1\in P_1$ uses look-ahead test $\cA_2\in P_1$, then 
$\cA'_1$ uses look-ahead test $\cA'_2$. Note that if 
$\cA\in P_1$ uses look-ahead test $\cB\in P_0$, then 
$\cA'$ uses an \abb{mso} test instead. Thus, clearly, 
the depth of every $\cA'$ is one less than the depth of $\cA$,
and so the depth of the resulting \abb{v$_k$i-ptt} $\cC'$ is $n-1$. 
Finally, we remove the stack tests 
and \abb{mso} tests from $\cC'$ and its iterated look-ahead tests as 
explained above for $\cD$. 
\end{proof}

Although this result does not seem practically useful, 
it will become important when we propose the query language Pebble XPath
in the next section, as an extension of Regular XPath. 
Intuitively, Pebble XPath expressions are similar 
to \abb{i-pta} with iterated look-ahead tests.
We note that \abb{ta} with iterated look-ahead tests are used in~\cite{CatSeg}
to prove that Regular XPath is not \abb{mso} complete.

\section{Document Navigation}\label{sec:xpath}

We define \emph{Pebble XPath}, an extension of Regular XPath \cite{Mar05} with pebbles.
Due to its potential application to navigation in XML documents, 
it works on (nonempty) forests rather than trees. 
We prove that the trips defined by the path expressions of Pebble XPath
are exactly the \abb{mso} definable trips on forests.

Pebble XPath has path expressions (denoted $\alpha,\beta$) and node expressions (denoted $\phi,\psi$). 
These expressions concern forests over an (unranked) alphabet~$\Sigma$ 
of node labels, or tags, that can be chosen arbitrarily. 
Since we are mainly interested in path expressions, we view the node expressions as auxiliary. 
A path expression describes a walk through a given nonempty forest $f$ over $\Sigma$ during which invisible coloured pebbles can be dropped on and lifted from the nodes of~$f$, in a nested (stack-like) manner. Such a walk steps through $f$ from node to node following both the vertical and horizontal edges in either direction. 
The context in which a path expression is evaluated (i.e., the situation at the start of the walk) is a pair $\tup{u,\pi}$ consisting of a node $u$ of $f$ and a stack $\pi$ of pebbles that lie on the nodes of $f$. Formally, a \emph{context}, or \emph{situation}, on a forest $f$ is an element of the set $\sit(f) = \nod{f} \times (\nod{f}\times C)^*$, where $\nod{f}$ is the set of nodes of $f$ and $C$ is the finite set of colours of the pebbles (that can be chosen arbitrarily). The walk ends in another context. Thus, the semantics of a path expression is a binary relation on $\sit(f)$. Similarly, the semantics of a node expression is a subset of $\sit(f)$, i.e., a test on a given context. Note that the notion of a context on a forest is entirely similar to that of a situation on a ranked tree for an \abb{i-pta} with (invisible) colour set $C$. 

For the syntax of Pebble XPath, we start with the basic path expressions, with $c\in C$:
$$\alpha_0 ::= \child \mid \parent \mid \rght \mid \lft \mid \dropt_c \mid \liftt_c$$ 
The first four expressions operate on the context node only (in the usual way, moving to a child, the parent, the next sibling, and the previous sibling, respectively), whereas the last two also operate on the pebble stack (dropping/lifting a pebble of colour $c$ on/from the context node $u$, which is modeled by pushing/popping the pair $(u,c)$ on/off the stack). The syntax of path expressions is
$$\alpha ::= \alpha_0 \mid \;?\phi \mid \alpha \cup \beta \mid \alpha/\beta \mid \alpha^*$$
where $\beta$ is an alias of $\alpha$.
The three last expressions show the usual regular operations on binary relations: union, composition, and transitive-reflexive closure. The expression $?\phi$ denotes the identity relation on the set of contexts defined by the node expression $\phi$, i.e., it filters the current context by requiring that $\phi$ is true. 

We now turn to the node expressions and start with the basic ones, with $\sigma\in\Sigma$:
$$\phi_0 ::=  \haslab_\sigma \mid \isleaf \mid \isroot \mid \isfirst \mid \islast \mid \haspeb_c$$
The first five expressions test whether the context node has label $\sigma$, whether it is a leaf, a root, the first among its siblings, or the last among its siblings. The last expression (which is the only one that also uses the pebble stack) tests whether the topmost pebble, i.e., the most recently dropped pebble, lies on the context node and has colour $c$. The syntax of node expressions is
$$\phi ::= \phi_0 \mid \tup{\alpha} \mid \neg\phi \mid \phi\wedge\psi \mid \phi\vee\psi$$
where $\psi$ is an alias of $\phi$.
The last three expressions show the usual boolean operations. The expression $\tup{\alpha}$ is like a predicate $[\alpha]$ in XPath 1.0, which filters the context by requiring the existence of at least one successful $\alpha$-walk starting from this context. In terms of tree-walking automata it is a look-ahead test.
We will also consider the language \emph{Pebble CAT}, which is obtained from Pebble XPath by dropping the filter tests $\phi ::= \langle\alpha\rangle$. The expressions of Pebble CAT are \emph{caterpillar expressions} extended with pebbles. 

The formal semantics of Pebble XPath expressions is given in Tables~\ref{tab:sempath} and~\ref{tab:semnode}. For every nonempty forest $f$ over $\Sigma$, the semantics $\semf{\alpha}\subseteq \sit(f)\times\sit(f)$ and $\semf{\phi}\subseteq \sit(f)$ of path and node expressions are defined, where $u,u'$ vary over $\nod{f}$, $\pi,\pi'$ vary over $(\nod{f}\times C)^*$, and $p$ varies over $\nod{f}\times C$. 
Note that $\semf{ \parent }=\semf{\child}^{-1}$, $\semf{ \lft }=\semf{ \rght }^{-1}$, and
$\semf{ \liftt_c }=\semf{ \dropt_c }^{-1}$.
Note also that the set $\semf{ \tup{\alpha} }$ is 
the domain of the binary relation $\semf{ \alpha }$.

\begin{table}
\[
\begin{array}{l@{\;}cl}
\semf{\child} &=& \{(\tup{u,\pi},\tup{u',\pi})\mid u' \mbox{ is a child of }u\} \\
\semf{ \parent } &=& \{(\tup{u,\pi},\tup{u',\pi})\mid u' \mbox{ is the parent of }u\} \\
\semf{ \rght } &=& \{(\tup{u,\pi},\tup{u',\pi})\mid 
    u' \mbox{ is the next sibling of }u\} \\
\semf{ \lft } &=& \{(\tup{u,\pi},\tup{u',\pi})\mid 
    u' \mbox{ is the previous sibling of }u\} \\
\semf{ \dropt_c } &=& \{(\tup{u,\pi},\tup{u,\pi p}) \mid p=(u,c)\} \\ 
\semf{ \liftt_c } &=& \{(\tup{u,\pi p},\tup{u,\pi}) \mid p=(u,c)\} \\[1mm]
\semf{ ?\phi } &=& \{(\tup{u,\pi},\tup{u,\pi})\mid 
    \tup{u,\pi}\in \semf{ \phi }\} \\
\semf{ \alpha\cup\beta } &=& \semf{ \alpha } \cup 
    \semf{ \beta } \\
\semf{ \alpha/\beta } &=& \semf{ \alpha } \circ \semf{ \beta } \\
\semf{ \alpha^* } &=& \semf{ \alpha }^* 
\end{array}
\]
\caption{Semantics of Pebble XPath path expressions}
\label{tab:sempath}
\end{table}

\begin{table}
\[
\begin{array}{l@{\;}cl}
\semf{ \haslab_\sigma } &=& \{\tup{u,\pi}\mid u \mbox{ has label }\sigma\} \\
\semf{ \isleaf } &=& \{\tup{u,\pi}\mid u \mbox{ is a leaf}\} \\
\semf{ \isroot } &=& \{\tup{u,\pi}\mid u \mbox{ is a root}\} \\
\semf{ \isfirst } &=& \{\tup{u,\pi}\mid u \mbox{ is a first sibling}\} \\
\semf{ \islast } &=& \{\tup{u,\pi}\mid u \mbox{ is a last sibling}\} \\
\semf{ \haspeb_c } &=& \{\tup{u,\pi p}
    \mid p=(u,c)\} \\[1mm]
\semf{ \tup{\alpha} } &=& \{\tup{u,\pi}\mid 
   \exists \tup{u',\pi'}\colon 
   (\tup{u,\pi},\tup{u',\pi'})\in \semf{ \alpha }\} \\
\semf{ \neg\phi } &=& \sit(f)\setminus \semf{ \phi } \\
\semf{ \phi\wedge\psi } &=& \semf{ \phi } \cap \semf{ \psi } \\
\semf{ \phi\vee\psi } &=& \semf{ \phi } \cup \semf{ \psi }
\end{array}
\]
\caption{Semantics of Pebble XPath node expressions}
\label{tab:semnode}
\end{table}

The filtering XPath expression $\alpha[\beta]$ of XPath 1.0 can here be defined as $\alpha[\beta] = \alpha/?\langle\beta\rangle$. Also, the node expression $\loops(\alpha)$ from \cite{GorMar05,Cat06} can be defined as $\loops(\alpha)= \tup{\dropt_c/\alpha/\liftt_c}$ where $c$ is a colour not occurring in $\alpha$. Then $\semf{\loops(\alpha)}= 
\{\tup{u,\pi}\mid (\tup{u,\pi},\tup{u,\pi})\in \semf{\alpha}\}=
\{\tup{u,\pi}\mid (\tup{u,\epsilon},\tup{u,\epsilon})\in \semf{\alpha}\}$, 
because $\alpha$ cannot inspect the stack $\pi$
and it must return to $u$ in order to lift pebble $c$.

Two path expressions $\alpha$ and $\beta$ are \emph{equivalent}, denoted by $\alpha\equiv\beta$,
if $\semf{\alpha}=\semf{\beta}$ for every nonempty forest $f$ over $\Sigma$, 
and similarly for node expressions. 
Note that $?(\phi\wedge\psi)\equiv \;?\phi/?\psi$ and $?(\phi\vee\psi)\equiv \;?\phi\;\cup \;?\psi$. Hence, using De~Morgan's laws, the syntax for node expressions can be replaced by $\phi::= \phi_0 \mid \neg\phi_0 \mid \langle\alpha\rangle \mid \neg\langle\alpha\rangle$ for Pebble XPath, and $\phi::= \phi_0 \mid \neg\phi_0$ for Pebble CAT.
Thus, keeping only the basic node expressions, 
we can always assume that the syntax for path expressions is 
$$\alpha ::= \alpha_0 \mid 
\;?\phi_0 \mid \;?\neg\phi_0 \mid \;?\tup{\beta} \mid \;?\neg\tup{\beta}
\mid \alpha \cup \beta \mid \alpha/\beta \mid \alpha^*$$
for Pebble XPath, and hence
$$\alpha ::= \alpha_0 \mid \;?\phi_0 \mid \;?\neg\phi_0 
\mid \alpha \cup \beta \mid \alpha/\beta \mid \alpha^*$$
for Pebble CAT. In that case we will say that we assume 
the syntax to be in \emph{normal form}. 

Note also that all basic node expressions except $\haslab_\sigma$ are redundant, because $\isleaf \equiv \neg\langle\child\rangle$ (a node is a leaf if and only if it has no children), $\isroot \equiv \neg\langle\parent\rangle$, 
$\isfirst \equiv \neg\langle\lft\rangle$, $\islast \equiv \neg\langle\rght\rangle$, and $\haspeb_c \equiv \langle\liftt_c\rangle$. However, these basic node expressions were kept in the syntax, because we also wish to consider the subset Pebble CAT in which there are no filter tests $\tup{\alpha}$. 
Note finally that when $\dropt_c$, $\liftt_c$, and $\haspeb_c$ are removed from Pebble XPath, the resulting formalism is exactly Regular XPath \cite{Mar05} (and in the semantics the stack can, of course, be disregarded). 

The purpose of Pebble XPath is the same as that of XPath: 
to define trips, i.e., binary patterns. Recall from Section~\ref{sec:trees} that 
a trip $T$ over an unranked alphabet $\Sigma$ is a set $T\subseteq\{(f,u,v)\mid f\in F_\Sigma, u,v\in \nod{f}\}$ where $F_\Sigma$ is the set of forests over $\Sigma$. Note that $f$ is always a nonempty forest. 
For a path expression~$\alpha$ (based on $\Sigma$ and some $C$) we say that $\alpha$ 
\emph{defines the trip} $T(\alpha) = \{(f,u,v)\mid 
\exists\, \pi\in(N(f)\times C)^*: (\tup{u,\epsilon},\tup{v,\pi})\in \semf{\alpha}\}$. 
We now define a trip~$T$ over $\Sigma$ 
to be \emph{definable in Pebble XPath} if there exists a Pebble XPath path expression $\alpha$ such that $T = T(\alpha)$. And similarly for Pebble CAT. The next theorem states that Pebble XPath and Pebble CAT have the same expressive power as \abb{MSO} logic on forests. 

\begin{theorem}\label{thm:xpath}
A trip is definable in Pebble XPath if and only if it is definable in Pebble CAT if and only if it is \abb{MSO} definable. 
\end{theorem}

As such our expressions have the desirable property of being a Core 
(and even Regular) XPath extension that is complete for 
\abb{MSO} definable binary patterns.
Other such extensions were considered in \cite{GorMar05}
(TMNF caterpillar expressions) and \cite{Cat06} ($\mu$Regular XPath).
Pebble CAT is similar to PCAT of \cite{GorMar05}
which has the same expressive power as the \abb{v-pta}
(and thus less than \abb{mso} by \cite{expressive}).
In PCAT the nesting of pebbles is defined syntactically
rather than semantically.

The proof of Theorem~\ref{thm:xpath} is given in the remainder of this section. 
It should be clear that Pebble CAT is closely related to the \abb{i-pta}.
In fact, we will show later that their relationship can be viewed as the classical equivalence of regular expressions and finite automata. The remainder of the proof is then directly based on the fact that the \abb{i-pta} has the same expressive power as \abb{mso} logic for defining trips on trees (Theorem~\ref{thm:trips}), and on the fact that the \abb{i-pta} can perform iterated look-ahead tests (Theorem~\ref{thm:iterated}). One technical problem is that these theorems are formulated for ranked trees rather than unranked forests. Thus we start by adapting Pebble XPath to ranked trees and showing that it suffices to prove Theorem~\ref{thm:xpath} for ranked trees instead of forests.

\smallpar{Pebble XPath on ranked trees}
Since ranked trees are a special case of unranked forests,
we need not change Pebble XPath for its use on ranked trees.
However, for its comparison to the \abb{i-pta} it is more convenient to 
change its basic path expressions $\alpha_0$ and basic node expressions $\phi_0$
as follows:
$$\alpha_0 ::= \downt_1 \mid \downt_2 \mid \upt \mid \dropt_c \mid \liftt_c$$
$$\phi_0 ::=  
   \haslab_\sigma \mid \ischild_0 \mid \ischild_1 \mid \ischild_2 \mid \haspeb_c$$

\noindent
The semantics of these basic expressions for a tree $t$ over $\Sigma$ is given 
in Tables~\ref{tab:sempathrank} and~\ref{tab:semnoderank}. Since we will only 
be interested in ranked trees that encode forests, we assume that $\Sigma$ is 
a ranked alphabet and that the rank of each element of $\Sigma$ is at most 2. 
Note that $\upt$ has the same semantics as $\parent$, and that the semantics of 
$\dropt_c$, $\liftt_c$, $\haslab_\sigma$, and $\haspeb_c$ is unchanged.
The remaining expressions of Pebble XPath, and their semantics 
(for $t$ instead of $f$), are the same as for forests, 
cf. the last four lines of Tables~\ref{tab:sempath} and~\ref{tab:semnode}.

\begin{table}
\[
\begin{array}{l@{\;}cl}
\semt{\downt_1} &=& \{(\tup{u,\pi},\tup{u',\pi})\mid u' 
                            \mbox{ is the first child of }u\} \\
\semt{\downt_2} &=& \{(\tup{u,\pi},\tup{u',\pi})\mid u' 
                            \mbox{ is the second child of }u\} \\               
\semt{ \upt } &=& \{(\tup{u,\pi},\tup{u',\pi})\mid u' \mbox{ is the parent of }u\} \\
\semt{ \dropt_c } &=& \{(\tup{u,\pi},\tup{u,\pi p}) \mid p=(u,c)\} \\ 
\semt{ \liftt_c } &=& \{(\tup{u,\pi p},\tup{u,\pi}) \mid p=(u,c)\} \\[1mm]
\end{array}
\]
\caption{Basic path expressions $\alpha_0$ for a ranked tree $t$}
\label{tab:sempathrank}
\end{table}

\begin{table}
\[
\begin{array}{l@{\;}cl}
\semt{ \haslab_\sigma } &=& \{\tup{u,\pi}\mid u \mbox{ has label }\sigma\} \\
\semt{ \ischild_0 } &=& \{\tup{u,\pi}\mid u \mbox{ is the root}\} \\
\semt{ \ischild_1 } &=& \{\tup{u,\pi}\mid u \mbox{ is a first child}\} \\
\semt{ \ischild_2 } &=& \{\tup{u,\pi}\mid u \mbox{ is a second child}\} \\
\semt{ \haspeb_c } &=& \{\tup{u,\pi p}
    \mid p=(u,c)\} \\[1mm]
\end{array}
\]
\caption{Basic node expressions $\phi_0$ for a ranked tree $t$}
\label{tab:semnoderank}
\end{table}

We first show that for every path expression $\alpha$ on forests there is a 
path expression $\alpha'$ that computes the same trip as $\alpha$ on the binary 
encoding of the forests as ranked trees.  
We use the encoding $\enc'$ defined in Section~\ref{sec:trees}, 
which encodes forests over the alphabet $\Sigma$ as ranked trees 
over the associated ranked alphabet $\Sigma'$. 
Note that for every forest~$f$, $\enc'(f)$ has the same nodes as $f$.
For a trip $T$ on forests, we define the \emph{encoded trip} $\enc'(T)$ 
on ranked trees by $\enc'(T)=\{(\enc'(f),u,v)\mid (f,u,v)\in T\}$.
 
\begin{lemma}\label{lem:xpathft} 
For every Pebble XPath path expression $\alpha$ on forests over $\Sigma$, 
a Pebble XPath path expression $\alpha'$ on ranked trees over $\Sigma'$ 
can be constructed in polynomial time such that 
$T(\alpha')=\enc'(T(\alpha))$. If $\alpha$ is a Pebble CAT expression, 
then so is $\alpha'$. 
\end{lemma}

\begin{proof}
The proof is an elementary coding exercise. 
Let us start with Pebble XPath. 
We will, in fact, define $\alpha'$ such that 
$\sem{\alpha'}_{\enc'(f)}=\semf{\alpha}$ for every $f\in F_\Sigma$,
which implies the result. It clearly suffices to do this for basic path expressions 
$\alpha_0$, and similarly for basic node expressions $\phi_0$. 
As observed before, all basic node expressions except $\haslab_\sigma$ are redundant, 
so it suffices to define $\haslab_\sigma' \equiv 
\haslab_{\sigma^{11}}\vee\haslab_{\sigma^{10}}\vee
\haslab_{\sigma^{01}}\vee\haslab_{\sigma^{00}}$.
We now turn to the basic path expressions. 
We will use the auxiliary basic path expressions $\child_1$ and $\parent_1$
with the semantics $\semf{\child_1}=\{(\tup{u,\pi},\tup{u',\pi})\mid 
u' \text{ is the first child of }u\}$ and $\semf{\parent_1}=\semf{\child_1}^{-1}$. 
Since clearly $\child\equiv \child_1/\rght^*$ and 
$\parent\equiv \lft^*/\parent_1$, it suffices to define 
$\child'_1$ and $\parent'_1$ instead of $\child'$ and $\parent'$, as follows: 
$\child'_1\equiv \text{?}\phi_1/\downt_1$ where $\phi_1$ is the disjunction of  
$\haslab_{\sigma^{11}}$ and $\haslab_{\sigma^{10}}$ for all $\sigma\in\Sigma$, and 
$\parent'_1\equiv \text{?}\ischild_1/\upt/\text{?}\phi_1$. 
Then we define $\rght'\equiv \downt_2 \cup \text{?}\phi_2/\downt_1$ 
where $\phi_2$ is the disjunction of all $\haslab_{\sigma^{01}}$ 
for $\sigma\in\Sigma$. Since $\semf{\lft}$ is the inverse of $\semf{\rght}$,
we define $\lft'\equiv 
\text{?}\ischild_2/\upt \cup \text{?}\ischild_1/\upt/\text{?}\phi_2$. 
Finally, $\dropt'_c \equiv \dropt_c$ and $\liftt'_c \equiv \liftt_c$.

To prove the result for Pebble CAT, we also have to consider 
the other basic node expressions $\phi_0$.
Obviously, we define $\haspeb'_c\equiv\haspeb_c$. 
We define $\isleaf'$ to be the disjunction of $\haslab_{\sigma^{01}}$ and 
$\haslab_{\sigma^{00}}$ for all \mbox{$\sigma\in\Sigma$}, and similarly, 
$\islast'$ to be the disjunction of $\haslab_{\sigma^{10}}$ and 
$\haslab_{\sigma^{00}}$ for all $\sigma\in\Sigma$. 
It remains to consider $\isfirst$ and $\isroot$. 
Since we may assume the syntax of $\alpha$ to be in normal form, 
it suffices to define $(?\phi_0)'$ and $(?\neg\phi_0)'$. 
We define $(?\isfirst)'\equiv \text{?}\ischild_0 \cup \ischild_1/\upt/\child'_1$ 
and $(?\neg\isfirst)'\equiv \upt/\rght'$ 
where $\child'_1$ and $\rght'$ are defined above. 
For $\isroot$, we first note that 
$?\isroot\equiv \dropt_c/\lft^*/\text{?}\isroot/\text{?}\isfirst/\rght^*/\liftt_c$
where $c$ is any element of $C$. Intuitively, we walk from the current node 
to the left until we arrive at the first root, and then walk back. 
Thus, since the first root of a forest $f$ is encoded as the root of $\enc'(f)$, 
we define $(?\isroot)'\equiv 
\dropt_c/(\lft')^*/\text{?}\ischild_0/(\rght')^*/\liftt_c$.
Finally, we define $(?\neg\isroot)'\equiv \dropt_c/\parent'/\child'/\liftt_c$. 
\end{proof}

Next we prove the reverse direction of Lemma~\ref{lem:xpathft}, for Pebble CAT. 

\begin{lemma}\label{lem:cattf} 
For every Pebble CAT path expression $\alpha$ on ranked trees over $\Sigma'$
there is a Pebble CAT path expression $\alpha'$ on forests over $\Sigma$ such that 
$$\enc'(T(\alpha'))=T(\alpha).$$
\end{lemma}

\begin{proof}
This is also an elementary coding exercise.
We assume the syntax of $\alpha$ to be in normal form, 
whereas for $\alpha'$ we keep the full syntax.
As in the previous lemma, we will define $\alpha'$ such that 
$\semf{\alpha'}=\sem{\alpha}_{\enc'(f)}$. 
It suffices to do this for path expressions $\alpha_0$, $?\phi_0$, and $?\neg\phi_0$. 
We start with $\alpha_0$ and we define 
$\downt'_1 \equiv \child/\text{?}\isfirst \cup \text{?}\isleaf/\rght$ and
$\downt'_2 \equiv \text{?}\neg\isleaf/\rght$. Moreover, 
$\up'\equiv \text{?}\isfirst/\parent \cup \lft$. Finally, 
$\dropt'_c \equiv \dropt_c$ and $\liftt'_c \equiv \liftt_c$.
We now turn to the basic node expressions. 
For $\phi_0\equiv \haslab_{\sigma^{10}}$ we define 
$(?\phi_0)'\equiv \text{?}\phi'_0$ and 
$(?\neg\phi_0)'\equiv \text{?}\neg\phi'_0$, where 
$\phi_0'\equiv \haslab_\sigma \wedge \neg\isleaf \wedge \islast$,
and similarly for $\haslab_{\sigma^{11}}$, $\haslab_{\sigma^{01}}$, 
and $\haslab_{\sigma^{00}}$. We do this also for $\phi_0\equiv\ischild_0$
with $\phi'_0\equiv\isroot\wedge\isfirst$, and for 
$\phi_0\equiv\haspeb_c$ with $\phi'_0\equiv \haspeb_c$.
It remains to consider $\ischild_1$ and $\ischild_2$. 
We define $(?\ischild_2)'\equiv \lft/?\neg\isleaf/\rght$ and hence 
$(?\neg\ischild_2)'\equiv \text{?}\isfirst \cup \lft/?\isleaf/\rght$.
For $\ischild_1$ the definitions of $(?\ischild_1)'$ and $(?\neg\ischild_1)'$
now follow from the fact that 
$?\ischild_1\equiv \text{?}\neg\ischild_0/?\neg\ischild_2$ and 
$?\neg\ischild_1\equiv \text{?}\ischild_0 \cup \text{?}\ischild_2$. 
\end{proof}

Lemmas~\ref{lem:xpathft} and~\ref{lem:cattf} together show that 
if the first equivalence of Theorem~\ref{thm:xpath} holds for ranked trees, 
then it also holds for forests. To show this also for the second equivalence,
we need the next elementary lemma. 

\begin{lemma}\label{lem:encmso}
For every trip $T$ on forests, $T$ is \abb{mso} definable if and only if 
$\enc'(T)$ is \abb{mso} definable.
\end{lemma}

\begin{proof}
(Only if) 
Since $f$ and $\enc'(f)$ have the same nodes, for every 
forest $f$ over $\Sigma$, it suffices to show that the atomic formulas 
$\lab_\sigma(x)$, $\down(x,y)$, and $\nex(x,y)$ for forests can be expressed by 
an \abb{mso} formula for the ranked trees that encode the forests.
Clearly, $\lab_\sigma(x)$ can be expressed by the disjunction of all
$\lab_{\sigma^{k\ell}}(x)$ for $k,\ell\in\{0,1\}$, as in the proof of
Lemma~\ref{lem:xpathft}.
For $\down(x,y)$ we show that the trip 
$T=\{(\enc'(f),u,v)\mid f\models \down(u,v)\}$ is \abb{mso} definable. 
This follows from Proposition~\ref{prop:trips} because $T=T(\cB)$ 
for the \abb{ta} $\cB$ that has the rules
(for all $k,\ell\in\{0,1\}$, $j\in\{0,1,2\}$, and $\sigma\in\Sigma$):
\[
\begin{array}{lll}
\tup{p_0,\sigma^{1\ell},j} & \to & \tup{p,\down_1}, \\
\tup{p,\sigma^{11},j} & \to & \tup{p,\down_2}, \\
\tup{p,\sigma^{01},j} & \to & \tup{p,\down_1}, \\
\tup{p,\sigma^{k\ell},j} & \to & \tup{p_\infty,\stay},
\end{array}
\]
where $p_0$ is the initial and $p_\infty$ the final state of $\cB$. 
Thus, there is a formula $\phi(x,y)$ such that 
$\enc'(f)\models \phi(u,v)$ if and only if $f\models \down(u,v)$, 
for every forest $f$, which means that $\phi(x,y)$ expresses $\down(x,y)$
on the encoding of~$f$.\footnote{For the reader familiar with \abb{mso} logic 
we note that it is also easy to write down the formula $\phi(x,y)$ 
using the equivalences in the proof of Lemma~\ref{lem:xpathft} and the fact that 
the transitive-reflexive closure of an \abb{mso} definable relation 
is \abb{mso} definable. 
} 
The formula $\nex(x,y)$ can be treated in the same way, 
where $\cB$ now has the rules
$\tup{p_0,\sigma^{11},j} \to \tup{p_\infty,\down_2}$ and 
$\tup{p_0,\sigma^{01},j} \to \tup{p_\infty,\down_1}$,
and hence $T(\cB)=\{(\enc'(f),u,v)\mid f\models \nex(u,v)\}$.

(If) For the same reason as above, it suffices to show  
that the atomic formulas $\down_i(x,y)$ and $\lab_{\sigma^{k\ell}}(x)$
for ranked trees over $\Sigma'$ 
can be expressed by an \abb{mso} formula for the forests they encode.
For this we consider the path expressions $\downt'_i$ and
$\haslab'_{\sigma^{10}}$ in the proof of Lemma~\ref{lem:cattf},
and we define 
\[
\begin{array}{lll}
\phi_1(x,y) & \!\!\equiv\!\! & (\down(x,y)\wedge \first(y)) \vee
     (\leaf(x) \wedge \nex(x,y)), \\[1mm]
\phi_2(x,y) & \!\!\equiv\!\! & \neg\,\leaf(x) \wedge \nex(x,y),  \\[1mm]
\phi_{10}(x) & \!\!\equiv\!\! & 
     \lab_\sigma(x) \wedge \neg\,\leaf(x)\wedge \last(x),
\end{array}
\]
and similarly for the other $\phi_{k\ell}(x,y)$. 
Then $\enc'(f)\models \down_i(u,v)$ if and only if $f\models \phi_i(u,v)$, and
$\enc'(f)\models \lab_{\sigma^{k\ell}}(u)$ if and only if 
$f\models \phi_{k\ell}(u)$. 
\end{proof}

From now on, when we refer to Pebble XPath or Pebble CAT
we always mean their version on ranked trees.

\smallpar{Directive I-PTA's}
For the purpose of the proof of Theorem~\ref{thm:xpath} on ranked trees, 
we formulate the \abb{i-pta} in an alternative way and, for lack of a better name,
call it the \emph{directive} \abb{i-pta}. 
For an alphabet~$\Sigma$ (of which every element has rank at most $2$)
and a finite set $C$ of colours,
we define a \emph{directive} over $\Sigma$ and $C$ to be 
a path expression $\tau$ with the syntax 
$\tau ::= \alpha_0 \mid \;?\phi_0 \mid \;?\neg\,\phi_0$ 
for the same $\Sigma$ and~$C$ (where $\alpha_0$ and $\phi_0$ are 
as in Tables~\ref{tab:sempathrank} and~\ref{tab:semnoderank}). 
The finite set of directives over $\Sigma$ and $C$ is denoted $D_{\Sigma,C}$. 

A \emph{directive} \emph{\abb{i-pta}} is a tuple 
$\cA = (\Sigma, Q, Q_0, F, C, R)$,
where $\Sigma$, $Q$, $Q_0$, $F$, and~$C$ are 
as for an ordinary \abb{i-pta} (with $C=C_\mathrm{i}$), 
and $R$ is a finite set of rules of the form
$\tup{q,\tau,q'}$ where $q,q'\in Q$ and $\tau\in D_{\Sigma,C}$.
Thus, syntactically, $\cA$ can be viewed as a finite automaton of which 
each state transition is labeled by a directive, i.e., 
either by a basic path expression of Pebble XPath,
or by a basic node expression of Pebble XPath, or its negation, 
where the node expressions are turned into path expressions by the ?-operator. 
Intuitively, $?\phi_0$ and $?\neg\,\phi_0$ represent a basic test on the current situation, 
whereas $\alpha_0$ is a basic instruction to be executed on the current situation. Just as 
for an ordinary \abb{i-pta}, a situation on a tree $t\in T_\Sigma$ is a pair $\tup{u,\pi}\in \sit(t)$ 
and a configuration is a triple $\tup{q,u,\pi}$ with $q\in Q$ and $\tup{u,\pi}\in \sit(t)$. 
We write $\tup{q,u,\pi}\Rightarrow_{t,\cA} \tup{q',u',\pi'}$ if there is a rule $\tup{q,\tau,q'}$ 
such that $(\tup{u,\pi},\tup{u',\pi'})\in \semt{\tau}$, where $\semt{\tau}$
is the semantics of path expression $\tau$ on~$t$ (cf. Tables~\ref{tab:sempathrank}
and~\ref{tab:semnoderank} for $\alpha_0$ and $\phi_0$, 
and Table~\ref{tab:sempath} for the \mbox{?-operator}).
To indicate the directive $\tau$ that is executed by $\cA$ in this computation step
we also write $\tup{q,u,\pi}\Rightarrow^\tau_{t,\cA} \tup{q',u',\pi'}$.
Moreover, we define the semantics $\semt{\cA}$ of $\cA$ on tree $t$ as 
$\semt{\cA} = \{(\tup{u,\pi},\tup{u',\pi'})\in \sit(t)\times \sit(t) \mid
\exists\, q_0\in Q_0, \,q_\infty\in F: 
\tup{q_0,u,\pi}\Rightarrow^*_{t,\cA} \tup{q_\infty,u',\pi'}\}$.
Finally, the trip computed by~$\cA$ on $T_\Sigma$ is 
$T(\cA)=\{(t,u,v)\mid \exists\, \pi\in(N(t)\times C)^*: 
(\tup{u,\epsilon},\tup{v,\pi})\in \semt{\cA}\}$.

For the sake of the proofs below we also define 
$\sem{\cA}_t$ for an ordinary \abb{i-pta}~$\cA$ on a tree $t$,
in entirely the same way as above for a directive \abb{i-pta}. 

A directive \abb{i-pta} $\cA$ \emph{with look-ahead tests} 
is defined similarly to the ordinary case in Section~\ref{sec:look-ahead} 
(restricted to automata), 
by additionally allowing rules of the form $\tup{q,?\tup{\cB},q'}$ and 
$\tup{q,?\neg\,\tup{\cB},q'}$
where $\cB$ is another directive \abb{i-pta}. 
The above semantics stays the same, with (as expected) 
\[
\semt{?\tup{\cB}} = \{(\tup{u,\pi},\tup{u,\pi})\mid 
   \exists \tup{u',\pi'}\colon 
   (\tup{u,\pi},\tup{u',\pi'})\in \semt{\cB}\}
\]
and similarly for $\semt{?\neg\,\tup{\cB}}$ (with $\neg\,\exists$).
A directive \abb{i-pta} \emph{with iterated look-ahead tests}
is defined as in Section~\ref{sec:look-ahead}.
We will use \abb{i-pta$^\text{la}$} as an abbreviation of 
`\abb{i-pta} with iterated look-ahead tests'.

We now show that the directive \abb{i-pta} has the same expressive power 
as the \abb{i-pta} (and similarly with iterated look-ahead tests). 
Hence Theorems~\ref{thm:trips} and~\ref{thm:iterated} also hold for the directive \abb{i-pta},
i.e., it computes the \abb{mso} definable trips, and 
it can perform iterated look-ahead tests.  
In what follows, we only consider \abb{i-pta}'s of which every input symbol 
has at most rank~$2$. 
 
\begin{lemma}\label{lem:iptaft}
For every directive \abb{i-pta$^\text{la}$} $\cA$ there is an 
\abb{i-pta$^\text{la}$} $\cA'$ such that $T(\cA')=T(\cA)$.
\end{lemma}

\begin{proof}
Let $\cA = (\Sigma, Q, Q_0, F, C, R)$ be a directive \abb{i-pta}. 
We will, in fact, define the \abb{i-pta} $\cA'$ such that 
$\semt{\cA'}=\semt{\cA}$ for every $t\in T_\Sigma$,
which implies the result. 

We let $\cA' = (\Sigma, Q, Q_0, F, C, \nothing, C_\mathrm{i}, R', 0)$ 
where $C_\mathrm{i}=C$ and $R'$ is defined as follows.
If $\tup{q,\alpha_0,q'}$ is a rule of $\cA$, 
where $\alpha_0$ is a basic path expression,
then $\cA'$ has all rules $\tup{q,\sigma,j,b}\to \tup{q',\alpha_0}$. 
We now turn to the basic node expressions. 
A rule $\tup{q,?\haslab_\sigma,q'}$ is simulated by 
all rules $\tup{q,\sigma,j,b}\to \tup{q',\stay}$, 
and a rule $\tup{q,?\neg\,\haslab_\sigma,q'}$ by 
all rules $\tup{q,\tau,j,b}\to \tup{q',\stay}$ 
with $\tau\in\Sigma\setminus\{\sigma\}$. 
A rule $\tup{q,?\ischild_j,q'}$ is simulated by 
all rules $\tup{q,\sigma,j,b}\to \tup{q',\stay}$,
and a rule $\tup{q,?\neg\,\ischild_j,q'}$ by 
the two rules $\tup{q,\sigma,j',b}\to \tup{q',\stay}$ 
with $j'\in\{0,1,2\}\setminus\{j\}$.
A rule $\tup{q,?\haspeb_c,q'}$ is simulated by 
all rules $\tup{q,\sigma,j,\{c\}}\to \tup{q',\stay}$,
and a rule $\tup{q,?\neg\,\haspeb_c,q'}$ by
all rules $\tup{q,\sigma,j,\nothing}\to \tup{q',\stay}$
and all rules $\tup{q,\sigma,j,\{c'\}}\to \tup{q',\stay}$
with $c'\in C\setminus\{c\}$. 

Finally we consider look-ahead.
If $\tup{q,?\tup{\cB},q'}$ is a rule of $\cA$, and $\cB'$ is an
\abb{i-pta$^\text{la}$} such that 
$\semt{\cB'}=\semt{\cB}$ for every $t\in T_\Sigma$, 
then $\cA'$ has all rules 
$\tup{q,\sigma,j,b,\cB'}\to \tup{q',\stay}$ 
that use $\cB'$ as a look-ahead test. 
Similarly, the rule $\tup{q,?\neg\,\tup{\cB},q'}$ is simulated by all rules
$\tup{q,\sigma,j,b,\neg\,\cB'}\to \tup{q',\stay}$.
\end{proof}

\begin{lemma}\label{lem:iptatf}
For every \abb{i-pta} $\cA$ 
there is a directive \abb{i-pta} $\cA'$ such that 
$T(\cA')=T(\cA)$.
\end{lemma}

\begin{proof}
Let $\cA = (\Sigma, Q, Q_0, F, C, \nothing, C_\mathrm{i}, R, 0)$ 
be an \abb{i-pta} with $C_\mathrm{i}=C$. 
To simplify the proof we extend the
syntax of the directive \abb{i-pta} by allowing rules $\tup{q,\tau,q'}$ with 
$\tau ::= \alpha_0 \mid \;?\phi_0 \mid \;?\neg\,\phi_0 \mid \tau/\tau'$, 
where $\tau'$ is an alias of $\tau$.
This clearly does not extend their power, 
because a rule $\tup{q,\tau/\tau',q'}$ can be replaced by the two rules
$\tup{q,\tau,p}$ and $\tup{p,\tau',q'}$ where $p$ is a new state. 
We now construct $\cA' = (\Sigma, Q, Q_0, F, C, R')$ 
where $R'$ is defined as follows. 
If $\cA$ has a rule $\tup{q,\sigma,j,b}\to \tup{q',\alpha}$ 
then $\cA'$ has the rule $\tup{q,\tau,q'}$ such that
$\tau = \tau_\sigma/\tau_j/\tau_b/\alpha$ if $\alpha\neq\stay$, and 
$\tau = \tau_\sigma/\tau_j/\tau_b$ if $\alpha=\stay$,
where
$\tau_\sigma = \text{?}\haslab_\sigma$,
$\tau_j = \text{?}\ischild_j$, 
$\tau_{\{c\}} = \text{?}\haspeb_c$, and
$\tau_\nothing = \text{?}\neg\,\haspeb_{c_1}/\cdots/?\neg\,\haspeb_{c_n}$,
where $C=\{c_1,\dots,c_n\}$. 
\end{proof}

As observed before, a directive \abb{i-pta} $\cA$ 
can be viewed as a finite automaton of which each state transition is labeled by 
a directive. Thus, viewing the set $D_{\Sigma,C}$ as an alphabet,
$\cA$ accepts a string language 
$L_\str(\cA)\subseteq D^*_{\Sigma,C}$. We now show the (rather obvious) fact that 
the semantics $\semt{\cA}$ of $\cA$ (for every tree $t$ over $\Sigma$) 
depends only on the language $L_\str(\cA)$, 
cf.~\cite[Theorem~3.11]{Eng74} and~\cite[Lemma~3]{bloem}. 
We do this (as in~\cite[Definition~2.7]{Eng74} and~\cite[Section~4]{bloem}) 
by associating a semantics 
$\semt{L}$ with every language $L\subseteq D^*_{\Sigma,C}$. 
Intuitively, a string $w=\tau_1\cdots \tau_n$ 
of directives can be viewed as the path expression $\tau_1/\cdots/\tau_n$ and 
a language $L=\{w_1,w_2,\dots\}$ of such strings can be viewed 
as the (possibly infinite) 
path expression $w_1\cup w_2\cup\cdots$. Thus, for a tree $t$ over $\Sigma$ 
we formally define $\semt{\epsilon}$ to be the identity on $\sit(t)$,
$\semt{\tau_1\cdots\tau_n}=\semt{\tau_1}\circ\cdots\circ\semt{\tau_n}$,
and $\semt{L}=\bigcup_{w\in L}\semt{w}$. 
The next lemma is a special case of~\cite[Theorem~3.11]{Eng74}. 
Its proof is entirely similar to the one of~\cite[Lemma~3]{bloem}. 

\begin{lemma}\label{lem:langipta}
$\semt{\cA}=\semt{L_\str(\cA)}$.
\end{lemma}

\begin{proof}
A string $w$ of directives induces a state transition relation 
$R_\cA(w)\subseteq Q\times Q$ as follows. For $\tau\in D_{\Sigma,C}$,
$R_\cA(\tau)=\{(q,q')\mid \tup{q,\tau,q'}\in R\}$. For the empty string,
$R_\cA(\epsilon)$ is the identity on $Q$, and 
$R_\cA(\tau_1\cdots\tau_n)=R_\cA(\tau_1)\circ\cdots\circ R_\cA(\tau_n)$.
Then $L_\str(\cA)=
\{w\in D^*_{\Sigma,C} \mid R_\cA(w)\cap (Q_0\times F)\neq\nothing\}$.

It is straightforward to show by induction that, 
for all configurations $\tup{q,u,\pi}$ and $\tup{q',u',\pi'}$
and for every $w=\tau_1\cdots\tau_n$ over $D_{\Sigma,C}$, 
there is a computation 
\[
\tup{q_1,u_1,\pi_1} \Rightarrow^{\tau_1}_{t,\cA} 
\tup{q_2,u_2,\pi_2} \Rightarrow^{\tau_2}_{t,\cA}
\cdots \Rightarrow^{\tau_n}_{t,\cA}
\tup{q_{n+1},u_{n+1},\pi_{n+1}}
\]
with $\tup{q_1,u_1,\pi_1}=\tup{q,u,\pi}$ and 
$\tup{q_{n+1},u_{n+1},\pi_{n+1}}=\tup{q',u',\pi'}$
if and only if 
$(\tup{u,\pi},\tup{u',\pi'})\in \semt{w}$ and $(q,q')\in R_\cA(w)$. 
From this equivalence it follows that 
$\semt{\cA}$ consists of all $(\tup{u,\pi},\tup{u',\pi'})$ such that
\[
\exists\, q_0\in Q_0, q_\infty\in F, w\in D^*_{\Sigma,C}: 
(\tup{u,\pi},\tup{u',\pi'})\in \semt{w},\;(q,q')\in R_\cA(w)
\]
i.e., such that 
$\exists\, w\in L_\str(\cA): (\tup{u,\pi},\tup{u',\pi'})\in \semt{w}$,
which means that it equals $\semt{L_\str(\cA)}$.
\end{proof}

\smallpar{Proof of Theorem~\ref{thm:xpath}}
We assume the syntax for path expressions $\alpha$ of Pebble XPath and Pebble CAT
to be in normal form.
We also add $\alpha ::= \nothing$ to the syntax,
with $\semt{\nothing}=\nothing$ for every tree $t$. 
That is possible because, e.g., 
$\semt{?\ischild_0/?\neg\ischild_0}=\nothing$. 

We first show that Pebble CAT has the same power as \abb{MSO}.
Let us recall that the set $D_{\Sigma,C}$ of directives $\tau$ 
of the directive \abb{i-pta} 
was defined by the syntax $\tau ::= \alpha_0 \mid \;?\phi_0 \mid \;?\neg\phi_0$. 
Thus, the path expressions of Pebble CAT are, in fact, exactly the usual regular expressions 
over the ``alphabet'' $D_{\Sigma,C}$.
Accordingly, we define for such a path expression $\alpha$ the string language 
$L_\str(\alpha)\subseteq D^*_{\Sigma,C}$ in the obvious way, 
interpreting the operators $\cup$, $/$, and $^*$ as union, concatenation, and 
Kleene star of string languages, respectively. The next lemma is the analogue of 
Lemma~\ref{lem:langipta}, with a straightforward proof.  

\begin{lemma}\label{lem:langcat}
$\semt{\alpha}=\semt{L_\str(\alpha)}$.
\end{lemma}

\begin{proof}
It is easy to see, for string languages $L_1,L_2\subseteq D_{\Sigma,C}^*$, 
that $\semt{L_1\cup L_2}= \semt{L_1}\cup\semt{L_2}$,
$\semt{L_1L_2}= \semt{L_1}\circ\semt{L_2}$, and 
$\semt{L_1^*}= \semt{L_1}^*$, cf.~\cite[Lemma~2.9]{Eng74}. 
Then the proof is by induction on the structure of $\alpha$.
\end{proof}

By Kleene's classical theorem, a string language can be accepted by a finite automaton 
if and only if it can be defined by a regular expression. 
Thus, by Lemmas~\ref{lem:langipta} and~\ref{lem:langcat}, 
a trip is definable in Pebble CAT if and only if it 
can be computed by a directive \abb{i-pta}, and hence, 
by Theorem~\ref{thm:trips} (for $k=0$) and 
Lemmas~\ref{lem:iptaft} and~\ref{lem:iptatf},
if and only if it is \abb{mso} definable. 

It remains to show that if a trip is definable in Pebble XPath, then it can be computed by a directive \abb{i-pta}. We will prove below that for every Pebble XPath path expression $\alpha$ there is a directive \abb{i-pta$^\text{la}$} $\cA$, i.e., a directive \abb{i-pta} with iterated look-ahead tests, such that $\semt{\cA} = \semt{\alpha}$ for every $t$. This  implies that $\alpha$ and $\cA$ define the same trip, and then we obtain from Theorem~\ref{thm:iterated} (and Lemmas~\ref{lem:iptaft} and~\ref{lem:iptatf}) a directive \abb{i-pta} (without look-ahead) computing that same trip. 

Let $n_\alpha$ be the nesting depth of subexpressions of $\alpha$ of the form $\tup{\beta}$. The proof is by induction on $n_\alpha$, and $\cA$ will be of look-ahead depth $n_\alpha$. If $n_\alpha=0$, i.e., there are no such subexpressions at all, then $\alpha$ is a Pebble CAT expression, and we are done by the first part of the proof. 
Now suppose that the result holds for nesting depth $n$, and let $n_\alpha=n+1$. For every subexpression $\tup{\beta}$ of~$\alpha$ that is not nested within another such subexpression, let $\cA_\beta$ be a directive \abb{i-pta$^\text{la}$} of look-ahead depth $n$ (or less) such that $\semt{\cA_\beta} = \semt{\beta}$ for all $t$. We now define the extended ``alphabet'' $D^n_{\Sigma,C}$ to consist of all path expressions $\tau$ with the syntax 
$\tau ::= \alpha_0 \mid 
\;?\phi_0 \mid \;?\neg\phi_0 \mid \;?\tup{\beta} \mid \;?\neg\tup{\beta}$
where $\tup{\beta}$ ranges over the above subexpressions of $\alpha$. 
Then $\alpha$ can be viewed as a regular expression over the alphabet 
$D^n_{\Sigma,C}$, and it should be clear that Lemma~\ref{lem:langcat}
is also valid in this case. 
Also, using $D^n_{\Sigma,C}$ instead of $D_{\Sigma,C}$ in the rules of 
the directive \abb{i-pta}, and identifying each ``symbol'' $?\tup{\beta}$ with the ``symbol'' $?\tup{\cA_\beta}$ (and similarly for the negated tests), we obtain a subclass of the directive \abb{i-pta$^\text{la}$} of look-ahead depth $n+1$,
because the semantics of the path expression $?\tup{\beta}$ is exactly the same as the meaning of the look-ahead test $?\tup{\cA_\beta}$.
Again, it should be clear that Lemma~\ref{lem:langipta}
is also valid for these directive \abb{i-pta}'s, 
which are finite automata over $D^n_{\Sigma,C}$. 
Hence, by the same Kleene argument as in the first part of the proof, 
there is a directive \abb{i-pta$^\text{la}$} $\cA$ of look-ahead depth $n+1$ such that 
$\semt{\cA} = \semt{\alpha}$ for every tree $t$.  

This ends the proof of Theorem~\ref{thm:xpath}, both for ranked trees and
(by Lemmas~\ref{lem:xpathft}, \ref{lem:cattf}, and~\ref{lem:encmso}) 
for unranked forests.

\smallpar{Two remarks}
(1) Although the MSO definable trips are, of course, closed under complement and intersection, we do not know whether the XPath~2.0 operations {\tt intersect} and {\tt except} can be added to the syntax of path expressions of Pebble XPath ($\alpha ::= \alpha \cap \beta \mid \alpha \setminus \beta$). 
That is because it is not clear whether for every \abb{i-pta} $\cA$ there is an \abb{i-pta} $\cB$ such that $\semt{\cB} = \sit(t) - \semt{\cA}$ for every tree $t$.

(2) The language Pebble XPath meets the requirements as listed in \cite{GorMar05}.
It is \emph{simple}, defined in an algebraic language 
using simple operators: in particular we believe that 
pebbles form a user friendly concept.
It is \emph{understandable},
as its expressive power can be characterized in terms of automata.
It is \emph{useful} in the sense that the query evaluation problem 
`given path expression $\alpha$ and two nodes $u,v$ in forest $f$,
is $(f,u,v)\in T(\alpha)$?' is tractable.
At least, the latter property holds for Pebble CAT,
as $\alpha$ can be transformed into an \abb{I-PTA} in polynomial time,
and the problem `$(f,u,v)\in T(\alpha)$?' can then be
translated into the emptiness problem for push-down automata.
For Pebble XPath the query evaluation problem is tractable 
for every fixed path expression $\alpha$. 
This is explained in more detail in the next two paragraphs.

\smallpar{Query evaluation}
For a directive \abb{i-pta} $\cA = (\Sigma, Q, Q_0, C, R)$, the binary node relation $T$ computed by $\cA$ on an input tree $t$ can be evaluated in polynomial time as follows. It is straightforward to construct from $\cA$ and $t$ an ordinary pushdown automaton $\cP$ with state set $Q \times \nod{t}$ and pushdown alphabet $\nod{t}\times C$ in such a way that $\cP$ (with the empty string as input) has the same computation steps as $\cA$ on $t$. Note that the configurations of $\cP$ are exactly the configurations $\langle q,u,\pi\rangle$ of $\cA$ on $t$. Dropping and lifting a pebble corresponds to pushing and popping a pushdown symbol. Moving around in $t$ corresponds to a change of state. To decide whether $(t,u,v)\in T$, with $u,v\in\nod{t}$, decide whether $\cP$ has a computation from configuration 
$\langle q_0,u,\epsilon\rangle$ (for some $q_0\in Q_0$) 
to some final configuration $\langle q,v,\pi\rangle$. Clearly, $\cP$ can be constructed in polynomial time from $\cA$ and $t$, and the existence of such a computation can be verified in polynomial time.

By Lemma~\ref{lem:xpathft}, path expressions on forests can be translated into path expressions on ranked trees in polynomial time. 
Since for a Pebble CAT path expression on ranked trees the corresponding directive \abb{i-pta} can be constructed in polynomial time, using Kleene's construction, Pebble CAT path expressions can be evaluated in polynomial time. This does not seem to hold for Pebble XPath, as the construction in the proof of Theorem~\ref{thm:look-ahead} (which implements a look-ahead test by calling an \abb{i-pta} $\cB$) is at least 2-fold exponential (because determining the domain of the related \abb{i-pta} $\cB'$
takes 2-fold exponential time by Theorem~\ref{thm:typecheck}). 
However, the data complexity of the problem is of course polynomial, 
i.e., for a fixed path expression $\alpha$ 
we obtain a fixed directive \abb{i-pta} $\cA$ for which the binary node relation 
can be evaluated in polynomial time.

\section{Pattern Matching}\label{sec:pattern}

One of the basic tree transformations in the
context of XML is pattern matching. The transducer
must find all sequences of nodes satisfying a certain
description and generate the subtrees rooted at these
nodes, for each match.
More precisely, we consider queries of the form 
\[
\mbox{\tt for } \cX
\mbox{ \tt where } \phi
\mbox{ \tt return } r
\] 
in which $\cX$ is a finite set of node variables, 
$\phi$ is an \abb{mso} formula with its free variables in $\cX$,
and $r$ is a tree of which the leaves may be labeled with 
the variables in $\cX$.
In what follows we assume that $\cX$ and $r$ are fixed. 
Let $\cX=\{x_1,\dots,x_n\}$, 
where $x_1,\dots,x_n$ is an arbitrary order of the elements of $\cX$. 
The transducer must find all sequences of nodes $u_1,\dots,u_n$
of the input tree $t$ that match the pattern defined by $\phi(x_1,\dots,x_n)$,  
i.e., such that $t\models \phi(u_1,\dots,u_n)$, and for each match
it must generate the output tree~$r$ in which each occurrence of 
the variable $x_i$ is replaced by the subtree of $t$ with root~$u_i$. 
Usually the variables in~$\cX$ are indeed specified in a specific
order $\lambda=(x_1,\dots,x_n)$, and it is required that the transducer
finds (and generates) the matches in the lexicographic document order
induced by $\lambda$. We will, however, also consider the case where 
this requirement is dropped, and the most efficient order $\lambda$
can be selected.  

For convenience we assume that $r$ is of the form $\mu(x_1,\dots,x_n)$ 
for some symbol $\mu$ of rank~$n$, and so the output for each match is  
$\mu(t|_{u_1},\dots,t|_{u_n})$ where $t|_u$ is the subtree of $t$ with root $u$.
For convenience we also assume that the input tree $t$ is ranked. 
Moreover, we assume that 
the output alphabet is also ranked and contains the binary symbol $\at$ that allows
us to list the various output trees $\mu(t|_{u_1},\dots,t|_{u_n})$,
and the nullary symbol $e$ to indicate the end of the list of output trees
(similar to the binary tag \texttt{<result>} and the nullary tag 
\texttt{<endofresults>} of Example~\ref{ex:siberie}). 
In Section~\ref{sec:pft} we will consider pattern matching in forests. 

We now describe a total deterministic \abb{ptt} $\cA$ that executes the above query. 
In order to find all $n$-tuples of nodes
matching the $n$-ary pattern defined by the \abb{mso}
formula $\phi(x_1,\dots,x_n)$, and generate the corresponding output, 
the \abb{ptt}~$\cA$ systematically enumerates \emph{all} $n$-tuples 
of nodes of the input tree $t$. To do this, $\cA$~uses visible pebbles 
$c_1,\dots,c_n$ on the stack, 
representing the variables $x_1,\dots,x_n$, respectively.\footnote{It is not
necessary that all pebbles are visible,
as we will discuss below, but it simplifies the description of $\cA$.
} 
It drops them in this order and moves each of them
through the input tree $t$ in document order (i.e., in pre-order), 
in a nested fashion. Inductively speaking, 
$\cA$ moves pebble~$c_1$ in pre-order through $t$ 
(alternately dropping and lifting $c_1$), 
and for each position $u_1$ of~$c_1$ it uses pebbles $c_2,\dots,c_n$
to enumerate all possible $(n\!-\!1)$-tuples $u_2,\dots,u_n$ of nodes of $t$. 
For each enumerated $n$-tuple $u_1,\dots,u_n$,  
with pebble $c_i$ at position~$u_i$,  
$\cA$ performs the test~$\phi$, 
using an \abb{mso} test on the visible
configuration (Lemma~\ref{lem:visiblesites}), and, 
in case of success, spawns a process
that outputs the corresponding $n$-tuple of subtrees.  

More precisely, if the ranked input alphabet is $\Sigma$, then 
$\phi$ is an \abb{mso} formula over $\Sigma$, and 
$\cA$ has the ranked output alphabet $\Delta=\Sigma \cup\{\mu,\at,e\}$ where $\mu$ has rank~$n$, and $\at$ and $e$ have rank~2 and~0 respectively. For input tree $t$, the output tree $s$ is of the form $s=\at(r_1,\at(r_2,\dots \at(r_k,e)\cdots))$ where each $r_i$ corresponds to a match, i.e., there is a sequence of nodes $u_1,\dots,u_n$ of $t$ such that $t\models \phi(u_1,\dots,u_n)$ and $r_i=\mu(t|_{u_1},\dots,t|_{u_n})$. Moreover, the sequence $r_1,\dots,r_k$ corresponds to the sequence of all matches, in lexicographic document order. 
As explained above, the visible colour set of the \abb{ptt} $\cA$ 
is $C_\mathrm{v}=\{c_1,\dots,c_n\}$, and $\cA$
generates $s$ by enumerating 
all sequences $u_1,\dots,u_n$ of nodes of $t$ using pebbles $c_1,\dots,c_n$. 
To find out whether this sequence is a match, $\cA$ performs the \abb{mso} test
$\psi(x)$ on the visible configuration, defined by  
\[
\psi(x) \equiv \forall x_1,\dots,x_n(
    (\peb_{c_1}(x_1) \wedge \cdots\wedge \peb_{c_n}(x_n))\to 
        \phi'(x_1,\dots,x_n))
\]
where $\peb_c(x)$ is the disjunction of all $\lab_{(\sigma,b)}(x)$ 
such that $c\in b$, and where $\phi'$ is obtained from $\phi$ by changing 
every atomic subformula $\lab_\sigma(y)$ into 
the disjunction of all $\lab_{(\sigma,b)}(y)$. 
Note that $\psi(x)$ is an \abb{mso} formula over $\Sigma\times 2^C$, where $C$ is the colour set of $\cA$.
Note also that the variable $x$ (for the head position) does not, and need not, 
occur in $\psi(x)$. 
If the sequence $u_1,\dots,u_n$ is not a match, then 
$\cA$ continues the enumeration of $n$-tuples. 
If the sequence \emph{is} a match, then 
$\cA$ outputs the symbol $\at$ and branches into two subprocesses 
(as in the 5-th rule of Example~\ref{ex:siberie}). 
In the second (main) branch it continues the enumeration of $n$-tuples.
In the first branch it outputs the symbol $\mu$ and branches into $n$ subprocesses,
where the $i$-th process 
searches for visible pebble~$c_i$ and then outputs $t|_{u_i}$.
Note that, in this first branch, $\cA$ could easily output an arbitrary tree $r$ 
in which every occurrence of the variable $x_i$ is replaced by $t|_{u_i}$.
This ends the description of~$\cA$.

As the complexity of typechecking the transducer $\cA$
depends critically on the number of visible pebbles used 
(see Theorem~\ref{thm:typecheck}), we 
wish to minimize that number and use as few visible
pebbles as possible for matching.
It should be clear that, instead of using $n$ visible pebbles, 
$\cA$ can also use $n\!-\!2$ visible pebbles $c_1,\dots,c_{n-2}$, 
one invisible pebble $c_{n-1}$ on top (which is therefore always observable), 
and the head instead of the last pebble $c_n$. 
Then $\cA$ can perform the \abb{mso} test $\chi(x)$ 
on the observable configuration, defined by $\chi(x) \equiv $
\[
\forall x_1,\dots,x_{n-1}(
    (\peb_{c_1}(x_1) \wedge \cdots\wedge \peb_{c_{n-1}}(x_{n-1}))\to 
        \phi'(x_1,\dots,x_{n-1},x))
\]
where $x_n$ is renamed into $x$ in $\phi'$. 
Thus, from Theorem~\ref{thm:mso} we obtain the following result on 
the matching of arbitrary \abb{MSO} definable patterns.

\begin{theorem}\label{thm:matchall}
For $n\geq 2$, 
every \abb{MSO} definable $n$-ary pattern
can be matched by a total deterministic \abb{v$_{n-2}$i-ptt}.
Moreover, and in particular, 
every \abb{MSO} definable unary or binary pattern
can be matched by a total deterministic \abb{i-ptt}.
\end{theorem}

To further reduce the number of visible pebbles,
we consider the more specific case of queries of the form 
\[
\mbox{\tt for } \cX
\mbox{ \tt where } \bool(\phi_1, \dots, \phi_m) 
\mbox{ \tt return } r
\]
in which $\bool(\phi_1, \dots, \phi_m)$ is a boolean combination
(using $\wedge,\vee,\neg$) of the \abb{mso} formulas
$\phi_1, \dots, \phi_m$, $m\geq 2$, and 
each $\phi_\ell$, $\ell\in[1,m]$, has its free variables in~$\cX$. 
We will make use of the fact that not all variables in $\cX$
need actually occur in each formula~$\phi_\ell$.
As discussed in the Introduction,
the {\tt for} $\cdots$\ {\tt where} construct in
XQuery often induces patterns $\phi_1 \wedge \cdots \wedge \phi_m$
such that each $\phi_\ell$ contains just two free variables,
cf. \cite{GotKocSch}.

Consider an arbitrary query as displayed above. 
Let $G_\phi=(V_\phi,E_\phi)$ be the undirected graph induced by the
pattern $\phi \equiv \bool(\phi_1, \dots, \phi_m)$, 
by which we mean that  
the set $V_\phi$ of vertices of $G_\phi$ consists of the free variables 
of $\phi$, i.e., $V_\phi=\cX$, 
and that the set $E_\phi$ of edges of $G_\phi$ consists of 
the unordered pairs $\{x,y\}$ 
(with $x,y\in V_\phi$, $x\neq y$)
for which there exists $\ell\in[1,m]$ 
such that both $x$ and $y$ occur (free) in $\phi_\ell$. 
Note that $G_\phi$ does not depend on any order 
of the variables in $\cX$. 
Note also that for every finite undirected graph $G$ there exists 
$\phi \equiv \phi_1 \wedge \cdots \wedge \phi_m$ 
such that $G$ is isomorphic to $G_\phi$.  

Pattern matching $\phi$, and executing the above query, 
can be done by a total deterministic \abb{ptt} $\cA$
as follows, similarly to the general \abb{ptt} $\cA$ above
(as discussed before Theorem~\ref{thm:matchall}).
Again, let $\lambda=(x_1,\dots,x_n)$ be an arbitrary order 
of the variables in $\cX$. 
Pebbles with distinct colours $c_1,\dots,c_{n-1}$ are used to
represent $x_1, \dots, x_{n-1}$, dropping them in that
order. For every $j\in[1,n]$, 
when pebbles $c_1,\dots,c_{j-1}$ are dropped on the tree 
and the head is at a candidate position $u_j$ for the variable $x_j$,  
all \abb{mso} tests $\phi_\ell$
are performed of which the free variables are in $\{x_1,\dots,x_j\}$
(and that have not been tested before).
Thus, when $\cA$ has enumerated a sequence $u_1,\dots,u_n$,
it can compute the boolean value of $\phi(u_1,\dots,u_n)$. 
For each match $u_1,\dots,u_n$ the tree~$r$ is generated,
such that for every occurrence of the variable $x_i$ in $r$
the subtree rooted at $u_i$ is generated,
by a separate process;
that is straightforward, even when $c_i$ is invisible:
lift pebbles $c_{n-1},\dots,c_{i+1}$ one by one (in that order), 
and then access $c_i$ and output $t|_{u_i}$.
Note that, as before, the matches are generated 
in the lexicographic document order
induced by the order $\lambda$. 

It remains to determine which are the visible and invisible pebbles,
keeping in mind that we wish to 
use as many invisible pebbles as possible for matching.
To do the \abb{mso} tests at position $u_j$ all pebbles $c_i$
for which $\{x_i,x_j\}\in E_\phi$ and $i<j$ should be observable.
Hence all such pebbles under the topmost pebble $c_{j-1}$
must be visible.
These are the pebbles corresponding to the set
\[
\mathrm{vis}(\lambda) = \{x_i \mid 
  \text{there exists }\{x_i,x_j\}\in E_\phi \text{ such that } i+1<j \}.
\]
Thus, for $\cA$ we define 
$C_\mathrm{v}=\{c_i\mid x_i\in \mathrm{vis}(\lambda)\}$
and $C_\mathrm{i}= \{c_i\mid x_i\notin \mathrm{vis}(\lambda)\}$.
Note that $c_{n-1}\in C_\mathrm{i}$.

In the case where the order $\lambda=(x_1,\dots,x_n)$ 
of the variables is irrelevant, 
we may want to determine an optimal order.
A finite undirected graph $G=(V,E)$ will be called 
a \emph{union of paths} if it is acyclic 
and has only vertices of degree at most $2$. Intuitively this means 
that each connected component of $G$ is a path. Thus, clearly, 
there is an order $v_1,\dots,v_p$ of the vertices of $G$ such that 
for all $i,j\in[1,p]$ with $i< j$, if $\{v_i,v_j\}\in E$ then $i+1=j$ 
(repeatedly pick a vertex of degree 0 or 1, and remove it from the graph
together with all its incident edges). We will call this 
an \emph{invisible order} of the vertices of $G$. Note that
a graph is a union of paths if and only if it has an invisible order.
Note also that every subgraph of $G$ is also a union of paths.

For an arbitrary finite undirected graph $G=(V,E)$, let us now say that 
a set $W \subseteq V$ of vertices of $G$ is a \emph{visible set} of $G$ if
the subgraph of $G$ induced by $V\setminus W$,
denoted by $G[V\setminus W]$, is a union of paths.
By the last sentence of the previous paragraph, every superset of a 
visible set is also a visible set. 

\begin{lemma}\label{lem:vis}
A set of variables $W \subseteq V_\phi$ is a visible set of $G_\phi$ 
if and only if there is an order $\lambda$ of $V_\phi$ 
such that $\mathrm{vis}(\lambda)\subseteq W$.
\end{lemma}

\begin{proof}
(If) It is easy to verify that every $\mathrm{vis}(\lambda)$ 
is a visible set of $G_\phi$. 
In fact, for all $i<j$, if $x_i,x_j\notin \mathrm{vis}(\lambda)$
and $\{x_i,x_j\}\in E_\phi$, then $i+1=j$. 

(Only if) Define the order $\lambda$ on $V_\phi$ as follows.
First list the vertices of $W$ in any order.
Then list the remaining vertices according to an invisible order
of the vertices of $G_\phi[V_\phi\setminus W]$. 
Obviously $\mathrm{vis}(\lambda) \subseteq W$. 
\end{proof}
 
\begin{theorem}\label{thm:matching}
Pattern $\phi \equiv \bool(\phi_1, \dots, \phi_m)$
can be matched by a total deterministic \abb{v$_k$i-ptt}
where $k = \#(W)$ for a visible set $W$ of $G_\phi$.
In particular, if $G_\phi$ is a union of paths, then $\phi$ 
can be matched by a total deterministic \abb{i-ptt}. 
\end{theorem}

\begin{proof}
By Lemma~\ref{lem:vis} there is an order $\lambda$ of $V_\phi$ 
such that $\mathrm{vis}(\lambda)\subseteq W$.
Hence at most $\#(W)$ visible pebbles suffice.
If $G_\phi$ is a union of paths, then $W=\nothing$ is a visible set. 
\end{proof}

Lemma~\ref{lem:vis} shows that finding an order $\lambda$ 
for which $\mathrm{vis}(\lambda)$ is of minimal size, is the same as 
finding a visible set $W$ of minimal size. 
Unfortunately, this is an NP-complete problem. More precisely,
the problem whether for a given graph $G=(V,E)$ and a given number $k$
there is a set of vertices $V'\subseteq V$ with \mbox{$\#(V')\geq k$} such that
$G[V']$ is a union of paths, is NP-complete (see Problem~GT21 of~\cite{GJ}).

We now give some examples of visible sets of a graph $G$.
It suffices to take as visible vertices
those of degree $\ge 3$ in $G$ 
(plus one vertex in each connected component that is a cycle). 
But often one can choose a smaller set.
\newcommand{\ladder}{
\draw (0,2) -- (0,0) -- (4,0) -- (4,2) -- (0,2);
\draw (1,2) -- (1,0);
\draw (2,2) -- (2,0);
\draw (3,2) -- (3,0);
\draw [fill] (0,2) circle [radius=0.07];
\draw [fill] (0,1) circle [radius=0.07];
\draw [fill] (0,0) circle [radius=0.07];
\draw [fill] (1,2) circle [radius=0.07];
\draw [fill] (1,1) circle [radius=0.07];
\draw [fill] (1,0) circle [radius=0.07];
\draw [fill] (2,2) circle [radius=0.07];
\draw [fill] (2,1) circle [radius=0.07];
\draw [fill] (2,0) circle [radius=0.07];
\draw [fill] (3,2) circle [radius=0.07];
\draw [fill] (3,1) circle [radius=0.07];
\draw [fill] (3,0) circle [radius=0.07];
\draw [fill] (4,2) circle [radius=0.07];
\draw [fill] (4,1) circle [radius=0.07];
\draw [fill] (4,0) circle [radius=0.07];
}
\newcommand{\smallladder}{
\draw (0,2) -- (0,0) -- (3,0) -- (3,2) -- (0,2);
\draw (1.5,2) -- (1.5,0);
\draw [fill] (0,2) circle [radius=0.07];
\draw [fill] (0,1) circle [radius=0.07];
\draw [fill] (0,0) circle [radius=0.07];
\draw [fill] (1.5,2) circle [radius=0.07];
\draw [fill] (1.5,1) circle [radius=0.07];
\draw [fill] (1.5,0) circle [radius=0.07];
\draw [fill] (3,2) circle [radius=0.07];
\draw [fill] (3,1) circle [radius=0.07];
\draw [fill] (3,0) circle [radius=0.07];
}
\begin{figure}[b]
\begin{center}
\begin{tikzpicture}
\begin{scope}[scale=0.8]
\ladder
\draw (0,1) circle [radius=0.152];
\draw (1,1) circle [radius=0.152];
\draw (2,1) circle [radius=0.152];
\draw (3,1) circle [radius=0.152];
\begin{scope}[shift={(6,0)}]
\ladder
\end{scope}
\draw (7,2) circle [radius=0.152];
\draw (8,0) circle [radius=0.152];
\draw (10,2) circle [radius=0.152];
\begin{scope}[shift={(0,3)}]
\ladder
\end{scope}
\draw (1,3) circle [radius=0.152];
\draw (1,5) circle [radius=0.152];
\draw (2,3) circle [radius=0.152];
\draw (2,5) circle [radius=0.152];
\draw (3,3) circle [radius=0.152];
\draw (3,5) circle [radius=0.152];
\begin{scope}[shift={(6,3)}]
\ladder
\end{scope}
\draw (7,3) circle [radius=0.152];
\draw (7,5) circle [radius=0.152];
\draw (9,3) circle [radius=0.152];
\draw (9,5) circle [radius=0.152];
\end{scope}
\end{tikzpicture}
\end{center}
  \caption{Visible sets of different sizes.}
  \label{fig:vissets}
\end{figure}
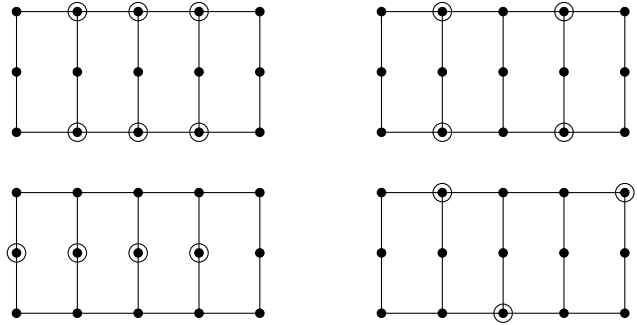
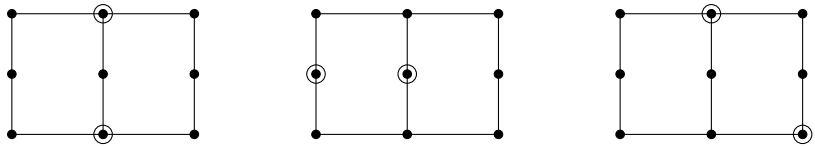
\begin{figure}
\begin{center}
\begin{tikzpicture}
\begin{scope}[scale=0.8]
\smallladder
\draw (1.5,2) circle [radius=0.152];
\draw (1.5,0) circle [radius=0.152];
\begin{scope}[shift={(5,0)}]
\smallladder
\end{scope}
\draw (5,1) circle [radius=0.152];
\draw (6.5,1) circle [radius=0.152];
\begin{scope}[shift={(10,0)}]
\smallladder
\end{scope}
\draw (11.5,2) circle [radius=0.152];
\draw (13,0) circle [radius=0.152];
\end{scope}
\end{tikzpicture}
\end{center}
  \caption{Three visible sets of minimal size.}
  \label{fig:minsizevis}
\end{figure}

\begin{example}
If $G$ is a cycle or a star, then it has a visible set $W$ with $\#(W)=1$
(for a cycle any singleton is a visible set, and for a star the visible set $W$ consists of 
the centre vertex). 

In Figs.~\ref{fig:vissets} and~\ref{fig:minsizevis} 
we show graphs with the vertices 
of a visible set $W$ encircled. For the graph $G$ in Fig.~\ref{fig:vissets}, 
the upper left $W$ consists of all vertices of degree~3. It is not minimal, 
in the sense that it has a proper subset that is also a visible set, 
as shown at the upper right. This one \emph{is} minimal, 
because dropping one of the vertices from $W$ produces a vertex of degree~3 
in the complement. Another minimal visible set (of the same size) 
is shown at the lower left: dropping the leftmost vertex of $W$ produces a cycle, 
and dropping one of the other vertices produces two vertices of degree~3. 
Finally, a visible set of size~3 is shown at the lower right. 
It is of minimal size, i.e., $\#(W)\geq 3$ for every visible set $W$ of $G$. 
In fact, removing a vertex of degree~2 from $G$ leaves a graph with two disjoint 
cycles that both must be broken, whereas removing a vertex of degree~3 from~$G$ 
either leaves a graph with two disjoint cycles or a graph with a cycle and a 
vertex of degree~3 of which the neighbourhood is disjoint with that cycle.
Thus, any pattern $\phi$ such that $G_\phi$ is isomorphic to $G$ 
can be matched with three visible pebbles.  

Visible sets of minimal size need not be unique. 
For the graph in Fig.~\ref{fig:minsizevis}, 
three different visible sets of minimal size are shown. 
\end{example}

If we allow matches to occur more than once in the output, then
Theorem~\ref{thm:matching} is not optimal (still assuming that 
the order $\lambda$ is irrelevant). Using the boolean laws,
the \abb{mso} formula $\phi \equiv \bool(\phi_1, \dots, \phi_m)$ can be written 
as a disjunction $\phi \equiv \psi_1\vee \cdots \vee \psi_k$ 
where each $\psi_i$ is a conjunction of some of the formulas 
$\phi_1, \dots, \phi_m$ or their negations. Now the \abb{ptt} $\cA$ 
can execute the queries  
`$\mbox{\tt for } \cX \mbox{ \tt where } \psi_i \mbox{ \tt return } r$'
consecutively for $i=1,\dots,k$. 
Obviously, $G_{\psi_i}$ is a subgraph of $G_\phi$ for every $i\in[1,k]$. 
Hence every visible set of $G_\phi$ is also a visible set of $G_{\psi_i}$,
and so the minimal size of the visible sets of $G_{\psi_i}$ is at most
the minimal size of the visible sets of $G_\phi$. Thus, pattern matching 
formulas $\psi_1,\dots,\psi_k$ consecutively needs at most 
as many visible pebbles as pattern matching $\phi$, but it may need less. 
As a simple example, let 
$\phi \equiv \phi_1(x,y)\wedge (\phi_2(y,z) \vee \phi_3(x,z))$.
Then $G_\phi$ is a triangle, which needs one visible pebble. 
But $\phi \equiv \psi_1\vee \psi_2$ where 
$\psi_1\equiv \phi_1(x,y)\wedge \phi_2(y,z)$ and
$\psi_2\equiv \phi_1(x,y)\wedge \phi_3(x,z)$. 
Both $G_{\psi_1}$ and $G_{\psi_2}$ are (unions of) paths, 
which do not need visible pebbles. 
Thus, $\phi$ can be matched by an \abb{i-ptt}. 
However, all matches for which $\phi_1\wedge \phi_2\wedge \phi_3$ 
holds occur twice in the output. 

We finally discuss another way to reduce the number of visible pebbles.
Suppose that, for some $i\in[1,m]$, the formula $\phi_i$
has exactly two free variables $x,y\in \cX$. 
Thus, the edge $\{x,y\}$ is in $E_\phi$.   
Suppose moreover that the trip defined by $\phi_i(x,y)$ is functional. 
Suppose finally that $W$ is a visible set of $G_\phi$ with $x,y\in W$. 
Then all other edges of $G_\phi$ incident with $y$ can be redirected
to $x$, and $y$ can be dropped from $W$. To be precise, every formula 
$\phi_j$ that contains the free variable $y$ can be changed into the 
formula $\forall y (\phi_i(x,y) \to \phi_j)$ that contains 
the free variable $x$ instead of $y$. 
The resulting query is obviously equivalent to the given one.

\section{Pebble Forest Transducers}\label{sec:pft}

The \abb{ptt} transforms ranked trees, whereas XML documents 
are unranked forests. However, it is not difficult to use, or slightly adapt, 
the \abb{ptt} for the transformation of forests. The most obvious, 
and well-known way to do this, is to encode the forests as binary trees. 
Let $\family{enc}'$ be the class of all encodings $\enc'$
(one encoding for each input alphabet $\Sigma$), and 
let $\family{dec}$ be the class of all decodings $\dec$
(one decoding for each output alphabet $\Delta$).
Then we can view the class $\family{enc}'\circ\VIPTT{k}\circ\family{dec}$ 
as the class of forest transductions realized by \abb{v$_k$i-ptt}'s. 
For the input forest $f$ this is a natural definition, because it is 
quite easy to visualize a \abb{ptt} walking on $\enc'(f)$ 
as actually walking on $f$ itself. For the output forest $g$ this is also 
a natural definition, as it is, in fact, easy to transform a \abb{ptt} 
that outputs $\enc(g)$ into a (slightly adapted type of) \abb{ptt} that 
directly outputs $g$ itself: change every output rule 
$\tup{q,\sigma,j,b}\to \delta(\tup{q_1,\stay},\tup{q_2,\stay})$ into 
$\tup{q,\sigma,j,b}\to \delta(\tup{q_1,\stay})\tup{q_2,\stay}$, 
and every output rule $\tup{q,\sigma,j,b}\to e$ into 
$\tup{q,\sigma,j,b}\to \epsilon$.
The definition is also natural with respect to typechecking, because 
a forest language $L$ is regular if and only if the tree language $\enc(L)$
is regular, and similarly for $\enc'(L)$. Since the transformation of the 
involved grammars can obviously be done in polynomial time, 
Theorem~\ref{thm:typecheck} in Section~\ref{sec:typechecking} also holds for 
\abb{v$_k$i-ptt} as forest transducers. 

We observe here that the class $\family{enc}'\circ\VIPTT{k}\circ\family{dec}$
does not depend on the chosen encodings and decodings, i.e., $\family{enc}'$
can be replaced by the class $\family{enc}$ of all encodings $\enc$, 
and $\family{dec}$ by the class $\family{dec}'$ of all decodings $\dec'$.
In fact, a \abb{ptt} that walks on $\enc'(f)$ can easily be simulated 
by one that walks on $\enc(f)$: the original label $\sigma^{kl}$ can be 
determined by inspecting the children of the node with label~$\sigma$.
Vice versa, a \abb{ptt} that walks on $\enc(f)$ can be simulated by one 
that walks on $\enc'(f)$: a node with label, e.g., $\sigma^{01}$ represents 
the original node and its first child with label $e$; the difference 
between these nodes can be stored in the finite state and in the pebble colours 
of the simulating \abb{ptt}. Moreover, a \abb{ptt} that outputs $\enc'(g)$ 
can easily be simulated by one that outputs $\enc(g)$: change, e.g., 
the rule $\tup{q,\sigma,j,b}\to \delta^{01}(\tup{q',\stay})$ into the two rules
$\tup{q,\sigma,j,b}\to \delta(\tup{p,\stay},\tup{q',\stay})$ and 
$\tup{p,\sigma,j,b}\to e$ where $p$ is a new state. 
Vice versa, a \abb{ptt} $\cA$ that outputs $\enc(g)$ can be simulated by  
a \abb{ptt} $\cA$ that outputs $\enc'(g)$, 
but that requires look-ahead (Theorem~\ref{thm:look-ahead}), as follows. 
If $\cA$ has an output rule
$\tup{q,\sigma,j,b}\to \delta(\tup{q_1,\stay},\tup{q_2,\stay})$,  
then $\cA'$ has the rule 
$\tup{q,\sigma,j,b,\cB_{01}}\to \delta^{01}(\tup{q_2,\stay})$
where $\cB_{01}$ is a look-ahead test that finds out whether $\cA$ 
can generate $e$ when started in state $q_1$ 
in the current situation. To be precise, 
$\cB_{01}$ is obtained from $\cA$ by changing its set of initial states 
into $\{q_1\}$ and removing all output rules that do not output $e$. 
And of course, $\cA'$ has similar rules for the other symbols~$\delta^{ij}$. 

So far so good, in particular for the input forest $f$. There is, however, 
another natural possibility for the output forest $g$, 
as introduced and investigated in~\cite{PerSei} for macro tree transducers. 
It is quite natural to allow a \abb{ptt} 
that directly outputs~$g$, as discussed above, to not only have 
output rules with right-hand sides $\delta(\tup{q_1,\stay})\tup{q_2,\stay}$ 
and $\epsilon$, but also right-hand sides 
$\tup{q_1,\stay}\tup{q_2,\stay}$ and $\delta(\tup{q',\stay})$ that
realize the concatenation of forests and the formation of a tree out of a forest.
  
Accordingly we define a \emph{tree-walking forest transducer with nested pebbles}
(abbreviated \abb{pft}) to be the same as a \abb{ptt} $\cM$, 
except that its output alphabet is unranked, and its output rules are of the form 
$\tup{q,\sigma,b,j} \to \zeta$
with  $\zeta = \delta(\tup{q',\mathrm{stay}})$
introducing a new node with label $\delta$
and generating a forest from state $q'$,
or $\zeta = \tup{q_1,\mathrm{stay}}\, \tup{q_2,\mathrm{stay}}$
concatenating two forests,
or $\zeta = \epsilon$ generating the empty forest. 
Note that a right-hand side $\delta(\tup{q_1,\stay})\tup{q_2,\stay}$
is also allowed, as it can easily be simulated in two steps. 

Formally, an output form of the \abb{pft} $\cM$ on an input tree $t$ is defined to be
a forest in $F_\Delta(\con(t))$. 
Let $s$ be an output form and let $v$ be a leaf of $s$ 
with label $\tup{q,u,\pi}\in \con(t)$. If the rule 
$\tup{q,\sigma,b,j} \to \zeta$ is relevant to $\tup{q,u,\pi}$
then we write $s\Rightarrow_{t,\cM} s'$ 
where $s'$ is obtained from $s$ as follows. 
If the rule is not an output rule, then the label of $v$ is changed 
in the same way as for the \abb{pta} and \abb{ptt}. 
If $\;\zeta = \delta(\,\tup{q',\mathrm{stay}}\,)$
then node $v$ is replaced by the subtree $\delta(\tup{q',u,\pi})$.
If $\;\zeta = \tup{q_1,\mathrm{stay}}\, \tup{q_2,\mathrm{stay}}$
then node $v$ is replaced by the two-node forest 
$\tup{q_1,u,\pi}\tup{q_2,u,\pi}$. 
And if $\zeta = \epsilon$ then the node $v$ is removed from~$s$. 
The transduction realized by $\cM$ consists of all 
$(t,s)\in T_\Sigma\times F_\Delta$ such that 
$\tup{q_0,\rt_t} \Rightarrow^*_{t,\cM} s$ for some $q_0\in Q_0$. 
Thus, we have defined the \abb{pft} as a transformer of ranked trees 
into unranked forests. The corresponding classes of transductions 
are denoted by $\VIPFT{k}$.
For forest transformations one can of course consider the classes 
$\family{enc}'\circ \VIPFT{k}$.

\begin{lemma}\label{lem:pftvsptt}
For every $k\geq 0$,
\[
(1) \;\;\VIPTT{k}\circ\family{dec} \subseteq \VIPFT{k} \quad\text{and}\quad
(2) \;\;\VIPFT{k}\circ\family{enc} \subseteq \VIPTT{k}\circ\IdPTT
\]
and similarly for the deterministic case. 
\end{lemma}

\begin{proof}
Inclusion~(1) is obvious from the discussion above:
change every rule 
$\tup{q,\sigma,j,b}\to \delta(\tup{q_1,\stay},\tup{q_2,\stay})$ into 
$\tup{q,\sigma,j,b}\to \delta(\tup{q_1,\stay})\tup{q_2,\stay}$, 
and every rule $\tup{q,\sigma,j,b}\to e$ into 
$\tup{q,\sigma,j,b}\to \epsilon$.

The proof of inclusion~(2) is similar to the proof in~\cite{PerSei}
that every macro forest transducer can be simulated by two macro tree transducers.   
Let $\cM$ be a \abb{v$_k$i-pft} with (unranked) output alphabet $\Delta$.
Let $\Delta_1$ be the ranked alphabet $\Delta\cup\{\at^{(2)},e^{(0)}\}$,
where every element of $\Delta$ has rank~1. 
We now obtain the \mbox{\abb{v$_k$i-ptt}} $\cM'$ from $\cM$ by changing every output rule
$\tup{q,\sigma,b,j} \to \tup{q_1,\stay}\, \tup{q_2,\stay}$ into 
$\tup{q,\sigma,b,j} \to \at(\tup{q_1,\stay},\tup{q_2,\stay})$ and every output rule
$\tup{q,\sigma,b,j} \to \epsilon$ into $\tup{q,\sigma,b,j} \to e$.
Let `$\fl$' be the mapping from $T_{\Delta_1}$ to $F_\Delta$ defined by
$\fl(\at(t_1,t_2)=\fl(t_1)\fl(t_2)$, $\fl(\delta(t))=\delta(\fl(t))$ and $\fl(e)=\epsilon$. 
Then obviously $\tau_\cM=\tau_{\cM'}\circ \fl$. 
Thus, it remains to show that the mapping $\fl\circ\enc$ is in $\IdPTT$.
We will prove this after Theorem~\ref{thm:dtl}. 
It is, in fact, not hard to see that $\fl\circ\enc$ is even in $\dTT$.
\end{proof}

\smallpar{Typechecking}
The inverse type inference problem and the typechecking problem
are defined for \abb{pft}'s as in Section~\ref{sec:typechecking},
except that $G_\mathrm{out}$ is a regular forest grammar rather 
than a regular tree grammar.
It follows from Lemma~\ref{lem:pftvsptt}(2),
together with Lemma~\ref{lem:nul-decomp}, Theorem~\ref{thm:decomp}, 
and Propositions~\ref{prop:invtypeinf} and~\ref{prop:typecheck}
that these problems can be solved for \abb{v$_k$i-pft}'s in 
$(k\!+\!4)$-fold and $(k\!+\!5)$-fold exponential time.
However, it is shown in~\cite[Section~7]{Eng09} that they
can be solved for \abb{v$_k$-pft}'s in 
the same time as for \abb{v$_k$-ptt}'s, i.e., in ($k+1$)-fold and 
($k+2$)-fold exponential time, respectively. 
This is due to the fact (shown in~\cite[Lemma~4]{Eng09}) 
that inverse type inference for the mapping $\fl\circ\enc$
can be solved in polynomial time, cf. the proof of Lemma~\ref{lem:pftvsptt}.
For exactly the same reason a similar result holds for \abb{v$_k$i-pft}'s.
In other words, Theorem~\ref{thm:typecheck} also holds for \abb{v$_k$i-pft}'s. 

\begin{theorem}\label{thm:typecheck-pft}
For fixed $k\geq 0$, 
the inverse type inference problem and the typechecking problem are solvable
for \abb{v$_k$i-pft}'s in $(k\!+\!2)$-fold and $(k\!+\!3)$-fold exponential time,
respectively.
\end{theorem}

\smallpar{\small{MSO} tests}
It should be clear that Theorem~\ref{thm:mso} also holds for the \abb{pft},
as \abb{mso} tests only concern the input tree. 

\smallpar{Pattern matching}
Pattern matching for forests can be defined in exactly 
the same way as we did for trees in Section~\ref{sec:pattern}.
Since, obviously, Lemma~\ref{lem:encmso} also holds for 
arbitrary $n$-ary patterns instead of trips, we may however assume that 
the input forest $f$ over $\Sigma$ of the query 
\[
\mbox{\tt for } \cX
\mbox{ \tt where } \phi
\mbox{ \tt return } r
\] 
is encoded as a binary tree $t=\enc'(f)$ over $\Sigma'$ for which
we execute the query 
\[
\mbox{\tt for } \cX
\mbox{ \tt where } \phi'
\mbox{ \tt return } r
\] 
where $\phi'$ is the encoding of the formula $\phi$ 
according to Lemma~\ref{lem:encmso}.
Consequently, we can use a \abb{pft} to execute this query 
and produce for each match of $\phi'(x_1,\dots,x_n)$ the required output $r$.
We may now assume that $r$ is a forest rather than a tree, and we may 
for simplicity assume that $r$ is of the form $\mu(x_1\cdots x_n)$ 
for some output symbol $\mu$. Thus, the output for each match 
$\phi'(u_1,\dots,u_n)$ is $\mu(f|_{u_1}\cdots f|_{u_n})$,  
and the output forest is of the form $s= r_1r_2\cdots r_ke$
where $r_1,\dots,r_k$ are the outputs corresponding to all the matches.
Note that $e$ is another output symbol, and so $\Delta=\Sigma\cup\{\mu,e\}$. 
It should be clear how the total deterministic \abb{ptt} $\cA$ 
in Section~\ref{sec:pattern} can be changed into 
a total deterministic \abb{pft} that executes this query. The only small problem 
is that $\cA$ outputs the encoded subtrees $t|_{u_i}$ rather than the required 
subtrees $f|_{u_i}$. However, a \abb{pft} can easily transform an encoded forest 
$\enc'(f|_u)$ into the forest~$f|_u$, using rules 
$\tup{q,\sigma^{11},j,b}\to \sigma(\tup{q,\down_1})\tup{q,\down_2}$,
$\tup{q,\sigma^{01},j,b}\to \sigma\tup{q,\down_1}$,
$\tup{q,\sigma^{10},j,b}\to \sigma(\tup{q,\down_1})$, and
$\tup{q,\sigma^{00},j,b}\to \sigma$.

From this it should be clear that Theorems~\ref{thm:matchall} 
and~\ref{thm:matching} also hold for forest pattern matching and \abb{pft}. 

\smallpar{Expressive power}
As in~\cite{PerSei}, the \abb{pft} is more powerful than the \abb{ptt}. 
In particular, the \abb{i-pft} is more powerful than the \abb{i-ptt} 
that generates encoded forests, i.e., 
$\IPTT\circ \family{dec}$ is a proper subclass of $\IPFT$. 
In fact, it is well known 
(cf.~\cite[Lemma~7]{EngMan03} and~\cite[Lemma~5.40]{FulVog}), and easy to see, 
that the height of the output tree of a functional \abb{tt} $\cM$ (which means 
that $\tau_\cM$ is a function) is linearly bounded by the size of the input tree:
otherwise $\cM$ would be in a loop and would generate infinitely many output trees 
for that input tree. Since $\IPTT \subseteq \TT \circ \TT$ 
by Lemma~\ref{lem:nul-decomp}, this implies that for a functional \abb{i-ptt}
the height of the output tree is exponentially bounded by 
the size of the input tree. However, the following total deterministic \abb{i-pft} 
$\cM_\text{2exp}$ outputs, for an input tree of size $n$, a forest of length 
double exponential in $n$. Since the height of the encoded output forest is at 
least the length of that forest, this transformation cannot be realized by 
an \abb{i-ptt} that generates encoded forests. 
The transducer~$\cM_\text{2exp}$ is similar to the \abb{i-ptt} $\cM_\text{sib}$ 
of Example~\ref{ex:siberie}, assuming that there are large cities only. 
Thus, using its pebbles, it enumerates $2^n$ itineraries (where $n$ is 
the number of intermediate cities). However, after marking an itinerary, 
it does not output the itinerary, but instead branches into two identical 
subprocesses that continue the enumeration. 
After the last itinerary, $\cM_\text{2exp}$ is branched into 
a forest of $2^{2^n}$ copies of itself, each of which finally outputs one symbol. 
Imitating~$\cM_\text{sib}$, the \abb{i-pft} $\cM_{\text{2exp}}$ 
first walks to the leaf:

$\tup{q_\mathrm{start},\sigma_1,j,\nothing} \to \tup{q_\mathrm{start},\down_1}$

$\tup{q_\mathrm{start},\sigma_0,1,\nothing} \to \tup{q_1,\up}$

\noindent
Then, in state $q_1$, it marks as many cities as possible: 

$\tup{q_1,\sigma_1,1,\nothing} \to \tup{q_1,\drop_c;\up}$

$\tup{q_1,\sigma_1,0,\nothing} \to   
  \tup{q_\mathrm{next},\down_1} 
  \tup{q_\mathrm{next},\down_1}$

\noindent
In state $q_\mathrm{next}$ it continues the search for itineraries
by unmarking the most recently marked city; 
when arriving at the leaf it outputs $e$:

$\tup{q_\mathrm{next},\sigma_1,1,\nothing} \to \tup{q_\mathrm{next},\down_1}$

$\tup{q_\mathrm{next},\sigma_1,1,\{c\}} \to \tup{q_1,\lift_c;\up}$

$\tup{q_\mathrm{next},\sigma_0,1,\nothing} \to e$

\noindent
This ends the description of the \abb{i-pft} $\cM_{\text{2exp}}$.

\section{Document Transformation}\label{sec:document}

In this section we compare the \abb{i-ptt} and \abb{i-pft} to 
the document transformation languages \abb{DTL} and \abb{TL},
which transform (unranked) forests. 
We prove that \abb{DTL} can be simulated by the \abb{i-ptt},
and that \abb{TL} has the same expressive power as the \abb{i-pft}. 

The \emph{Document Transformation Language}
\abb{DTL} was introduced and studied in \cite{ManNev00}.
A \emph{program}  in the \abb{DTL} framework
is a tuple $\cP=(\Sigma,\Delta,Q,Q_0,R)$ where $\Sigma$ and $\Delta$ 
are unranked alphabets, $Q$ is a finite set of states, $Q_0\subseteq Q$ is 
the set of initial states, and $R$ is a finite set of \emph{template rules} 
of the form $\tup{q,\phi(x)} \to f$,
where $f$ is a forest over $\Delta$,
the leaves of which 
can additionally be labelled by a \emph{selector}
of the form $\tup{q',\psi(x,y)}$; $q$ and $q'$ are states in $Q$, and 
$\phi$ and $\psi$ are \abb{mso} formulas over $\Sigma$,
with one and two free variables respectively.
Such a rule can be applied in state $q$ at an input node
$x$ that matches $\phi$, i.e., 
satisfies $\phi(x)$.
Then program $\cP$ outputs forest $f$,
where each selector $\tup{q',\psi(x,y)}$ 
is recursively computed as the result of a
sequence of copies of $\cP$,
started in state $q'$ at each of the nodes $y$
that satisfy $\psi(x,y)$,  
the nodes taken in pre-order (i.e., document order).
Thus, $\cP$ ``jumps'' from node $x$ to each node $y$,
according to the trip defined by the \abb{mso} formula~$\psi$.

Formally, a configuration of $\cP$ on input forest $t$ 
is a pair $\tup{p,u}$ where $u$ is a node of $t$ 
and $p$ is either a state or a selector of $\cP$.
An output form of $\cP$ on $t$ is a forest in $F_\Delta(\con(t))$,
where $\con(t)$ is the set of configurations of $\cP$ on $t$. 
As usual, the computation steps of $\cP$ on $t$ are formalized 
as a binary relation $\Rightarrow_{t,\cP}$ on $F_\Delta(\con(t))$.
Let $s$ be an output form and let $v$ be a leaf of $s$ with label 
$\tup{q,u}\in \con(t)$, where $q$ is a state of $\cP$. 
Moreover, let $\tup{q,\phi(x)} \to f$ be a template rule of $\cP$ 
such that $t\models \phi(u)$. Let $\theta_u(f)$ be the forest obtained 
from~$f$ by changing every selector $\tup{q',\psi(x,y)}$ into 
$\tup{\tup{q',\psi(x,y)},u}$. 
Then we write $s\Rightarrow_{t,\cP} s'$
where $s'$ is obtained from $s$ by replacing the node $v$ by the 
forest $\theta_u(f)$. 
Now let $s$ be an output form and let $v$ be a leaf of $s$ with label 
$\tup{\tup{q',\psi(x,y)},u}$. Then we write $s\Rightarrow_{t,\cP} s'$
where $s'$ is obtained from $s$ by replacing the node $v$ by the 
forest $\tup{q',u'_1}\cdots\tup{q',u'_\ell}$ where $u'_1,\dots,u'_\ell$
is the sequence of all nodes $u'$ of $t$, in document order, such that
$t\models \psi(u,u')$.  
The transduction $\tau_\cP$ realized by $\cP$ is defined by 
$\tau_\cP = \{(t,s)\in F_\Sigma\times F_\Delta \mid \exists\,q_0\in Q_0: 
\tup{q_0,\rt_t} \Rightarrow^*_{t,\cP} s\}$. 

The \abb{dtl} program $\cP$ is \emph{deterministic} if for every two rules 
$\tup{q,\phi(x)} \to f$ and $\tup{q,\phi'(x)} \to f'$ with the same state $q$,
the tests $\phi(x)$ and $\phi'(x)$ are exclusive, 
meaning that the sites they define are disjoint. 

We observe here that in \cite{ManNev00} the selectors 
have a more complicated form, which we will discuss after the next lemma. 

We have defined the \abb{dtl} program such that 
the input $t$ is an unranked forest,
and thus it can in particular be a ranked tree. 
It should be clear from Lemma~\ref{lem:encmso} 
(which also holds for sites instead of trips) that we may in fact 
restrict ourselves to ranked trees and assume that input forests 
are encoded as binary trees. Thus, \emph{from now on we assume that} 
in the above definition $\Sigma$ is a ranked alphabet and 
$t\in T_\Sigma$ is a ranked input tree. This allows us to compare 
\abb{dtl} programs with \abb{pft}'s. 

Let $\DTL$ denote the transductions 
realized by \abb{dtl} programs and $\dDTL$ those realized 
by deterministic \abb{dtl} programs, from ranked trees to unranked forests. 
Thus, the class of forest transductions realized by \abb{dtl} programs
is equal to $\family{enc}'\circ \DTL$, and similarly for the deterministic case. 

\begin{lemma}\label{lem:dtl}
$\DTL \subseteq \IPFT$ and $\dDTL \subseteq \IdPFT$.
\end{lemma}

\begin{proof}
Let $\cP=(\Sigma,\Delta,Q,Q_0,R)$ be a \abb{dtl} program. 
We construct an equivalent \abb{I-PFT} $\cM$ with \abb{mso} tests,
cf.\ Theorem~\ref{thm:mso}. 
It has the same alphabets $\Sigma$ and $\Delta$ as $\cP$.
Since $\cM$ stepwise simulates $\cP$, 
its set of states consists of the states and selectors of $\cP$,
plus the states that it needs to execute the subroutines discussed below. 
It has the same initial states $Q_0$ as $\cP$. 
Moreover, it uses invisible pebbles of a single colour $\odot$,
and never lifts its pebbles.

For an input tree $t$, the transducer $\cM$ simulates a template rule
$\tup{q,\phi(x)} \to f$ in state $q$ at node $u$ of $t$ 
by first using an \abb{mso} head test to check whether $t\models \phi(u)$.
With a positive test result, it calls a subroutine $S$ that outputs 
the $\Delta$-labelled nodes of the right-hand side $f$. 
The subroutine $S$ is started by $\cM$ in state~$[f]$.
If its state is of the form $[sf']$, for a tree $s$
and a forest $f'$, it uses a rule 
$\tup{[sf'],\sigma,j,b} \to \tup{[s],\mathrm{stay}}\,\tup{[f'],\mathrm{stay}}$,
branching the computation.
If the state is of the form $[\delta(f')]$,
the rule is 
$\tup{[\delta(f')],\sigma,j,b} \to \delta(\tup{[f'],\mathrm{stay}})$, 
and if it is of the form $[\epsilon]$,
the rule is $\tup{[\epsilon],\sigma,j,b}\to \epsilon$. 
If the state is of the form
$[\tup{q',\psi(x,y)}]$, for a selector $\tup{q',\psi(x,y)}$, 
the subroutine $S$ returns control to 
(this copy of) $\cM$ in state $\tup{q',\psi(x,y)}$. 
In that state, $\cM$ first drops a pebble $\odot$ on the current node~$u$
and then calls a subroutine $S_{q',\psi}$
that finds all nodes $u'$ in the input tree $t$ for which $\psi(u,u')$ holds.
The subroutine does this by performing a depth-first traversal of $t$,
starting at the root, checking in each node $u'$
whether $t\models \psi(u,u')$ 
using an \abb{mso} test on the observable configuration.
If true, then $S_{q',\psi}$ branches into two
concatenated processes. The left branch returns control to $\cM$ in
state~$q'$, and the right branch continues the depth-first search.
When the search ends, $S_{q',\psi}$ outputs $\epsilon$.
Thus, $S_{q',\psi}$ transforms the configuration 
$\tup{\tup{q',\psi(x,y)},u,\pi}$ of $\cM$ into the forest of configurations
$\tup{q',u'_1,\pi}\cdots\tup{q',u'_\ell,\pi}$, 
where $u'_1,\dots,u'_\ell$ are all such nodes $u'$, in document order. 
With this definition of $\cM$, it should be clear that $\tau_\cM=\tau_\cP$. 
\end{proof}

The selectors in \cite{ManNev00} are more general than those defined above. 
They can be of the form $\tup{q'_1,\psi_1(x,y), \dots, q'_m,\psi_m(x,y)}$,
such that the \abb{mso} formulas $\psi_1(x,y),\dots,\psi_m(x,y)$ 
are mutually exclusive, i.e., 
the trips they define are mutually disjoint. Let $\psi(x,y)$ be the disjunction 
of all $\psi_i(x,y)$, $i\in[1,m]$. The execution of the above selector 
at node $u$ of the input tree
results in the forest $\tup{q'_{i_1},u'_1}\cdots\tup{q'_{i_\ell},u'_\ell}$ 
where $u'_1,\dots,u'_\ell$ is the sequence of all nodes $u'$ of $t$ 
in document order such that $t\models \psi(u,u')$, and for every $j\in[1,\ell]$,
$i_j$ is the unique number in $[1,m]$ such that $t\models \psi_{i_j}(u,u'_j)$.
It should be clear that Lemma~\ref{lem:dtl} is still valid 
with these more general selectors. To execute the above selector, 
the \abb{i-pft} $\cM$ calls subroutine $S_{q'_1,\psi_1,\dots,q'_m,\psi_m}$
which in each node $u'$ tests each of the formulas $\psi_i(u,u')$;
if $\psi_i(u,u')$ is true, then the subroutine branches in two, 
in the first branch returning control to $\cM$ in state $q_i$. 

To compare $\DTL$ to $\IPTT$ rather than $\IPFT$ we also consider \abb{dtl} programs 
that transform ranked trees. 
A \abb{dtl} program $\cP=(\Sigma,\Delta,Q,Q_0,R)$ is \emph{ranked} if 
$\Sigma$ and $\Delta$ are both ranked alphabets, and 
every rule $\tup{q,\phi(x)}\to f$ satisfies the following two restrictions:
\begin{enumerate}
\item[(R1)] $f$ is a ranked tree in $T_\Delta(S)$ where $S$ is the set of selectors, and 
\item[(R2)] for every selector $\tup{q',\psi(x,y)}$ that occurs in $f$,
every input tree $t\in T_\Sigma$, and every node $u\in N(t)$,
if $t\models \phi(u)$ then there is a unique node $v\in N(t)$ such that $t\models \psi(u,v)$.
\end{enumerate}
In other words, the trip $T(\psi(x,y))$ is functional and,  
for fixed input tree $t\in T_\Sigma$, it is defined for every node of $t$ 
that satisfies $\phi(x)$. 
Thus, execution of the selector $\tup{q',\psi(x,y)}$ results in a ``jump'' 
from node $x$ to exactly one node $y$. This clearly implies that all reachable output forms
of $\cP$ are ranked trees in $T_\Delta(\con(t))$. 
Thus $\tau_\cP\subseteq T_\Sigma\times T_\Delta$ is a ranked tree transformation.
The class of transductions realized by ranked \abb{tl} programs will be denoted by $\DTLr$,
and by $\dDTLr$ in the deterministic case. 

\begin{corollary}\label{cor:dtlr}
$\DTLr \subseteq \IPTT$ and $\dDTLr \subseteq \IdPTT$.
\end{corollary}

\begin{proof}
The proof is the same as the one of Lemma~\ref{lem:dtl}, except for 
the subroutines~$S$ and~$S_{q',\psi}$. 
The states of $S$ are now of the form $[s]$ 
where $s$ is a subtree of a right-hand side of a rule. 
Instead of the rules for states $[sf']$, $[\delta(f')]$, and $[\epsilon]$, 
subroutine~$S$ has rules $\tup{[\delta(s_1,\dots,s_m)],\sigma,j,b}\to 
\delta(\tup{[s_1],\stay},\dots,\tup{[s_m],\stay})$ for every $\delta$ of rank $m$
and all trees $s_1,\dots,s_m$ (restricted to subtrees of right-hand sides). 
When subroutine $S_{q',\psi}$ finds a node $u'$ such that $t\models \psi(u,u')$
(and it always finds one by restriction (R2)), it returns control to $\cM$ and 
does not continue the depth-first search. 
\end{proof}

It can, in fact, be shown that when output forests
are encoded as binary trees, $\DTL$ is included in $\IPTT$. 
Thus, instead of $\IPFT$ we consider the class $\IPTT \circ \family{dec}$
(which equals the class $\IPTT \circ \family{dec}'$), cf. Section~\ref{sec:pft}.
The next theorem will not be used in what follows (except in the paragraph
directly after the theorem). 

\begin{theorem}\label{thm:dtl}
$\DTL \subseteq \IPTT \circ \family{dec}$ and
$\dDTL \subseteq \IdPTT \circ \family{dec}$.
\end{theorem}

\begin{proof}
Let $\cP=(\Sigma,\Delta,Q,Q_0,R)$ be a \abb{dtl} program. 
The main difficulty in outputting the binary encoding $\enc(f)$ of a forest $f$
as opposed to the construction in 
the proof of Lemma~\ref{lem:dtl} 
is that here the first symbol $\delta$ of $f$ has to be determined before
any other output can be generated.
We reconsider that construction, and here essentially 
make a depth-first sequential search over 
nodes in the computation tree 
(implemented using a stack of pebbled nodes)
instead of the recursive approach.
In that way an \abb{i-ptt} $\cM$ can simulate the leftmost computations 
of the \abb{dtl} program $\cP$. 

As unranked forests with selectors 
can be generated by the recursive definition
$f ::= \delta(f)f' \mid \tup{q,\psi} f \mid \epsilon$,
where $f'$ is an alias of $f$,  
\abb{DTL} rules are of the form $\tup{q,\phi(x)} \to f$,
where $f$ is $\delta(f_1)f_2$,  $\tup{q,\psi} f'$, or $\epsilon$.
The set of states of the transducer $\cM$ to be constructed 
consists of the states of $\cP$ and all states $[f]$ 
where $f$ is a subforest of a right-hand side of a rule of $\cP$,
plus the states of the subroutines $S'_{q',\psi}$ and $S''_{q',\psi}$ 
discussed below. 
The state $[f]$ is used to generate the binary encoding of
the subforest $f$, similarly to its use by the subroutine $S$ 
in the proof of Lemma~\ref{lem:dtl}. 
The initial states of $\cM$ are those of $\cP$. 
The pebble colours used by $\cM$ are $\tup{q,\psi,f}$ 
where $\tup{q,\psi} f$ occurs in the right-hand side of a rule of $\cP$,
and the special colour $\bot$. The state and pebble stack of $\cM$ 
store a part of the output form of $\cP$ that still has to be evaluated. 
The output alphabet of $\cM$ is $\Delta\cup\{e\}$ 
where each $\delta\in\Delta$ has rank~2 and $e$ has rank~$0$.

The transducer $\cM$ starts by dropping $\bot$ on the root.
To simulate, in state~$q$, a rule
$\tup{q,\phi(x)} \to f$ of $\cP$, it uses an \abb{MSO} head test 
to check whether $\phi$ holds for the current node, and goes into state $[f]$.
We consider the above three cases for $[f]$.

In state $[\tup{q',\psi} f']$, 
pebble $\tup{q',\psi,f'}$ is dropped on the current node $u$.
As in the proof of Lemma~\ref{lem:dtl},
$\cM$ then calls a subroutine $S'_{q',\psi}$ which, this time, finds \\[0.6mm]
the first node $u'$ (in document order) for which $\psi(u,u')$ holds,
where it returns control to $\cM$ in state $q'$.
If $S'_{q',\psi}$ does not find such a matching node $u'$, then 
it moves to the topmost pebble $\tup{q',\psi,f'}$, lifts it, 
and returns control to $\cM$ in state $[f']$. 

In state $[f]=[\delta(f_1)f_2]$,
the root $\delta$ of the first tree of the forest 
is explicitly given, and this is captured by the 
\abb{I-PTT} output rule 
$\tup{[f],\sigma,j,b} \to 
\delta(\, \tup{[f_1],\mathrm{drop}_\bot}, 
\tup{[f_2],\mathrm{stay}}\,)$.
The symbol $\bot$ is pushed, and never popped afterwards, 
making the stack of pebbles effectively empty:
the first copy of the transducer evaluates $f_1$ as left child of $\delta$.
The second copy inherits the stack and evaluates $f_2$
as right child of $\delta$, together with all postponed duties
as stored in the stack of pebbles. 
This will generate the siblings of $\delta$ in the original forest.

In state $[\epsilon]$, the transducer $\cM$ determines the colour of 
the topmost pebble, using an \abb{mso} test on the observable configuration. 
If it is $\bot$, it outputs $e$ for the empty forest.
Otherwise it calls subroutine $S''_{q',\psi}$ to continue the search 
corresponding to the topmost pebble $\tup{q',\psi,f'}$.
That subroutine finds the first node $u'$ after the current node $u$ 
(in document order) for which $\psi(v,u')$ holds, 
where $v$ is the position of the topmost pebble. Similar to $S'_{q',\psi}$,
if a matching node is found it returns control to $\cM$ in state $q'$,
and otherwise it lifts the topmost pebble
and returns control to $\cM$ in state $[f']$.

This ends the description of $\cM$. To understand its correctness, we show how the 
output forms of $\cM$ represent output forms of $\cP$. We disregard 
the output forms of $\cM$ that contain states of the subroutines 
$S'_{q',\psi}$ and $S''_{q',\psi}$, and view the execution of such a subroutine 
as one big computation step of $\cM$ that (deterministically) changes 
one configuration into another. The mapping `$\rep$' from such restricted 
output forms of $\cM$ to output forms of $\cP$ is defined as follows. 
The $\Delta$-labelled part of the output form of $\cM$ is decoded, i.e., 
$\rep(e)=\epsilon$ and $\rep(\delta(s_1,s_2))=\delta(\rep(s_1))\rep(s_2)$.
It remains to define `$\rep$' for the configurations on an input tree $t$ 
that occur in the restricted 
output forms of~$\cM$, i.e., for every configuration $\tup{p,u,\pi}$ where $p$ is 
a state $q$ of $\cP$ or a state $[f]$. We will write $\rep(p,u,\pi)$ 
instead of $\rep(\tup{p,u,\pi})$. The definition is by induction on the structure 
of $\pi$, of which the topmost pebble is of the form 
$(v,\bot)$ or $(v,\tup{q',\psi,f'})$. 
For a state $[f]$, we define $\rep([f],u,\pi(v,\bot))=\theta_u(f)$ and 
\[
\rep([f],u,\pi(v,\tup{q',\psi,f'}))=
\theta_u(f)\tup{q',u'_1}\cdots\tup{q',u'_\ell}\rep([f'],v,\pi)
\]
where $u'_1,\dots,u'_\ell$ are all nodes $u'$ after $u$ (in document order)
such that $t\models \psi(v,u')$. 
Note that $\rep([f],u,\pi)=\theta_u(f)\rep([\epsilon],u,\pi)$ 
because $\theta_u(\epsilon)=\epsilon$, and hence 
$\rep([f_1f_2],u,\pi)=\theta_u(f_1)\rep([f_2],u,\pi)$. 
For a state $q$ of $\cP$ we define $\rep(q,u,\pi)=\tup{q,u}\rep([\epsilon],u,\pi)$. 

It is now straightforward to prove, for every initial state $q_0$ of $\cP$, 
every input tree $t$, and every output form $s$ of $\cP$, that 
$\tup{q_0,\rt_t} \Rightarrow^*_{t,\cP} s$ if and only if 
there exists a restricted output form $s'$ of $\cM$ such that 
$\tup{q_0,\rt_t,(\rt_t,\bot)} \Rightarrow^*_{t,\cM} s'$ and $\rep(s')=s$.
The proof of the if-direction of this equivalence is by induction on 
the length of the computation, and consists of four cases,
depending on the state of the configuration of $\cM$ that is rewritten, 
as discussed above,
viz., $q$, $[\tup{q',\psi} f']$, $[\delta(f_1)f_2]$, or $[\epsilon]$. 
From the last two cases it follows 
that for every restricted output form 
$s'$ of $\cM$ there exists a restricted output form $s''$ of $\cM$ 
such that $s' \Rightarrow^*_{t,\cM} s''$,
$\rep(s'')=\rep(s')$, and the states of $\cM$ that occur in $s''$ are 
either states $q$ of $\cP$ or states of the form $[\tup{q',\psi}f']$.
In the only-if-direction we only consider leftmost computations of $\cP$,
i.e., computations in which always the first configuration 
of the output form (in pre-order) is rewritten.  
If $\rep(s')=\rep(s'')=s$, with $s''$ as above, then the first configuration 
of $\cM$ in $s''$ corresponds to the first configuration of $\cP$ in $s$, 
and the proof is similar to the first two cases of the proof of the if-direction.
The details are left to the reader. 
Since $\rep(s')=\dec(s')$ for every output tree $s'$ of $\cM$,  
the above equivalence implies that $\tau_\cM \circ \dec=\tau_\cP$. 
\end{proof}

We are now able to finish the proof of Lemma~\ref{lem:pftvsptt}(2). 
Consider the mapping $\fl: T_{\Delta_1}\to F_\Delta$ defined in that proof.
It can be realized by the one-state deterministic \abb{dtl} program with rules 
$\tup{q,\lab_\at(x)}\to \tup{q,\down_1(x,y)}\tup{q,\down_2(x,y)}$, 
$\tup{q,\lab_\delta(x)}\to \delta(\tup{q,\down_1(x,y)})$ for every $\delta\in\Delta$, and
$\tup{q,\lab_e(x)}\to \epsilon$. Hence, by Theorem~\ref{thm:dtl}, 
it is in $\IdPTT \circ \family{dec}$, which means that the mapping 
$\fl\circ\enc$ is in $\IdPTT$.

\smallskip
In \cite{ManBerPerSei} the language \abb{dtl} was
extended to the \emph{Transformation Language}~\abb{tl} where the states
have parameters that hold unevaluated
forests, similar to macro tree
transducers with outside-in parameter evaluation~\cite{EngVog85}.
In a \abb{tl} program~$\cP=(\Sigma,\Delta,Q,Q_0,R)$, 
the set of states $Q$ is a ranked alphabet
such that the initial states in $Q_0$ have rank~$0$. 
The rules of \abb{tl} program $\cP$ are of the form
$\tup{q,\phi(x)}(z_1,\dots,z_n) \to f$, 
where $n=\rank_Q(q)$ and $z_1,\dots,z_n$ are the formal parameters of $q$,
taken from a fixed infinite parameter set $Z=\{z_1,z_2,\dots\}$.
The right-hand side $f$ of the rule is a forest of
which the nodes can be labeled by a symbol from $\Delta$,
by a selector $\tup{q',\psi(x,y)}$, or by a formal parameter 
$z_i$ with $i\in[1,n]$.  
A node labeled by $\tup{q',\psi(x,y)}$ must have $\rank(q')$ children,
and a node labeled by parameter $z_i$ must be a leaf. 
Thus, in such a forest (called an \emph{action} in~\cite{ManBerPerSei}), 
selectors can be nested.
We could as well allow in \abb{tl} the more general selectors discussed after
Lemma~\ref{lem:dtl}, but we restrict ourselves to the usual selectors for simplicity
(and because they are the selectors in~\cite{ManBerPerSei}). 
Determinism of program $\cP$ is defined as for \abb{dtl}.

An output form of $\cP$ on input forest $t$ is a forest 
of which the nodes can be labeled 
either by a symbol from $\Delta$, 
or by a configuration $\tup{q,u}$ or $\tup{\tup{q,\psi(x,y)},u}$ of $\cP$
in which case the node must have $\rank(q)$ children. 
A node of an output form, or of a right-hand side of a rule, is said to be 
\emph{outermost} if all its proper ancestors are labelled 
by a symbol from $\Delta$. 
The computation steps of $\cP$ are formalized as 
a binary relation on output forms, as follows (similar to the \abb{dtl} case). 
Let $s$ be an output form, and let $v$ be an outermost node of $s$
with label $\tup{q,u}$, where $q$ is a state of $\cP$. Moreover,
let $\tup{q,\phi(x)}(z_1,\dots,z_n) \to f$ be a rule of $\cP$ 
such that $t\models \phi(u)$. Let $\theta_u(f)$ be defined as in the \abb{dtl} case.
Then we write $s\Rightarrow_{t,\cP} s'$
where $s'$ is obtained from $s$ by replacing the subtree $s|_v$ with root $v$ 
by the forest $\theta_u(f)$ in which every parameter $z_i$ is replaced by the subtree
$s|_{vi}$, for $i\in[1,\rank(q)]$. Intuitively, the subtree $s|_{vi}$ 
rooted at the $i$-th child $vi$ of $v$ is 
the $i$-th actual parameter of (this occurrence of) the state $q$. 
Now let $s$ be an output form and let $v$ be an outermost node 
of $s$ with label $\tup{\tup{q',\psi(x,y)},u}$ and $\rank(q')=m$. 
Then we write $s\Rightarrow_{t,\cP} s'$
where $s'$ is obtained from $s$ by replacing the subtree $s|_v$ with root $v$ 
by the forest 
$\tup{q',u'_1}(s|_{v1},\dots,s|_{vm})\cdots\tup{q',u'_\ell}(s|_{v1},\dots,s|_{vm})$ 
where $u'_1,\dots,u'_\ell$ is the sequence of all nodes $u'$ of $t$, 
in document order, such that $t\models \psi(u,u')$.
Intuitively, the actual parameters of (this occurrence of) 
the selector $\tup{q',\psi(x,y)}$ 
are passed to each new occurrence of the state $q'$. 
As in the \abb{dtl} case, the transduction realized by $\cP$ 
is defined by 
$\tau_\cP = \{(t,s)\in F_\Sigma\times F_\Delta \mid \exists\,q_0\in Q_0: 
\tup{q_0,\rt_t} \Rightarrow^*_{t,\cP} s\}$. 

In \cite{ManBerPerSei} the denotational semantics of a \abb{tl} program is given 
as a least fixed point. It is straightforward to show that that semantics is 
equivalent to the above operational semantics.\footnote{It is similar to the ``alternative''
fixed point characterization of the OI context-free tree languages mentioned 
after~\cite[Definition~5.19]{EngSch78}.
}
Also, in \cite{ManBerPerSei} the syntactic formulation of \abb{tl} is such that 
in the right-hand side of a rule the states can have forests as parameters 
rather than trees. Such a forest parameter $s_1\cdots s_m$, 
where each $s_i$ is a tree, can be expressed in our syntactic formulation of \abb{tl}
as the tree $\tup{@_m,x=y}(s_1,\dots,s_m)$, where $@_m$ is a special state 
of rank $m$ that has the unique rule 
$\tup{@_m,x=x}(z_1,\dots,z_m)\to z_1\cdots z_m$. 
 
\begin{example}
\label{ex:tlsiberie}
The transformation from Example~\ref{ex:siberie} can be computed
by a deterministic \abb{TL} program $\cP_\mathrm{sib}$
with the following rules, 
where the variables $i$, $\sigma_i$, $c$, and $\lambda_i$
range over the same values as in Example~\ref{ex:siberie}, 
with $c=1$ or $i=1$ in rule $\rho_4$.

\smallskip
$\rho_1: \tup{q_\mathrm{start},\rt(x)}  \to  \tup{q_\mathrm{start},\leaf(y)}$

\smallskip
$\rho_2: \tup{q_\mathrm{start},\neg\rt(x)\wedge\lab_{\sigma_0}(x)}  
\to  \tup{q_1,\up(x,y)}(\sigma_0,e)$

\smallskip
$\rho_3: \tup{q_0,\neg\rt(x)\wedge\lab_{\lambda_0}(x)}(z_1,z_2)  
\to  \tup{q_0,\up(x,y)}(z_1,z_2)$

\smallskip
$\rho_4: \tup{q_c,\neg\rt(x)\wedge\lab_{\lambda_i}(x)}(z_1,z_2)$ 

$\hspace{2cm}\to 
\tup{q_i,\up(x,y)}(\lambda_i(z_1),\tup{q_c,\up(x,y)}(z_1,z_2))$

\smallskip
$\rho_5: \tup{q_c,\rt(x)\wedge\lab_{\sigma_1}(x)}(z_1,z_2) \to r(\sigma_1(z_1),z_2)$

\smallskip
\noindent
Intuitively, $z_1$ represents an itinerary from some city to Vladivostok,
and $z_2$ represents a list of itineraries from Moscow to Vladivostok
(viz. all itineraries that do not have $z_1$ as postfix), where
we only consider itineraries that do not visit a small city twice in a row.

The selectors in the right-hand sides of the rules all define functional trips,
and hence select just one node. 
Rule $\rho_1$ jumps from the root to the 
leaf, and rules $\rho_2$, $\rho_3$, $\rho_4$ just move to the parent.

To show the correctness of $\cP_\mathrm{sib}$, 
let $u$ be a node of an input tree $t$, such that $u$ is not the leaf of $t$.
Moreover, let $\zeta_1$ be an output tree that is an itinerary 
from the child of $u$ to the leaf, of which the first stop 
is large ($c=1$) or small ($c=0$), and let $\zeta_2$ be an arbitrary output form.
Then $\tup{q_c,u}(\zeta_1,\zeta_2)$ generates the output form  
$r(s_1(\zeta_1),r(s_2(\zeta_1),\dots r(s_n(\zeta_1),\zeta_2)\cdots ))$ 
where $s_1,\dots,s_n$
are all possible itineraries from the root to $u$ 
such that every $s_i(\zeta_1)$ is an itinerary from root to leaf.
This can be proved by induction on the number of nodes between the root and $u$. 
The base of the induction is by rule $\rho_5$, which generates 
the root label $\sigma_1$, and the induction step is by
rules $\rho_3$ and $\rho_4$. In rule $\rho_3$ a small city is skipped. In rule $\rho_4$,
the outermost selector $\tup{q_i,\up(x,y)}$ generates all itineraries~$s_i$ 
from the root to $x$ that include $x$ (or rather, its label $\lambda_i$), 
whereas the innermost selector 
$\tup{q_c,\up(x,y)}$ generates all those that do not include $x$.
Taking $c=1$, $u$ equal to the parent of the leaf, and 
$\sigma_0$ to the label of the leaf, shows that 
$\tup{q_1,u}(\sigma_0,e)$ generates all required itineraries.
That implies the correctness of $\cP_\mathrm{sib}$ by rule $\rho_2$.  

An XSLT~1.0 program with exactly the same structure as 
$\cP_\mathrm{sib}$ is given in Section~\ref{sec:siberie}. 
\end{example}

As in the case of $\DTL$, we will assume that in the above definition 
of \abb{TL} program, the input alphabet $\Sigma$ is ranked
and the input forest $t$ is a ranked tree in $T_\Sigma$. 
Also, \emph{ranked} \abb{tl} programs are defined as for \abb{dtl} programs.
In particular, for every rule
$\tup{q,\phi(x)}(z_1,\dots,z_n)\to f$, the right-hand side 
$f$ is a ranked tree in $T_\Delta(S\cup Z_n)$ where 
$S$ is the set of selectors and $Z_n=\{z_1,\dots,z_n\}$.
The program $\cP_\mathrm{sib}$ of Example~\ref{ex:tlsiberie} is ranked.

Let $\TL$ denote the class of transductions 
realized by \abb{tl} programs and $\dTL$ the class of those realized 
by deterministic \abb{tl} programs, from ranked trees to unranked forests.
Moreover, $\TLr$ and $\dTLr$ denote the classes of transductions 
realized by ranked programs, from ranked trees to ranked trees.

In what follows we will prove that
$\TL = \IPFT$, and similarly for the deterministic case and for the ranked case 
(Theorem~\ref{thm:tl}). 
Note that this also implies that \abb{TL} programs and \abb{i-pft}'s 
realize the same forest transductions, i.e., 
$\family{enc}'\circ \TL = \family{enc}'\circ \IPFT$. 
These equalities
are variants of the well-known fact that macro grammars are
equivalent to indexed grammars \cite{Fis},
see also~\cite[Theorem~5.24]{EngVog}.

\begin{lemma}\label{lem:tlipft}
$\TL \subseteq \IPFT$ and $\dTL \subseteq \IdPFT$.
Moreover, $\TLr \subseteq \IPTT$ and $\dTLr \subseteq \IdPTT$.
\end{lemma}

\begin{proof}
The construction extends the one in the proof of Lemma~\ref{lem:dtl}.
The main idea is to use pebbles to store the actual parameters.
Thus, the pebble colours are of the form 
$([s_1],\dots,[s_m])$ where $m\geq 0$ and $s_1,\dots,s_m$ are subtrees 
of a right-hand side of a rule (in particular, the subtrees 
rooted at the children of a node that is labelled by a selector).
 
As in the \abb{dtl} case, for an input tree $t$, the transducer $\cM$
simulates a rule $\tup{q,\phi(x)}(z_1,\dots,z_n) \to f$ in state $q$ 
at node $u$ of $t$ by testing whether $t\models \phi(u)$ and, if successful, 
calling subroutine $S$. In this (nested) case, $S$ outputs 
the outermost $\Delta$-labelled nodes of $f$, plus the outermost 
$\Delta$-labelled nodes of the actual parameters that are the values of the 
formal parameters $z_i$ that occur outermost in $f$. 
For the states $[sf']$, $[\delta(f')]$, and $[\epsilon]$, 
the rules of $S$ are as in the proof of Lemma~\ref{lem:dtl}
(and see the proof of Corollary~\ref{cor:dtlr} for the ranked case). 
If the state of $S$ is of the form
$[\tup{q',\psi(x,y)}(s_1,\dots,s_m)]$,
then it drops a pebble $([s_1],\dots,[s_m])$
on the current node $u$ to represent the parameters, 
and returns control to (this copy of) $\cM$ in state $\tup{q',\psi(x,y)}$.
In that state, $\cM$ calls subroutine $S_{q',\psi}$, 
which works as in the \abb{dtl} case. 
Note that $\cM$ need not drop a pebble~$\odot$, 
as $S_{q',\psi}$ can use the pebble $([s_1],\dots,[s_m])$ instead. 
Finally, if the state of $S$ is of the form $[z_i]$ for
some formal parameter~$z_i$, this means that the corresponding actual parameter
has to be evaluated. To do this, the subroutine~$S$ searches for the topmost pebble,
which has some colour $([s_1],\dots,[s_m])$. Then $S$ lifts that pebble
and changes its state to $[s_i]$, ready to evaluate $s_i$. 

It is easy to show, for every $i\in\nat$, that 
whenever $\cM$ is in state $q$ or state $\tup{q,\psi(x,y)}$ 
with $i\in[1,\rank(q)]$,  
and whenever $S$ is in state $[f]$ and $z_i$ occurs in~$f$, 
then the top pebble with colour $([s_1],\dots,[s_m])$ satisfies $i\in[1,m]$.
Hence the last sentence of the previous paragraph never fails.  

To understand the correctness of $\cM$, we show how the output forms of $\cM$ 
represent output forms of $\cP$, similar to the correctness proof of 
Theorem~\ref{thm:dtl}. 
We restrict ourselves to output forms
in which all the states of $\cM$ are states of~$\cP$ or selectors of $\cP$
or states of the subroutine $S$, i.e., we disregard the states of 
the subroutines $S_{q',\psi}$ and view the execution of such a subroutine 
as one big step in the computation of $\cM$, changing a configuration
$\tup{\tup{q',\psi(x,y)},u,\pi}$ deterministically into a forest 
$\tup{q',u'_1,\pi}\cdots\tup{q',u'_\ell,\pi}$ (which is just a one-node tree
$\tup{q',u',\pi}$ in the ranked case).   
Thus, we define a mapping `$\rep$' from such restricted 
output forms of $\cM$ to the output forms of $\cP$.
The $\Delta$-labelled part of the output form is not changed by `$\rep$', 
i.e., $\rep(sf)=\rep(s)\rep(f)$, $\rep(\epsilon)=\epsilon$, and 
$\rep(\delta(f))=\delta(\rep(f))$ for $\delta\in\Delta$, 
where $s$ is a tree and $f$ a forest (or 
$\rep(\delta(s_1,\dots,s_m))= \delta(\rep(s_1),\dots,\rep(s_m))$ in the ranked case). 
It remains to define `$\rep$' for the configurations of $\cM$ 
that occur in restricted output forms,
i.e., for every configuration $\tup{p,u,\pi}$ where~$p$ is a state $q$ of $\cP$,
or a selector $\tup{q',\psi(x,y)}$ of $\cP$, or a state $[f]$ of $S$ 
(where~$f$ is a subforest of a right-hand side of a rule of $\cP$). 
As before, we will write $\rep(p,u,\pi)$ instead of $\rep(\tup{p,u,\pi})$.
The definition is by induction on the structure of $\pi$, 
of which we consider the topmost pebble:  
let $\pi=\pi'(v,([s_1],\dots,[s_m]))$. 
If $p=q$ or $p=\tup{q',\psi(x,y)}$, then 
$\rep(p,u,\pi)=\tup{p,u}(\rep([s_1],v,\pi'),\dots,\rep([s_m],v,\pi'))$. 
For $p=[f]$ we define
$\rep([f],u,\pi)$ to be the forest $\theta_u(f)$ in which 
every parameter $z_i$ is replaced by $\rep([s_i],v,\pi')$,
Finally, for $\pi=\epsilon$, we define $\rep(p,u,\epsilon)=\tup{p,u}$ 
in the first case, and $\rep([f],u,\epsilon)=\theta_u(f)$ in the second case. 
If we consider only reachable output forms of $\cM$, 
then `$\rep$' is well defined (cf. the previous paragraph).

It is now straightforward to prove, for every initial state $q_0$ of $\cP$,
every input tree $t$, and every output form $s$ of $\cP$, 
that $\tup{q_0,\rt_t} \Rightarrow^*_{t,\cP} s$ if and only if 
there exists a restricted output form $s'$ of $\cM$ such that 
$\tup{q_0,\rt_t,\epsilon} \Rightarrow^*_{t,\cM} s'$ and $\rep(s')=s$. 
In the proof one should use the rather obvious fact that 
for every restricted output form 
$s'$ of $\cM$ there exists a restricted output form $s''$ of $\cM$ 
such that $s' \Rightarrow^*_{t,\cM} s''$,
$\rep(s'')=\rep(s')$, and no states $[f]$ of $S$ occur in $s''$. 
The above equivalence implies that $\tau_\cM=\tau_\cP$. 
\end{proof}

\begin{example}
\label{ex:pttsiberie}
The \abb{i-ptt} $\cM$ corresponding to the (ranked) \abb{tl} program $\cP_\mathrm{sib}$
of Example~\ref{ex:tlsiberie}, according to the proof of Lemma~\ref{lem:tlipft},
works in essentially the same way as the \abb{i-ptt} $\cM_{\text{sib}}$ of 
Example~\ref{ex:siberie}. Rules $\rho_1$ to $\rho_5$ are translated into rules for~$\cM$ 
that are similar to the first 5 rules of $\cM_{\text{sib}}$. 
Rule $\rho_1$ can be translated into the first rule of $\cM_{\text{sib}}$, 
which implements the jump to the leaf. Rule $\rho_2$ can be translated into 
the rule $\tup{q_\mathrm{start},\sigma_0,1,\nothing}\to 
\tup{q_1,\drop_{([\sigma_0],[e])};\up}$. Thus, $\cM$ drops the special pebble 
$([\sigma_0],[e])$ at the leaf, where $\cM_{\text{sib}}$ does not drop a pebble. 
Rule $\rho_3$ can be translated into the rule 
$\tup{q_0,\lambda_0,1,\nothing}\to \tup{q_0,\drop_{([z_1],[z_2])};\up}$. 
Thus, $\cM$ drops the ``empty'' pebble $([z_1],[z_2])$ whenever 
$\cM_{\text{sib}}$ does not drop a pebble.
Rule $\rho_4$ can be translated into the rule 
$\tup{q_c,\lambda_i,1,\nothing}\to \tup{q_i,\drop_{c(\lambda_i)};\up}$, 
where $c(\lambda_i)$ is the pebble
$([\lambda_i(z_1)],[\tup{q_c,\up(x,y)}(z_1,z_2)])$ 
which is dropped by~$\cM$ instead of the pebble $c$. Note that the pebble colours 
$c(\lambda_i)$ and $([\sigma_0],[e])$ include the label 
($\lambda_i$ or $\sigma_0$) of the node on which the 
pebble is dropped, which is of course superfluous information. 
Finally, rule $\rho_5$ can be translated into the rule 
$\tup{q_c,\sigma_1,0,\nothing}\to 
r(\tup{[\sigma_1(z_1)],\stay},\tup{[z_2],\stay})$,
which calls the states $[\sigma_1(z_1)]$ and $[z_2]$ of the subroutine $S$.
In state $[\sigma_1(z_1)]$, $S$ outputs $\sigma_1$ and 
goes into state $[z_1]$. 
We note that at any moment of time, when $\cM$ is at node $u$ of the input tree,
all descendants of $u$, possibly including $u$ itself, carry a pebble. Thus, 
in state~$[z_i]$, $S$ moves down to the child of $u$, lifts pebble 
$([s_1],[s_2])$ and goes into state $[s_i]$. It is now easy to see that states 
$[z_1]$ and $[z_2]$ of $\cM$ correspond to states $q_\mathrm{out}$ and 
$q_\mathrm{next}$ of $\cM_{\text{sib}}$, respectively. 
In state $[z_1]$, $S$ moves down and outputs the labels of all nodes 
that are marked by some pebble $c(\lambda_i)$ or $([\sigma_0],[e])$, 
lifting those pebbles one by one.
In state $[z_2]$, $S$ moves down to the first pebble $c(\lambda_i)$,
replaces that pebble by the ``empty'' pebble $([z_1],[z_2])$, and returns control 
to~$\cM$, which then goes into state $q_c$ and moves up to the parent.
When, in state $[z_2]$, $S$ reaches the leaf with pebble $([\sigma_0],[e])$,
it lifts that pebble and outputs $e$. 
\end{example}

Lemma~\ref{lem:tlipft} and Theorem~\ref{thm:typecheck-pft} (for $k=0$) 
together provide an alternative proof
of the main result of \cite{ManBerPerSei}: 
the inverse type inference problem and the typechecking problem are solvable 
for \abb{tl} programs.
The proofs are, however, similar. In~\cite{ManBerPerSei}
every \abb{tl} program is decomposed into three macro tree transducers,
whereas we have decomposed every \abb{i-ptt} into two \abb{tt}'s. 
In general, decomposition into
\abb{tt}'s leads to more efficient typechecking
than decomposition into macro tree transducers, because
(cf. Proposition~\ref{prop:invtypeinf}) 
inverse type inference of a macro tree transducer takes double exponential
time, unless the number of parameters is bounded and 
the output type is fixed \cite{PerSei}.
Let us define a \abb{tl}$^\text{\abb{db}}$ program to be a \abb{tl} program 
in which the \abb{mso} formulas $\phi(x)$ and $\psi(x,y)$ in the 
template rules of the program are represented
by deterministic bottom-up finite-state tree automata that recognize 
the corresponding regular sites $\tmark(T(\phi))$ and trips $\tmark(T(\psi))$. 

\begin{theorem}\label{thm:tltypecheck} 
The inverse type inference problem and the typechecking problem are solvable
for \abb{tl}$^\text{\abb{db}}$ programs in $3$-fold and $4$-fold exponential time,
respectively.
\end{theorem}

\begin{proof}
By Theorem~\ref{thm:typecheck-pft},
these problems are solvable for \abb{i-pft}'s in 2-fold and \mbox{3-fold} exponential time.
Let us now assume that the regular sites and trips used in \abb{mso} tests of \abb{i-pft}'s
are also represented by deterministic bottom-up finite-state tree automata.
Then it is easy to see that the construction in
the proof of Lemma~\ref{lem:tlipft} takes polynomial time.
However, the \abb{mso} tests that are used by the resulting \abb{i-pft} have to be removed,
and the construction in the proof of Theorem~\ref{thm:mso} takes exponential time, as can be 
checked in a straightforward way. That involves
checking that the constructions in the proofs of 
Lemmas~\ref{lem:regular},~\ref{lem:sites}, and~\ref{lem:visiblesites} 
take polynomial time, and so does the 
construction in the proof of Proposition~\ref{prop:trips} (for the nonfunctional case),  
i.e., in the proof of~\cite[Theorem~8]{bloem}. 
The exponential in the proof of Theorem~\ref{thm:mso} is due to the use 
of the sets of states~$S$ of $\cB_d$ in the colours of the beads. 
Hence, solving the above problems 
takes one more exponential for \abb{tl}$^\text{\abb{db}}$ programs than for \abb{i-pft}. 
\end{proof}

A \abb{tl} program $\cP=(\Sigma,\Delta,Q,Q_0,R)$ is a \emph{macro tree transducer}, 
more precisely an \abb{OI} macro tree transducer (see~\cite{EngVog85}), if it is ranked, 
and for every rule $\tup{q,\phi(x)}(z_1,\dots,z_n)\to f$ the following hold. 
First, $\phi(x)\equiv \lab_\sigma(x)$ for some $\sigma\in\Sigma$. 
Second, for every selector $\tup{q',\psi(x,y)}$ that occurs in $f$, we have
$\psi(x,y)\equiv \down_i(x,y)$ for some $i\in[1,\rank_\Sigma(\sigma)]$.
It follows immediately from Lemma~\ref{lem:tlipft} that macro tree transducers 
can be simulated by \abb{i-ptt}. 
Let $\family{MT}_\text{\abb{OI}}$ denote the class of tree transductions realized by 
\abb{OI} macro tree transducers, and $\family{dMT}_\text{\abb{OI}}$ 
the corresponding deterministic class. 

\begin{corollary}\label{cor:mtiptt}
$\family{MT}_\text{\abb{OI}} \subseteq \IPTT$ and
$\family{dMT}_\text{\abb{OI}} \subseteq \IdPTT$. 
\end{corollary}

The inclusions are proper because for every \abb{OI} macro tree transduction
the height of the output tree is exponentially bounded by the height of the input tree
\cite[Theorem~3.24]{EngVog85}, whereas it is not difficult to construct 
a deterministic \abb{i-ptt} $\cM$ with input alphabet $\{\sigma,e\}$, 
where $\rank(\sigma)=2$ and $\rank(e)=0$, such that the height of the output tree 
is exponential in the \emph{size} of the input tree. 
The transducer~$\cM$ is similar to the \abb{i-ptt} $\cM_\text{sib}$ 
of Example~\ref{ex:siberie}, viewing the nodes of the input tree as large cities that
are ordered by document order; thus, the number of itineraries is indeed
exponential in the size of the input tree.
Note that by~\cite[Corollary~7.2]{EngMan99} and~\cite[Theorem~6.18]{EngVog85}, 
$\family{dMT}_\text{\abb{OI}}$ properly contains the class $\family{DMSOT}$ of deterministic 
\emph{\abb{mso} definable tree transductions} (see also~\cite[Section~8]{thebook}).  
Note also that, since $\family{dB}$ is properly contained in $\family{dMT}_\text{\abb{OI}}$
by~\cite[Corollary~6.16]{EngVog85}, the second part of Corollary~\ref{cor:mtiptt} strengthens 
the second part of Theorem~\ref{thm:bottomup}. It is open whether or not $\family{B}$ 
is contained in $\family{MT}_\text{\abb{OI}}$.  

We now turn to the inclusion $\family{I-PFT}\subseteq\family{TL}$.
To prove that, we need a normal form for \abb{i-pft}. 
We say that a rule of an \abb{i-pft} is \emph{initial} 
if the state in its left-hand side is an initial state. 
We define an \abb{i-pft} 
$\cM = (\Sigma, \Delta, Q, Q_0, C, \nothing, C_\mathrm{i}, R, 0)$ 
with $C=C_\mathrm{i}$
to be in \emph{normal form} 
if its rules satisfy the following five requirements:

(1) Initial states do not appear in the right-hand side of a rule.

(2) All initial rules are of the form 
$\tup{q_0,\sigma,0,\nothing}\to \tup{q,\drop_c}$ for some 
$q_0\in Q_0$, $\sigma\in\Sigma$, $q\in Q\setminus Q_0$, and $c\in C$.
Intuitively,  $\cM$ starts its computation by dropping a pebble 
on the root of the input tree. 

(3) All non-initial rules have a left-hand side of the form 
$\tup{q,\sigma,j,\{c\}}$ with $c\in C$.
Intuitively, $\cM$ always observes the topmost pebble, i.e., 
that pebble is always at the position of the head. 

(4) All non-initial non-output rules have a right-hand side $\tup{q',\alpha}$
with $q'\in Q\setminus Q_0$ and $\alpha=\stay$ or $\alpha=\mu;\drop_c$ or
$\alpha=\lift_c;\mu$ where $c\in C$ and 
$\mu\in\{\up,\stay\}\cup\{\down_i\mid i\in[1,\m_\Sigma]\}$. 
We will identify $\stay;\drop_c$ with $\drop_c$ and $\lift_c;\stay$ with $\lift_c$. 
Intuitively, to force that $\cM$ always observes the topmost pebble, 
$\cM$ always drops a pebble after moving, and always moves after lifting a pebble.
Note that, in a successful computation, 
$\cM$ never lifts the pebble that it dropped with an initial rule.

(5) There is a function $\delta$ from $C$ to 
$\{\up,\stay\}\cup\{\down_i\mid i\in[1,\m_\Sigma]\}$ such that 
(i) if a rule of $\cM$ has right-hand side $\tup{q',\lift_c;\mu}$, 
then $\mu=\delta(c)$, and
(ii)~for every rule $\tup{q,\sigma,j,\{d\}}\to \tup{q',\mu;\drop_c}$ of $\cM$,
if $\mu=\up$ then $\delta(c)=\down_j$, 
if $\mu=\stay$ then $\delta(c)=\stay$, and 
if $\mu=\down_i$ then $\delta(c)=\up$.  
Intuitively this means that $\cM$, after lifting a pebble, always knows
where to find the new topmost pebble. 
 
This ends the definition of normal form. Obviously, it can also be defined
for \abb{i-ptt}'s and for \abb{i-pta}'s. The \abb{i-pta} in normal form 
can be viewed as a reformulation of the two-way
backtracking pushdown tree automaton of~\cite{Slu}.
The \abb{i-ptt} in normal form can be viewed as a reformulation of the 
RT(P($S$))-transducer of~\cite{Eng86,EngVog}, 
where $S$ is the storage type Tree-walk of~\cite{Eng86}.\footnote{See 
also~\cite[Section~3.3]{EngMan03} where the \abb{tt} is related to 
the RT($S$)-transducer for $S =$ Tree-walk. 
}

\begin{lemma}\label{lem:ipftnf}
For every \abb{i-pft} $\cM$ an equivalent \abb{i-pft} $\cM'$ in normal form 
can be constructed. If $\cM$ is deterministic, then so is $\cM'$. 
The same holds for \abb{i-ptt}. 
\end{lemma}

\begin{proof}
The idea of the construction is a simplified version of the one in the proof of 
Theorem~\ref{thm:mso}, where ``beads'' are used to cover the shortest path between 
the head and the topmost pebble. Assuming that the \abb{i-pta} $\cA$ in that proof 
starts by dropping a pebble on the root (which is never lifted), 
the constructed \abb{i-pta} $\cA'$ satisfies the above requirements on the rules. 
To show the details, we will repeat that construction, in a simplified form.
Here, the only information a bead has to carry is the position 
of the previous pebble or bead. Moreover, we do not have to drop a bead 
on the position of the topmost pebble. 

Let $\cM$ be an \abb{i-pft} with colour set $C$. We may obviously assume 
that~$\cM$ already satisfies the first two requirements above. 
We construct $\cM'$ with the same states and initial states as $\cM$,
and with the colour set $C\cup B$ where 
$B=\{\up\}\cup\{\down_i\mid [1,\m_\Sigma]\}$. 
The function $\delta$ of requirement~(5) is defined by 
$\delta(d)=d$ for every $d\in B$,
and $\delta(c)=\stay$ for every $c\in C$. 
The rules of $\cM'$ are obtained from those of $\cM$ as follows. 
The initial rules of $\cM$ are also rules of~$\cM'$. 

If $\tup{q,\sigma,j,\nothing}\to \tup{q',\up}$ is a rule of $\cM$,
then $\cM'$ has the rules $\tup{q,\sigma,j,\{\up\}}\to \tup{q',\lift_\up;\up}$
and $\tup{q,\sigma,j,\{\down_i\}}\to \tup{q',\up;\drop_{\down_j}}$ for every $i$.
Also, if $\tup{q,\sigma,j,\{c\}}\to \tup{q',\up}$ is a rule of $\cM$,
then $\cM'$ has the rule $\tup{q,\sigma,j,\{c\}}\to \tup{q',\up;\drop_{\down_j}}$.

Similarly, if $\tup{q,\sigma,j,\nothing}\to \tup{q',\down_i}$ is a rule of $\cM$,
then $\cM'$ has the rules 
$\tup{q,\sigma,j,\{\down_i\}}\to \tup{q',\lift_{\down_i};\down_i}$ and
$\tup{q,\sigma,j,\{\mu\}}\to \tup{q',\down_i;\drop_\up}$ 
for every $\mu\in\{\up\}\cup\{\down_k\mid k\neq i\}$. 
Also, if $\tup{q,\sigma,j,\{c\}}\to \tup{q',\down_i}$ is a rule of $\cM$,
then $\cM'$ has the rule $\tup{q,\sigma,j,\{c\}}\to \tup{q',\down_i;\drop_\up}$.

The remaining rules of $\cM$ (viz. rules with right-hand side $\tup{q',\stay}$, 
output rules, rules that lift, and non-initial rules that drop) are treated as follows. 
If $\tup{q,\sigma,j,\nothing}\to \zeta$ is such a rule of $\cM$, then 
$\cM'$ has the rules $\tup{q,\sigma,j,\{\mu\}}\to \zeta$
for every bead $\mu\in B$.
If $\tup{q,\sigma,j,\{c\}}\to \zeta$ is such a rule of $\cM$, then  
it is also a rule of $\cM'$.

It should be clear that $\cM'$ is equivalent to $\cM$. 
Whenever $\cM$ observes the topmost pebble $c$, so does $\cM'$. 
Whenever $\cM$ does not observe $c$, $M'$ observes a bead 
that indicates the direction of the topmost pebble. 
Note that if $\cM'$ lifts pebble $c$ of~$\cM$, 
the new topmost pebble/bead is always at the same position,
because when $c$ was dropped $\cM'$ was observing the topmost pebble/bead. 
\end{proof}

The \abb{tl} program that we will construct to simulate a given \abb{i-pft} $\cM$
will use \abb{mso} formulas $\phi(x)$ and $\psi(x,y)$ that closely resemble 
the tests and instructions in the left-hand and right-hand sides 
of the rules of $\cM$, respectively.  Those tests and instructions are 
``local'' in the sense that they only concern the node~$x$, its parent, and its children.
Thus, we say that a \abb{tl} program $\cP$ is \emph{local} if in the left-hand side 
of a rule it only uses a formula $\phi_{\sigma,j}(x)$ for $\sigma\in\Sigma$
and $j\in[0,\m_\Sigma]$, where $\phi_{\sigma,0}(x)\equiv \lab_\sigma(x)\wedge \rt(x)$
and $\phi_{\sigma,j}(x)\equiv \lab_\sigma(x)\wedge \childno_j(x)$ for $j\neq 0$,
and in the right-hand side of that rule it only uses the formulas $\up(x,y)$ 
(provided $j\neq 0$), $\stay(x,y)$, and $\down_i(x,y)$ 
for $i\in[1,\rank_\Sigma(\sigma)]$.\footnote{Recall that
$\rt(x) \equiv \neg\,\exists z(\down(z,x))$, 
$\childno_i(x) \equiv \exists z(\down_i(z,x))$, 
$\up(x,y) \equiv \down(y,x)$,
and $\stay(x,y)\equiv x=y$.
} 
Thus, $\cP$~also satisfies restriction (R2) in the definition of a ranked \abb{tl} program.
Note that macro tree transducers, as defined before Corollary~\ref{cor:mtiptt}, 
are local ranked \abb{tl} programs. 
The classes of transductions realized by local \abb{tl} programs will be decorated 
with a subscript~$\ell$.

\begin{lemma}\label{lem:ipfttl}
$\IPFT\subseteq\TLl$ and $\IdPFT\subseteq\dTLl$. 
Moreover, $\IPTT\subseteq\TLlr$ and $\IdPTT\subseteq\dTLlr$.
\end{lemma}

\begin{proof}
Let $\cM = (\Sigma, \Delta, Q, Q_0, C, \nothing, C_\mathrm{i}, R, 0)$ 
with $C=C_\mathrm{i}$ be an \abb{i-pft} in normal form. 
We construct a \abb{tl} program $\cP$ that is equivalent to $\cM$. 
The set of states of $\cP$ is 
\[
Q_0\cup ((Q\setminus Q_0)\times C) \cup \{q_\bot\}.
\]
Each initial state has rank~$0$, 
each pair $\tup{q,c}$ has rank~$\#(Q\setminus Q_0)$,
and $q_\bot$ has rank~$0$.  
The set of initial states of $\cP$ is $Q_0$.
The rules of $\cP$ are defined as follows, 
where we denote a state $\tup{q,c}$ as $q^c$.
Let $Q\setminus Q_0=\{q_1,\dots,q_n\}$ 
where we fix the order $q_1,\dots,q_n$. 

First, if $\tup{q_0,\sigma,0,\nothing}\to \tup{q,\drop_c}$ 
is an initial rule of $\cM$, then $\cP$ has the rule 
$\tup{q_0,\phi_{\sigma,0}(x)}\to \tup{q^c,\stay(x,y)}(\bot,\dots,\bot)$,
where $\bot$ abbreviates $\tup{q_\bot,\stay(x,y)}$.
There are no rules of $\cP$ with $q_\bot$ in the left-hand side. 

Second, let $\tup{q,\sigma,j,\{c\}}\to \zeta$ be a (non-initial) rule of $\cM$ 
that does not contain a drop- or lift-instruction. Thus, $\zeta$ is of the form 
$\tup{p,\stay}$, $\tup{p_1,\stay}\tup{p_2,\stay}$, $\delta(\tup{p,\stay})$, 
or $\epsilon$, with $p,p_1,p_2\in Q$ and $\delta\in\Delta$.\footnote{In the case where 
$\cM$ is an \abb{i-ptt}, $\zeta$ is of the form 
$\tup{p,\stay}$ or $\delta(\tup{p_1,\stay},\dots,\tup{p_m,\stay})$.
}
Then $\cP$ has the rule
$\tup{q^c,\phi_{\sigma,j}(x)}(z_1,\dots,z_n)\to\zeta'$, where 
$\zeta'$ is obtained from $\zeta$ by replacing every $\tup{p,\stay}$ by 
$\tup{p^c,\stay(x,y)}(z_1,\dots,z_n)$.  

Third, let $\tup{q,\sigma,j,\{d\}}\to \tup{p,\mu;\drop_c}$ be a rule 
of $\cM$. Note that for every $\mu\in\{\up,\stay\}\cup\{\down_i\mid i\in[1,\m_\Sigma]\}$,
there is an \abb{mso} formula $\mu(x,y)$. 
Then $\cP$ has the rule 
\[
\tup{q^d,\phi_{\sigma,j}(x)}(z_1,\dots,z_n)\to\tup{p^c,\mu(x,y)}(s_1,\dots,s_n)
\]
where $s_i=\tup{q_i^d,\stay(x,y)}(z_1,\dots,z_n)$ for every $i\in[1,n]$;
thus, the rule is
\begin{eqnarray*}
\lefteqn{\tup{q^d,\phi_{\sigma,j}(x)}(z_1,\dots,z_n)\to} \\[1mm]
& & \tup{p^c,\mu(x,y)}
(\tup{q_1^d,\stay(x,y)}(z_1,\dots,z_n),\dots,
\tup{q_n^d,\stay(x,y)}(z_1,\dots,z_n)).
\end{eqnarray*}

Fourth and final, if $\tup{q,\sigma,j,\{c\}}\to \tup{q_i,\lift_c;\mu}$ is a rule 
of $\cM$, then $\cP$ has the rule 
$\tup{q^c,\phi_{\sigma,j}(x)}(z_1,\dots,z_n)\to z_i$. 

Intuitively, $\cP$ is in state $q^c$ when $\cM$ is in state $q$ 
and the topmost pebble of~$\cM$ is $c$. The parameter $z_i$ of $q^c$
contains the continuation of $\cM$'s computation 
just after pebble $c$ is lifted and $\cM$ goes into state $q_i$.
At the moment that $\cM$ drops pebble~$c$, $\cP$ does not know what the state 
$q_i$ of $\cM$ will be after lifting $c$ and thus prepares the continuation 
for every possible state. The correct continuation is then chosen 
by $\cP$ when it simulates $\cM$'s lifting of $c$. 
Note that due to requirement~(5) of the normal form, when $\cM$ lifts a pebble, 
it returns to the same node where it decided to drop the pebble 
(at that node, or at the parent or at one of the children of that node). 

Formally, we define a mapping `$\rep$' from the output forms of $\cM$ 
(except the initial one) to those restricted output forms of $\cP$ 
of which the outermost nodes are labelled by a symbol from $\Delta$ 
or by a configuration $\tup{q,u}$ where $q$~is a state of $\cP$ (thus, they are 
not labelled by a configuration $\tup{p,u}$ where $p$~is a selector of $\cP$). 
As in the proof of Lemma~\ref{lem:tlipft}, 
the $\Delta$-labelled part of the output form is not changed. 
Thus, it remains to define `$\rep$' for the configurations of $\cM$ that contain 
non-initial states, which are of the form $\tup{q,u,\pi(u,c)}$ because the 
topmost pebble is always at the position of the head. 
We define $\rep(q,u,\pi(u,c))=\tup{q^c,u} \rep'(\pi)$, 
where $\rep'$ maps the pebble stacks of $\cM$ to sequences of output forms of $\cP$,
recursively as follows:
$\rep'(\epsilon)=(\bot,\dots,\bot)$ and $\rep'(\pi(u,c))=(s_1,\dots,s_n)$ where 
$s_i=\tup{\tup{q_i^c,\stay(x,y)},u}\rep'(\pi)$ for every $i\in[1,n]$.
Note that `$\rep$' is injective. 

It is now straightforward to prove, for every $q\in Q\setminus Q_0$, 
every $c\in C$, every input tree $t$, 
and every output form $s$ of $\cP$ (restricted as described above), 
that $\tup{q^c,\rt_t}(\bot,\dots,\bot) \Rightarrow^*_{t,\cP} s$ if and only if 
there exists an output form $s'$ of $\cM$ such that 
$\tup{q,\rt_t,(\rt_t,c)} \Rightarrow^*_{t,\cM} s'$ and $\rep(s')=s$.
Since `$\rep$' is injective, $s'$ is in fact unique. 
Note that each computation step of $\cM$  
is simulated by two (or three) computation steps of $\cP$, 
where the second (and third) step executes a selector 
to satisfy the restriction on the output forms of $\cP$. 
Due to its special form, the execution of such a selector $\psi(x,y)$
changes the label $\tup{\tup{q',\psi(x,y)},u}$ of a node of the output form 
into $\tup{q',u'}$ where $u'$ is the unique node of the input tree for which
$\psi(u,u')$ holds. 

Taking into account the initial rules of $\cM$, it should be clear that the above 
equivalence proves that $\tau_\cP=\tau_\cM$. 
\end{proof}

\begin{example}
We illustrate Lemma~\ref{lem:ipfttl} with the deterministic 
\abb{i-ptt} $\cM_\mathrm{sib}$
of Example~\ref{ex:siberie}. 
We first construct an \abb{i-ptt} $\cM'_\mathrm{sib}$ in normal form 
that is equivalent to $\cM_\mathrm{sib}$. We also allow tuples
$\tup{q',\lift_d;\mu}$ in the output rules for any colour~$d$, 
which can easily be handled too. 
The transducer $\cM'_\mathrm{sib}$ has a new initial state~$q_\mathrm{in}$, 
in which it drops pebble $\odot$ on the root, which also serves as the pebble `$\up$'.
The pebble `$\down_1$' is denoted by $\downarrow$. 
The normal form function $\delta$ is defined by $\delta(\odot)=\up$,
$\delta(\downarrow)=\down_1$, and $\delta(c)=\stay$ for $c\in\{0,1\}$. 
There are new states $\overline{q}_0$ and $\overline{q}_1$ in which $\cM'_\mathrm{sib}$
moves up, drops pebble $\downarrow$, and goes into the corresponding unbarred state.
Thus the rules for them are  \\[1mm]
\indent
$\rho_{c,d}: \tup{\overline{q}_c,\sigma,1,\{d\}} \to 
\tup{q_c,\up;\drop_\downarrow}$\\[1mm]
\noindent
with $\sigma\in\Sigma$ and $d\in\{\odot,\downarrow,0,1\}$. 
The other rules (with $c=1$ or $i=0$ in rule $\rho_4$ as usual) are \\[1mm]
\indent
$\rho_0: \tup{q_\mathrm{in},\sigma_1,0,\nothing} \to 
   \tup{q_\mathrm{start},\drop_\odot}$ \\[1mm]
\indent
$\rho_1: \tup{q_\mathrm{start},\sigma_1,j,\{\odot\}} \to 
  \tup{q_\mathrm{start},\mathrm{down}_1;\drop_\odot}$ \\[1mm]
\indent
$\rho_2: \tup{q_\mathrm{start},\sigma_0,1,\{\odot\}} \to 
   \tup{\overline{q}_1,\stay}$ \\[1mm]
\indent
$\rho_3: \tup{q_0,\lambda_0,1,\{\downarrow\}} \to \tup{\overline{q}_0,\stay}$ \\[1mm]
\indent
$\rho_4: \tup{q_c,\lambda_i,1,\{\downarrow\}} \to 
   \tup{\overline{q}_i,\mathrm{drop}_c}$ \\[1mm]
\indent
$\rho_5: \tup{q_c,\sigma_1,0,\{\downarrow\}} \to 
r(\tup{q_\mathrm{out},\mathrm{stay}}, 
  \tup{q_\mathrm{next},\lift_\downarrow;\mathrm{down}_1})$ \\[1mm]
\indent
$\rho_6: \tup{q_\mathrm{out},\sigma_1,0,\{\downarrow\}} \to 
\sigma_1(\tup{q_\mathrm{out},\lift_\downarrow;\mathrm{down}_1})$ \\[1mm]
\indent
$\rho_7: \tup{q_\mathrm{out},\sigma_1,1,\{\downarrow\}} \to 
\tup{q_\mathrm{out},\lift_\downarrow;\mathrm{down}_1}$ \\[1mm]
\indent
$\rho_8: \tup{q_\mathrm{out},\sigma_1,1,\{c\}} \to 
\sigma_1(\tup{q_\mathrm{out},\lift_c})$ \\[1mm]
\indent
$\rho_9: \tup{q_\mathrm{out},\sigma_0,1,\{\odot\}} \to \sigma_0$ \\[1mm]
\indent
$\rho_{10}: \tup{q_\mathrm{next},\sigma_1,1,\{\downarrow\}} \to 
\tup{q_\mathrm{next},\lift_\downarrow;\mathrm{down}_1}$ \\[1mm]
\indent
$\rho_{11}: \tup{q_\mathrm{next},\sigma_1,1,\{c\}} \to 
  \tup{\overline{q}_c,\lift_c}$ \\[1mm]
\indent
$\rho_{12}: \tup{q_\mathrm{next},\sigma_0,1,\{\odot\}} \to e$ \\[1mm]
\noindent
We now construct the deterministic \abb{tl} program $\cP$ 
corresponding to $\cM'_\mathrm{sib}$.
The states of $\cM'_\mathrm{sib}$ after lifting $\downarrow$ are 
$q_\mathrm{out}$ and $q_\mathrm{next}$. Thus, the states of $\cP$ 
that are active when the topmost pebble is $\downarrow$ only need two parameters 
$z_1,z_2$ corresponding to $q_\mathrm{out}$ and $q_\mathrm{next}$. 
Similarly, the states of $\cP$ that are active when the topmost pebble is $c$ 
only need two parameters $z_1,z_2$ corresponding to 
$q_\mathrm{out}$ and~$\overline{q}_c$. The states of $\cP$ 
that are active when the topmost pebble is $\odot$ do not need parameters, 
because $\odot$ is never lifted. 
Program $\cP$ has the states $q_\mathrm{in}$, $q_\mathrm{start}^\odot$,
$q_c^{\downarrow}$, $\overline{q}_c^d$, $q_\mathrm{out}^d$, and $q_\mathrm{next}^d$,
where $c\in\{0,1\}$ and $d\in\{\odot,\downarrow,0,1\}$. 
Note that the state $q_\bot$ is superfluous. The initial state 
$q_\mathrm{in}$ and all states with superscript $\odot$ have 
rank~$0$, and the other states have rank~$2$. 

Program $\cP$ has the following rule corresponding to rule
$r_{c,d}$ of $\cM'_\mathrm{sib}$, with $d\neq\odot$:
\begin{eqnarray*}
\lefteqn{\rho_{c,d}: \tup{\overline{q}_c^d,\phi_{\sigma,1}(x)}(z_1,z_2)\to} \\[1mm]
& & \quad\quad\tup{q_c^{\downarrow},\up(x,y)}
      (\tup{q_\mathrm{out}^d,\stay(x,y)}(z_1,z_2),
       \tup{q_\mathrm{next}^d,\stay(x,y)}(z_1,z_2))
\end{eqnarray*}
and for $d=\odot$ the same rule without the parameters $(z_1,z_2)$. 
The other rules of $\cP$ are \\[1mm]
\indent
$\rho_0: \tup{q_\mathrm{in},\phi_{\sigma_1,0}(x)}\to 
  \tup{q_\mathrm{start}^\odot,\stay(x,y)}$ \\[1mm]
\indent
$\rho_1: \tup{q_\mathrm{start}^\odot,\phi_{\sigma_1,j}(x)}\to 
  \tup{q_\mathrm{start}^\odot,\down_1(x,y)}$ \\[1mm]
\indent
$\rho_2: \tup{q_\mathrm{start}^\odot,\phi_{\sigma_0,1}(x)}\to 
  \tup{\overline{q}_1^\odot,\stay(x,y)}$ \\[1mm]
\indent
$\rho_3: \tup{q_0^{\downarrow},\phi_{\lambda_0,1}(x)}(z_1,z_2)\to
  \tup{\overline{q}_0^\downarrow,\stay(x,y)}(z_1,z_2)$ \\[1mm]
\indent
$\rho_4: \tup{q_c^{\downarrow},\phi_{\lambda_i,1}(x)}(z_1,z_2)\to$ \\[1mm]
\indent
$\quad\quad\quad\tup{\overline{q}_i^c,\stay(x,y)}
      (\tup{q_\mathrm{out}^\downarrow,\stay(x,y)}(z_1,z_2),
       \tup{\overline{q}_c^\downarrow,\stay(x,y)}(z_1,z_2))$ \\[1mm]
\indent
$\rho_5: \tup{q_c^{\downarrow},\phi_{\sigma_1,0}(x)}(z_1,z_2)\to
         r(\tup{q_\mathrm{out}^\downarrow,\stay(x,y)}(z_1,z_2),z_2)$ \\[1mm]
\indent
$\rho_6: \tup{q_\mathrm{out}^\downarrow,\phi_{\sigma_1,0}(x)}(z_1,z_2)\to 
   \sigma_1(z_1)$ \\[1mm]
\indent
$\rho_7: \tup{q_\mathrm{out}^\downarrow,\phi_{\sigma_1,1}(x)}(z_1,z_2)\to z_1$ \\[1mm]
\indent
$\rho_8: \tup{q_\mathrm{out}^c,\phi_{\sigma_1,1}(x)}(z_1,z_2)\to \sigma_1(z_1)$ \\[1mm]
\indent
$\rho_9: \tup{q_\mathrm{out}^\odot,\phi_{\sigma_0,1}(x)}\to \sigma_0$ \\[1mm]
\indent
$\rho_{10}: \tup{q_\mathrm{next}^\downarrow,\phi_{\sigma_1,1}(x)}(z_1,z_2)\to z_2$ \\[1mm]
\indent
$\rho_{11}: \tup{q_\mathrm{next}^c,\phi_{\sigma_1,1}(x)}(z_1,z_2)\to z_2$ \\[1mm]
\indent
$\rho_{12}: \tup{q_\mathrm{next}^\odot,\phi_{\sigma_0,1}(x)}\to e$ \\[1mm]
\noindent
Applying rule $\rho_6$ to the right-hand side of rule $\rho_5$, we obtain the rule
$\rho'_5: \tup{q_c^{\downarrow},\phi_{\sigma_1,0}(x)}(z_1,z_2)\to
         r(\sigma_1(z_1),z_2)$,
which is in fact rule $\rho_5$ of program $\cP_\mathrm{sib}$ of 
Example~\ref{ex:tlsiberie}, if we identify the states $q_c^{\downarrow}$ and $q_c$.  
Rules $\rho_0$ and $\rho_1$ of $\cP$ correspond to rule $\rho_1$ of $\cP_\mathrm{sib}$ 
in an obvious way (with $q_\mathrm{start}^\odot$ and $q_\mathrm{start}$ identified).
Since program $\cP$ is deterministic, and its states generate trees 
(rather than forests), we can also apply rules $\rho_7-\rho_{12}$ 
to the right-hand side of rule $\rho_{c,d}$, and we obtain the rules \\[1mm]
\indent
$\rho'_{c,\downarrow}: \tup{\overline{q}_c^\downarrow,\phi_{\sigma_1,1}(x)}(z_1,z_2)\to 
   \tup{q_c^{\downarrow},\up(x,y)}(z_1,z_2)$ \\[1mm]
\indent
$\rho'_{i,c}: \tup{\overline{q}_i^c,\phi_{\sigma_1,1}(x)}(z_1,z_2)\to 
   \tup{q_i^{\downarrow},\up(x,y)}(\sigma_1(z_1),z_2)$ \\[1mm]
\indent
$\rho'_{c,\odot}: \tup{\overline{q}_c^\odot,\phi_{\sigma_0,1}(x)}\to 
   \tup{q_c^{\downarrow},\up(x,y)}(\sigma_0,e)$ \\[1mm]
\noindent
Applying $\rho'_{1,\odot}$ to the right-hand side of $\rho_2$ we obtain 
$\rho'_2: \tup{q_\mathrm{start}^\odot,\phi_{\sigma_0,1}(x)}\to 
  \tup{q_1^{\downarrow},\up(x,y)}(\sigma_0,e)$, 
which is rule $\rho_2$ of $\cP_\mathrm{sib}$. 
Applying $\rho'_{0,\downarrow}$ to the right-hand side of $\rho_3$ we obtain 
$\rho'_3: \tup{q_0^{\downarrow},\phi_{\lambda_0,1}(x)}(z_1,z_2)\to
  \tup{q_0^{\downarrow},\up(x,y)}(z_1,z_2)$
which is rule $\rho_3$ of $\cP_\mathrm{sib}$.
Finally, applying rules $\rho'_{i,c}$, $\rho_7$, and $\rho'_{c,\downarrow}$ to the selectors 
in the right-hand side of rule $\rho_4$, respectively, we obtain the right-hand side 
$\tup{q_i,\up(x,y)}(\lambda_i(z_1),\tup{q_c,\up(x,y)}(z_1,z_2))$
of rule $\rho_4$ of $\cP_\mathrm{sib}$.
Thus, program $\cP$ is essentially the same as program $\cP_\mathrm{sib}$
of Example~\ref{ex:tlsiberie}. 
\end{example}

Lemmas~\ref{lem:tlipft} and~\ref{lem:ipfttl} together prove that \abb{tl} programs 
have the same expressive power as \abb{i-pft}'s.  
Additionally, they prove that for every \abb{tl} program there is an equivalent local one. 

\begin{theorem}\label{thm:tl}
$\TL=\TLl=\IPFT$ and $\dTL=\dTLl=\IdPFT$. 
Moreover, $\TLr=\TLlr=\IPTT$ and $\dTLr=\dTLlr=\IdPTT$.
\end{theorem}

Since local \abb{tl} programs satisfy restriction (R2) in the definition 
of a ranked \abb{tl} program, the equation $\TL=\TLl$ shows that 
the pattern matching aspect that is involved in the execution 
of selectors, can be viewed as an extended feature. Moreover, even the ``jumps''
in the execution of selectors,
and the arbitrary \abb{mso} head tests in the left-hand sides of rules,
can be viewed as extended features of~$\TLl$.

Note that for \abb{tl}$^\text{\abb{db}}_\ell$ programs the construction in
the proof of Lemma~\ref{lem:tlipft} can easily be simplified to one that
takes polynomial time and that results in an \abb{i-pft} that
does not use \abb{mso} tests. 
That implies that the inverse type inference problem for such programs
is solvable in \mbox{$2$-fold} exponential time, and hence typechecking can be done in
$3$-fold exponential time (cf. Theorem~\ref{thm:tltypecheck}). 

The local ranked \abb{tl} program is an obvious reformulation of 
the ``macro tree-walking transducer'' (2-mtt) of~\cite{ManBerPerSei}.
The inclusion $\TLr\subseteq\TLlr$ is a (slightly stronger) version 
of~\cite[Theorem~5]{ManBerPerSei}. 
Moreover, the local ranked \abb{tl} program is the same as the 
\mbox{``$0$-pebble} macro tree transducer'' 
of~\cite[Section~5.1]{EngMan03} and it is the CFT($S$)-transducer of \cite{EngVog} 
for the storage type $S=$ Tree-walk, both of which generalize the 
macro attributed tree transducer of~\cite{KuhVog,FulVog} which additionally satisfies
a noncircularity condition. 
It follows from Lemma~\ref{lem:nul-decomp} and Theorem~\ref{thm:tl} that 
$\TLlr\subseteq \TT^2$, which was stated as an open problem in~\cite[Section~8]{EngMan03}
(where $\TLlr$ and $\TT$ are denoted 0-PMTT and 0-PTT, respectively). 
In view of Lemma~\ref{lem:ipftnf}, the equality $\TLlr=\IPTT$ is the same as the equality 
CFT($S$) $=$ RT(P($S$)) of~\cite[Theorem~5.24]{EngVog} for $S= \text{Tree-walk}$,
and similarly for the deterministic case.

\section{A TL Program in XSLT}\label{sec:siberie}

In Tables~\ref{tab:input} and~\ref{tab:output} we listed 
a possible input document and 
the resulting output document for 
the \abb{i-ptt} $\cM_\mathrm{sib}$ of Example~\ref{ex:siberie}.
In this section we present in Table~\ref{tab:xslt} an XSLT~1.0 program
with the same structure as the \abb{TL} program $\cP_\mathrm{sib}$ of 
Example~\ref{ex:tlsiberie}. 
In what follows we comment on the XSLT program 
and its relationship to $\cP_\mathrm{sib}$, abbreviated as $\cP$.

\begin{table*}[p]
\begin{scriptsize}
\begin{verbatim}
<xsl:stylesheet xmlns:xsl="http://www.w3.org/1999/XSL/Transform" version="1.0"> 
<xsl:output method="xml"/> 

    <xsl:template match="/">
        <xsl:for-each select="//stop[@final=1]">
            <xsl:call-template name="start" />
        </xsl:for-each>
    </xsl:template>

    <xsl:template name="start">
        <xsl:apply-templates select="parent::stop">            
            <xsl:with-param name="nextstoplarge" select="@large" />       
            <xsl:with-param name="stoplist">
                <xsl:copy>
                    <xsl:copy-of select="attribute::*" />
                </xsl:copy>
            </xsl:with-param>
            <xsl:with-param name="additionalresults">
                <endofresults />
            </xsl:with-param>
        </xsl:apply-templates>
    </xsl:template>
  
    <xsl:template match="stop">
        <xsl:param name="nextstoplarge" />
        <xsl:param name="stoplist" />
        <xsl:param name="additionalresults" />
        <xsl:if test="@initial = 1">
            <result>
                <xsl:copy>
                    <xsl:copy-of select="attribute::*" />
                    <xsl:copy-of select="$stoplist" />
                </xsl:copy>
                <xsl:copy-of select="$additionalresults" />
            </result>
        </xsl:if>
        <xsl:if test="not(@initial = 1)">
            <xsl:variable name="results">
                <xsl:apply-templates select="parent::stop">
                    <xsl:with-param name="nextstoplarge" select="$nextstoplarge" />
                    <xsl:with-param name="stoplist" select="$stoplist" />
                    <xsl:with-param name="additionalresults" select="$additionalresults" />
                </xsl:apply-templates>
            </xsl:variable>
            <xsl:if test="@large = 1 or $nextstoplarge = 1">
                <xsl:apply-templates select="parent::stop">
                    <xsl:with-param name="nextstoplarge" select="@large" />
                    <xsl:with-param name="stoplist">
                        <xsl:copy>
                            <xsl:copy-of select="attribute::*" />
                            <xsl:copy-of select="$stoplist" />
                        </xsl:copy>
                    </xsl:with-param>
                    <xsl:with-param name="additionalresults" select="$results" />
                </xsl:apply-templates>
            </xsl:if>
            <xsl:if test="@large = 0 and $nextstoplarge = 0">
                <xsl:copy-of select="$results" />
            </xsl:if>
        </xsl:if>
    </xsl:template>

</xsl:stylesheet>
\end{verbatim}
\end{scriptsize}
\caption{XSLT Program}\label{tab:xslt}
\end{table*}

The first rule $\rho_1$ of $\cP$ corresponds to the first template
of the XSLT program: this template initalizes the algorithm by matching the root 
of the input document, 
jumping to the leaf by selecting the final stop, 
and invoking named template \texttt{start} on it.

The second rule $\rho_2$ of $\cP$ corresponds to template \texttt{start}: 
it moves up, using the \texttt{apply-templates} instruction 
which selects the parent, 
and thus invokes the third template on that parent, 
which is the only template for nonroot document elements. 
It invokes that template with the appropriate parameters: 
\texttt{nextstoplarge} is 1 
because $\mathtt{large = 1}$ for the final stop, 
\texttt{stoplist} is a list containing only the final stop, and 
\texttt{additionalresults} is the single element \texttt{<endofresults />}. 

The remaining rules of $\cP$ correspond to the third template, 
which is applied to all nonfinal stops. 
That template takes a partial stop list \texttt{stoplist} 
(from the current stop to the final stop) 
and generates all allowed ways to complete that stop list 
using the stops between the current one and the initial one. 
Nested below the deepest element of the output, 
it includes the result tree fragment passed in \texttt{additionalresults}. 
The third template has three parameters:
\begin{enumerate}
\item[] \texttt{nextstoplarge}: a boolean indicating whether or not 
the ``next'' stop 
(i.e., the stop at the front of \texttt{stoplist}) is a large stop; 
it corresponds to states $q_1$ and $q_0$ in $\cP$, respectively,
\item[] \texttt{stoplist}: a partial list of stops 
(taken from the current stop to the final stop)
for which this template will recursively generate 
all (allowed) ways in which it can be completed;
it corresponds to parameter $z_1$ in $\cP$,
\item[] \texttt{additionalresults}: results that are to be appended to 
the results that this template generates;
it corresponds to parameter $z_2$ in $\cP$,
\end{enumerate}
where both \texttt{stoplist} and \texttt{additionalresults} are of type `result tree fragment'.

Corresponding to rule $\rho_5$ of $\cP$, the third template, 
when invoked on the initial stop (for which $\mathtt{initial = 1}$),
has computed a complete stop list (after adding this stop) and
outputs it: it copies the initial stop  
and nests the remainder of the stop list (i.e., the value of 
its parameter \texttt{stoplist}) in it; 
it also includes the additional results 
(i.e., the value of parameter \texttt{additionalresults}).

Corresponding to rules $\rho_3$ and $\rho_4$ of $\cP$, the third template, 
when invoked on an intermediate stop
(for which $\mathtt{not(initial = 1)}$), 
has not yet computed a complete stop list, and now
calculates all allowed ways to complete it. 
Intuitively, it computes two result sets: 
one that \emph{does not} add the current stop, and one that \emph{does}.
They are combined by passing the first result set as 
``additional results'' to the calculation of the second one. 
Thus, the third template starts by computing the first result set,
and, to abbreviate the remaining code, it assigns its value
to a variable called \texttt{results}. 
In rules $\rho_3$ and $\rho_4$ of $\cP$ this result set corresponds to 
the selector $\tup{q_c,\up(x,y)}(z_1,z_2)$, where $c=0$ in $\rho_3$. 
In the case that \texttt{large = 0} and \texttt{nextstoplarge = 0},
we are not allowed to stop here 
because that would create two consecutive small stops. 
Thus the template only outputs the results 
that it just stored in the variable 
(corresponding to rule $\rho_3$ of $\cP$).
In the case that \texttt{large = 1} or \texttt{nextstoplarge = 1},
the template calculates all possible ways to complete the stop list 
that contain this stop, and includes as additional results 
those that are stored in the variable
(corresponding to rule $\rho_4$ of $\cP$).

\section{Data Complexity}

In this section we show that 
the transduction of a deterministic \abb{ptt} $\cM$ 
can be realized in (1-fold) exponential time, 
in the sense that there is an exponential time algorithm that, 
for every given input tree $t$, computes a regular tree grammar $G$
that generates the language $\{\tau_\cM(t)\}$. 
If $t$ is in the domain of $\cM$, then~$G$ can be viewed as a 
DAG (directed acyclic graph) that defines the output tree $\tau_\cM(t)$, in the usual sense.
Thus, producing the actual output tree would take 2-fold exponential time. 
If $t$ is not in the domain of $\cM$, then $G$ generates the empty tree language 
(which can be decided in time linear in the size of $G$). 

\begin{theorem}\label{thm:expocom}
For every deterministic \abb{ptt} $\cM$ 
there is an exponential time algorithm that, for given input tree $t$,
computes a regular tree grammar $G$ such that $L(G)=\{s\mid (t,s)\in\tau_\cM\}$. 
\end{theorem}

\begin{proof}
Let $\cM = (\Sigma, \Delta, Q, \{q_0\},C, C_\mathrm{v}, C_\mathrm{i}, R,k)$ 
be a deterministic \abb{v$_k$i-ptt}. 
For an input tree $t\in T_\Sigma$ in the domain of $\cM$, 
let us consider the computation $\tup{q_0,\rt_t,\epsilon}\Rightarrow^*_{t,\cM} s$,
where $s=\tau_\cM(t)$, and let $\tup{q,u,\pi}$
be a configuration of~$\cM$ that occurs in that computation. 
We claim that the length of $\pi$ is at most $N=|Q|\cdot(|C|+1)^{k+1}\cdot n^{k+2}$, 
where $n$ is the size of $t$. 

To prove this claim we define, 
as an auxiliary tool,  
the nondeterministic \abb{v$_k$i-pta} $\cA$ 
that is obtained from~$\cM$ by changing every output rule 
$\tup{q,\sigma,j,b}\to \delta(\tup{q_1,\stay},\dots,\tup{q_m,\stay})$ of $\cM$ 
into the rules $\tup{q,\sigma,j,b}\to \tup{q_i,\stay}$ for all $i\in[1,m]$.
Intuitively, whenever $\cM$ branches, $\cA$ nondeterministically 
follows one of those branches. Thus, all computations of $\cA$ 
that start with $\tup{q_0,\rt_t,\epsilon}$ are finite. 
Obviously, $\tup{q,u,\pi}$ occurs in such a computation of~$\cA$. 
Let $\pi=(v_1,c_1)\cdots (v_m,c_m)$ and suppose that $m > N$.
For every $\ell\in[1,m]$ we define $\pi_\ell=(v_1,c_1)\cdots (v_\ell,c_\ell)$. 
Then there exist configurations $\tup{q_\ell,u_\ell,\pi_\ell}$, $\ell\in[1,m]$, such that 
$\tup{q_0,\rt_t,\epsilon} \Rightarrow^*_{t,\cA} \tup{q_1,u_1,\pi_1}$ and 
$\tup{q_\ell,u_\ell,\pi_\ell} \Rightarrow^*_{t,\cA} \tup{q_{\ell+1},u_{\ell+1},\pi_{\ell+1}}$
for every $\ell\in[1,m-1]$, and such that, moreover, 
every configuration occurring in the computation
$\tup{q_\ell,u_\ell,\pi_\ell} \Rightarrow^*_{t,\cA} \tup{q_{\ell+1},u_{\ell+1},\pi_{\ell+1}}$
has a pebble stack with prefix~$\pi_\ell$.  
Due to the choice of $m$, there exist $i,j\in[1,m]$ with $i<j$ such that 
$q_i=q_j$, $u_i=u_j$, $(v_i,c_i)=(v_j,c_j)$, and 
for every $v\in N(t)$ and \mbox{$c\in C_\mathrm{v}$}: 
$(v,c)$~occurs in $\pi_i$ if and only $(v,c)$ occurs in $\pi_j$. 
This implies that the computation 
$\tup{q_i,u_i,\pi_i} \Rightarrow^*_{t,\cA} \tup{q_j,u_j,\pi_j}$ can be repeated arbitrarily many times,
leading to an infinite computation of $\cA$, which is a contradiction and proves the claim.  

We now construct the regular tree grammar $G$. Its nonterminals are 
the configurations $\tup{q,u,\pi}$ of $\cM$ on $t$ such that $|\pi|\leq N$.
Since $N$ is polynomial in $n$, the number of nonterminals of $G$ is exponential in $n$. 
The initial nonterminal of $G$ is $\tup{q_0,\rt_t,\epsilon}$. 
If $\tup{q,u,\pi} \Rightarrow^*_{t,\cM} \tup{q',u',\pi'} \Rightarrow_{t,\cM}
\delta(\tup{q_1,u',\pi'},\dots,\tup{q_m,u',\pi'})$, then 
$\tup{q,u,\pi} \to \delta(\tup{q_1,u',\pi'},\dots,\tup{q_m,u',\pi'})$
is a rule of $G$. To decide whether 
$\tup{q',u',\pi'} \Rightarrow_{t,\cM}
\delta(\tup{q_1,u',\pi'},\dots,\tup{q_m,u',\pi'})$
it suffices to inspect the output rules of $\cM$. 
To decide whether $\tup{q,u,\pi} \Rightarrow^*_{t,\cM} \tup{q',u',\pi'}$ 
we construct from~$\cM$ and $t$
an ordinary pushdown automaton $\cP$ that simulates the non-output behaviour of $\cM$ on $t$, 
as in the query evaluation paragraph at the end of Section~\ref{sec:xpath}.
Since, as opposed to that paragraph, $\cM$ also has visible pebbles,
$\cP$~should keep track of those pebbles in its finite state.
Let $\Gamma$ be the set of all mappings $\gamma: C_\mathrm{v}\to N(t)\cup\{\bot\}$ 
such that $\#(\{c\in C_\mathrm{v}\mid \gamma(c)\neq \bot\})\leq k$.
During $\cP$'s computation, the mapping $\gamma$ in its finite state indicates for every visible pebble 
whether it occurs in the current stack and, if so, on which node it is dropped.
Thus, we define $\cP$ to have state set $Q\times N(t)\times \Gamma$ and pushdown alphabet $N(t)\times C$. 
A configuration $\tup{q,u,\pi}$ of $\cM$ is simulated by the configuration $\cP(\tup{q,u,\pi})=\tup{p,\pi}$
of $\cP$ such that $p=(q,u,\gamma)$ where, for every $c\in C_\mathrm{v}$, 
if $\gamma(c)\in N(t)$ then $(\gamma(c),c)$ occurs in $\pi$,
and if $\gamma(c)=\bot$ then $c$ does not occur in $\pi$. 
The transitions of the automaton $\cP$ are defined in such a way that $\cP$ (with the empty string as input)
has the same computation steps as $\cM$ (without its output rules), i.e., such that 
$\tup{q,u,\pi} \Rightarrow_{t,\cM} \tup{q',u',\pi'}$ if and only if 
$\cP(\tup{q,u,\pi}) \Rightarrow_\cP \cP(\tup{q',u',\pi'})$, where $\Rightarrow_\cP$ is the 
computation step relation of $\cP$. 
For instance, let $\cP$ be in state $(q,u,\gamma)$ and let the top element of its stack be $(v,c)$.
Let $u$ have label $\sigma$ and child number $j$, and let
$b$ consist of all $c'\in C_\mathrm{v}$ with $\gamma(c')= u$ plus $c$ if $v=u$.
If $\tup{q,\sigma,j,b}\to\tup{q',\drop_d}$ is a rule of $\cM$ such that $d\in C_\mathrm{v}$, 
$\gamma(d)=\bot$, and $\#(\{c'\in C_\mathrm{v}\mid \gamma(c')\neq \bot\}) < k$, 
then $\cP$ pushes $(u,d)$ on its stack and goes into state $(q',u,\gamma')$ 
where $\gamma'(d)=u$ and $\gamma'(c')=\gamma(c')$ for all $c'\neq d$.
If $\tup{q,\sigma,j,b}\to\tup{q',\lift_c}$ is a rule of $\cM$ 
such that $c\in C_\mathrm{i}$ and $v=u$,
then $\cP$ pops $(v,c)$ from its stack and goes into state $(q',u,\gamma)$. 
The transitions of $\cP$ are defined similarly for the other non-output rules of $\cM$. 
It should be clear that $\cP$ can be constructed in time polynomial in $n$. 
Since it can be decided in polynomial time for configurations $\tup{p,\pi}$ and $\tup{p',\pi'}$ 
of $\cP$ whether $\tup{p,\pi} \Rightarrow^*_\cP \tup{p',\pi'}$, it can be decided 
whether $\tup{q,u,\pi} \Rightarrow^*_{t,\cM} \tup{q',u',\pi'}$ in polynomial time.
Hence the total time to construct $G$ is exponential. 
\end{proof}

Note that the first part of the above proof also shows that for every deterministic \abb{ptt}
the height of the output tree is exponential in the size of the input tree. 

A natural question is whether Theorem~\ref{thm:expocom} also holds for forest transducers, 
i.e., for deterministic \abb{pft}'s. That is indeed the case (as the reader can easily verify), 
except that $G$ is not a regular forest grammar, but a forest generating context-free grammar. 
To be precise, $G$ is a context-free grammar of which every rule is of the form 
$X_0\to \delta(X_1)$ or $X_0\to X_1X_2$ or $X\to\epsilon$ where $\delta$~is a symbol
from an unranked alphabet. If $L(G)=\{f\}$, then $G$ 
can still be viewed as a DAG that defines the forest $f$. 
Thus, in this sense, by Theorem~\ref{thm:tl}, deterministic \abb{tl} programs 
can be executed in exponential time,
in accordance with the result of~\cite{JanKorBus} that XSLT~1.0 programs 
can be executed in exponential time.

Another natural question is whether there exist interesting subclasses of \abb{ptt}'s
that can be realized in polynomial time. Here we discuss one such subclass. 
We define a \abb{ptt} to be \emph{bounded} if there exists $m\in \nat$ such that output rules 
can only be applied when the pebble stack contains at most $m$ pebbles. 
Intuitively it means that the infinitely many invisible pebbles are mainly used to check 
\abb{mso} properties of the observable configuration. 
Formally it can either be required as a dynamic property of the (successful) computations 
of the \abb{ptt} or be incorporated statically in the semantics of the \abb{ptt}. 
We now show that bounded \abb{ptt}'s can be realized in polynomial time, 
even in the nondeterministic case.

\begin{theorem}\label{thm:polycom}
For every bounded \abb{ptt} $\cM$ 
there is a polynomial time algorithm that, for given input tree $t$,
computes a regular tree grammar~$G$ such that $L(G)=\{s\mid (t,s)\in\tau_\cM\}$. 
\end{theorem}

\begin{proof}
The construction of $G$ is exactly the same as in the proof of Theorem~\ref{thm:expocom},
except that its nonterminals are now the configurations 
$\tup{q,u,\pi}$ of $\cM$ on $t$ such that $|\pi|\leq m$.\footnote{Additionally, 
$G$ has an initial nonterminal $S$ with rules $S\to \tup{q_0,\rt_t,\epsilon}$
for every initial state $q_0$ of $\cM$. 
} 
The number of nonterminals of $G$ is therefore polynomial in the size of $t$,
and since the pushdown automaton $\cP$ can also be constructed (and tested) in polynomial time,
the total time to construct $G$ is polynomial.  
\end{proof}

Again, the same result holds for \abb{pft}'s, 
taking $G$ to be a forest generating context-free grammar. 
Note that for a nondeterministic \abb{pft} $\cM$ and an input tree $t$, 
the set $\{s\mid (t,s)\in\tau_\cM\}$ is not necessarily a regular forest language. 

Also, the same result holds for bounded \abb{ptt}'s 
that use \abb{mso} tests on the observable configuration. 
That is not immediate, because the construction in the proof of Theorem~\ref{thm:mso}
does not preserve boundedness, due to the use of beads. 
However, it is easy to adapt the construction of the pushdown automaton $\cP$
in the proof of Theorem~\ref{thm:expocom} to incorporate the \abb{mso} tests of 
the \abb{v$_k$i-ptt} $\cM$. In fact, the observable configuration tree $\obs(t,\pi)$
can be constructed from $t$, from the mapping $\gamma$ in the state of $\cP$, and 
from the top element of its stack, and then $\obs(t,\pi)$ can be tested in linear time
using a deterministic bottom-up finite-state tree automaton. 
An example of bounded \abb{ptt}'s (with \abb{mso} tests) are 
the pattern matching \abb{ptt}'s of Section~\ref{sec:pattern}. In that section, 
every \abb{ptt} that matches an $n$-ary pattern is bounded, with bound $n$ 
or even $n-1$. Hence, pattern matching \abb{ptt}'s can be evaluated in polynomial time. 
And the same is true for pattern matching \abb{pft}'s, see Section~\ref{sec:pft}.

\section{Variations of Decomposition}\label{sec:variations}

In this section we present a number of results the proofs of which are based on 
variations of the decomposition techniques used in Section~\ref{sec:decomp}.
In the first part of the section we consider deterministic \abb{ptt}'s,
and in the second part we consider \abb{ptt}'s with strong (visible) pebbles. 

\smallpar{Deterministic PTT's}
As observed at the end of Section~\ref{sec:decomp} it is open whether 
$\IdPTT \subseteq \dTT \circ \dTT$. We first show that a subclass of $\IdPTT$
is included in $\dTT \circ \dTT$ and then we show that $\IdPTT \subseteq \dTT^3$.
Hence, every deterministic \abb{ptt} can be decomposed into deterministic \abb{tt}'s. 

Recall that $\dTTmso$ denotes the class of transductions that are realized 
by deterministic \abb{tt}'s with \abb{mso} head tests. By Lemma~\ref{lem:sites}
it is a subclass of $\IdPTT$. We will show that such transducers can be decomposed 
into two deterministic \abb{tt}'s of which the first never moves up. 
To do this we need a lemma with an alternative proof of the inclusion 
$\dTTmso \subseteq \IdPTT$, showing that the resulting \abb{i-ptt}
uses its pebbles in a restricted way. The \abb{i-ptt} that is constructed in the proof of 
Lemma~\ref{lem:sites} does not satisfy that restriction. 

For the definition of normal form of an \abb{i-ptt} 
see the paragraphs before Lemma~\ref{lem:ipftnf}.
We now define an \abb{i-ptt} (or \abb{i-pta}) to be \emph{root-oriented} if 
it satisfies requirements (1)$-$(3) of the normal form,
and all non-initial non-output rules have a right-hand side 
of one of the following five forms:
$\tup{q',\down_i;\drop_c}$, $\tup{q',\lift_c;\up}$, 
$\tup{q',\lift_c;\drop_d}$, or $\tup{q',\stay}$, 
where $q'\in Q\setminus Q_0$, $i\in\nat$ and \mbox{$c,d\in C$}.
Thus, except in an initial configuration, 
every pebble stack is of the form $(u_1,c_1)\cdots(u_n,c_n)$ 
where $u_1,\dots,u_n$ is the path from the root to the current node.
The \abb{i-pta} in the proof of Lemma~\ref{lem:regular} is root-oriented. 

The next lemma follows from~\cite[Theorem~8.12]{thebook}, 
but we provide its proof for completeness sake. 
Let $\rIdPTT$ denote the class of transductions realized by 
root-oriented deterministic \abb{i-ptt}'s.\footnote{In~\cite[Chapter~8]{thebook} 
root-oriented \abb{i-ptt}'s are called tree-walking pushdown transducers,
and $\rIdPTT$ is denoted $\family{P-DTWT}$.
They are the \abb{RT(P(TR))}-transducers of~\cite{EngVog}, also called 
indexed tree transducers. 
}

\begin{lemma}\label{lem:in-ridptt}
$\dTTmso \subseteq \rIdPTT$.
\end{lemma}

\begin{proof}
Let $\cM$ be a deterministic \abb{tt} that uses a regular site $T$ as \abb{mso} head test.
For simplicity we will assume that $\cM$ tests $T$ in every rule. 
Let $\cA=(\Sigma\times\{0,1\},P,F,\delta)$ be a deterministic bottom-up finite-state tree automaton 
that recognizes $\tmark(T)$. As usual we identify the symbols $(\sigma,0)$ and $\sigma$.
For every tree $t\in T_\Sigma$ and every node $u\in N(t)$, we define the set $\suc_t(u)$
of \emph{successful states} of $\cA$ at $u$ to consist of all states $p\in P$ such that 
$\cA$ recognizes~$t$ when started at $u$ in state $p$. To be precise, $\suc_t(\rt_t)=F$ and 
if $u$ has label $\sigma\in\Sigma^{(m)}$ and $i\in[1,m]$, then $\suc_t(ui)$ is 
the set of all states $p\in P$
such that $\delta(\sigma,p_1,\dots,p_{i-1},p,p_{i+1},\dots,p_m)\in \suc_t(u)$,
where $p_j$ is the state in which $\cA$~arrives at $uj$ for every $j\in[1,m]\setminus\{i\}$.

We construct a root-oriented deterministic \abb{i-ptt} $\cM'$ that stepwise simulates~$\cM$ 
and simultaneously keeps track of $\suc_t(v)$ for all nodes $v$ 
on the path from the root to the current node $u$,
by storing that information in its pebble colours. 
It uses the \abb{i-pta} $\cA'$ of Lemma~\ref{lem:regular} 
(with $\cA$ restricted to $\Sigma\times\{0\}$)
as a subroutine to compute the states in which $\cA$ arrives at the children of $u$. 
Using these states and $\suc_t(u)$, it can easily test whether $(t,u)\in T$. 
Morover, when moving down to a child $ui$ of $u$ it can use this information to 
compute $\suc_t(ui)$.

Formally, in addition to the pebble colours $p_1\cdots p_m$ of $\cA'$, 
the transducer~$\cM'$ uses pebble colours $(S,p_1\cdots p_m)$ where $S\subseteq P$. 
As states it uses (apart from its initial state) the states of $\cM$
and states of the form $(\tilde{q},q)$ where $\tilde{q}$ is a state of~$\cM$ 
and $q$ a state of $\cA'$;
in fact, $q$ is either the main state $q_\circ$ of $\cA'$ 
or it is $\bar{q}_p$ for some $p\in P$. 
Initially, $\cM'$ drops pebble $(F,\epsilon)$ on the root and  
goes into state $(\tilde{q}_0,q_\circ)$ where $\tilde{q}_0$ is the initial state of $\cM$. 
This incorporates rule $\rho_1$ of $\cA'$. 
The other rules of $\cM'$ that correspond to $\cA'$ are as follows. 
First, the rule $\rho_2$ of $\cA'$ 
together with the corresponding rule for pebble colour $(S,p_1\cdots p_m)$,
both for $m<\rank(\sigma)$:
\[
\begin{array}{lll}
\tup{(\tilde{q},q_\circ),\sigma,j,\{p_1\cdots p_m\}} & \to & 
          \tup{(\tilde{q},q_\circ),\down_{m+1};\drop_\epsilon} \\[1mm]
\tup{(\tilde{q},q_\circ),\sigma,j,\{(S,p_1\cdots p_m)\}} & \to & 
          \tup{(\tilde{q},q_\circ),\down_{m+1};\drop_\epsilon}. 
\end{array}
\]
Second, the rule $\rho_3$ of $\cA'$,
for $m=\rank(\sigma)$ and $p=\delta(\sigma,p_1,\dots,p_m)$: 
\[
\begin{array}{llll}
\tup{(\tilde{q},q_\circ),\sigma,j,\{p_1\cdots p_m\}} & \to & 
       \tup{(\tilde{q},\bar{q}_p),\lift_{p_1\cdots p_m};\up}
       & \text{if }  j\neq 0.
\end{array}
\]
Third, the rule $r_6$ of $\cA'$ 
together with the corresponding rule for pebble colour $(S,p_1\cdots p_m)$,
both for $m<\rank(\sigma)$:
\[
\begin{array}{lll}
\tup{(\tilde{q},\bar{q}_p),\sigma,j,\{p_1\cdots p_m\}} & \to 
              & \tup{(\tilde{q},q_\circ),\lift_{p_1\cdots p_m};\drop_{p_1\cdots p_mp}} \\[1mm]
\tup{(\tilde{q},\bar{q}_p),\sigma,j,\{(S,p_1\cdots p_m)\}} & \to 
              & \tup{(\tilde{q},q_\circ),\lift_{(S,p_1\cdots p_m)};\drop_{(S,p_1\cdots p_mp)}}. 
\end{array}
\]
The subroutine $\cA'$ is always called at a node $u$ where $\cM'$ observes a pebble 
of the form $(S,\epsilon)$, and when $\cA'$ is finished $\cM'$ is back at the same node $u$
and observes the pebble $(S,p_1\cdots p_m)$ where $p_1,\dots,p_m$ are the states at which 
$\cA$~arrives at the children of $u$. 

Finally we consider the simulation of a step of $\cM$, which either occurs when 
the subroutine $\cA'$ is finished (instead of its rules $\rho_4$ and $\rho_5$), 
or just after the simulation of another step of $\cM$, in which it does not move down. 
Suppose that $\cM$ has a rule $\tup{\tilde{q},\sigma,j,T}\to \zeta$
and that $\delta((\sigma,1),p_1,\dots,p_m)\in S$, or 
suppose that it has a rule $\tup{\tilde{q},\sigma,j,\neg T}\to \zeta$
and $\delta((\sigma,1),p_1,\dots,p_m)\notin S$.
Then $\cM'$ has the following two rules, for $m=\rank(\sigma)$:
\[
\begin{array}{lll}
\tup{(\tilde{q},q_\circ),\sigma,j,\{(S,p_1\cdots p_m)\}} & \to & \zeta'  \\[1mm]
\tup{\tilde{q},\sigma,j,\{(S,p_1\cdots p_m)\}} & \to & \zeta'     
\end{array}
\]
such that 
\begin{enumerate}
\item[(1)] if $\zeta=\tup{\tilde{q}',\up}$, then 
$\zeta'=\tup{\tilde{q}',\lift_{(S,p_1\cdots p_m)};\up}$,
\item[(2)] if $\zeta=\tup{\tilde{q}',\down_i}$, then 
$\zeta'=\tup{(\tilde{q}',q_\circ),\down_i;\drop_{(S',\epsilon)}}$ \\
where $S'=\{p\in P \mid \delta(\sigma,p_1,\dots,p_{i-1},p,p_{i+1},\dots,p_m)\in S\}$, and
\item[(3)] $\zeta'=\zeta$ otherwise. 
\end{enumerate}
This ends the formal description of $\cM'$. 
In general, $\cM$ uses regular sites $T_1,\dots,T_n$ as \abb{mso} head tests,
and correspondingly $\cM'$ has pebble colours of the form $(S_1,\dots,S_n,p_1\cdots p_m)$ 
where $S_i$ is a set of states of an automaton $\cA_i$ recognizing $\tmark(T_i)$.  
\end{proof}

Let $\dTT\!_\downarrow$ denote the class of transductions realized by 
deterministic \abb{tt}'s that do not use the $\up$-instruction. 
Such transducers are equivalent to classical deterministic top-down tree transducers.
The next lemma is shown in~\cite[Theorem~8.15]{thebook} 
but we provide its proof again, 
to show the connection to Lemma~\ref{lem:nul-decomp}. 

\begin{lemma}\label{lem:ridptt-in}
$\rIdPTT \subseteq \dTT\!_\downarrow \circ \dTT$. 
\end{lemma}

\begin{proof}
Let $\cM$ be a root-oriented deterministic \abb{i-ptt}. 
Looking at the proof of Lemma~\ref{lem:nul-decomp}, it should be clear that,
for every input tree $t$, 
the simulating transducer $\cM'$ only visits those nodes of $t'$ that correspond to 
a sequence of instructions of $\cM$ that starts with a drop-instruction 
and then consists alternatingly of a down-instruction and a drop-instruction. 
Consequently, the ``preprocessor'' $\cN$ can be adapted so as to generate just 
that part of $t'$. The new $\cN$ does not need the states $f_i$ any more,
but just has the initial state $g$ and the state $f$. Its rules are  
\[
\begin{array}{lll}
\tup{\g,\sigma,j} & \to & \sigma'(\bot^m, \tup{f,\stay}^\gamma) \\[1mm]
\langle f,\sigma,j\rangle & \to &
\sigma'_{0,j}(\tup{\g,\down_1},\dots,\tup{\g,\down_m}, \bot^\gamma, \bot)
\end{array}
\]
where $m$ is the rank of $\sigma$ 
and $\bot^n$ abbreviates the sequence $\bot,\dots,\bot$ of length~$n$. 
Note that the child number $j$ is irrelevant. 
With this new, total deterministic preprocessor $\cN$ the proof of 
Lemma~\ref{lem:nul-decomp} is still valid. 
\end{proof}

The following corollary was shown in~\cite[Theorem~8.22]{thebook}, 
but we repeat it here for completeness sake, cf. Corollary~\ref{cor:mtiptt}. 

\begin{corollary}\label{cor:ridpttmt}
$\rIdPTT = \dTT\!_\downarrow \circ \dTT = \family{dMT}_\text{\abb{OI}}$.
\end{corollary}

\begin{proof}
The inclusion $\dTT\!_\downarrow \circ \dTT \subseteq \family{dMT}_\text{\abb{OI}}$
follows from the inclusions $\dTT\subseteq \family{dMT}_\text{\abb{OI}}$, 
shown in~\cite[Theorem~35 for $n=0$]{EngMan03}, and
$\dTT\!_\downarrow \circ \family{dMT}_\text{\abb{OI}}\subseteq \family{dMT}_\text{\abb{OI}}$,
shown in~\cite[Theorem~7.6(3)]{EngVog85}. By Lemma~\ref{lem:ridptt-in} it now suffices 
to show that $\family{dMT}_\text{\abb{OI}} \subseteq \rIdPTT$ (which strengthens 
the second inclusion of Corollary~\ref{cor:mtiptt}). 
There are two ways of proving this, which are essentially the same. 
First, the proof of Lemma~\ref{lem:tlipft}
can be adapted in a straightforward way.\footnote{The transducer $\cM$ uses 
an additional pebble $\odot$, which it drops initially on the root 
and whenever it moves down (instead of calling subroutine $S_{q',\psi}$). 
When necessary it replaces $\odot$ by a pebble $([s_1],\dots,[s_m])$. 
When subroutine $S$ is in state $[z_i]$ for some parameter $z_i$,
it lifts $\odot$ and moves up where it finds a pebble $([s_1],\dots,[s_m])$.
}  
Second, the equality $\rIdPTT = \family{dMT}_\text{\abb{OI}}$ is shown 
for total functions in~\cite[Theorem~5.16]{EngVog}. 
By~\cite[Theorem~6.18]{EngVog85}, 
every transduction $\tau\in \family{dMT}_\text{\abb{OI}}$ is of the form 
$\tau_1\circ\tau_2$ where $\tau_1$ is the identity on a regular tree language $R$
and $\tau_2\in\family{dMT}_\text{\abb{OI}}$ is a total function. 
Thus, $\tau_2$ is in $\rIdPTT$. This implies that 
$\tau_1 \circ\tau_2$ is in $\rIdPTT$: the \abb{i-ptt} just starts by checking that 
the input tree is in $R$, using the root-oriented \abb{i-ptt} $\cA'$ 
in the proof of Lemma~\ref{lem:regular} as a subroutine. 
\end{proof}

We now turn to the decomposition of an arbitrary deterministic \abb{i-ptt}
into deterministic \abb{tt}'s.  

\begin{lemma}\label{lem:decompidptt}
$\IdPTT \subseteq \tdTTmso\circ \dTT$.
\end{lemma}

\begin{proof}
Let $\cM = (\Sigma,\Delta,Q,\{q_0\},C,\nothing,C_\mathrm{i},R,0)$ 
be a deterministic \abb{i-ptt} with $C=C_\mathrm{i}$.
We may assume that there is a mapping $\chi: C\to Q$ such that $\chi(c)=q'$ 
for every rule $\tup{q,\sigma,j,b}\to\tup{q',\drop_c}$ of $\cM$. 
If not, then we change $C$ into $C\times Q$ and we change every rule 
$\tup{q,\sigma,j,b}\to\tup{q',\drop_c}$ into $\tup{q,\sigma,j,b}\to\tup{q',\drop_{(c,q')}}$
and every rule $\tup{q,\sigma,j,\{c\}}\to\tup{q',\lift_c}$ into all the rules
$\tup{q,\sigma,j,\{(c,p)\}}\to\tup{q',\lift_{(c,p)}}$. 
Moreover, we may assume that $C=[1,\gamma]$ for some $\gamma\in\nat$.

As in the proof of Lemma~\ref{lem:ridptt-in} we consider the proof of 
Lemma~\ref{lem:nul-decomp} and adapt the preprocessor $\cN$ to the needs of $\cM$. 
Every copy of the input tree that is generated by $\cN$ corresponds 
to a unique potential pebble stack $\pi$ of $\cM$. The simulating deterministic \abb{tt} $\cM'$
walks on that copy whenever $\cM$ has pebble stack~$\pi$.
The idea is now to construct a variation $\cN'$ of $\cN$ that only generates 
those copies of the input tree $t$ that correspond to reachable pebble stacks. A~pebble stack  
$\pi$ is \emph{reachable} (on $t$) if $\cM$ has a reachable output form 
that contains a configuration $\tup{q,v,\pi}$ for some $q\in Q$ and $v\in N(t)$.
For a given~$t$ in the domain of $\cM$, the number of reachable stacks
is finite because $\cM$ is deterministic and thus has a unique computation on~$t$.
Consequently $\cN'$ can preprocess $t$ deterministically. 
Then we can define a total deterministic preprocessor $\cN''$ that 
starts by performing an \abb{mso} head test whether or not the input tree 
is in the domain of $\cM$ (which is regular by Corollary~\ref{cor:domptt}). 
If it is, then $\cN''$ calls $\cN'$, and if it is not, 
then $\cN''$ outputs~$\bot$ and halts.

As an auxiliary tool, we define (as in the proof of Theorem~\ref{thm:expocom}) 
the nondeterministic \abb{i-pta} $\cA$ 
that is obtained from~$\cM$ by changing every output rule 
$\tup{q,\sigma,j,b}\to \delta(\tup{q_1,\stay},\dots,\tup{q_m,\stay})$ of $\cM$ 
into the rules $\tup{q,\sigma,j,b}\to \tup{q_i,\stay}$ for $i\in[1,m]$.
Intuitively, whenever $\cM$ branches, $\cA$ nondeterministically 
follows one of those branches. 
Obviously a nonempty pebble stack $\pi$ with top element $(u,c)$ 
is reachable if and only if $\tup{\chi(c),u,\pi}$ is a 
reachable configuration of $\cA$ (see footnote~\ref{foot:reachcon}). 
Note that $\tup{\chi(c),u,\pi}$ is the configuration of $\cM$ 
just after dropping pebble~$c$ at node $u$. 

For pebble colour $c$, we consider the site $T_c$ consisting of all pairs $(t,u)$
such that one-pebble stack $(u,c)$ is reachable, i.e., such that
$\cA$ has a computation starting in the initial configuration and ending 
in the configuration $\tup{\chi(c),u,(u,c)}$. 
It is not difficult to see that $T_c$ is a regular site. 
In fact, $\tmark(T_c)$ is the domain of an \abb{i-pta} $\cB$ with stack tests that 
simulates $\cA$; whenever it arrives at the marked node $u$ in state $\chi(c)$
and it observes pebble $c$, then it may lift the pebble, 
check that its stack is empty, and accept. Stack tests are allowed by 
Lemma~\ref{lem:stacktests}, and the domain of $\cB$ is regular by Corollary~\ref{cor:domptt}.

We now turn to reachable pebble stacks with more than one pebble, i.e., of the form 
$\pi(u,c)(v,d)$. Assuming that we already know that $\pi(u,c)$ is reachable, 
we can find out whether $\pi(u,c)(v,d)$ is reachable through a regular trip, as follows. 
For pebble colours $c$ and $d$, we consider the trip $T_{c,d}$ consisting of all triples 
$(t,u,v)$ such that $\cA$ has a computation on $t$ 
starting in configuration $\tup{\chi(c),u,(u,c)}$
and ending in configuration $\tup{\chi(d),v,(u,c)(v,d)}$; 
moreover, in every intermediate configuration 
the bottom element of the pebble stack must be $(u,c)$. 
The trip $T_{c,d}$ is regular because $\tmark(T)$ is the domain of an \abb{i-pta}~$\cB'$ 
with stack tests that first walks to the marked node $u$. Then $\cB'$ simulates $\cA$, 
starting in state $\chi(c)$, interpreting the mark of $u$ as pebble $c$ (which cannot be lifted). 
Similar to $\cB$ above, whenever $\cB'$ arrives at the marked node $v$ in state $\chi(d)$
and it observes pebble $d$, then it may lift the pebble, 
check that the stack is empty, and accept. 
Obviously, if $\pi(u,c)$ is reachable, then 
$\pi(u,c)(v,d)$ is reachable if and only if $(t,u,v)\in T_{c,d}$. 
Let $\cB_{c,d}$ be a (nondeterministic) \abb{ta} with \abb{mso} head tests 
that computes $T_{c,d}$, as in Proposition~\ref{prop:trips}. 

The new preprocessor $\cN'$ is a deterministic \abb{tt} with \abb{mso} head tests
that works in the same way as $\cN$ but only creates the copies of the input tree $t$ that
correspond to reachable pebble stacks. Initially it uses the test $T_c$ at node $u$ to decide 
whether it has to create a copy of $t$ corresponding to pebble stack $(u,c)$.
If the test is positive, then, just as $\cN$, it creates a copy of $t$ 
by walking from $u$ to every other node $v$ of $t$, copying $v$ to the output.
Now recall that $\cN$ walks from $u$ to $v$ 
along the shortest (undirected) path in $t$. Thus, by Proposition~\ref{prop:trips}, 
$\cN'$ can simulate the behaviour of \abb{ta} $\cB_{c,d}$ from $u$ to $v$,
for every pebble colour $d$ (using a subset construction 
as in the proof of Theorem~\ref{thm:mso}). Thus, arriving at $v$ it can use 
the trip $T_{c,d}$ to decide whether it has to create a copy of $t$ corresponding 
to pebble stack $(u,c)(v,d)$. At the next level it simulates all $\cB_{d,d'}$
for every $d'\in C$, etcetera. 

More formally, $\cN'$ has initial state $\g$, and all other states are of the form 
$(q,c,S_1,\dots,S_\gamma)$ where $q$ is a state of $\cN$, $c\in C$, 
and $S_d$ is a set of states of $\cB_{c,d}$ for every $d\in C=[1,\gamma]$.
We will call them ``extended'' states in what follows. 
To describe the rules of $\cN'$, we 
first recall the rules of the transducer $\cN$ from the proof of Lemma~\ref{lem:nul-decomp}.
Apart from the rules $\langle f,\sigma,j\rangle \to \bot$, $\cN$ has the rules
\[
\begin{array}{llll}
\rho_\g: & \langle \g,\sigma,j\rangle & \to &
\sigma'(\langle \g,\down_1\rangle,\dots,\langle \g,\down_m\rangle, 
      \langle f,\stay\rangle^\gamma) \\[1mm]
\rho_f: & \langle f,\sigma,j\rangle & \to &
\sigma'_{0,j}(\langle \g,\down_1\rangle,\dots,\langle \g,\down_m\rangle, 
      \tup{f,\stay}^\gamma, \xi_j) \\[1mm]
\rho_{f_i}: & \langle f_i,\sigma,j\rangle & \to & \sigma'_{i,j}(
    \langle \g,\down_1\rangle,\dots,\langle \g,\down_{i-1}\rangle, 
    \bot, \\
&&& \quad\langle \g,\down_{i+1}\rangle,\dots,\langle \g,\down_m\rangle, 
    \tup{f,\stay}^\gamma, 
    \xi_j)    
\end{array}
\]
where $\xi_j = \langle f_j,\up\rangle$ for $j\neq 0$, and $\xi_0 = \bot$.

The rules of $\cN'$ for state $\g$ are obtained from rule $\rho_\g$ 
by adding all possible combinations 
of the \abb{mso} head tests $T_c$ and their negations to the left-hand side.
In the right-hand side, the sequence $\tup{f,\stay}^\gamma$ should be replaced by
the sequence $\zeta_1,\dots,\zeta_\gamma$
where $\zeta_c=\tup{(f,c,I_{c,1},\dots,I_{c,\gamma}),\stay}$ if $T_c$ is true,
$I_{c,d}$ being the set of initial states of $\cB_{c,d}$, and 
$\zeta_c=\bot$ if $T_c$ is false.\footnote{More precisely, 
$I_{c,d}$ consists of all initial states of $\cB_{c,d}$, plus all states that $\cB_{c,d}$
can reach from an initial state by applying a relevant rule with a stay-instruction.
}
The rules of $\cN'$ for an ``extended'' state $(q,c,S_1,\dots,S_\gamma)$ are obtained 
from rule $\rho_q$ as follows.
In the left-hand side change $q$ into $(q,c,S_1,\dots,S_\gamma)$.
Moreover, add all \abb{mso} head tests of $\cB_{c,d}$ for every $d\in C$.   
In the right-hand side change every occurrence of a state $q'\neq f$ into the extended state 
$(q',c,S'_1,\dots,S'_\gamma)$ where the set $S'_d$ is obtained from the set $S_d$ by simulating
$\cB_{c,d}$ appropriately, moving down to the \mbox{$\ell$-th} child 
if $q'=\g$ in $\tup{\g,\down_\ell}$ and moving up if $q'=f_j$. Moreover, 
the sequence $\tup{f,\stay}^\gamma$ should be replaced by $\zeta_1,\dots,\zeta_\gamma$
where $\zeta_d=\tup{(f,d,I_{d,1},\dots,I_{d,\gamma}),\stay}$ if $S_d$ contains a final state 
of $\cB_{c,d}$, and $\zeta_d=\bot$ otherwise 
(where $I_{d,d'}$ is defined similarly to $I_{c,d}$ above). 

It should be clear that $\cN'$ produces an output for every input tree $t$ on which $\cM$ 
has finitely many reachable pebble stacks. Thus, $\cN'$ preprocesses $t$ appropriately 
and the deterministic \abb{tt} $\cM'$ in the proof of Lemma~\ref{lem:nul-decomp} 
can simulate $\cM$ on $\tau_{\cN'}(t)$. 
Hence $\tau_{\cM'}(\tau_{\cN'}(t))=\tau_\cM(t)$ for every $t$ in the domain of $\cM$. 
\end{proof}

It is easy to adapt the proof of Theorem~\ref{thm:composition}
to the case where the first (deterministic) \abb{tt} $\cM_1$ uses \abb{mso} head tests; 
those tests can also be executed by the constructed \abb{i-ptt} $\cM$, by Lemma~\ref{lem:sites}. 
Moreover, that proof can also easily be adapted to the case where the second transducer $\cM_2$
is a root-oriented \abb{i-ptt}. From this and Lemmas~\ref{lem:in-ridptt} and~\ref{lem:decompidptt}
we obtain the following characterizations of $\IdPTT$ as a corollary.
We do not know whether similar characterizations hold for $\IPTT$.

\begin{theorem}\label{thm:charidptt}
$\IdPTT = \dTTmso\circ \dTT = \dTTmso\circ \dTTmso = \dTTmso\circ \rIdPTT$.
\end{theorem}

\begin{proof}
Let us show for completeness sake that $\dTT\circ \rIdPTT \subseteq \IdPTT$.
The proof of Theorem~\ref{thm:composition} can easily be generalized to 
a root-oriented \abb{i-ptt} $\cM_2$, because the path from the root of $s$ to the current node $v$ of $\cM_2$
is represented by the pebble stack of the constructed transducer $\cM$, and so the pebbles of $\cM_2$
can also be stored in the pebble stack of $\cM$. For each node on that path, the stack contains a pebble 
with the rule of $\cM_1$ that generates that node, with its child number, and with the pebble that $\cM_2$
drops on that node. 

Formally, the pebble colours of $\cM$ are now triples $(\rho,i,c)$ where $c$ is a pebble colour of $\cM_2$,
and the states of $\cM$ are the states of $\cM_2$ and all 4-tuples $(p,i,c,q)$ where $c$ is again 
a pebble colour of $\cM_2$. The initial state of $\cM$ is now the one of $\cM_2$, and if 
$\cM_2$ has an initial rule $\tup{q_0,\delta,0,\nothing}\to \tup{q,\drop_c}$, then $\cM$ has the rule
$\tup{q_0,\delta,0,\nothing}\to \tup{(p_0,0,c,q),\stay}$. 
The rules of $\cM$ that simulate $\cM_1$ are defined as in the proof of Theorem~\ref{thm:composition},
replacing $i$ by $i,c$ everywhere for each~$c$. 
The rules of $\cM$ that simulate the non-initial rules of $\cM_2$ are defined as follows.
Let $\tup{q,\delta,i,\{c\}}\to \zeta$ be a non-initial rule of~$\cM_2$ and 
let $\rho: \tup{p,\sigma,j}\to \delta(\tup{p_1,\stay},\dots,\tup{p_m,\stay})$ be an output rule of~$\cM_1$. 
Then $\cM$ has the rule $\tup{q,\sigma,j,\{(\rho,i,c)\}}\to \zeta'$ 
where $\zeta'$ is defined as follows. 
If $\zeta=\tup{q',\down_\ell;\drop_d}$, then $\zeta'= \tup{(p_\ell,\ell,d,q'),\stay}$.
If $\zeta=\tup{q',\lift_c;\up}$, then $\zeta'= \tup{q',\lift_{(\rho,i,c)};\totop}$. 
If $\zeta=\tup{q',\lift_c;\drop_d}$, then $\zeta'= \tup{q',\lift_{(\rho,i,c)};\drop_{(\rho,i,d)}}$. 
In the remaining cases, $\zeta'=\zeta$.
\end{proof}

As another corollary we obtain from the three Lemmas~\ref{lem:in-ridptt}, \ref{lem:ridptt-in}, 
and~\ref{lem:decompidptt} that $\IdPTT \subseteq \dTT^3$. Moreover, 
$\IdPTT \subseteq \family{dMT}_\text{\abb{OI}}^2$ by 
the second equality of Corollary~\ref{cor:ridpttmt}. 
Together with Theorem~\ref{thm:tl}, that implies that 
$\dTLlr\subseteq \family{dMT}_\text{\abb{OI}}^2$, which was stated 
as an open problem in~\cite[Section~8]{EngMan03}
(where $\dTLlr$ and $\family{dMT}_\text{\abb{OI}}$ are denoted 
0-DPMTT and DMTT, respectively). 

\begin{corollary}\label{cor:idptt-tt3}
$\IdPTT \subseteq \dTT\!_\downarrow \circ \dTT \circ \dTT \subseteq 
\family{dMT}_\text{\abb{OI}}\circ\family{dMT}_\text{\abb{OI}}$.
\end{corollary}

We are now able to prove the deterministic analogue of Theorem~\ref{thm:decomp} 
for \abb{ptt}'s with at least one visible pebble. 

\begin{theorem}\label{thm:detdecomp}
For every $k\ge 1$,
$\VIdPTT{k} \subseteq \dTT^{k+2}$.
\end{theorem}

\begin{proof} 
Since it follows from Lemma~\ref{lem:decomp} and Corollary~\ref{cor:idptt-tt3} that 
$\VIdPTT{k} \subseteq \tdTT^{k-1}\circ \tdTT\circ \dTT_\downarrow \circ \dTT \circ \dTT$,  
it suffices to show that $\tdTT \circ \dTT_\downarrow \subseteq \dTT$.
For the sake of the proof of Lemma~\ref{lem:tdttmso}, we will show more generally that 
for all deterministic \abb{tt}'s $\cM_1$ and $\cM_2$
such that $\cM_2$ does not use the up-instruction, 
a deterministic \abb{tt} $\cM$ can be constructed such that 
$\tau_\cM(t)=\tau_{\cM_2}(\tau_{\cM_1}(t))$ for every input tree~$t$ in the domain of $\cM_1$. 
This can be proved by a straightforward product construction, which is an easy 
adaptation of the construction in the proof of Theorem~\ref{thm:composition}. 
Since transducer $\cM_2$ never moves up,
there is no need to backtrack on the computation of $\cM_1$.
Therefore, the constructed transducer $\cM$ only considers the topmost pebble.
Since that pebble is always at the position of the head, its colour 
can as well be stored in the finite state of $\cM$. 
Hence $\cM$~can be turned into a \abb{tt} rather than an \abb{i-ptt}. 

Formally, let $\cM_1=(\Sigma,\Delta,P,\{p_0\},R_1)$ and $\cM_2=(\Delta,\Gamma,Q,\{q_0\},R_2)$. 
The deterministic \abb{tt} $\cM$ has input alphabet $\Sigma$ and output alphabet $\Gamma$. 
Its states are of the form $(p,i,q)$ or $(\rho,i,q)$, 
where $p\in P$, $i\in[0,\m_\Delta]$, $q\in Q$, and 
$\rho$ is an output rule of $\cM_1$, i.e., a rule of the form 
$\tup{p,\sigma,j}\to \delta(\tup{p_1,\stay},\dots,\tup{p_m,\stay})$. 
Its initial state is $(p_0,0,q_0)$. 
As in the proof of Theorem~\ref{thm:composition}, 
state $(p,i,q)$ is used by $\cM$ when simulating the computation of $\cM_1$ 
that generates the $i$-th child of the current node of $\cM_2$ 
(keeping the state $q$ of $\cM_2$ in memory).
A state $(\rho,i,q)$ is used by $\cM$ when simulating a computation step of $\cM_2$
on the node that $\cM_1$ has generated with rule $\rho$. 
The rules of $\cM$ are defined as follows. 

First the rules that simulate $\cM_1$. Let $\rho: \tup{p,\sigma,j}\to\zeta$ be a rule in~$R_1$.
If $\zeta=\tup{p',\alpha}$, where $\alpha$ is a move instruction, 
then $\cM$ has the rules $\tup{(p,i,q),\sigma,j}\to \tup{(p',i,q),\alpha}$ 
for every $i\in[0,\m_\Delta]$ and $q\in Q$. 
If $\rho$ is an output rule, then  
$\cM$ has the rules $\tup{(p,i,q),\sigma,j}\to \tup{(\rho,i,q),\stay}$
for every $i$ and $q$ as above.

Second the rules that simulate $\cM_2$.  
Let $\tup{q,\delta,i}\to \zeta$ be a rule in~$R_2$ and 
let $\rho: \tup{p,\sigma,j}\to \delta(\tup{p_1,\stay},\dots,\tup{p_m,\stay})$
be an output rule in~$R_1$ (with the same $\delta$). 
Then $\cM$ has the rule $\tup{(\rho,i,q),\sigma,j}\to \zeta'$ 
where $\zeta'$ is obtained from~$\zeta$ by changing 
every $\tup{q',\stay}$ into $\tup{(\rho,i,q'),\stay}$,
and every $\tup{q',\down_\ell}$ into $\tup{(p_\ell,\ell,q'),\stay}$. 
\end{proof}

Since the topmost pebble of a \abb{v-ptt} can be replaced by 
an invisible pebble, we obtain from Theorem~\ref{thm:detdecomp} that 
$\VdPTT{k} \subseteq \dTT^{k+1}$, which was proved in~\cite[Theorem~10]{EngMan03}.

Theorem~\ref{thm:detdecomp} allows us to show that, in the deterministic case,
$k+1$ visible pebbles are more powerful than $k$ visible pebbles.

\begin{theorem}\label{thm:dethier}
For every $k\ge 0$,
$\VIdPTT{k} \subsetneq \VIdPTT{k+1}$.
\end{theorem}

\begin{proof}
It follows from Theorem~\ref{thm:detdecomp} and Corollary~\ref{cor:idptt-tt3}
(and the inclusion $\dTT\subseteq \family{dMT}_\text{\abb{OI}}$ 
in Corollary~\ref{cor:ridpttmt}) 
that $\VIdPTT{k} \subseteq \family{dMT}_\text{\abb{OI}}^{k+2}$ for every $k\geq 0$. 
But it is proved in~\cite[Theorem~41]{EngMan03} that, for every $k\geq 1$, 
$\VdPTT{k}$ is not included in $\family{dMT}_\text{\abb{OI}}^k$. Hence, 
since the topmost pebble of a \abb{v-ptt} can be replaced by an invisible pebble,
$\VIdPTT{k}$ is not included in $\family{dMT}_\text{\abb{OI}}^{k+1}$.
\end{proof}

The above proof also shows that Theorem~\ref{thm:detdecomp} is optimal, in the sense that, 
for every $k\geq 1$, $\VIdPTT{k}$ is not included in $\dTT^{k+1}$. 

Another consequence of Theorem~\ref{thm:detdecomp} is that, by the results of~\cite{man},
all total deterministic \abb{vi-pft} transformations for which the size of the
output document is linear in the size of the input document,
can be programmed in $\TL$. Let $\family{LSI}$ be the class of all total functions $\phi$
for which there exists a constant $c\in\nat$ such that $|\phi(t)|\leq c\cdot|t|$ 
for every input tree $t$. 

\begin{theorem}\label{thm:lsi}
For every $k\geq 0$,
\[ 
\VIdPTT{k} \cap \family{LSI} \subseteq \IdPTT = \dTL_\mathrm{r} \quad\text{and}\quad
\VIdPFT{k} \cap \family{LSI} \subseteq \IdPFT = \dTL.
\]
\end{theorem}

\begin{proof}
It is shown in~\cite{man} that 
$\family{dMT}_\text{\abb{OI}}^k \cap \family{LSI} \subseteq \family{dMT}_\text{\abb{OI}}$
for every $k\geq 1$. 
By Theorem~\ref{thm:detdecomp} and Corollary~\ref{cor:ridpttmt},
$\VIdPTT{k} \subseteq \family{dMT}_\text{\abb{OI}}^{k+2}$. 
And by Corollary~\ref{cor:mtiptt} and Theorem~\ref{thm:tl}, 
$\family{dMT}_\text{\abb{OI}} \subseteq \IdPTT = \dTL_\mathrm{r}$. This proves the first inclusion.
To prove the second inclusion, let $\phi\in\VIdPFT{k} \cap \family{LSI}$. 
Obviously, $\phi\circ\enc$ is also in $\family{LSI}$, and 
$\phi\circ\enc\in \VIdPTT{k}\circ\IdPTT$ by Lemma~\ref{lem:pftvsptt}(2). 
Hence $\phi\circ\enc \in \family{dMT}_\text{\abb{OI}}^{k+4}\subseteq\IdPTT$, as above. 
In other words, $\phi\in \IdPTT\circ\dec$. Consequently, by Lemma~\ref{lem:pftvsptt}(1)
and Theorem~\ref{thm:tl}, $\phi\in\IdPFT = \dTL$. 
\end{proof}

In fact, $\VIdPTT{k} \cap \family{LSI}$ is the class 
of total functions in the class
$\family{DMSOT}$ of deterministic 
\abb{mso} definable tree transductions
discussed after Corollary~\ref{cor:mtiptt}, and similarly, 
$\VIdPFT{k} \cap \family{LSI}$ is the class of
total functions in the class of
deterministic \abb{mso} definable
tree-to-forest transductions (which equals $\family{DMSOT}\circ\family{dec}$,
because both $\dec$ and $\enc$ are \abb{mso} definable). 

For the reader familar with results about attribute grammars 
(which are a well-known compiler construction tool) and related formalisms, 
we now briefly discuss the relationship between those results and some of the above.
As explained in detail in~\cite[Section~3.2]{EngMan03}, 
the total deterministic tree-walking tree transducer, i.e., the td\abb{TT}, is essentially 
a notational variant of the attributed tree transducer (\abb{AT}) of~\cite{Ful,FulVog}, 
except that the \abb{at} is in addition required to be ``noncircular'', 
which means that no configuration can generate an output form 
in which that same configuration occurs. 
As observed at the end of Section~\ref{sec:document}, the deterministic \abb{i-ptt} has the 
same expressive power as the deterministic \abb{tl} program that is local and ranked,
which corresponds to the macro attributed tree transducer (\abb{MAT}) of~\cite{KuhVog,FulVog} 
in the same way, i.e., the \abb{MAT} is the ``noncircular'' td\abb{tl}$_{\ell {\sf r}}$ program.
Since $\rIdPTT = \family{dMT}_\text{\abb{OI}}$ by Corollary~\ref{cor:ridpttmt},
Lemma~\ref{lem:in-ridptt} ($\dTTmso \subseteq \rIdPTT$) is closely related to 
the well-known fact that \abb{at} (with look-ahead) can be simulated 
by deterministic macro tree transducers.
Lemma~\ref{lem:ridptt-in} ($\rIdPTT \subseteq \dTT_\downarrow \circ \dTT$) is related 
to the fact that every total deterministic macro tree transducer can be decomposed into 
a deterministic top-down tree transducer followed by a YIELD mapping, which can be realized
by an \abb{at}.  
Theorem~\ref{thm:charidptt} ($\IdPTT = \dTTmso\circ \dTT = \dTTmso\circ \rIdPTT$) 
is closely related to 
the fact that every \abb{mat} can be decomposed into two \abb{at}'s, 
and that the composition of an \abb{at} and a total deterministic macro tree transducer 
can be simulated by a \abb{mat},
as shown in~\cite[Theorem~4.8]{KuhVog} and its proof (see also~\cite[Corollary~7.30]{FulVog}). 
The inclusion $\tdTT \circ \dTT_\downarrow \subseteq \dTT$ in the proof of 
Theorem~\ref{thm:detdecomp} is closely related to the closure of \abb{AT} under right-composition 
with deterministic top-down tree transducers, as shown in~\cite[Theorem~4.3]{Ful} 
(see also~\cite[Lemma~4.11]{KuhVog} and~\cite[Lemma~5.46]{FulVog}).
We finally mention that the class $\family{DMSOT}$ of deterministic \abb{mso} 
definable tree transductions is properly included in $\dTTmso$ 
(see~\cite[Theorems~8.6 and~8.7]{thebook}),
as shown for attribute grammars (with look-ahead) in~\cite{EngBlo}.

\smallpar{Strong pebbles}
In the litterature there are pebble automata with weak and strong pebbles.
Weak pebbles (which are the pebbles considered until now) 
can only be lifted when the reading head is at the position where they were dropped,
whereas strong pebbles can also be lifted from a distance, i.e., 
when the reading head is at any other position.
So, strong pebbles are more like dogs that can be whistled back, or like pointers 
that can be erased from memory. Formally, we define a pebble colour $c$ to be \emph{strong}
as follows. For a rule $\tup{q,\sigma,j,b}\to\tup{q',\lift_c}$
we do not require any more that $c\in b$. If the rule is 
relevant to configuration $\tup{q,u,\pi}$, then it is applicable whenever the topmost
element of the pebble stack is $(v,c)$ for some node $v$ (not necessarily equal to~$u$).
That top pebble is then popped from the stack, i.e., 
$\pi=\pi'(v,c)$ where $\pi'$ is the new stack. 
Strong pebbles were investigated, e.g., in~\cite{fo+tc,MusSamSeg,expressive,SamSeg,FulMuzFI}. 

It turns out that strong invisible pebbles are too strong, in the sense that 
they allow the recognition of nonregular tree languages, 
cf. the paragraph after Theorem~\ref{thm:regt}.  
For example, the nonregular language $\{a^n\#b^n\mid n\in\nat\}$
can be accepted by an \abb{i-pta} with strong pebbles as follows. 
After checking that the input string $w$ is in $a^*\#b^*$, 
the automaton drops a pebble on $\#$ and walks to the left, 
dropping a pebble on every $a$. Next it walks to the end of $w$,
and then walks to the left, lifting a pebble (from a distance) 
for every $b$ it passes. It accepts $w$ if it arrives at $\#$ and 
observes a pebble on $\#$. 

Thus, we will only consider the \abb{pta} and \abb{ptt} with strong \emph{visible} pebbles, 
abbreviated as \abb{v$^+$i-pta} and \abb{v$^+$i-ptt} (and similarly for the classes 
of transductions they realize). Obviously, $\VIPTT{k} \subseteq \VsIPTT{k}$ for every $k\geq 0$.
We do not know whether the inclusion is proper. 

Let us first show that the \abb{v$^+$i-pta} and \abb{v$^+$i-ptt} can perform stack tests. 

\begin{lemma}\label{lem:stacktestsplus}
Let $k\geq 0$. For every \abb{v$^+_k$i-pta} with stack tests $\cA$ 
an equivalent (ordinary) \abb{v$^+_k$i-pta} $\cA'$ 
can be constructed in polynomial time. 
The construction preserves determinism and
the absence of invisible pebbles.
The same holds for the corresponding \abb{ptt}'s.  
\end{lemma}

\begin{proof}
Let $\cA = (\Sigma, Q, Q_0, F, C, C_\mathrm{v}, C_\mathrm{i}, R,k)$.
We construct $\cA'$ in the same way as in the proof of Lemma~\ref{lem:stacktests},
except that it additionally keeps track of the visible pebbles in its \emph{own} stack, 
in the order in which they were dropped, 
cf. the construction of a counting \abb{pta} after Lemma~\ref{lem:stacktests}. 
Thus, its states are of the form 
$(q,\gamma,\phi)$ where $q\in Q$, $\gamma\in C\cup\{\epsilon\}$, and 
$\phi\in (C'_\mathrm{v})^*=(C_\mathrm{v}\times (C\cup\{\epsilon\}))^*$
is a string without repetitions of length $\leq k$. 
Its initial states are $(q,\epsilon,\epsilon)$ with $q\in Q_0$.

The rules of $\cA'$ are defined as follows. 
Let $\tup{q,\sigma,j,b,\gamma} \to \tup{q',\alpha}$ be a rule of~$\cA$,
let $\phi$ be a string over $C'_\mathrm{v}$ as above, 
and let $b'$ be (the graph of) a mapping from $b$ to $C\cup\{\epsilon\}$.
If $\alpha$ is a move instruction, then $\cA'$ has the rule 
$\tup{(q,\gamma,\phi),\sigma,j,b'} \to \tup{(q',\gamma,\phi),\alpha}$ 
(and similarly for an output rule of a \abb{ptt}).
If $\alpha = \drop_c$, then $\cA'$ has the rule
$\tup{(q,\gamma,\phi),\sigma,j,b'} \to \tup{(q',c,\phi'),\drop_{(c,\gamma)}}$ 
where $\phi'=\phi$ if $c\in C_\mathrm{i}$ and $\phi'=\phi(c,\gamma)$ 
otherwise (provided $|\phi|< k$ and $(c,\gamma)$ does not occur in $\phi$). 
Now let $\alpha=\lift_c$ and $\gamma=c$.
If $c \in C_\mathrm{i}$ and $(c,\gamma')\in b'$,
then $\cA'$~has the rule  
$\tup{(q,\gamma,\phi),\sigma,j,b'} \to \tup{(q',\gamma',\phi),\lift_{(c,\gamma')}}$.
If $c \in C_\mathrm{v}$, then $\cA'$~has the rule  
$\tup{(q,\gamma,\phi(c,\gamma')),\sigma,j,b'} \to \tup{(q',\gamma',\phi),\lift_{(c,\gamma')}}$
for every $\gamma'\in C\cup\{\epsilon\}$ such that $(c,\gamma')$ does not occur in $\phi$
(with $|\phi|< k$).
\end{proof}

Using this lemma, we now show that every \abb{v$^+$-ptt} can be decomposed into \abb{tt}'s,
as already shown in~\cite{FulMuzFI} in a different way.\footnote{In that paper the authors 
``think that those proofs cannot be generalized for the strong pebble case 
because the mapping EncPeb [$\cdots$] is strongly based on weak pebble handling'', where `those proofs' 
mainly refers to the proof of~\cite[Lemma~9]{EngMan03} in which the preprocessor is called EncPeb, 
see Lemma~\ref{lem:decomp}.
}
 
\begin{lemma}\label{lem:decompplus}
For every $k\ge 1$, 
$\VsPTT{k} \subseteq \TT \circ \VsPTT{k-1}$.
For fixed $k$, the construction takes polynomial time. 
\end{lemma}

\begin{proof}
Let $\cM = (\Sigma, \Delta, Q, Q_0,C, C_\mathrm{v}, C_\mathrm{i}, R,k)$ 
be a \abb{v$^+_k$-ptt} with $C_\mathrm{i}=\nothing$.
The construction is similar to the one in the proof of Lemma~\ref{lem:decomp},
except that we use the nondeterministic ``multi-level'' preprocessor $\cN$ of the proof
of Lemma~\ref{lem:nul-decomp}, for which we assume that  $C_\mathrm{v}=[1,\gamma]$.

By Lemma~\ref{lem:stacktestsplus} we may assume that the simulating transducer $\cM'$
can perform stack tests. 
As in the proof of Lemma~\ref{lem:decomp},
$\cM'$ starts by simulating $\cM$ on the top level of 
the preprocessed version $t'$ of the input tree $t$. 
When $\cM$ drops the first pebble $c$ on node $u$,
$\cM'$ enters the second level copy $\hat{t}_u$ of $t$ corresponding to $c$, 
but it also stores $c$ in its finite state.
When $\cM$ wants to lift pebble $c$ and can actually do so because 
the pebble stack of $\cM'$ is empty, $\cM'$ removes $c$ from its finite state
and continues simulating $\cM$ on $\hat{t}_u$. Note that since $c$ can be lifted from a distance,
$\cM'$ cannot return to the top level without loosing its current position. 
When $\cM$ again drops a pebble $d$ on some second-level node that corresponds 
to a node $v$ of $t$,
$\cM'$ enters the third level copy $\hat{t}_v$ corresponding to~$d$, 
and stores $d$ in its finite state. Thus, whenever $\cM$ drops a bottom pebble,
$\cM'$~moves one level down in the ``tree of trees'' $t'$. 

It should be noted that we could as well have taken $\gamma=1$ for $\cN$
and let $\cM'$ enter the unique copy of $t$ when it drops a pebble $c$,
because $\cM'$ keeps~$c$ in its finite state. However, the present construction 
simplifies the proof of Theorem~\ref{thm:detdecompplus}. 

Although the above description should be clear, let us give the formal details.
As in the proof of Lemma~\ref{lem:nul-decomp},
the output alphabet $\Gamma$ of $\cN$ is the union of $\{\bot\}$, 
$\{\sigma'\mid \sigma\in\Sigma\}$, and 
$\{\sigma'_{i,j}\mid \sigma\in\Sigma, i\in[0,\rank_\Sigma(\sigma)], j\in[0,\m_\Sigma]\}$ 
where, for every $\sigma\in\Sigma$ of rank~$m$, $\sigma'$ has rank $m+\gamma$ and 
$\sigma'_{i,j}$ has rank $m+\gamma+1$.
As in the proof of Lemma~\ref{lem:decomp},
the \abb{v$^+_{k-1}$-ptt} $\cM'$ has input alphabet~$\Gamma$, 
set of states $Q\cup (Q\times C_{\mathrm{v}})$, 
and the same initial states and pebble colours as $\cM$. 
The rules of~$\cM'$ are defined as follows.
Let $\langle q,\sigma,j,b\rangle\to \zeta$ be a rule of $\cM$ 
with $\rank_\Sigma(\sigma)=m$. 

First, we consider the behaviour of $\cM'$ in state $q$, where we assume that $b=\nothing$.
Then $\cM'$ has the rules $\langle q,\sigma',j,\nothing\rangle\to \zeta_1$,
$\langle q,\sigma'_{0,j},j',\nothing\rangle\to \zeta_2$, and
$\langle q,\sigma'_{i,j},j',\nothing\rangle\to \zeta_{3,i}$ 
for every $i\in[1,m]$ and $j'\in[1,\m_\Gamma]$,
where $\zeta_1$ is obtained from $\zeta$ by changing 
$\tup{q',\drop_c}$ into $\tup{(q',c),\down_{m+c}}$ for every $q'\in Q$ and $c\in C_{\mathrm{v}}$,
$\zeta_2$ is obtained from $\zeta_1$ by changing $\up$ into $\down_{m+\gamma+1}$, and 
$\zeta_{3,i}$ is obtained from $\zeta_2$ by changing $\down_i$ into $\up$.
Thus, whenever the pebble stack of $\cM$ is empty, $\cM'$ simulates $\cM$ 
on a copy of the input tree $t$ in $t'$, 
until $\cM$ drops a pebble~$c\in C_{\mathrm{v}}$. Then $\cM'$ steps to the next level, 
and stores $c$ in its finite state. 

Second, we consider the behaviour of $\cM'$ in state $(q,c)$, where $c\in C_{\mathrm{v}}$.
Rules of $\cM'$ that have $\sigma'_{0,j}$ in their left-hand side are defined 
under the proviso that \mbox{$c\in b$}, and the other rules under the proviso that $c\notin b$.
If $\zeta=\tup{q',\lift_c}$, then $\cM'$ has the rules 
$\langle (q,c),\sigma',j,b,\epsilon\rangle\to \tup{q',\stay}$,
$\langle (q,c),\sigma'_{0,j},m+c,b\setminus\{c\},\epsilon\rangle\to \tup{q',\stay}$, and
$\langle (q,c),\sigma'_{i,j},j',b,\epsilon\rangle\to \tup{q',\stay}$ 
for every $i\in[1,m]$ and $j'\in[1,\m_\Gamma]$,
where $\epsilon$ is the stack test that checks emptiness of the stack of $\cM'$.
Thus, when $\cM$ lifts pebble~$c$ (at the position of $c$ or from a distance), 
$\cM'$~removes~$c$ from memory and knows that the pebble stack of $\cM$ is empty. 
Otherwise,  $\cM'$ has  the rules 
$\langle (q,c),\sigma',j,b\rangle\to \zeta_{c,1}$,
$\langle (q,c),\sigma'_{0,j},m+c,b\setminus\{c\}\rangle\to \zeta_{c,2}$, and
$\langle (q,c),\sigma'_{i,j},j',b\rangle\to \zeta_{c,3,i}$ 
for every $i\in[1,m]$ and $j'\in[1,\m_\Gamma]$,
where $\zeta_{c,1}$ is obtained from $\zeta$ by changing 
every occurrence of a state $q'$ 
into $(q',c)$,
$\zeta_{c,2}$ is obtained from $\zeta_{c,1}$ by changing $\up$ into $\down_{m+\gamma+1}$, and 
$\zeta_{c,3,i}$ is obtained from $\zeta_{c,2}$ by changing $\down_i$ into $\up$.
Thus, $\cM'$ simulates $\cM$ on a copy of the input tree in $t'$, 
assuming that $c$ is present on the node with label $\sigma'_{0,j}$ 
and absent on the other nodes, until $c$ is lifted by $\cM$. 
\end{proof}

The next result is an immediate consequence of Lemma~\ref{lem:decompplus}.
It was proved in~\cite[Theorem~6.5(5)]{FulMuzFI}, generalizing the same result
for weak pebbles in~\cite[Theorem~10]{EngMan03} 
(cf. Theorem~\ref{thm:detdecomp}).
It implies that Propositions~\ref{prop:invtypeinf}(2) and~\ref{prop:typecheck}(2) 
also hold for strong pebbles. Thus, for \abb{ptt}'s without invisible pebbles, 
the inverse type inference problem and the typechecking problem
are solvable for strong pebbles in the same time as for weak pebbles
(cf.~\cite[Theorem~6.7 and~6.9]{FulMuzFI}).
Note that it also implies that the domains of \abb{v$^+$-ptt}'s are regular
(cf.~Corollary~\ref{cor:domptt}), 
which was proved in~\cite[Corollary~6.8]{FulMuzFI} and~\cite[Theorem~4.7]{Muz}.  

\begin{theorem}\label{thm:decompplus}
For every $k\geq 0$, $\VsPTT{k}\subseteq \TT^{k+1}$.  
For fixed $k$, the construction takes polynomial time.
\end{theorem}

To prove a similar result for deterministic \abb{ptt}'s with strong pebbles, 
we need the next small lemma. 

\begin{lemma}\label{lem:tdttmso}
For every $k\geq 1$, $(\tdTTmso)^k \subseteq \dTT\!_\downarrow\circ \dTT^k$.
\end{lemma}

\begin{proof}
We will show by induction on $k$ that for every $\tau\in (\tdTTmso)^k$ 
there exist $\tau_0\in \dTT\!_\downarrow$ and $\tau_1,\dots,\tau_k\in\dTT$ such that 
$\tau=\tau_0\circ\tau_1\circ\cdots\circ\tau_k$. Note that since $\tau$ is a total function, 
every output tree of $\tau_0\circ\tau_1\circ\cdots\circ\tau_{i-1}$ 
is in the domain of~$\tau_i$, for every $i\in[1,k]$. 
For $k=1$ the statement is immediate from the inclusion 
$\dTTmso\subseteq \dTT_\downarrow \circ \dTT$, which follows from
Lemmas~\ref{lem:in-ridptt} and~\ref{lem:ridptt-in}. 
Now consider $\tau\in (\tdTTmso)^{k+1}$. By induction and the case $k=1$,
$\tau=\tau_0\circ\tau_1\circ\cdots\circ\tau_k\circ \tau'_0\circ \tau'_1$
with $\tau_0,\tau'_0\in\dTT\!_\downarrow$ and $\tau_1,\dots,\tau_k,\tau'_1\in\dTT$.
Since every output tree of $\tau_0\circ\tau_1\circ\cdots\circ\tau_{k-1}$ is in the domain
of $\tau_k$, we can replace $\tau_k\circ \tau'_0$ by any transduction $\tau'$ 
such that $\tau'(t)=\tau'_0(\tau_k(t))$ for every $t$ in the domain of $\tau_k$.
Since $\tau_k\in\dTT$ and $\tau'_0\in\dTT\!_\downarrow$, we can take $\tau'\in\dTT$
by the proof of Theorem~\ref{thm:detdecomp}.
\end{proof}

Theorem~\ref{thm:decompplus} was also shown in~\cite[Theorem~10]{EngMan03} for weak pebbles 
in the deterministic case. Here we need one more deterministic \abb{tt}.

\begin{theorem}\label{thm:detdecompplus}
For every $k\geq 1$, 
$\VsdPTT{k}\subseteq \dTT\!_\downarrow\circ \dTT^{k+1}$.  
\end{theorem}

\begin{proof}
By Lemma~\ref{lem:tdttmso} it suffices to show that
$\VsdPTT{k} \subseteq \tdTTmso \,\circ \,\VsdPTT{k-1}$ for every \mbox{$k\ge 1$}.
The proof of this inclusion is obtained from the proof of Lemma~\ref{lem:decompplus}
by changing the preprocessor $\cN$ in a similar way as in 
the proof of Lemma~\ref{lem:decompidptt}.

For the given deterministic \abb{v$^+_k$-ptt} $\cM$ we assume that 
$C_\mathrm{i}=\nothing$ and $C=C_\mathrm{v}=[1,\gamma]$.
As in the proof of Lemma~\ref{lem:decompidptt}, 
we may assume that there is a mapping $\chi:C\to Q$ 
that specifies the state of $\cM$ 
after dropping a pebble (because we may also assume that $\cM$ keeps track 
in its finite state of the pebbles in its stack, in the order in which they were dropped,
cf. the proof of Lemma~\ref{lem:stacktestsplus}).

The new preprocessor $\cN'$ is constructed in the same way as in the proof of 
Lemma~\ref{lem:decompidptt}, with a different definition of the trips $T_{c,d}$.
For $c\in C$, the site $T_c$ is defined as in that proof, i.e., it consists of 
all pairs $(t,u)$ such that the configuration $\tup{\chi(c),u,(u,c)}$ is reachable
by the automaton $\cA$ associated with~$\cM$, which now is a 
nondeterministic \abb{v$^+_k$-pta}. The automaton $\cB$ recognizing $\tmark(T_c)$
is a \abb{v$^+_k$-pta} with stack tests (see Lemma~\ref{lem:stacktestsplus}). 
When it arrives at the marked node $u$
in state $\chi(c)$ and observes $c$, it may check that $c$ is the top pebble, lift it,
check that the stack is now empty, and accept. 
For $c,d\in C$, the trip $T_{c,d}$ now consists of all triples 
$(t,u,v)$ such that $\cA$ has a computation on $t$ 
starting in configuration $\tup{\chi(c),u,(u,c)}$
and ending in configuration $\tup{\chi(d),v,(v,d)}$,
with at least one computation step.
It should be clear that there is a \abb{v$^+_k$-pta} $\cB'$ with stack tests
that recognizes $\tmark(T_{c,d})$: it starts by dropping $c$ on marked node~$u$ 
in state $\chi(c)$, and then behaves similarly to $\cB$ (with respect to $v$ and $d$). 

For every input tree $t$ in the domain of $\cM$, the preprocessor $\cN'$ produces 
the appropriate output. In fact, if $\cN'$ would not produce output, 
then there would be an infinite sequence $(u_1,c_1),(u_2,c_2),\dots$
such that $(t,u_1)\in T_{c_1}$ and $(t,u_i,u_{i+1})\in T_{c_i,c_{i+1}}$ for every $i\geq 1$.
But that would imply the existence of an infinite computation of $\cM$ on $t$ that 
starts in the initial configuration, contradicting the determinism of $\cM$.
\end{proof}

Next, we decompose arbitrary \abb{v$^+$i-ptt}'s. To do that we need two \abb{tt}'s 
at each decomposition step rather than one. 

\begin{lemma}\label{lem:decompplusI}
For every $k\geq 1$, $\VsIPTT{k}\subseteq \TT \circ \TT \circ \VsIPTT{k-1}$.
For fixed $k$, the construction takes polynomial time.
\end{lemma}

\begin{proof}
The proof of Lemma~\ref{lem:decompplus} is also valid for \abb{v$^+$i-ptt},
provided every reachable pebble stack of the given transducer  
has a visible bottom pebble (for the definition of reachable pebble stack 
see the proof of Lemma~\ref{lem:decompidptt}).
Thus, it suffices to construct for every \abb{v$^+_k$i-ptt} $\cM$ 
a \abb{tt} $\cN$ and a \abb{v$^+_k$i-ptt} $\cM'$ with that property, such that 
$\tau_\cN\circ\tau_{\cM'}=\tau_\cM$.

Let $\cM = (\Sigma, \Delta, Q, Q_0,C, C_\mathrm{v}, C_\mathrm{i}, R,k)$. 
The construction is similar to the one in the proof of Lemma~\ref{lem:nul-decomp}.
In particular, we assume that $C_\mathrm{i}=[1,\gamma]$ and 
we use the same nondeterministic ``multi-level'' preprocessor $\cN$ of that proof.
The simulating transducer $\cM'$ works in the same way as the one in the proof 
of Lemma~\ref{lem:nul-decomp} as long as the pebble stack of $\cM$ 
consists of invisible pebbles only. Thus, during that time the pebble stack of $\cM'$ is empty. 
As soon as $\cM$ drops a visible pebble $c$, $\cM'$ stays in the same copy of the input tree
and also drops~$c$. After that, $\cM'$ just simulates $\cM$ on that copy 
until $\cM$ lifts $c$. Then $\cM'$ also lifts~$c$ and returns to the first mode until $\cM$ 
again drops a visible pebble. Note that when $\cM$ drops $c$, 
all invisible pebbles on the input tree become unobservable until $c$ is lifted. 

Formally, the set of states of $\cM'$ is the union of $Q$ (used in the first mode) and
$Q\times C_\mathrm{v}$ (used in the second mode). The rules for the first mode are 
the same as in the proof of Lemma~\ref{lem:nul-decomp}, 
with the empty set of pebble colours added to each left-hand side. 
Now let $\tup{q,\sigma,j,b}\to \zeta$ be a rule of $\cM$ and $\rank_\Sigma(\sigma)=m$.
In what follows, $i$ ranges over $[1,m]$ and $j'$ over $[1,\m_\Gamma]$, as usual.
With the following rules $\cM'$ switches from the first to the second mode,
where we assume that $\zeta = \tup{q',\drop_c}$ with $c\in C_\mathrm{v}$: 
if $b=\{d\}$ with $d\in C_\mathrm{i}$, then it uses the rule
$\tup{q,\sigma_{0,j},m+d,\nothing}\to \tup{(q',c),\drop_c}$, and 
if $b=\nothing$, then it uses the rules
$\tup{q,\sigma',j,\nothing}\to \tup{(q',c),\drop_c}$ and 
$\tup{q,\sigma_{i,j},j',\nothing}\to \tup{(q',c),\drop_c}$.
The rules for the second mode are as follows, for every $c\in C_\mathrm{v}$. 
We first assume that $\zeta$ does not contain the instruction $\lift_c$.
Then $\cM'$ has the rules $\tup{(q,c),\sigma',j,b}\to \zeta_1$,
$\tup{(q,c),\sigma_{0,j},j',b}\to \zeta_2$, and 
$\tup{(q,c),\sigma_{i,j},j',b}\to \zeta_{3,i}$,
where $\zeta_1$ is obtained from $\zeta$ by changing every state $q'$ into $(q',c)$, 
$\zeta_2$ is obtained from $\zeta_1$ by changing $\up$ into $\down_{m+\gamma+1}$,
and $\zeta_{3,i}$ is obtained from $\zeta_2$ by changing $\down_i$ into $\up$.
Finally, if $\zeta=\tup{q',\lift_c}$, then 
$\cM'$ switches from the second to the first mode with the following rules:
$\tup{(q,c),\sigma',j,b}\to \zeta$,   
$\tup{(q,c),\sigma_{0,j},j',b}\to \zeta$, and
$\tup{(q,c),\sigma_{i,j},j',b}\to \zeta$. 
\end{proof}

The next result is immediate from Lemmas~\ref{lem:decompplusI} and~\ref{lem:nul-decomp}.
It implies, by Propositions~\ref{prop:invtypeinf}(1) and~\ref{prop:typecheck}(1), that
the inverse type inference problem and the typechecking problem
are solvable for \abb{ptt}'s with $k$ strong visible pebbles, 
in $(2k+2)$-fold and $(2k+3)$-fold exponential time, respectively.
It also implies that the domains of \abb{v$^+$i-ptt}'s are regular,
cf. Corollary~\ref{cor:domptt}. 

\begin{theorem}\label{thm:decompplusI}
For every $k\geq 0$, $\VsIPTT{k}\subseteq \TT^{2k+2}$.  
For fixed $k$, the construction takes polynomial time.
\end{theorem}

Applying the techniques in the proofs of Lemma~\ref{lem:decompidptt} 
and Theorem~\ref{thm:detdecompplus} to the proof of 
Lemma~\ref{lem:decompplusI}, and using Lemmas~\ref{lem:decompidptt} and~\ref{lem:tdttmso}, 
we obtain that every deterministic \abb{v$^+$i-ptt} 
can be decomposed into deterministic \abb{tt}'s, cf. Theorem~\ref{thm:detdecomp}. 
The formal proof is straightforward. 

\begin{theorem}\label{thm:detdecompplusI}
For every $k\geq 0$, $\VsIdPTT{k}\subseteq \dTT\!_\downarrow\circ \dTT^{2k+2}$.
\end{theorem}

We do not know whether these results are optimal, 
i.e., whether the exponent $2k+2$ can be lowered.

\section{Conclusion}\label{sec:conc}

We have shown in Theorem~\ref{thm:decomp} that
$\VIPTT{k} \subseteq \TT^{k+2}$, but we do not know whether this is optimal,
i.e., whether or not $\VIPTT{k} \subseteq \TT^{k+1}$. 
Since the results on typechecking in Section~\ref{sec:typechecking} 
are based on this decomposition, we also do not know whether the time bound 
for typechecking \abb{v$_k$i-ptt}'s, as stated in Theorem~\ref{thm:typecheck}, is optimal.
Using the results of~\cite{SamSeg}, it can be shown that the time bound 
for inverse type inference is optimal, cf. the discussion after~\cite[Corollary~1]{Eng09}. 

We have shown in Theorem~\ref{thm:matchall} that all \abb{mso} definable $n$-ary patterns can be 
matched by deterministic \abb{v$_{n-2}$i-ptt}'s, but we do not know whether this is optimal,
i.e., whether or not it can be done with less than $n-2$ pebbles. In particular, 
we do not know whether or not all \abb{mso} definable ternary patterns can be matched by \abb{i-ptt}'s
(or, by \abb{tl} programs), cf. Theorem~\ref{thm:lsi}. 
In Section~\ref{sec:pattern} we have suggested ways of reducing 
the number of visible pebbles in special cases. 
Given an \abb{mso} formula $\phi$, can one compute the minimal number of visible pebbles
that is needed to match the pattern $\phi$? 

The language \abb{tl} can be extended with visible pebbles, in an obvious way.
The resulting ``pebble \abb{tl} programs'' are closely related to the \emph{pebble macro tree transducers} 
that were introduced in~\cite{EngMan03}. 
What is the relationship between the $k$-pebble macro tree transducer
and the \abb{v$_k$i-ptt}? Is there an analogon of Theorem~\ref{thm:tl}? It is not clear 
whether the proof of Theorem~\ref{thm:tl} can be generalized to the addition of visible pebbles. 

We have shown in Theorem~\ref{thm:dethier} that
$\VIdPTT{k} \subsetneq \VIdPTT{k+1}$, i.e., that 
$k+1$ visible pebbles are more powerful than $k$, in the deterministic case. 
We do not know whether this holds for the nondeterministic transducers,
i.e., whether or not the inclusion $\VIPTT{k} \subseteq \VIPTT{k+1}$ is proper. 
We also do not know whether every functional \abb{v$_k$i-ptt} can be simulated 
by a deterministic one, where a \abb{ptt}~$\cM$ is functional 
if $\tau_\cM$ is a function. If so, then the inclusion would of course be proper. 

Is it decidable for a given deterministic \abb{v$_{k+1}$i-ptt} $\cM$ whether or not 
$\tau_\cM$ is in $\VIdPTT{k}$? If so, then one could compute the minimal number of visible pebbles
needed to realize the transformation $\tau_\cM$ by a \abb{ptt}. Obviously, 
that would answer the above question for the pattern $\phi$ in the affirmative. 

It is proved in~\cite{expressive} that the \abb{v$^+_k$-pta}
has the same expressive power as the \abb{v$_k$-pta}, i.e., 
that strong pebbles are not more powerful than weak pebbles.
We do not know whether or not the \abb{v$^+_k$-ptt} is more powerful than the \abb{v$_k$-ptt},
and neither whether or not the \abb{v$^+_k$i-ptt} is more powerful than the \abb{v$_k$i-ptt}.


\end{document}